\documentclass[11pt]{article}
\usepackage[T1]{fontenc}
\usepackage{lmodern}
\usepackage{microtype}
\usepackage{amsfonts,amsmath,amsthm,amssymb,mathtools}
\numberwithin{equation}{section}
\usepackage{mathrsfs}
\usepackage[margin=1in]{geometry}
\usepackage{enumitem}
\usepackage[linesnumbered,boxed,ruled,vlined]{algorithm2e}
\usepackage[colorlinks,citecolor=blue,linkcolor=blue,urlcolor=red,pagebackref]{hyperref}
\usepackage[capitalise]{cleveref}
\usepackage{url}
\usepackage{thmtools, thm-restate}
\usepackage{xcolor}
\usepackage{graphicx}
\usepackage{placeins}
\usepackage{xspace}
\xspaceaddexceptions{]\}}
\usepackage{framed}
\usepackage{tikz}

\theoremstyle{plain}
\newtheorem{theorem}{Theorem}[section]
\newtheorem{lemma}[theorem]{Lemma}
\newtheorem{proposition}[theorem]{Proposition}
\newtheorem{corollary}[theorem]{Corollary}
\newtheorem{claim}[theorem]{Claim}

\theoremstyle{definition}
\newtheorem{definition}[theorem]{Definition}

\Crefname{claim}{Claim}{Claims}
\Crefname{example}{Example}{Examples}

\setlist[enumerate]{nosep, topsep=1ex}
\setlist[itemize]{nosep, topsep=1ex}
\setlist[description]{nosep}
\allowdisplaybreaks
\def\ShowAuthNotes{1}
\ifnum\ShowAuthNotes=1
\newcommand{\authnote}[2]{\ \\ \textcolor{red}{\parbox{0.9\linewidth}{[{\footnotesize {\bf #1:} { {#2}}}]}}\newline}
\else
\newcommand{\authnote}[2]{}
\fi

\newcommand{\eps}{\varepsilon}
\newcommand{\vol}{\operatorname{vol}}

\newcommand{\extr}{\mathrm{ex}}

\newcommand{\poly}{\operatorname{\mathrm{poly}}}
\newcommand{\spa}{\operatorname{\mathrm{span}}}

\newcommand{\rk}{\operatorname{\mathrm{rank}}}

\newcommand{\supp}{\operatorname{\mathrm{supp}}}
\newcommand{\conv}{\operatorname{\mathrm{conv}}}

\newcommand{\fw}{\operatorname{width}_f}

\newcommand{\bflambda}{{\boldsymbol{\lambda}}}

\newcommand{\polylog}{\poly\log}
\newcommand{\maxbitlen}{\mathop{\mathsf{maxbitlen}}\nolimits}
\newcommand{\Tquery}{T_{\mathrm{query}}}
\newcommand{\R}{\mathbb{R}}
\newcommand{\Q}{\mathbb{Q}}
\newcommand{\Z}{\mathbb{Z}}
\newcommand{\Nonneg}{\Z_{\ge 0}}
\newcommand{\Pos}{\Z_{>0}}

\newcommand{\caD}{\mathcal{D}}

\newcommand{\caG}{\mathcal{G}}

\newcommand{\caI}{\mathcal{I}}

\newcommand{\caL}{\mathcal{L}}

\newcommand{\caO}{\mathcal{O}}
\newcommand{\caP}{\mathcal{P}}
\newcommand{\caQ}{\mathcal{Q}}

\newcommand{\caS}{\mathcal{S}}
\newcommand{\caT}{\mathcal{T}}

\newcommand{\caZ}{\mathcal{Z}}

\newcommand{\Oh}{\caO}
\newcommand{\tO}{\widetilde \caO}

\newcommand{\bfa}{{\boldsymbol{a}}}
\newcommand{\bfb}{{\boldsymbol{b}}}
\newcommand{\bfc}{{\boldsymbol{c}}}

\newcommand{\bfe}{{\boldsymbol{e}}}

\newcommand{\bfg}{{\boldsymbol{g}}}
\newcommand{\bfh}{{\boldsymbol{h}}}

\newcommand{\bfm}{{\boldsymbol{m}}}

\newcommand{\bfp}{{\boldsymbol{p}}}
\newcommand{\bfq}{{\boldsymbol{q}}}
\newcommand{\bfr}{{\boldsymbol{r}}}

\newcommand{\bft}{{\boldsymbol{t}}}
\newcommand{\bfu}{{\boldsymbol{u}}}
\newcommand{\bfv}{{\boldsymbol{v}}}
\newcommand{\bfw}{{\boldsymbol{w}}}
\newcommand{\bfx}{{\boldsymbol{x}}}
\newcommand{\bfy}{{\boldsymbol{y}}}
\newcommand{\bfz}{{\boldsymbol{z}}}
\newcommand{\one}{\mathbf{1}}
\newcommand{\zero}{{\boldsymbol{0}}}

\newcounter{problem}

\crefname{problem}{Problem}{Problems}
\Crefname{problem}{Problem}{Problems}
\newcommand{\defproblem}[4]{
  \par\vspace{3mm}
\noindent\fbox{
  \begin{minipage}{0.95\textwidth}
  \refstepcounter{problem}\label{#1}
  \textbf{Problem \theproblem.} #2 \\
  \textbf{Input:} #3  \\
  \textbf{Task:} #4
  \end{minipage}
  }
  \vspace{2mm}
}

\title{Subset Sum in Almost-Linear Time}
\author{Ce Jin\thanks{Supported by the Miller Research Fellowship at the Miller Institute for Basic Research in Science, UC Berkeley. Part of the work was done while the author was affiliated with MIT and supported by the Jane Street Graduate Research Fellowship, NSF grant CCF-2330048, and a Simons Investigator Award.} \\ UC Berkeley\\\url{cejin@berkeley.edu} \and Xiao Mao\thanks{ Supported by NSF CAREER Award CCF-2337901.} \\ Stanford University \\ \url{xiaomao@stanford.edu}}

\begin{document}

  \setcounter{page}{0} \clearpage
  \maketitle
  \thispagestyle{empty}
\begin{abstract}
  We give a deterministic algorithm for finding a subset of the input integers $1\le a_1,\ldots,a_n \le \Delta$ that sums to the target $t$ in $\mathcal{O}(n+\Delta^{1+o(1)})$ time. This improves upon the previous $\widetilde{\mathcal{O}}(n + \Delta^{1.5})$-time algorithm by Chen, Lian, Mao, and Zhang (SODA 2024), and achieves almost tight dependence on $\Delta$ assuming SETH.
  
  We use an inverse theorem for subset sums due to Conlon, Fox, and Pham (2023), and develop a quotient variant of economic cap covering from the polytope approximation literature.
\end{abstract}

  \clearpage
  \begingroup
    \hypersetup{pageanchor=false}
    \pagenumbering{gobble}
    \tableofcontents
    \clearpage
  \endgroup
  \hypersetup{pageanchor=true}
  \pagenumbering{arabic}

\section{Introduction}
\subsection{Background}

We revisit the \textsc{Subset Sum} problem, a fundamental NP-hard computational task with connections to combinatorics, discrete optimization, and cryptography.

\defproblem{prob:subset-sum}{\textsc{Subset Sum} (decision version)}{Integers $1\le a_1,\dots,a_n\le \Delta$ and a target $t$ (without loss of generality, $\Delta \le t\le n \Delta)$\footnote{To justify this,  remove the useless integers $a_i>t$, reset $\Delta \coloneqq \max_i a_i$, and return NO if $t>n \Delta$.}.}{Decide whether there exists $(x_1,\dots,x_n) \in \{0,1\}^n$ such that $\sum_{i=1}^n a_i x_i = t$.}

The \textsc{Subset Sum} problem (as well as the more difficult \textsc{0-1 Knapsack}\footnote{In \textsc{0-1 Knapsack}, given the knapsack capacity $t \in \Nonneg$ and $n$ items with weights $w_1,\dots,w_n \in \Z \cap [0,\Delta]$ and profits $p_1,\dots,p_n \in \Nonneg$, the task is to find $(x_1,\dots,x_n) \in \{0,1\}^n$ that maximizes $\sum_{i=1}^n p_i x_i$ subject to $\sum_{i=1}^n w_i x_i \le t$.})
 can be solved in pseudopolynomial time
$\Oh(nt) = \Oh(n^2 \Delta)$ by a standard dynamic program \cite{Bellman1957}
(see the textbook  \cite{bookkpp} for an extensive treatment of related problems and algorithms).
We study the fine-grained complexity of \textsc{Subset Sum} in the pseudopolynomial-time regime:

\begin{quote}
   \centering
    \itshape
 How fast can we solve \textup{\textsc{Subset Sum}} with $n$ input integers bounded by $\Delta$?
\end{quote}

Work on this question dates back more than three decades.
Early approaches were based on additive number theory and combinatorics.
In particular, the existence of long arithmetic progressions in subset sums \cite{alonfreiman,freiman93,Sarkozy2} was used to design faster  algorithms in certain regimes \cite{freimanparti,cfgalgo,freimanalgo,GalilM91,chaimovichalgo}.
For example, 
Chaimovich, Freiman, and Galil's algorithm runs in
$\tO(n + \Delta)$ time
 when $\Delta^{3/2} \sqrt{\log \Delta}\ll t \le \frac{1}{2} \sum_{i=1}^{n} a_i$ and all $a_i$ are distinct \cite{cfgalgo}. See \cref{tab:intro-dense-distinct} 
in \cref{subsec:relatedwork} and 
the survey \cite{chaimovichsurvey} for further results of this kind.

Later, Pisinger~\cite{Pisinger99}
solved \textsc{Subset Sum} in $\Oh(n \Delta)$ time via a refined dynamic program, achieving
the first general improvement over the baseline $\Oh(n^2 \Delta)$ time complexity.

\begin{table}[ht]
\centering
\renewcommand{\arraystretch}{1.15}
\begin{tabular}{|p{0.5\linewidth}|c|c|}
\hline
\textbf{Reference} & \textbf{Time in $n,\Delta$} & \textbf{Time in $n,t,\Delta$} \\
\hline
\raggedright Bellman~\cite{Bellman1957} & $\Oh(n^2 \Delta)$  & $\Oh(nt)$ \\
\hline
\raggedright Pisinger~\cite{Pisinger99} & $\Oh(n \Delta)$ & -- \\
\hline
\raggedright Koiliaris and Xu~\cite{KoiliarisX17proc} & -- &
$\tO(n + t \sqrt{n})$ \\
\hline
\raggedright Bringmann~\cite{Bringmann17} & -- & $\tO(n+t)$  \\
\hline
\raggedright Chen, Lian, Mao, and Zhang~\cite{clmzfocs24} & $\tO(n + \sqrt{n} \Delta)$ & $\tO(n+\sqrt{\Delta t})$ \\
\hline
\hline
\raggedright Galil and Margalit~\cite{GalilM91} + \cite{Bringmann17}$^{\triangle}$ &
$\tO(n+\Delta^{3/2})$ & -- \\
\hline
\raggedright Eisenbrand and Weismantel~\cite{EisenbrandW20} + \cite{Bringmann17} & $\tO(n+\Delta^2)$ & -- \\
\hline
\raggedright Polak, Rohwedder, and W\k{e}grzycki~\cite{icalp21} & $\tO(n+\Delta^{5/3})$ & -- \\
\hline
\raggedright Chen, Lian, Mao, and Zhang~\cite{ChenLMZsoda24}; \cite{jinunpublishedknapsack} & $\tO(n+\Delta^{3/2})$ & -- \\
\hline
\raggedright Chen, Lian, Mao, and Zhang~\cite{clmzfocs24}$^{\triangle}$ &
$\tO(n+\Delta^{5/4})$ & -- \\
\hline
\textbf{This work} & $\Oh(n + \Delta^{1+o(1)})$ & -- \\
\hline
\end{tabular}
\caption{Pseudopolynomial-time algorithms for \textsc{Subset Sum} (possibly randomized).
Entries marked with $^{\triangle}$ require distinct $a_1,\dots,a_n$. Throughout this paper, $\tO(T)$ denotes $\Oh(T\cdot (\log T)^{\Oh(1)})$.}
\label{tab:intro-fw}
\end{table}

Using fast Fourier transform (FFT) with divide and conquer, Koiliaris and Xu \cite{KoiliarisX17proc,KoiliarisX19} solved \textsc{Subset Sum} in
$\tO(n+t \sqrt{n} )$ deterministic time. 
This was improved in a landmark result of Bringmann \cite{Bringmann17} to an $\tO(n+t)$-time randomized algorithm  (derandomized by Chan \cite{chan26}).

Recent advances in pseudopolynomial-time \textsc{Subset Sum} algorithms came from a blend of proximity bounds (e.g., \cite{EisenbrandW20}),
FFT, and additive combinatorial structure (e.g., \cite{BringmannW21,GalilM91}).
The previous state of the art was an
$\tO(n+\Delta^{3/2})$-time algorithm by Chen, Lian, Mao, and Zhang~\cite{ChenLMZsoda24} (and an independent unpublished note \cite{jinunpublishedknapsack}).
Assuming $a_1,\dots,a_n$ are distinct, the same authors further improved the time complexity to $\tO(n+\Delta^{5/4})$ \cite{clmzfocs24,chen2026journal}. They also achieved $\tO(n+\sqrt{\Delta t})$ time without the distinctness assumption.
See \cref{tab:intro-fw} for a summary.

\begin{samepage}
Despite the exciting progress, the answer to the following question remains elusive:
\begin{quote}
   \centering
    \itshape
Main open problem: Can \textup{\textsc{Subset Sum}} be solved in $\tO(n + \Delta)$ time?
\end{quote}
\end{samepage}
This question has
been repeatedly asked in the literature on fine-grained algorithms for combinatorial optimization, e.g., \cite{AxiotisBJTW19,
BringmannW21,
icalp21,
AbboudBHS22jcss,
BringmannC22,
Cassis_2024,
Bringmann24,Jin24,
ChenLMZstoc24partition,
FischerW24,
clmzfocs24,
clmznotes,
linchensurvey}, and was described as ``the most important open question in this line of research'' in a survey talk given by Bringmann in 2021\footnote{Link: \url{https://youtu.be/lTw0flxXuKE?t=2473}}.

Under the Strong
Exponential Time Hypothesis (SETH) \cite{ImpagliazzoPZ01}, 
\textsc{Subset Sum} cannot be solved in $\tO(n + \Delta^{0.99})$ time. In fact, even a $2^{o(n)}\cdot t^{1-\delta}$-time algorithm (for any positive constant $\delta$) would refute SETH, as shown by Abboud, Bringmann, Hermelin, and Shabtay
\cite{AbboudBHS22}.

\subsection{Our result}

We \emph{almost} (up to a subpolynomial factor) answer the main open problem and achieve SETH-optimality. 
With little extra effort, our algorithm solves the following slightly more general problem (also called \textsc{Bounded Subset Sum}), and can find a witness when one exists.

\defproblem{prob:single-constraint-ilp}{\textsc{Single-Constraint Bounded Integer Feasibility} (search version)}{Integers $a_i\in [-\Delta,\Delta]$ and $u_i \in [1,U]$ for $1\le i\le n$, and an integer $b$ ($\lvert b\rvert\le nU \Delta$).}{Find a solution to $\sum_{i=1}^n a_i x_i = b, x_i\in \{0,1,\dots,u_i\}$, or report that none exists.}

\begin{theorem}
    \label{thm:main-main}
    There is a deterministic algorithm that solves
    \cref{prob:single-constraint-ilp} (which generalizes {\rm \textsc{Subset Sum}}) in $\Oh(n+\Delta^{1+o(1)})$ time (in the word-RAM model with $\Theta(\log(nU \Delta))$-bit words).
\end{theorem}

The time complexity is $\Oh(n + C_d \Delta^{1+\Oh(1/d)})$ for a parameter $d$ as discussed below, and a  computable function $C_d$ of $d$. Letting $d$ grow sufficiently slowly with $\Delta$ yields \cref{thm:main-main}.

Our approach for \textsc{Subset Sum} lies at the intersection of additive combinatorics and computational geometry.
We use the inverse theorem of 
 Conlon, Fox, and Pham \cite{conlon2023homogeneous} 
 to model (almost all of) the input integers $a_1,\dots,a_n$ as elements of a generalized arithmetic progression (GAP) $f(B\cap \Z^d)$ with $f\colon \R^d \to \R$ linear and a box $B\subseteq \R^d$, such that, writing $a_i = f(\bfx_i)$,  $\bfx_i \in B\cap \Z^d$, every integer point deep inside the zonotope $Z= \{\sum_{i} \eps_i \bfx_i : \eps_i \in [0,1]\}$ must be expressible as a subset sum of $(\bfx_i)_i$.

Then, the remaining possible subset sums correspond to points lying in a thin shell near the boundary of the zonotope $Z$.
 To compute them, we follow the polytope approximation literature (e.g., \cite{aryafonsecamount2017,AryaFM23}) to efficiently cover the shell by a collection of caps.
Due to convexity, we can compute the subset sums in each cap in near-linear time by adapting Bringmann's algorithm~\cite{Bringmann17,chan26}.
 The total time is thus near-linear in the shell volume.

To keep the shell volume small,
we remove the redundancy caused by the non-injectivity of $f$: 
define the lattice $\Lambda = \ker(f)\cap \Z^d$, and
measure shell volume in the quotient space $\R^d/\Lambda$. We then develop and apply a quotient variant of cap covering.
 We also have to control the thickness of the shell when $Z$ is not fat. This is done by repeatedly applying \cite{conlon2023homogeneous} to partition the input integers, so that groups with thin zonotopes also have thin shells.
  See \cref{sec:overview} for a more detailed overview. 
 
  \paragraph*{Independent and concurrent work.} Independently and concurrently to our work, Chen, Mao, and Zhang \cite{chen2026subsetsumintegerprogramming} (made public on arXiv on Oct 12, 2026)  obtained the same result as \cref{thm:main-main}.\footnote{For proof of independence, see \url{https://github.com/ce-jin/subset-sum-draft/blob/main/main.pdf} for a mathematically complete draft of this paper uploaded to GitHub on Oct 6, 2026, and shared to the authors of \cite{chen2026subsetsumintegerprogramming} on Oct 12, 2026 for coordination.}
  Both papers share the ingredients of the Conlon--Fox--Pham theorem and Macbeath regions, but the implementation strategies differ: Roughly speaking, our approach partitions the input items into groups and merges the answers bottom-up, while they use a recursive approach with a more top-down flavor. 
  They also generalized the result to $d$-dimensional \textsc{Subset Sum} (and consequently, integer linear programming). 
\subsection{Related work and open questions}
\label{subsec:relatedwork}
It remains an open question to improve our algorithm to $\tO(n+\Delta)$ time. To tackle it, one may start with the \textsc{Partition} problem, which is the special case of \textsc{Subset Sum} with $t=\frac{1}{2} \sum_{i=1}^n a_i$.
Analogous questions for several related problems have already been resolved:
\textsc{Unbounded Subset Sum}\footnote{\textsc{Unbounded Subset Sum} asks whether there exists $(x_1,\ldots,x_n)\in \Nonneg^{n}$ such that $\sum_{i=1}^n a_i x_i =t$.} can be solved in $\tO(n+\Delta)$ time \cite{JansenR23}.
\textsc{Modular Subset Sum}\footnote{\textsc{Modular Subset Sum} 
asks whether there exists $(x_1,\ldots,x_n)\in \{0,1\}^n$ such that $\sum_{i=1}^n a_i x_i \equiv t \pmod{m}$.} can be solved in $\tO(n+m)$ time \cite{AxiotisBJTW19,AxiotisBBJNTW21,CardinalI21,BringmannN21icalp,Potepa21,le2026deterministiclineartimemodularsubset}.
Chen, Lian, Mao, and Zhang \cite{ChenLMZstoc24partition,clmznotes} gave a  $(1-\eps)$-approximation algorithm for \textsc{Partition} in $\tO(n + 1/\eps)$ time.

For \textsc{0-1 Knapsack},
 Bringmann \cite{Bringmann24} and Jin \cite{Jin24} gave $\tO(n + \Delta^{2})$-time exact algorithms, and Chen, Lian, Mao, and Zhang \cite{clmzknapsackstoc24} and Mao \cite{xiaomaofptas} gave $\tO(n + 1/\eps^{2})$-time approximation algorithms. See the survey \cite{linchensurvey}, and more recent simplifications and advances \cite{vaughn2026simpleralgorithmsknapsacksubset,alman2026trulysubquadratic3sumtruly}.

We are optimistic that our geometric techniques will lead to further progress in combinatorial optimization, especially for pseudopolynomial-time algorithms.
Notably, several previous papers identified the lack of an $\tO(n+\Delta)$-time 
\textsc{Subset Sum} algorithm as a barrier for resolving their open problems:
Bringmann and
Wellnitz \cite[Section~6]{BringmannW21} on dense \textsc{Subset Sum} instances in the super-linear-time regime, 
Fischer and Wennmann \cite[Page 4]{FischerW24} on minimizing tardy processing time, and Chen, Hu, Mao, and Zhang \cite[Section~1.1]{ChenHMZ26} on removing the nondegeneracy assumption in their algorithm for solving multidimensional dense \textsc{Subset Sum} instances.

Bringmann's $\tO(n+t)$-time algorithm \cite{Bringmann17,chan26} actually computes \emph{all} attainable subset sums in $[0,t]$.
This all-target feature might seem impossible in $\tO(n+\Delta)$ time, since the output size may well exceed $\Delta$.
Nevertheless, the output set is actually expressible as the union of $\Oh(\Delta \log(n \Delta))$ arithmetic progressions of the same common difference (\cref{lem:aprepresentation}).
Can we efficiently compute this compact representation? Such an algorithm would have positive implications for W\k{e}grzycki's \textsc{Bounded Support Subset Sum} question \cite[Problem~2]{AlgUWOpenProblems25};
see \cref{app:subset-sum-representation} for details.

\textsc{Subset Sum} has been studied in the exponential-time regime and in the average-case setting as well; see e.g., \cite{horowitz1974computing,nederlofsurvey,ye2026subsetsumpartialmatch,oai2,oai1,LagariasO85,JouxW26} and the references therein.

\begin{table}[ht]
\centering
\renewcommand{\arraystretch}{1.15}
\begin{tabular}{|p{0.42\linewidth}|p{0.14\linewidth}|p{0.36\linewidth}|}
\hline
\textbf{Reference} & \textbf{Time} & \textbf{Assumptions (in addition to $^\star$)} \\
\hline
\raggedright Chaimovich, Freiman, and Galil~\cite{cfgalgo} & $\tO(n + \Delta^2/n^2)$ &
 $t\gg \Delta^{3/2} \sqrt{\log \Delta}$ \\
\hline
\raggedright Freiman~\cite{freimanalgo} & $\Oh_{\eps}(n)$ &
$t\gg \Delta^{5/3}(\log \Delta)^{4/3}$,
$S\ge \eps n \Delta$ \\ \hline
\raggedright Chaimovich~\cite{chaimovichalgo} & $\Oh(n^{7/4})$ &
$n\gg \sqrt{\Delta \log \Delta}$,
$\lvert t-S/2\rvert\ll n^2$ \\
\hline
\raggedright Galil and Margalit~\cite{GalilM91} & $\tO(n + \Delta^2/n^2)$ & $nt\gg \sqrt{\Delta}\,S \log \Delta$\\
\hline
\raggedright Bringmann and Wellnitz~\cite{BringmannW21} & $\tO(n)$ & $n^2 t\gg \Delta S \log(2n)$\\
\hline
\end{tabular}
\caption{A non-exhaustive list of previous algorithms for \textsc{Subset Sum} (decision version)  under the following assumptions$^\star$: the input integers $a_1,\dots,a_n$ are distinct, and $t\le S/2$, where $S\coloneqq \sum_{i=1}^n a_i$.
Here $f\gg g$ means $f\ge Cg$ for an absolute positive constant $C$.}
\label{tab:intro-dense-distinct}
\end{table}
    \section{Overview}
\label{sec:overview}
\newcounter{overviewpanel}[figure]
\renewcommand{\theoverviewpanel}{\thefigure(\alph{overviewpanel})}
\renewcommand{\theHoverviewpanel}{\theHfigure.\arabic{overviewpanel}}
\crefalias{overviewpanel}{figure}
\definecolor{ovblue}{RGB}{54,112,153}
\definecolor{ovamber}{RGB}{220,166,48}
\definecolor{ovgray}{RGB}{126,132,137}
\tikzset{
  ov body/.style={draw=ovblue,line width=.8pt,fill=ovblue!17},
  ov local/.style={draw=ovamber!85!black,line width=.8pt,fill=ovamber!32},
  ov shell/.style={draw=ovamber!85!black,line width=.8pt,fill=ovamber!23},
  ov guide/.style={draw=ovgray!70,line width=.45pt},
  ov cut/.style={draw=ovamber,line width=.8pt,dashed},
  ov arrow/.style={draw=ovgray!90,line width=.65pt,->},
  ov point/.style={fill=black,draw=none},
  ov figure/.style={x=1cm,y=1cm,font=\footnotesize,text=black,
    line cap=round,line join=round,>=stealth},
  ov subset figure/.style={ov figure,font=\normalsize,
    ov arrow/.append style={line width=.975pt}},
  ov subset grid/.style={x=.2618cm,y=.2618cm},
  ov subset boundary/.style={draw=ovamber!85!black,line width=.975pt,fill=white},
  ov subset vector/.style={->,draw=ovblue,line width=.975pt},
  ov subset attained/.style={fill=ovblue,draw=none,radius=1.4pt},
  ov subset missing/.style={draw=ovamber!85!black,fill=white,line width=.65pt,
    radius=1.4pt}
}

\newcommand{\OverviewRandomOneDimension}{\relax
\begin{tikzpicture}[ov figure,font=\small]
  \fill[ovblue!17] ({12*13/70},-.18) rectangle ({12*57/70},.18);
  \draw[ov guide] (0,0)--(12,0);
  \foreach \a in {0,1,...,70}
    \fill[ovblue] ({12*\a/70},0) circle[radius=1.7pt];
  \foreach \a in {1,2,3,4,7,12,58,63,66,67,68,69}
    \draw[draw=ovamber!85!black,fill=white,line width=.5pt]
      ({12*\a/70},0) circle[radius=1.7pt];
  \foreach \a in {0,70}
    \node[below] at ({12*\a/70},-.2) {$\a$};
\end{tikzpicture}\relax
}

\newcommand{\OverviewFlattenedGeometry}{\relax
\begin{tikzpicture}[ov figure]
  \draw[ov arrow] (7,1.05)--(7,.43);
  \node[right] at (7.05,.75) {$f(x_1,x_2)=x_1+7x_2$};
  \draw[ov guide] (0,0)--(14,0);
  \foreach \a in {0,1,...,102}
    \fill[ovgray!65] ({14*\a/102},0) circle[radius=.6pt];
  \foreach \a in {0,9,15,16,17,22,23,24,25,26,31,32,33,37,38,39,
      40,41,42,45,46,47,48,49,53,54,55,56,57,60,61,62,
      63,64,65,69,70,71,76,77,78,79,80,85,86,87,93,102}
    \fill[ov subset attained] ({14*\a/102},0) circle;
  \foreach \a in {8,30,34,68,72,94}
    \draw[ov subset missing] ({14*\a/102},0) circle;
  \foreach \a in {0,102}
    \node[below] at ({14*\a/102},-.14) {$\a$};
\end{tikzpicture}\relax
}

\newcommand{\OverviewGenerating}{\relax
\begin{tikzpicture}[ov subset figure]
  \begin{scope}[shift={(.616,1.848)},ov subset grid]
    \draw[draw=ovgray!22,line width=.45pt,fill=white]
      (-3,-3) rectangle (3,3);
    \foreach \a in {-2,-1,...,2}
      \draw[draw=ovgray!22,line width=.45pt] (\a,-3)--(\a,3);
    \foreach \b in {-2,-1,...,2}
      \draw[draw=ovgray!22,line width=.45pt] (-3,\b)--(3,\b);
    \draw[ov arrow] (-3.7,0)--(4,0);
    \draw[ov arrow] (0,-3.7)--(0,4);
    \node[right] at (4,0) {$x_1$};
    \node[above] at (0,4) {$x_2$};
    \foreach \a/\b in {2/-3,3/-1,-3/2,0/3,1/3,2/3,3/3} {
      \draw[ov subset vector] (0,0)--(\a,\b);
      \fill[ov subset attained] (\a,\b) circle;
    }
    \fill[ov subset attained] (0,0) circle;
    \node[above left] at (-3,3) {$H$};
    \node at (0,-4.5) {$X'$};
  \end{scope}
  \begin{scope}[shift={(3.85,0)},ov subset grid]
    \draw[ov subset boundary] (-3,2)--(-1,-1)--(2,-3)--(5,-4)--(8,-1)--(10,2)--(11,5)--(11,8)--(9,11)--(6,13)--(3,14)--(0,11)--(-2,8)--(-3,5)--cycle;
    \fill[ovblue!17] (1,2) rectangle (7,8);
    \draw[ov arrow] (-3.8,0)--(11.8,0);
    \draw[ov arrow] (0,-4.8)--(0,14.8);
    \node[right] at (11.8,0) {$x_1$};
    \node[above] at (0,14.8) {$x_2$};
    \node at (7,16.3) {$\caZ(X')$ and $\Sigma(X')$};
    \foreach \a/\b in {3/-3,4/-3,5/-3,6/-3,1/-2,3/-2,4/-2,5/-2,6/-2,7/-2,0/-1,1/-1,
      2/-1,4/-1,-1/0,1/0,6/0,7/0,8/0,-2/1,-1/1,1/1,6/1,7/1,
      8/1,9/1,-2/2,-3/3,-2/3,-1/3,8/3,9/3,10/3,-3/4,-2/4,-1/4,
      8/4,9/4,10/4,-2/6,-1/6,0/6,9/6,10/6,11/6,-2/7,-1/7,0/7,
      9/7,10/7,11/7,10/8,-1/9,0/9,1/9,2/9,7/9,9/9,10/9,0/10,
      1/10,2/10,7/10,9/10,4/11,6/11,7/11,8/11,1/12,2/12,3/12,4/12,
      5/12,7/12,2/13,3/13,4/13,5/13}
      \draw[ov subset missing] (\a,\b) circle;
    \foreach \a/\b in {5/-4,2/-3,2/-2,-1/-1,3/-1,5/-1,6/-1,7/-1,8/-1,0/0,2/0,3/0,
      4/0,5/0,0/1,2/1,3/1,4/1,5/1,-3/2,-1/2,0/2,8/2,9/2,
      10/2,0/3,0/4,-3/5,-2/5,-1/5,0/5,8/5,9/5,10/5,11/5,8/6,
      8/7,-2/8,-1/8,0/8,8/8,9/8,11/8,3/9,4/9,5/9,6/9,8/9,
      3/10,4/10,5/10,6/10,8/10,0/11,1/11,2/11,3/11,5/11,9/11,6/12,
      6/13,3/14}
      \fill[ov subset attained] (\a,\b) circle;
    \foreach \a in {1,2,...,7}
      \foreach \b in {2,3,...,8}
        \fill[ov subset attained] (\a,\b) circle;
  \end{scope}
\end{tikzpicture}\relax
}

\newcommand{\OverviewPositiveGeometry}{\relax
\begin{tikzpicture}[ov subset figure]
  \begin{scope}[shift={(.4,2)},ov subset grid]
    \foreach \a in {1,2,3}
      \draw[draw=ovgray!22,line width=.45pt] (\a,0)--(\a,3);
    \foreach \b in {1,2,3}
      \draw[draw=ovgray!22,line width=.45pt] (0,\b)--(3,\b);
    \draw[ov arrow] (-.55,0)--(4,0);
    \draw[ov arrow] (0,-.55)--(0,4);
    \node[right] at (4,0) {$x_1$};
    \node[above left] at (0,4) {$x_2$};
    \node[below left,inner sep=2pt] at (0,0) {$O$};
    \node at (1.5,4.6) {$X$};
    \foreach \a/\b in {1/2,1/3,2/1,2/2,2/3,3/2} {
      \draw[ov subset vector] (0,0)--(\a,\b);
      \fill[ov subset attained] (\a,\b) circle;
    }
    \fill[ov subset attained] (0,0) circle;
  \end{scope}
  \node[align=center] at (.85,1.1)
    {$A=\{9,15,16,$\\$17,22,23\}$};
  \begin{scope}[shift={(3.4,.1)},ov subset grid]
    \draw[ov subset boundary] (0,0)--(2,1)--(5,3)--(7,5)--(9,8)--(10,10)--(11,13)--(9,12)--(6,10)--(4,8)--(2,5)--(1,3)--cycle;
    \draw[ov arrow] (-.8,0)--(12,0);
    \draw[ov arrow] (0,-.8)--(0,14);
    \node[right] at (12,0) {$x_1$};
    \node[above left] at (0,14) {$x_2$};
    \node[below left,inner sep=2pt] at (0,0) {$O$};
    \node at (6,15.2) {$\caZ(X)$ and $\Sigma(X)$};
    \foreach \a/\b in {1/1,2/4,5/9,6/4,9/9,10/12}
      \draw[ov subset missing] (\a,\b) circle;
    \foreach \a/\b in {0/0,2/1,1/2,2/2,3/2,1/3,2/3,3/3,4/3,5/3,3/4,4/4,5/4,2/5,3/5,4/5,5/5,6/5,7/5,3/6,4/6,5/6,6/6,7/6,4/7,5/7,6/7,7/7,8/7,4/8,5/8,6/8,7/8,8/8,9/8,6/9,7/9,8/9,6/10,7/10,8/10,9/10,10/10,8/11,9/11,10/11,9/12,11/13}
      \fill[ov subset attained] (\a,\b) circle;
  \end{scope}
\end{tikzpicture}\relax
}

\newcommand{\OverviewShellObstructions}{\relax
\begin{tikzpicture}[ov figure,scale=.85]
  \pgfmathsetmacro{\ovLatticeStep}{1.6/sqrt(130)}
  \pgfmathsetmacro{\ovKernelTilt}{-atan(3/11)}
  \pgfmathsetmacro{\ovLatticeOriginX}{.55*cos(\ovKernelTilt)+.8*sin(\ovKernelTilt)}
  \pgfmathsetmacro{\ovLatticeOriginY}{.55*sin(\ovKernelTilt)-.8*cos(\ovKernelTilt)}
  \def\ovObstructionD{
    (-3,0)--(-2,-2)--(0,-3)--(2,-2)--
    (3,0)--(2,2)--(0,3)--(-2,2)--cycle}
  \def\ovObstructionThinZ{
    (-24,0)--(-23,-2)--(-21,-3)--(21,-3)--
    (23,-2)--(24,0)--(24,1)--(23,3)--
    (21,4)--(-21,4)--(-23,3)--(-24,1)--cycle}
  \def\ovObstructionThinInner{
    (-21,0)--(21,0)--(21,1)--(-21,1)--cycle}
  \def\ovObstructionRoundZ{
    (-20,1)--(-19,-2)--(-17,-5)--(-14,-8)--
    (-11,-10)--(-8,-11)--(3,-11)--(6,-10)--
    (9,-8)--(12,-5)--(14,-2)--(15,1)--
    (15,12)--(14,15)--(12,18)--(9,21)--
    (6,23)--(3,24)--(-8,24)--(-11,23)--
    (-14,21)--(-17,18)--(-19,15)--(-20,12)--cycle}
  \def\ovObstructionRoundInner{
    (-17,1)--({-49/3},-1)--(-15,-3)--(-12,-6)--
    (-10,{-22/3})--(-8,-8)--(3,-8)--(5,{-22/3})--
    (7,-6)--(10,-3)--({34/3},-1)--(12,1)--
    (12,12)--({34/3},14)--(10,16)--(7,19)--
    (5,{61/3})--(3,21)--(-8,21)--(-10,{61/3})--
    (-12,19)--(-15,16)--({-49/3},14)--(-17,12)--cycle}
  \begin{scope}[shift={(2.9,2.8)}]
    \begin{scope}[shift={(0,{-.5*\ovLatticeStep})},x=\ovLatticeStep cm,y=\ovLatticeStep cm]
      \foreach \i in {-27,...,27}
        \foreach \j in {-23,...,21}
          \fill[ovgray!50] (\i,\j) circle[radius=.45pt];
      \begin{scope}[shift={(0,12)}]
        \draw[draw=ovgray!85,line width=.8pt,fill=ovgray!9] \ovObstructionD;
        \begin{scope}
          \clip \ovObstructionD;
          \foreach \i in {-3,...,3}
            \foreach \j in {-3,...,3}
              \fill[ovgray!50] (\i,\j) circle[radius=.45pt];
        \end{scope}
        \draw[draw=ovgray!85,line width=.8pt] \ovObstructionD;
        \node[right,inner sep=3pt] at (3,0) {$D$};
      \end{scope}
      \draw[ov shell] \ovObstructionThinZ;
      \draw[ov body] \ovObstructionThinInner;
      \begin{scope}
        \clip \ovObstructionThinZ;
        \foreach \i in {-24,...,24}
          \foreach \j in {-3,...,4}
            \fill[ovgray!50] (\i,\j) circle[radius=.45pt];
      \end{scope}
      \draw[draw=ovamber!85!black,line width=.8pt] \ovObstructionThinZ;
    \end{scope}
    \node at (-1.65,.94) {$\caZ(X)$};
    \draw[ov guide] (-1.65,.73)--(-1.65,{3.5*\ovLatticeStep});
    \node at (1.05,-.96) {$\caZ(X)\ominus D$};
    \draw[ov guide] (1.05,-.73)--(.9,{-.5*\ovLatticeStep});
  \end{scope}

  \begin{scope}[shift={(10.8,2.8)}]
    \begin{scope}[shift={(\ovLatticeOriginX,\ovLatticeOriginY)},x=\ovLatticeStep cm,y=\ovLatticeStep cm]
      \foreach \i in {-23,...,19}
        \foreach \j in {-17,...,27}
          \fill[ovgray!50] (\i,\j) circle[radius=.45pt];
      \draw[ov shell] \ovObstructionRoundZ;
      \draw[ov body] \ovObstructionRoundInner;
      \begin{scope}[even odd rule]
        \clip \ovObstructionRoundZ;
        \foreach \k in {-2,...,2}
          \clip[shift={({3*\k},{11*\k})}]
            (-60,-60) rectangle (60,60) \ovObstructionRoundInner;
        \fill[ovamber!85!black] (-30,-30) rectangle (30,30);
      \end{scope}
      \begin{scope}
        \clip \ovObstructionRoundZ;
        \foreach \i in {-21,...,16}
          \foreach \j in {-12,...,25}
            \fill[ovgray!50] (\i,\j) circle[radius=.45pt];
      \end{scope}
      \draw[draw=ovamber!85!black,line width=.8pt] \ovObstructionRoundZ;
    \end{scope}
    \begin{scope}[rotate=\ovKernelTilt]
      \node at (-.45,1.05) {$\caZ(X)\ominus D$};
      \draw[->,draw=ovgray!90,line width=1.1pt] (.55,-.8)--(.55,.8);
      \fill[ovgray!90] (.55,-.8) circle[radius=1.4pt];
      \node[right] at (.55,0) {$\boldsymbol{\ell}$};
      \draw[ov arrow] (-1.5,-2.72)--(1.8,-2.72);
      \node[right] at (1.8,-2.72) {$f$};
    \end{scope}
  \end{scope}
\end{tikzpicture}\relax
}

\newcommand{\OverviewCapMacbeath}{\relax
\begin{tikzpicture}[ov figure,scale=.68]
  \fill[ovgray!8] (0,0) ellipse[x radius=2.35,y radius=1.45];
  \draw[draw=ovamber!65!black,line width=.85pt]
    (1.4,1.746903)
    arc[start angle=53.434263,end angle=-53.434263,x radius=3.525,y radius=2.175]
    arc[start angle=233.434263,end angle=126.565737,x radius=3.525,y radius=2.175]
    --cycle;
  \draw[ov local]
    (1.4,1.164602)
    arc[start angle=53.434263,end angle=-53.434263,x radius=2.35,y radius=1.45]
    arc[start angle=233.434263,end angle=126.565737,x radius=2.35,y radius=1.45]
    --cycle;
  \draw[draw=ovgray!85,line width=.8pt]
    (0,0) ellipse[x radius=2.35,y radius=1.45];
  \fill[ov point] (1.4,0) circle[radius=1.5pt];
  \node[left,fill=ovamber!32,inner sep=1.5pt] at (1.34,0) {$\bfx$};
  \node[above] at (-1.6,1.1) {$K$};
  \node[right] at (3.15,1.15) {$M_K^{3/2}(\bfx)$};
  \draw[ov guide] (3.12,1.15)--(2.6124,.7439);
  \node[right] at (3.15,-.75) {$M_K^1(\bfx)$};
  \draw[ov guide] (3.12,-.75)--(2.1298,-.6128);
\end{tikzpicture}\relax
}

\newcommand{\OverviewEconomicCovering}{\relax
\begingroup
\colorlet{ovcovershell}{ovamber!55}
\colorlet{ovcoverdark}{ovamber!35!black}
\colorlet{ovcoverlight}{ovamber!95}
\colorlet{ovcovermid}{ovamber!70!black}
\def\ovCoverZ{(1.04,0.49)--(1.038173,0.531835)--(1.032708,0.573351)--(1.023644,0.614233)--
    (1.011052,0.65417)--(0.995028,0.692857)--(0.975692,0.73)--(0.953193,0.765317)--
    (0.927701,0.798538)--(0.899411,0.829411)--(0.868538,0.857701)--(0.835317,0.883193)--
    (0.8,0.905692)--(0.762857,0.925028)--(0.72417,0.941052)--(0.684233,0.953644)--
    (0.643351,0.962708)--(0.601835,0.968173)--(0.56,0.97)--(-0.56,0.97)--
    (-0.601835,0.968173)--(-0.643351,0.962708)--(-0.684233,0.953644)--(-0.72417,0.941052)--
    (-0.762857,0.925028)--(-0.8,0.905692)--(-0.835317,0.883193)--(-0.868538,0.857701)--
    (-0.899411,0.829411)--(-0.927701,0.798538)--(-0.953193,0.765317)--(-0.975692,0.73)--
    (-0.995028,0.692857)--(-1.011052,0.65417)--(-1.023644,0.614233)--(-1.032708,0.573351)--
    (-1.038173,0.531835)--(-1.04,0.49)--(-1.04,-0.49)--(-1.038173,-0.531835)--
    (-1.032708,-0.573351)--(-1.023644,-0.614233)--(-1.011052,-0.65417)--(-0.995028,-0.692857)--
    (-0.975692,-0.73)--(-0.953193,-0.765317)--(-0.927701,-0.798538)--(-0.899411,-0.829411)--
    (-0.868538,-0.857701)--(-0.835317,-0.883193)--(-0.8,-0.905692)--(-0.762857,-0.925028)--
    (-0.72417,-0.941052)--(-0.684233,-0.953644)--(-0.643351,-0.962708)--(-0.601835,-0.968173)--
    (-0.56,-0.97)--(0.56,-0.97)--(0.601835,-0.968173)--(0.643351,-0.962708)--
    (0.684233,-0.953644)--(0.72417,-0.941052)--(0.762857,-0.925028)--(0.8,-0.905692)--
    (0.835317,-0.883193)--(0.868538,-0.857701)--(0.899411,-0.829411)--(0.927701,-0.798538)--
    (0.953193,-0.765317)--(0.975692,-0.73)--(0.995028,-0.692857)--(1.011052,-0.65417)--
    (1.023644,-0.614233)--(1.032708,-0.573351)--(1.038173,-0.531835)--(1.04,-0.49)--cycle}
\def\ovCoverInner{(0.97,0.488473)--(0.968373,0.525728)--(0.963706,0.561184)--(0.955965,0.596099)--
    (0.945211,0.630205)--(0.931526,0.663245)--(0.915013,0.694967)--(0.895798,0.725128)--
    (0.874027,0.7535)--(0.849867,0.779867)--(0.8235,0.804027)--(0.795128,0.825798)--
    (0.764967,0.845013)--(0.733245,0.861526)--(0.700205,0.875211)--(0.666099,0.885965)--
    (0.631184,0.893706)--(0.595728,0.898373)--(0.558473,0.9)--(-0.558473,0.9)--
    (-0.595728,0.898373)--(-0.631184,0.893706)--(-0.666099,0.885965)--(-0.700205,0.875211)--
    (-0.733245,0.861526)--(-0.764967,0.845013)--(-0.795128,0.825798)--(-0.8235,0.804027)--
    (-0.849867,0.779867)--(-0.874027,0.7535)--(-0.895798,0.725128)--(-0.915013,0.694967)--
    (-0.931526,0.663245)--(-0.945211,0.630205)--(-0.955965,0.596099)--(-0.963706,0.561184)--
    (-0.968373,0.525728)--(-0.97,0.488473)--(-0.97,-0.488473)--(-0.968373,-0.525728)--
    (-0.963706,-0.561184)--(-0.955965,-0.596099)--(-0.945211,-0.630205)--(-0.931526,-0.663245)--
    (-0.915013,-0.694967)--(-0.895798,-0.725128)--(-0.874027,-0.7535)--(-0.849867,-0.779867)--
    (-0.8235,-0.804027)--(-0.795128,-0.825798)--(-0.764967,-0.845013)--(-0.733245,-0.861526)--
    (-0.700205,-0.875211)--(-0.666099,-0.885965)--(-0.631184,-0.893706)--(-0.595728,-0.898373)--
    (-0.558473,-0.9)--(0.558473,-0.9)--(0.595728,-0.898373)--(0.631184,-0.893706)--
    (0.666099,-0.885965)--(0.700205,-0.875211)--(0.733245,-0.861526)--(0.764967,-0.845013)--
    (0.795128,-0.825798)--(0.8235,-0.804027)--(0.849867,-0.779867)--(0.874027,-0.7535)--
    (0.895798,-0.725128)--(0.915013,-0.694967)--(0.931526,-0.663245)--(0.945211,-0.630205)--
    (0.955965,-0.596099)--(0.963706,-0.561184)--(0.968373,-0.525728)--(0.97,-0.488473)--cycle}
\expandafter\def\csname ovCoverCap0\endcsname{(1.04,0.49)--(1.038173,0.531835)--(1.032708,0.573351)--(1.023644,0.614233)--
    (1.011052,0.65417)--(0.995028,0.692857)--(0.975692,0.73)--(0.953193,0.765317)--
    (0.927701,0.798538)--(0.91,0.817856)--(0.91,-0.817856)--(0.927701,-0.798538)--
    (0.953193,-0.765317)--(0.975692,-0.73)--(0.995028,-0.692857)--(1.011052,-0.65417)--
    (1.023644,-0.614233)--(1.032708,-0.573351)--(1.038173,-0.531835)--(1.04,-0.49)--cycle}
\expandafter\def\csname ovCoverRegion0\endcsname{(1.067097,0.735)--(1.064358,0.797752)--(1.056159,0.860027)--(1.042564,0.92135)--
    (1.023676,0.981255)--(0.999639,1.039285)--(0.985805,1.06586)--(0.971971,1.039285)--
    (0.947934,0.981255)--(0.929046,0.92135)--(0.915451,0.860027)--(0.907252,0.797752)--
    (0.904513,0.735)--(0.904513,-0.735)--(0.907252,-0.797752)--(0.915451,-0.860027)--
    (0.929046,-0.92135)--(0.947934,-0.981255)--(0.971971,-1.039285)--(0.985805,-1.06586)--
    (0.999639,-1.039285)--(1.023676,-0.981255)--(1.042564,-0.92135)--(1.056159,-0.860027)--
    (1.064358,-0.797752)--(1.067097,-0.735)--cycle}
\expandafter\def\csname ovCoverCap1\endcsname{(1.039316,0.505658)--(1.038173,0.531835)--(1.032708,0.573351)--(1.023644,0.614233)--
    (1.011052,0.65417)--(0.995028,0.692857)--(0.975692,0.73)--(0.953193,0.765317)--
    (0.927701,0.798538)--(0.899411,0.829411)--(0.868538,0.857701)--(0.835317,0.883193)--
    (0.8,0.905692)--(0.762857,0.925028)--(0.72417,0.941052)--(0.684233,0.953644)--
    (0.643351,0.962708)--(0.601835,0.968173)--(0.575658,0.969316)--cycle}
\expandafter\def\csname ovCoverRegion1\endcsname{(1.127297,0.482821)--(1.115437,0.53632)--(1.096549,0.596225)--(1.072512,0.654256)--
    (1.043509,0.70997)--(1.00976,0.762945)--(0.971522,0.812777)--(0.929087,0.859087)--
    (0.882777,0.901522)--(0.832945,0.93976)--(0.77997,0.973509)--(0.724256,1.002512)--
    (0.666225,1.026549)--(0.60632,1.045437)--(0.552821,1.057297)--(0.564681,1.003798)--
    (0.583569,0.943894)--(0.607606,0.885863)--(0.63661,0.830148)--(0.670359,0.777173)--
    (0.708596,0.727341)--(0.751031,0.681031)--(0.797341,0.638596)--(0.847173,0.600359)--
    (0.900148,0.56661)--(0.955863,0.537606)--(1.013894,0.513569)--(1.073798,0.494681)--cycle}
\expandafter\def\csname ovCoverCap2\endcsname{(0.887856,0.84)--(0.868538,0.857701)--(0.835317,0.883193)--(0.8,0.905692)--
    (0.762857,0.925028)--(0.72417,0.941052)--(0.684233,0.953644)--(0.643351,0.962708)--
    (0.601835,0.968173)--(0.56,0.97)--(-0.56,0.97)--(-0.601835,0.968173)--
    (-0.643351,0.962708)--(-0.684233,0.953644)--(-0.72417,0.941052)--(-0.762857,0.925028)--
    (-0.8,0.905692)--(-0.835317,0.883193)--(-0.868538,0.857701)--(-0.887856,0.84)--cycle}
\expandafter\def\csname ovCoverRegion2\endcsname{(1.167059,0.917124)--(1.144285,0.92898)--(1.086255,0.953017)--(1.02635,0.971904)--
    (0.965027,0.985499)--(0.902752,0.993698)--(0.84,0.996438)--(-0.84,0.996438)--
    (-0.902752,0.993698)--(-0.965027,0.985499)--(-1.02635,0.971904)--(-1.086255,0.953017)--
    (-1.144285,0.92898)--(-1.167059,0.917124)--(-1.144285,0.905269)--(-1.086255,0.881232)--
    (-1.02635,0.862344)--(-0.965027,0.848749)--(-0.902752,0.84055)--(-0.84,0.83781)--
    (0.84,0.83781)--(0.902752,0.84055)--(0.965027,0.848749)--(1.02635,0.862344)--
    (1.086255,0.881232)--(1.144285,0.905269)--cycle}
\begin{tikzpicture}[ov figure]
  \foreach \panel/\offset in {0/1.8,1/5.5} {
    \begin{scope}[shift={(\offset,0)},x=1.4cm,y=1.4cm]
      \fill[ovcovershell] \ovCoverZ;
      \fill[ovblue!17] \ovCoverInner;
      \foreach \sx/\sy in {1/1,-1/1} {
        \begin{scope}[xscale=\sx,yscale=\sy]
          \ifnum\panel=0
            \draw[draw=ovcoverdark,line width=.65pt]
              \csname ovCoverCap0\endcsname;
          \else
            \draw[draw=ovcoverdark,line width=.65pt]
              \csname ovCoverRegion0\endcsname;
            \fill[ovcoverdark] (0.985805,0) circle[radius=1pt];
          \fi
        \end{scope}
      }
      \foreach \sx/\sy in {1/1,-1/1,-1/-1,1/-1} {
        \begin{scope}[xscale=\sx,yscale=\sy]
          \ifnum\panel=0
            \draw[draw=ovcovermid,line width=.65pt]
              \csname ovCoverCap1\endcsname;
          \else
            \draw[draw=ovcovermid,line width=.65pt]
              \csname ovCoverRegion1\endcsname;
            \fill[ovcovermid] (0.840059,0.770059) circle[radius=1pt];
          \fi
        \end{scope}
      }
      \foreach \sx/\sy in {1/1,1/-1} {
        \begin{scope}[xscale=\sx,yscale=\sy]
          \ifnum\panel=0
            \draw[draw=ovcoverlight,line width=.65pt]
              \csname ovCoverCap2\endcsname;
          \else
            \draw[draw=ovcoverlight,line width=.65pt]
              \csname ovCoverRegion2\endcsname;
            \fill[ovcoverlight] (0,0.917124) circle[radius=1pt];
          \fi
        \end{scope}
      }
      \draw[draw=ovcovershell!75!black,line width=.4pt,dash pattern=on 2pt off 1.5pt] \ovCoverInner;
      \draw[draw=ovamber!65!black,line width=.65pt] \ovCoverZ;
      \node at (0,0) {$\caZ(X)\ominus D$};
    \end{scope}
  }
  \node at (1.8,-1.9) {(a) Caps $Q_i$};
  \node at (5.5,-1.9) {(b) $M_{\caZ(X)}^{3/2}(\bfx_i)$};
\end{tikzpicture}\relax
\endgroup
}

\newcommand{\OverviewClassClipping}{\relax
\begingroup
\def\ovClassZ{(0.612,0)--(0.442,0.442)--(0,0.612)--(-0.442,0.442)--
    (-0.612,0)--(-0.442,-0.442)--(0,-0.612)--(0.442,-0.442)--cycle}
\def\ovClassK{(0,0) circle[radius=.2622]}
\def\ovClassCore{(-0.331075,0)--(-0.23911,-0.23911)--(0,-0.331075)--(0.23911,-0.23911)--
    (0.331075,0)--(0.23911,0.23911)--(0,0.331075)--(-0.23911,0.23911)--cycle}
\def\ovClassSum{(-5.212,0)--(-4.938098,-0.712144)--(-4.930232,-0.718909)--(-4.908984,-0.734993)--
    (-4.885888,-0.750532)--(-4.860942,-0.765568)--(-4.834145,-0.780136)--(-4.805497,-0.794263)--
    (-4.774997,-0.807974)--(-4.742643,-0.821288)--(-4.300643,-0.991288)--(-4.266432,-1.004222)--
    (-4.230364,-1.016792)--(-4.192434,-1.029009)--(-4.152641,-1.040885)--(-4.110981,-1.052428)--
    (-4.067449,-1.063648)--(-4.02204,-1.074551)--(-3.974749,-1.085143)--(-3.925569,-1.095429)--
    (-3.874492,-1.105416)--(-3.82151,-1.115106)--(-3.766612,-1.124503)--(-3.709787,-1.133612)--
    (-3.651022,-1.142433)--(-3.590302,-1.150971)--(-3.52761,-1.159226)--(-3.462926,-1.167202)--
    (-3.396228,-1.1749)--(-3.327491,-1.18232)--(-3.256686,-1.189465)--(-3.18378,-1.196336)--
    (-3.108736,-1.202933)--(-3.03151,-1.209258)--(-2.952055,-1.21531)--(-2.870313,-1.221091)--
    (-2.786221,-1.226602)--(-2.699703,-1.231842)--(-2.610673,-1.236812)--(-2.519032,-1.241513)--
    (-2.424661,-1.245944)--(-2.327422,-1.250106)--(-2.227152,-1.253999)--(-2.123655,-1.257623)--
    (-2.016697,-1.260979)--(-1.905989,-1.264066)--(-1.791177,-1.266884)--(-1.671813,-1.269434)--
    (-1.547325,-1.271715)--(-1.416959,-1.273728)--(-1.27969,-1.275473)--(-1.134063,-1.276949)--
    (-0.977896,-1.278157)--(-0.807614,-1.279096)--(-0.616531,-1.279767)--(-0.388468,-1.28017)--
    (0,-1.280304)--(0.388468,-1.28017)--(0.616531,-1.279767)--(0.807614,-1.279096)--
    (0.977896,-1.278157)--(1.134063,-1.276949)--(1.27969,-1.275473)--(1.416959,-1.273728)--
    (1.547325,-1.271715)--(1.671813,-1.269434)--(1.791177,-1.266884)--(1.905989,-1.264066)--
    (2.016697,-1.260979)--(2.123655,-1.257623)--(2.227152,-1.253999)--(2.327422,-1.250106)--
    (2.424661,-1.245944)--(2.519032,-1.241513)--(2.610673,-1.236812)--(2.699703,-1.231842)--
    (2.786221,-1.226602)--(2.870313,-1.221091)--(2.952055,-1.21531)--(3.03151,-1.209258)--
    (3.108736,-1.202933)--(3.18378,-1.196336)--(3.256686,-1.189465)--(3.327491,-1.18232)--
    (3.396228,-1.1749)--(3.462926,-1.167202)--(3.52761,-1.159226)--(3.590302,-1.150971)--
    (3.651022,-1.142433)--(3.709787,-1.133612)--(3.766612,-1.124503)--(3.82151,-1.115106)--
    (3.874492,-1.105416)--(3.925569,-1.095429)--(3.974749,-1.085143)--(4.02204,-1.074551)--
    (4.067449,-1.063648)--(4.110981,-1.052428)--(4.152641,-1.040885)--(4.192434,-1.029009)--
    (4.230364,-1.016792)--(4.266432,-1.004222)--(4.300643,-0.991288)--(4.742643,-0.821288)--
    (4.774997,-0.807974)--(4.805497,-0.794263)--(4.834145,-0.780136)--(4.860942,-0.765568)--
    (4.885888,-0.750532)--(4.908984,-0.734993)--(4.930232,-0.718909)--(4.938098,-0.712144)--
    (5.212,0)--(4.938098,0.712144)--(4.930232,0.718909)--(4.908984,0.734993)--
    (4.885888,0.750532)--(4.860942,0.765568)--(4.834145,0.780136)--(4.805497,0.794263)--
    (4.774997,0.807974)--(4.742643,0.821288)--(4.300643,0.991288)--(4.266432,1.004222)--
    (4.230364,1.016792)--(4.192434,1.029009)--(4.152641,1.040885)--(4.110981,1.052428)--
    (4.067449,1.063648)--(4.02204,1.074551)--(3.974749,1.085143)--(3.925569,1.095429)--
    (3.874492,1.105416)--(3.82151,1.115106)--(3.766612,1.124503)--(3.709787,1.133612)--
    (3.651022,1.142433)--(3.590302,1.150971)--(3.52761,1.159226)--(3.462926,1.167202)--
    (3.396228,1.1749)--(3.327491,1.18232)--(3.256686,1.189465)--(3.18378,1.196336)--
    (3.108736,1.202933)--(3.03151,1.209258)--(2.952055,1.21531)--(2.870313,1.221091)--
    (2.786221,1.226602)--(2.699703,1.231842)--(2.610673,1.236812)--(2.519032,1.241513)--
    (2.424661,1.245944)--(2.327422,1.250106)--(2.227152,1.253999)--(2.123655,1.257623)--
    (2.016697,1.260979)--(1.905989,1.264066)--(1.791177,1.266884)--(1.671813,1.269434)--
    (1.547325,1.271715)--(1.416959,1.273728)--(1.27969,1.275473)--(1.134063,1.276949)--
    (0.977896,1.278157)--(0.807614,1.279096)--(0.616531,1.279767)--(0.388468,1.28017)--
    (0,1.280304)--(-0.388468,1.28017)--(-0.616531,1.279767)--(-0.807614,1.279096)--
    (-0.977896,1.278157)--(-1.134063,1.276949)--(-1.27969,1.275473)--(-1.416959,1.273728)--
    (-1.547325,1.271715)--(-1.671813,1.269434)--(-1.791177,1.266884)--(-1.905989,1.264066)--
    (-2.016697,1.260979)--(-2.123655,1.257623)--(-2.227152,1.253999)--(-2.327422,1.250106)--
    (-2.424661,1.245944)--(-2.519032,1.241513)--(-2.610673,1.236812)--(-2.699703,1.231842)--
    (-2.786221,1.226602)--(-2.870313,1.221091)--(-2.952055,1.21531)--(-3.03151,1.209258)--
    (-3.108736,1.202933)--(-3.18378,1.196336)--(-3.256686,1.189465)--(-3.327491,1.18232)--
    (-3.396228,1.1749)--(-3.462926,1.167202)--(-3.52761,1.159226)--(-3.590302,1.150971)--
    (-3.651022,1.142433)--(-3.709787,1.133612)--(-3.766612,1.124503)--(-3.82151,1.115106)--
    (-3.874492,1.105416)--(-3.925569,1.095429)--(-3.974749,1.085143)--(-4.02204,1.074551)--
    (-4.067449,1.063648)--(-4.110981,1.052428)--(-4.152641,1.040885)--(-4.192434,1.029009)--
    (-4.230364,1.016792)--(-4.266432,1.004222)--(-4.300643,0.991288)--(-4.742643,0.821288)--
    (-4.774997,0.807974)--(-4.805497,0.794263)--(-4.834145,0.780136)--(-4.860942,0.765568)--
    (-4.885888,0.750532)--(-4.908984,0.734993)--(-4.930232,0.718909)--(-4.938098,0.712144)--cycle}
\def\ovClassSumCore{(-4.931075,0)--(-4.722454,-0.542415)--(-4.713499,-0.547283)--(-4.693726,-0.557033)--
    (-4.671327,-0.567102)--(-4.645678,-0.577657)--(-4.207218,-0.746296)--(-4.176911,-0.757754)--
    (-4.147015,-0.768173)--(-4.11474,-0.778568)--(-4.080136,-0.788896)--(-4.043249,-0.799116)--
    (-4.004116,-0.809202)--(-3.962775,-0.819129)--(-3.919254,-0.828876)--(-3.873571,-0.838431)--
    (-3.825747,-0.847782)--(-3.775804,-0.856916)--(-3.723742,-0.865828)--(-3.669575,-0.874511)--
    (-3.613306,-0.882957)--(-3.554932,-0.891165)--(-3.49445,-0.899129)--(-3.43185,-0.906848)--
    (-3.367126,-0.914318)--(-3.300258,-0.921537)--(-3.231222,-0.928503)--(-3.159998,-0.935215)--
    (-3.086554,-0.941672)--(-3.010851,-0.947872)--(-2.932849,-0.953813)--(-2.852491,-0.959497)--
    (-2.769722,-0.964921)--(-2.68447,-0.970084)--(-2.596649,-0.974987)--(-2.506167,-0.979628)--
    (-2.412906,-0.984007)--(-2.31673,-0.988124)--(-2.217478,-0.991977)--(-2.114956,-0.995567)--
    (-2.008931,-0.998894)--(-1.899118,-1.001956)--(-1.78516,-1.004753)--(-1.666611,-1.007285)--
    (-1.542899,-1.009552)--(-1.413268,-1.011554)--(-1.276695,-1.01329)--(-1.13172,-1.014759)--
    (-0.976159,-1.015963)--(-0.806431,-1.016899)--(-0.615839,-1.017568)--(-0.388191,-1.01797)--
    (0,-1.018104)--(0.388191,-1.01797)--(0.615839,-1.017568)--(0.806431,-1.016899)--
    (0.976159,-1.015963)--(1.13172,-1.014759)--(1.276695,-1.01329)--(1.413268,-1.011554)--
    (1.542899,-1.009552)--(1.666611,-1.007285)--(1.78516,-1.004753)--(1.899118,-1.001956)--
    (2.008931,-0.998894)--(2.114956,-0.995567)--(2.217478,-0.991977)--(2.31673,-0.988124)--
    (2.412906,-0.984007)--(2.506167,-0.979628)--(2.596649,-0.974987)--(2.68447,-0.970084)--
    (2.769722,-0.964921)--(2.852491,-0.959497)--(2.932849,-0.953813)--(3.010851,-0.947872)--
    (3.086554,-0.941672)--(3.159998,-0.935215)--(3.231222,-0.928503)--(3.300258,-0.921537)--
    (3.367126,-0.914318)--(3.43185,-0.906848)--(3.49445,-0.899129)--(3.554932,-0.891165)--
    (3.613306,-0.882957)--(3.669575,-0.874511)--(3.723742,-0.865828)--(3.775804,-0.856916)--
    (3.825747,-0.847782)--(3.873571,-0.838431)--(3.919254,-0.828876)--(3.962775,-0.819129)--
    (4.004116,-0.809202)--(4.043249,-0.799116)--(4.080136,-0.788896)--(4.11474,-0.778568)--
    (4.147015,-0.768173)--(4.176911,-0.757754)--(4.207218,-0.746296)--(4.645678,-0.577657)--
    (4.671327,-0.567102)--(4.693726,-0.557033)--(4.713499,-0.547283)--(4.722454,-0.542415)--
    (4.931075,0)--(4.722454,0.542415)--(4.713499,0.547283)--(4.693726,0.557033)--
    (4.671327,0.567102)--(4.645678,0.577657)--(4.207218,0.746296)--(4.176911,0.757754)--
    (4.147015,0.768173)--(4.11474,0.778568)--(4.080136,0.788896)--(4.043249,0.799116)--
    (4.004116,0.809202)--(3.962775,0.819129)--(3.919254,0.828876)--(3.873571,0.838431)--
    (3.825747,0.847782)--(3.775804,0.856916)--(3.723742,0.865828)--(3.669575,0.874511)--
    (3.613306,0.882957)--(3.554932,0.891165)--(3.49445,0.899129)--(3.43185,0.906848)--
    (3.367126,0.914318)--(3.300258,0.921537)--(3.231222,0.928503)--(3.159998,0.935215)--
    (3.086554,0.941672)--(3.010851,0.947872)--(2.932849,0.953813)--(2.852491,0.959497)--
    (2.769722,0.964921)--(2.68447,0.970084)--(2.596649,0.974987)--(2.506167,0.979628)--
    (2.412906,0.984007)--(2.31673,0.988124)--(2.217478,0.991977)--(2.114956,0.995567)--
    (2.008931,0.998894)--(1.899118,1.001956)--(1.78516,1.004753)--(1.666611,1.007285)--
    (1.542899,1.009552)--(1.413268,1.011554)--(1.276695,1.01329)--(1.13172,1.014759)--
    (0.976159,1.015963)--(0.806431,1.016899)--(0.615839,1.017568)--(0.388191,1.01797)--
    (0,1.018104)--(-0.388191,1.01797)--(-0.615839,1.017568)--(-0.806431,1.016899)--
    (-0.976159,1.015963)--(-1.13172,1.014759)--(-1.276695,1.01329)--(-1.413268,1.011554)--
    (-1.542899,1.009552)--(-1.666611,1.007285)--(-1.78516,1.004753)--(-1.899118,1.001956)--
    (-2.008931,0.998894)--(-2.114956,0.995567)--(-2.217478,0.991977)--(-2.31673,0.988124)--
    (-2.412906,0.984007)--(-2.506167,0.979628)--(-2.596649,0.974987)--(-2.68447,0.970084)--
    (-2.769722,0.964921)--(-2.852491,0.959497)--(-2.932849,0.953813)--(-3.010851,0.947872)--
    (-3.086554,0.941672)--(-3.159998,0.935215)--(-3.231222,0.928503)--(-3.300258,0.921537)--
    (-3.367126,0.914318)--(-3.43185,0.906848)--(-3.49445,0.899129)--(-3.554932,0.891165)--
    (-3.613306,0.882957)--(-3.669575,0.874511)--(-3.723742,0.865828)--(-3.775804,0.856916)--
    (-3.825747,0.847782)--(-3.873571,0.838431)--(-3.919254,0.828876)--(-3.962775,0.819129)--
    (-4.004116,0.809202)--(-4.043249,0.799116)--(-4.080136,0.788896)--(-4.11474,0.778568)--
    (-4.147015,0.768173)--(-4.176911,0.757754)--(-4.207218,0.746296)--(-4.645678,0.577657)--
    (-4.671327,0.567102)--(-4.693726,0.557033)--(-4.713499,0.547283)--(-4.722454,0.542415)--cycle}
\def\ovClassThinCore{(-4.337391,0)--(-4.235123,-0.265896)--(-4.230181,-0.270146)--(-4.210053,-0.285382)--
    (-4.188255,-0.300047)--(-4.164798,-0.314186)--(-4.139693,-0.327834)--(-4.112952,-0.341021)--
    (-4.084583,-0.353774)--(-4.054599,-0.366112)--(-4.023002,-0.378058)--(-3.989802,-0.389629)--
    (-3.955003,-0.400838)--(-3.918607,-0.4117)--(-3.880615,-0.422226)--(-3.841024,-0.432431)--
    (-3.799833,-0.442321)--(-3.757035,-0.451906)--(-3.712612,-0.461197)--(-3.666553,-0.470203)--
    (-3.618852,-0.478928)--(-3.569469,-0.48738)--(-3.518393,-0.495568)--(-3.465595,-0.503493)--
    (-3.411038,-0.511165)--(-3.35469,-0.518584)--(-3.296501,-0.525759)--(-3.23645,-0.53269)--
    (-3.174475,-0.53938)--(-3.110514,-0.545835)--(-3.04453,-0.552053)--(-2.976455,-0.558038)--
    (-2.906226,-0.56379)--(-2.833774,-0.569308)--(-2.759002,-0.574596)--(-2.681848,-0.579653)--
    (-2.602217,-0.584475)--(-2.519985,-0.589066)--(-2.435052,-0.593423)--(-2.347283,-0.597544)--
    (-2.256515,-0.601429)--(-2.162578,-0.605076)--(-2.065254,-0.608484)--(-1.964304,-0.611651)--
    (-1.859433,-0.614576)--(-1.750261,-0.617255)--(-1.636332,-0.619689)--(-1.517062,-0.621874)--
    (-1.391674,-0.623811)--(-1.259137,-0.625495)--(-1.117968,-0.626926)--(-0.965953,-0.628102)--
    (-0.799473,-0.62902)--(-0.611769,-0.629679)--(-0.386562,-0.630077)--(0,-0.630211)--
    (0.386562,-0.630077)--(0.611769,-0.629679)--(0.799473,-0.62902)--(0.965953,-0.628102)--
    (1.117968,-0.626926)--(1.259137,-0.625495)--(1.391674,-0.623811)--(1.517062,-0.621874)--
    (1.636332,-0.619689)--(1.750261,-0.617255)--(1.859433,-0.614576)--(1.964304,-0.611651)--
    (2.065254,-0.608484)--(2.162578,-0.605076)--(2.256515,-0.601429)--(2.347283,-0.597544)--
    (2.435052,-0.593423)--(2.519985,-0.589066)--(2.602217,-0.584475)--(2.681848,-0.579653)--
    (2.759002,-0.574596)--(2.833774,-0.569308)--(2.906226,-0.56379)--(2.976455,-0.558038)--
    (3.04453,-0.552053)--(3.110514,-0.545835)--(3.174475,-0.53938)--(3.23645,-0.53269)--
    (3.296501,-0.525759)--(3.35469,-0.518584)--(3.411038,-0.511165)--(3.465595,-0.503493)--
    (3.518393,-0.495568)--(3.569469,-0.48738)--(3.618852,-0.478928)--(3.666553,-0.470203)--
    (3.712612,-0.461197)--(3.757035,-0.451906)--(3.799833,-0.442321)--(3.841024,-0.432431)--
    (3.880615,-0.422226)--(3.918607,-0.4117)--(3.955003,-0.400838)--(3.989802,-0.389629)--
    (4.023002,-0.378058)--(4.054599,-0.366112)--(4.084583,-0.353774)--(4.112952,-0.341021)--
    (4.139693,-0.327834)--(4.164798,-0.314186)--(4.188255,-0.300047)--(4.210053,-0.285382)--
    (4.230181,-0.270146)--(4.235123,-0.265896)--(4.337391,0)--(4.235123,0.265896)--
    (4.230181,0.270146)--(4.210053,0.285382)--(4.188255,0.300047)--(4.164798,0.314186)--
    (4.139693,0.327834)--(4.112952,0.341021)--(4.084583,0.353774)--(4.054599,0.366112)--
    (4.023002,0.378058)--(3.989802,0.389629)--(3.955003,0.400838)--(3.918607,0.4117)--
    (3.880615,0.422226)--(3.841024,0.432431)--(3.799833,0.442321)--(3.757035,0.451906)--
    (3.712612,0.461197)--(3.666553,0.470203)--(3.618852,0.478928)--(3.569469,0.48738)--
    (3.518393,0.495568)--(3.465595,0.503493)--(3.411038,0.511165)--(3.35469,0.518584)--
    (3.296501,0.525759)--(3.23645,0.53269)--(3.174475,0.53938)--(3.110514,0.545835)--
    (3.04453,0.552053)--(2.976455,0.558038)--(2.906226,0.56379)--(2.833774,0.569308)--
    (2.759002,0.574596)--(2.681848,0.579653)--(2.602217,0.584475)--(2.519985,0.589066)--
    (2.435052,0.593423)--(2.347283,0.597544)--(2.256515,0.601429)--(2.162578,0.605076)--
    (2.065254,0.608484)--(1.964304,0.611651)--(1.859433,0.614576)--(1.750261,0.617255)--
    (1.636332,0.619689)--(1.517062,0.621874)--(1.391674,0.623811)--(1.259137,0.625495)--
    (1.117968,0.626926)--(0.965953,0.628102)--(0.799473,0.62902)--(0.611769,0.629679)--
    (0.386562,0.630077)--(0,0.630211)--(-0.386562,0.630077)--(-0.611769,0.629679)--
    (-0.799473,0.62902)--(-0.965953,0.628102)--(-1.117968,0.626926)--(-1.259137,0.625495)--
    (-1.391674,0.623811)--(-1.517062,0.621874)--(-1.636332,0.619689)--(-1.750261,0.617255)--
    (-1.859433,0.614576)--(-1.964304,0.611651)--(-2.065254,0.608484)--(-2.162578,0.605076)--
    (-2.256515,0.601429)--(-2.347283,0.597544)--(-2.435052,0.593423)--(-2.519985,0.589066)--
    (-2.602217,0.584475)--(-2.681848,0.579653)--(-2.759002,0.574596)--(-2.833774,0.569308)--
    (-2.906226,0.56379)--(-2.976455,0.558038)--(-3.04453,0.552053)--(-3.110514,0.545835)--
    (-3.174475,0.53938)--(-3.23645,0.53269)--(-3.296501,0.525759)--(-3.35469,0.518584)--
    (-3.411038,0.511165)--(-3.465595,0.503493)--(-3.518393,0.495568)--(-3.569469,0.48738)--
    (-3.618852,0.478928)--(-3.666553,0.470203)--(-3.712612,0.461197)--(-3.757035,0.451906)--
    (-3.799833,0.442321)--(-3.841024,0.432431)--(-3.880615,0.422226)--(-3.918607,0.4117)--
    (-3.955003,0.400838)--(-3.989802,0.389629)--(-4.023002,0.378058)--(-4.054599,0.366112)--
    (-4.084583,0.353774)--(-4.112952,0.341021)--(-4.139693,0.327834)--(-4.164798,0.314186)--
    (-4.188255,0.300047)--(-4.210053,0.285382)--(-4.230181,0.270146)--(-4.235123,0.265896)--cycle}
\def\ovClassThin{(-4.496098,0.270144)--(-4.6,0)--(-4.496098,-0.270144)--(-4.488232,-0.276909)--
    (-4.466984,-0.292993)--(-4.443888,-0.308532)--(-4.418942,-0.323568)--(-4.392145,-0.338136)--
    (-4.363497,-0.352263)--(-4.332997,-0.365974)--(-4.300643,-0.379288)--(-4.266432,-0.392222)--
    (-4.230364,-0.404792)--(-4.192434,-0.417009)--(-4.152641,-0.428885)--(-4.110981,-0.440428)--
    (-4.067449,-0.451648)--(-4.02204,-0.462551)--(-3.974749,-0.473143)--(-3.925569,-0.483429)--
    (-3.874492,-0.493416)--(-3.82151,-0.503106)--(-3.766612,-0.512503)--(-3.709787,-0.521612)--
    (-3.651022,-0.530433)--(-3.590302,-0.538971)--(-3.52761,-0.547226)--(-3.462926,-0.555202)--
    (-3.396228,-0.5629)--(-3.327491,-0.57032)--(-3.256686,-0.577465)--(-3.18378,-0.584336)--
    (-3.108736,-0.590933)--(-3.03151,-0.597258)--(-2.952055,-0.60331)--(-2.870313,-0.609091)--
    (-2.786221,-0.614602)--(-2.699703,-0.619842)--(-2.610673,-0.624812)--(-2.519032,-0.629513)--
    (-2.424661,-0.633944)--(-2.327422,-0.638106)--(-2.227152,-0.641999)--(-2.123655,-0.645623)--
    (-2.016697,-0.648979)--(-1.905989,-0.652066)--(-1.791177,-0.654884)--(-1.671813,-0.657434)--
    (-1.547325,-0.659715)--(-1.416959,-0.661728)--(-1.27969,-0.663473)--(-1.134063,-0.664949)--
    (-0.977896,-0.666157)--(-0.807614,-0.667096)--(-0.616531,-0.667767)--(-0.388468,-0.66817)--
    (0,-0.668304)--(0.388468,-0.66817)--(0.616531,-0.667767)--(0.807614,-0.667096)--
    (0.977896,-0.666157)--(1.134063,-0.664949)--(1.27969,-0.663473)--(1.416959,-0.661728)--
    (1.547325,-0.659715)--(1.671813,-0.657434)--(1.791177,-0.654884)--(1.905989,-0.652066)--
    (2.016697,-0.648979)--(2.123655,-0.645623)--(2.227152,-0.641999)--(2.327422,-0.638106)--
    (2.424661,-0.633944)--(2.519032,-0.629513)--(2.610673,-0.624812)--(2.699703,-0.619842)--
    (2.786221,-0.614602)--(2.870313,-0.609091)--(2.952055,-0.60331)--(3.03151,-0.597258)--
    (3.108736,-0.590933)--(3.18378,-0.584336)--(3.256686,-0.577465)--(3.327491,-0.57032)--
    (3.396228,-0.5629)--(3.462926,-0.555202)--(3.52761,-0.547226)--(3.590302,-0.538971)--
    (3.651022,-0.530433)--(3.709787,-0.521612)--(3.766612,-0.512503)--(3.82151,-0.503106)--
    (3.874492,-0.493416)--(3.925569,-0.483429)--(3.974749,-0.473143)--(4.02204,-0.462551)--
    (4.067449,-0.451648)--(4.110981,-0.440428)--(4.152641,-0.428885)--(4.192434,-0.417009)--
    (4.230364,-0.404792)--(4.266432,-0.392222)--(4.300643,-0.379288)--(4.332997,-0.365974)--
    (4.363497,-0.352263)--(4.392145,-0.338136)--(4.418942,-0.323568)--(4.443888,-0.308532)--
    (4.466984,-0.292993)--(4.488232,-0.276909)--(4.496098,-0.270144)--(4.6,0)--
    (4.496098,0.270144)--(4.488232,0.276909)--(4.466984,0.292993)--(4.443888,0.308532)--
    (4.418942,0.323568)--(4.392145,0.338136)--(4.363497,0.352263)--(4.332997,0.365974)--
    (4.300643,0.379288)--(4.266432,0.392222)--(4.230364,0.404792)--(4.192434,0.417009)--
    (4.152641,0.428885)--(4.110981,0.440428)--(4.067449,0.451648)--(4.02204,0.462551)--
    (3.974749,0.473143)--(3.925569,0.483429)--(3.874492,0.493416)--(3.82151,0.503106)--
    (3.766612,0.512503)--(3.709787,0.521612)--(3.651022,0.530433)--(3.590302,0.538971)--
    (3.52761,0.547226)--(3.462926,0.555202)--(3.396228,0.5629)--(3.327491,0.57032)--
    (3.256686,0.577465)--(3.18378,0.584336)--(3.108736,0.590933)--(3.03151,0.597258)--
    (2.952055,0.60331)--(2.870313,0.609091)--(2.786221,0.614602)--(2.699703,0.619842)--
    (2.610673,0.624812)--(2.519032,0.629513)--(2.424661,0.633944)--(2.327422,0.638106)--
    (2.227152,0.641999)--(2.123655,0.645623)--(2.016697,0.648979)--(1.905989,0.652066)--
    (1.791177,0.654884)--(1.671813,0.657434)--(1.547325,0.659715)--(1.416959,0.661728)--
    (1.27969,0.663473)--(1.134063,0.664949)--(0.977896,0.666157)--(0.807614,0.667096)--
    (0.616531,0.667767)--(0.388468,0.66817)--(0,0.668304)--(-0.388468,0.66817)--
    (-0.616531,0.667767)--(-0.807614,0.667096)--(-0.977896,0.666157)--(-1.134063,0.664949)--
    (-1.27969,0.663473)--(-1.416959,0.661728)--(-1.547325,0.659715)--(-1.671813,0.657434)--
    (-1.791177,0.654884)--(-1.905989,0.652066)--(-2.016697,0.648979)--(-2.123655,0.645623)--
    (-2.227152,0.641999)--(-2.327422,0.638106)--(-2.424661,0.633944)--(-2.519032,0.629513)--
    (-2.610673,0.624812)--(-2.699703,0.619842)--(-2.786221,0.614602)--(-2.870313,0.609091)--
    (-2.952055,0.60331)--(-3.03151,0.597258)--(-3.108736,0.590933)--(-3.18378,0.584336)--
    (-3.256686,0.577465)--(-3.327491,0.57032)--(-3.396228,0.5629)--(-3.462926,0.555202)--
    (-3.52761,0.547226)--(-3.590302,0.538971)--(-3.651022,0.530433)--(-3.709787,0.521612)--
    (-3.766612,0.512503)--(-3.82151,0.503106)--(-3.874492,0.493416)--(-3.925569,0.483429)--
    (-3.974749,0.473143)--(-4.02204,0.462551)--(-4.067449,0.451648)--(-4.110981,0.440428)--
    (-4.152641,0.428885)--(-4.192434,0.417009)--(-4.230364,0.404792)--(-4.266432,0.392222)--
    (-4.300643,0.379288)--(-4.332997,0.365974)--(-4.363497,0.352263)--(-4.392145,0.338136)--
    (-4.418942,0.323568)--(-4.443888,0.308532)--(-4.466984,0.292993)--(-4.488232,0.276909)--cycle}
\def\ovClassThinK{(0,0) ellipse[x radius=.2622,y radius=.038093328]}
\begin{tikzpicture}[ov figure,scale=1.32]
  \node[anchor=west] at (.1,6.9) {(a)};
  \node at (6.1,6.9) {$\caZ(X_1)+\caZ(X_2)$};
  \begin{scope}[shift={(6.1,5.2)}]
    \draw[ov shell] \ovClassSum;
    \draw[ov body] \ovClassSumCore;
  \end{scope}
  \begin{scope}[shift={(3.2,3.55)}]
    \draw[draw=ovgray!85,line width=.65pt] \ovClassK;
  \end{scope}
  \node[right,inner sep=2pt] at (3.55,3.55) {$D_1=sH_1$};

  \node[anchor=west] at (.1,2.9) {(b)};
  \node at (1.1,2.9) {$\caZ(X_1)$};
  \begin{scope}[shift={(1.1,1.6)}]
    \draw[ov shell] \ovClassZ;
    \draw[ov body] \ovClassCore;
  \end{scope}
  \begin{scope}[shift={(1.1,.45)}]
    \draw[draw=ovgray!85,line width=.65pt] \ovClassK;
  \end{scope}
  \node at (1.1,-.02) {$D_1=sH_1$};
  \node at (2.25,1.6) {$+$};
  \fill[ovgray!8] (8.05,.6) rectangle (9.274,2.6);
  \begin{scope}[shift={(7.45,1.6)}]
    \draw[ov shell,line width=.45pt,fill=ovamber!23] \ovClassThin;
    \draw[ov body,line width=.45pt] \ovClassThinCore;
    \begin{scope}
      \clip (.6,-1) rectangle (1.824,1);
      \fill[ovblue!35] \ovClassThinCore;
    \end{scope}
    \draw[draw=ovblue,line width=.45pt] \ovClassThinCore;
  \end{scope}
  \begin{scope}[shift={(4.95,.45)}]
    \draw[draw=ovgray!85,line width=.4pt] \ovClassThinK;
  \end{scope}
  \node at (4.95,-.02) {$D_2=sH_2$};
  \node at (4.95,2.55) {$\caZ(X_2)$};
  \draw[draw=ovgray!90,line width=.8pt,dashed] (8.05,.6)--(8.05,2.6);
  \draw[draw=ovgray!90,line width=.8pt,dashed] (9.274,.6)--(9.274,2.6);
  \node at (8.662,2.9) {$f^{-1}(I_2)$};
  \draw[ov arrow] (2.85,-.35)--(12.15,-.35);
  \node[right] at (12.15,-.35) {$f$};
  \foreach \xx in {8.05,9.274}
    \draw[ov guide] (\xx,-.25)--(\xx,-.45);
  \node[below,inner sep=2pt] at (8.05,-.45) {$t-W_1$};
  \node[below,inner sep=2pt] at (9.274,-.45) {$t$};
  \draw[ov guide] (8.05,-.9)--(8.05,-1.05)--(9.274,-1.05)--(9.274,-.9);
  \node[below,inner sep=2pt] at (8.662,-1.05) {$I_2$};
\end{tikzpicture}\relax
\endgroup
}

\newcommand{\OverviewColumnSweep}{\relax
\begin{tikzpicture}[ov figure]
  \node at (1.9,7.65) {(a) Physical shell};
  \begin{scope}[shift={(.25,4.45)},x=.047cm,y=.047cm]
    \draw[ov shell] (0,32)--(32,0)--(64,32)--(32,64)--cycle;
    \draw[draw=ovblue!65,line width=.5pt,fill=ovblue!17]
      (2,32)--(32,2)--(62,32)--(32,62)--cycle;
    \node at (36,35) {$K\ominus D$};
    \draw[draw=ovblue,line width=.9pt,->] (16.5,16)--(16.5,32);
    \fill[ovblue] (16.5,16) circle[radius=1.7pt];
    \fill[ovblue] (16.5,32) circle[radius=1.7pt];
    \node[below left,inner sep=1pt] at (16.5,16) {$\bfz$};
    \node[above left,inner sep=1pt] at (16.5,32) {$\bfz'$};
    \node[right,inner sep=2pt,font=\scriptsize] at (16.5,24) {$+\boldsymbol{\beta}$};
  \end{scope}
  \draw[ov arrow] (3.5,6.1)--(4.1,6.1);
  \node[above,inner sep=2pt] at (3.8,6.1) {$\pi_\Lambda$};

  \node at (8.4,7.65) {(b) Quotient shell};
  \begin{scope}[shift={(4.35,5.4)},x=.13cm,y=.13cm]
    \draw[ov shell] (0,8)--(8,0)--(56,0)--(64,8)--(56,16)--(8,16)--cycle;
    \draw[draw=ovblue!65,line width=.5pt,fill=ovblue!17]
      (2,8)--(10,0)--(54,0)--(62,8)--(54,16)--(10,16)--cycle;
    \draw[ov guide,dashed] (0,0)--(64,0);
    \draw[ov guide,dashed] (0,16)--(64,16);
    \foreach \yy in {0,16}
      \draw[draw=ovgray,line width=.8pt,->] (31,\yy)--(35,\yy);
    \fill[ovblue] (16.5,8) circle[radius=1.8pt];
    \node[above,inner sep=2pt] at (34,8) {$\pi_\Lambda(K\ominus D)$};
    \node[below,inner sep=2pt,font=\scriptsize] at (26,7)
      {$\pi_\Lambda(\bfz)=\pi_\Lambda(\bfz')$};
  \end{scope}
  \node[font=\scriptsize] at (8.5,5.08) {top and bottom identified};
  \node[font=\scriptsize] at (8.5,4.72) {most of the physical shell has a deep translate};

  \node at (6.4,4.08) {(c) Boxes and depth jumps};
  \begin{scope}[x=.2cm,y=.2cm]
    \fill[white] (0,0) rectangle (64,16);
    \foreach \xx/\yy in {8/0,9/0,10/0,11/0,52/0,53/0,54/0,55/0,56/0,7/1,8/1,9/1,10/1,53/1,54/1,55/1,56/1,6/2,7/2,8/2,9/2,54/2,55/2,56/2,57/2,5/3,6/3,7/3,8/3,55/3,56/3,57/3,58/3,4/4,5/4,6/4,7/4,56/4,57/4,58/4,59/4,3/5,4/5,5/5,6/5,57/5,58/5,59/5,60/5,2/6,3/6,4/6,5/6,58/6,59/6,60/6,61/6,1/7,2/7,3/7,4/7,59/7,60/7,61/7,62/7,0/8,1/8,2/8,3/8,4/8,59/8,60/8,61/8,62/8,63/8,1/9,2/9,3/9,4/9,5/9,58/9,59/9,60/9,61/9,62/9,2/10,3/10,4/10,5/10,6/10,57/10,58/10,59/10,60/10,61/10,3/11,4/11,5/11,6/11,7/11,56/11,57/11,58/11,59/11,60/11,4/12,5/12,6/12,7/12,8/12,55/12,56/12,57/12,58/12,59/12,5/13,6/13,7/13,8/13,9/13,54/13,55/13,56/13,57/13,58/13,6/14,7/14,8/14,9/14,10/14,53/14,54/14,55/14,56/14,57/14,7/15,8/15,9/15,10/15,11/15,52/15,53/15,54/15,55/15,56/15}
      \fill[ovamber!38] (\xx,\yy) rectangle ++(1,1);
    \foreach \xx/\yy/\xe in {12/0/13,13/0/14,14/0/16,16/0/19,19/0/23,23/0/29,29/0/38,38/0/45,45/0/49,49/0/51,51/0/52,11/1/12,12/1/13,13/1/15,15/1/18,18/1/22,22/1/28,28/1/37,37/1/45,45/1/49,49/1/51,51/1/52,52/1/53,10/2/11,11/2/12,12/2/14,14/2/17,17/2/21,21/2/27,27/2/36,36/2/45,45/2/50,50/2/52,52/2/53,53/2/54,9/3/10,10/3/11,11/3/13,13/3/16,16/3/20,20/3/26,26/3/35,35/3/45,45/3/50,50/3/53,53/3/54,54/3/55,8/4/9,9/4/10,10/4/12,12/4/15,15/4/19,19/4/25,25/4/34,34/4/45,45/4/51,51/4/54,54/4/55,55/4/56,7/5/8,8/5/9,9/5/11,11/5/14,14/5/18,18/5/24,24/5/33,33/5/45,45/5/51,51/5/54,54/5/56,56/5/57,6/6/7,7/6/8,8/6/10,10/6/13,13/6/17,17/6/23,23/6/32,32/6/45,45/6/52,52/6/55,55/6/57,57/6/58,5/7/6,6/7/7,7/7/9,9/7/12,12/7/16,16/7/22,22/7/31,31/7/45,45/7/52,52/7/56,56/7/58,58/7/59,5/8/6,6/8/7,7/8/9,9/8/12,12/8/16,16/8/22,22/8/31,31/8/45,45/8/52,52/8/56,56/8/58,58/8/59,6/9/7,7/9/8,8/9/10,10/9/13,13/9/17,17/9/23,23/9/32,32/9/45,45/9/52,52/9/55,55/9/57,57/9/58,7/10/8,8/10/9,9/10/11,11/10/14,14/10/18,18/10/24,24/10/33,33/10/45,45/10/51,51/10/54,54/10/56,56/10/57,8/11/9,9/11/10,10/11/12,12/11/15,15/11/19,19/11/25,25/11/34,34/11/45,45/11/51,51/11/54,54/11/55,55/11/56,9/12/10,10/12/11,11/12/13,13/12/16,16/12/20,20/12/26,26/12/35,35/12/45,45/12/50,50/12/53,53/12/54,54/12/55,10/13/11,11/13/12,12/13/14,14/13/17,17/13/21,21/13/27,27/13/36,36/13/45,45/13/50,50/13/52,52/13/53,53/13/54,11/14/12,12/14/13,13/14/15,15/14/18,18/14/22,22/14/28,28/14/37,37/14/45,45/14/49,49/14/51,51/14/52,52/14/53,12/15/13,13/15/14,14/15/16,16/15/19,19/15/23,23/15/29,29/15/38,38/15/45,45/15/49,49/15/51,51/15/52}
      \fill[ovblue!18] (\xx,\yy) rectangle (\xe,\yy+1);
    \draw[step=1,draw=ovgray!30,line width=.12pt] (0,0) grid (64,16);
    \foreach \xx/\yy/\xe in {12/0/13,13/0/14,14/0/16,16/0/19,19/0/23,23/0/29,29/0/38,38/0/45,45/0/49,49/0/51,51/0/52,11/1/12,12/1/13,13/1/15,15/1/18,18/1/22,22/1/28,28/1/37,37/1/45,45/1/49,49/1/51,51/1/52,52/1/53,10/2/11,11/2/12,12/2/14,14/2/17,17/2/21,21/2/27,27/2/36,36/2/45,45/2/50,50/2/52,52/2/53,53/2/54,9/3/10,10/3/11,11/3/13,13/3/16,16/3/20,20/3/26,26/3/35,35/3/45,45/3/50,50/3/53,53/3/54,54/3/55,8/4/9,9/4/10,10/4/12,12/4/15,15/4/19,19/4/25,25/4/34,34/4/45,45/4/51,51/4/54,54/4/55,55/4/56,7/5/8,8/5/9,9/5/11,11/5/14,14/5/18,18/5/24,24/5/33,33/5/45,45/5/51,51/5/54,54/5/56,56/5/57,6/6/7,7/6/8,8/6/10,10/6/13,13/6/17,17/6/23,23/6/32,32/6/45,45/6/52,52/6/55,55/6/57,57/6/58,5/7/6,6/7/7,7/7/9,9/7/12,12/7/16,16/7/22,22/7/31,31/7/45,45/7/52,52/7/56,56/7/58,58/7/59,5/8/6,6/8/7,7/8/9,9/8/12,12/8/16,16/8/22,22/8/31,31/8/45,45/8/52,52/8/56,56/8/58,58/8/59,6/9/7,7/9/8,8/9/10,10/9/13,13/9/17,17/9/23,23/9/32,32/9/45,45/9/52,52/9/55,55/9/57,57/9/58,7/10/8,8/10/9,9/10/11,11/10/14,14/10/18,18/10/24,24/10/33,33/10/45,45/10/51,51/10/54,54/10/56,56/10/57,8/11/9,9/11/10,10/11/12,12/11/15,15/11/19,19/11/25,25/11/34,34/11/45,45/11/51,51/11/54,54/11/55,55/11/56,9/12/10,10/12/11,11/12/13,13/12/16,16/12/20,20/12/26,26/12/35,35/12/45,45/12/50,50/12/53,53/12/54,54/12/55,10/13/11,11/13/12,12/13/14,14/13/17,17/13/21,21/13/27,27/13/36,36/13/45,45/13/50,50/13/52,52/13/53,53/13/54,11/14/12,12/14/13,13/14/15,15/14/18,18/14/22,22/14/28,28/14/37,37/14/45,45/14/49,49/14/51,51/14/52,52/14/53,12/15/13,13/15/14,14/15/16,16/15/19,19/15/23,23/15/29,29/15/38,38/15/45,45/15/49,49/15/51,51/15/52}
      \draw[draw=ovblue!80,line width=.4pt]
        (\xx,\yy) rectangle (\xe,\yy+1);
    \draw[draw=black!75,line width=.85pt]
      (0,8)--(8,0)--(56,0)--(64,8)--(56,16)--(8,16)--cycle;
    \draw[ov guide,dashed] (0,0)--(64,0);
    \draw[ov guide,dashed] (0,16)--(64,16);
    \foreach \yy in {0,16}
      \draw[draw=ovgray,line width=1pt,->] (31,\yy)--(35,\yy);
    \draw[draw=ovblue!65!black,line width=1.2pt] (16,8) rectangle (22,9);
    \fill[ovblue!65!black] (16.5,8) circle[radius=1.6pt];
    \draw[draw=ovblue!80,line width=.55pt] (19,9)--(19,18.1);
    \node[above,inner sep=1pt,font=\scriptsize] at (19,18.1)
      {$t=7$: skip six boxes};
    \draw[ov arrow] (44,18)--(56,18);
    \node[above,inner sep=1pt,font=\scriptsize] at (50,18)
      {enumeration direction};
  \end{scope}
  \draw[fill=ovblue!18,draw=ovblue!80,line width=.4pt]
    (.5,-.64) rectangle (1.15,-.39);
  \node[right,inner sep=3pt,font=\scriptsize] at (1.15,-.51) {deep: skip interval};
  \draw[fill=ovamber!38,draw=ovamber!95!black,line width=.4pt]
    (5,-.64) rectangle (5.25,-.39);
  \node[right,inner sep=3pt,font=\scriptsize] at (5.25,-.51) {shallow: mark box};
  \draw[fill=white,draw=ovgray!50,line width=.4pt]
    (9.2,-.64) rectangle (9.45,-.39);
  \node[right,inner sep=3pt,font=\scriptsize] at (9.45,-.51) {unoccupied};
\end{tikzpicture}\relax
}

\newcommand{\OverviewBallQuotientShell}{\relax
\begin{tikzpicture}[ov figure]
  \pgfdeclareradialshading{ovfigeightgold}{\pgfpoint{-9bp}{10bp}}{
    color(0bp)=(ovamber!22!white);
    color(14bp)=(ovamber!65!white);
    color(28bp)=(ovamber);
    color(40bp)=(ovamber!82!black);
    color(50bp)=(ovamber!82!black)}
  \pgfdeclareradialshading{ovfigeightcore}{\pgfpoint{-9bp}{10bp}}{
    color(0bp)=(ovblue!5!white);
    color(12bp)=(ovblue!13!white);
    color(25bp)=(ovblue!27!white);
    color(38bp)=(ovblue!43!white);
    color(50bp)=(ovblue!43!white)}
  \shade[shading=ovfigeightgold] (2.15,1.95) circle[radius=1.6];
  \shade[shading=ovfigeightcore] (2.15,1.95) circle[radius=1.408];
  \draw[ovamber!90!black,line width=.75pt] (2.15,1.95) circle[radius=1.6];
  \draw[ovblue!70,line width=.55pt] (2.15,1.95) circle[radius=1.408];
  \draw[ovamber!55!black,densely dashed,line width=.45pt]
    (3.75,1.95) arc[start angle=0,end angle=180,x radius=1.6,y radius=.49];
  \draw[ovamber!70!black,line width=.45pt]
    (.55,1.95) arc[start angle=180,end angle=360,x radius=1.6,y radius=.49];
  \draw[ovblue!75,densely dashed,line width=.5pt]
    (2.15,1.95) ellipse[x radius=1.408,y radius=.4312];
  \draw[ovamber!85!black,line width=.5pt] (3.24,3.12)--(3.66,3.36);
  \node[right,inner sep=2pt] at (3.66,3.36) {$\caZ(X)$};
  \draw[ovblue!70,line width=.5pt] (3.38,2.62)--(3.95,2.83);
  \node[right,inner sep=2pt] at (3.95,2.83) {$\caZ(X)\ominus D$};
  \begin{scope}[xshift=1.3cm]
  \path[shade,top color=ovamber!70!white,bottom color=ovamber!90!black,shading angle=35,draw=ovamber!90!black,line width=.45pt]
    (6.1404,1.3563)--(6.2354,1.2056)--(6.3304,1.0799)--(6.4254,0.9726)--(6.5204,0.8803) --
    (6.6154,0.8007)--(6.7103,0.7326)--(6.8053,0.6751)--(6.9003,0.6278)--(6.9953,0.5906) --
    (7.0903,0.5636)--(7.1853,0.5473)--(7.2803,0.5424)--(7.3753,0.5503)--(7.4703,0.5731) --
    (7.5653,0.6143)--(7.6602,0.6805)--(7.6602,0.6805)--(7.7039,0.6738)--(7.7477,0.6730) --
    (7.7914,0.6773)--(7.8351,0.6863)--(7.8788,0.6998)--(7.9225,0.7174)--(7.9662,0.7391) --
    (8.0099,0.7647)--(8.0536,0.7943)--(8.0973,0.8278)--(8.1410,0.8654)--(8.1848,0.9071) --
    (8.2285,0.9533)--(8.2722,1.0042)--(8.3159,1.0603)--(8.3596,1.1222)--(8.3596,1.1222) --
    (8.2646,1.0560)--(8.1696,1.0147)--(8.0746,0.9919)--(7.9796,0.9840)--(7.8846,0.9889) --
    (7.7897,1.0053)--(7.6947,1.0323)--(7.5997,1.0695)--(7.5047,1.1168)--(7.4097,1.1742) --
    (7.3147,1.2424)--(7.2197,1.3219)--(7.1247,1.4143)--(7.0297,1.5216)--(6.9347,1.6473) --
    (6.8398,1.7980)--(6.8398,1.7980)--(6.7961,1.7361)--(6.7523,1.6800)--(6.7086,1.6291) --
    (6.6649,1.5829)--(6.6212,1.5412)--(6.5775,1.5036)--(6.5338,1.4701)--(6.4901,1.4405) --
    (6.4464,1.4149)--(6.4027,1.3932)--(6.3590,1.3756)--(6.3152,1.3621)--(6.2715,1.3531) --
    (6.2278,1.3488)--(6.1841,1.3496)--(6.1404,1.3563) --cycle;
  \path[shade,top color=ovamber!90!white,bottom color=ovamber!85!black,shading angle=20,draw=ovamber!90!black,line width=.4pt]
    (7.6602,0.4318)--(7.7039,0.4338)--(7.7477,0.4397)--(7.7914,0.4493)--(7.8351,0.4624) --
    (7.8788,0.4788)--(7.9225,0.4985)--(7.9662,0.5215)--(8.0099,0.5475)--(8.0536,0.5767) --
    (8.0973,0.6090)--(8.1410,0.6445)--(8.1848,0.6832)--(8.2285,0.7253)--(8.2722,0.7709) --
    (8.3159,0.8202)--(8.3596,0.8735)--(8.3596,1.1222)--(8.3159,1.0603)--(8.2722,1.0042) --
    (8.2285,0.9533)--(8.1848,0.9071)--(8.1410,0.8654)--(8.0973,0.8278)--(8.0536,0.7943) --
    (8.0099,0.7647)--(7.9662,0.7391)--(7.9225,0.7174)--(7.8788,0.6998)--(7.8351,0.6863) --
    (7.7914,0.6773)--(7.7477,0.6730)--(7.7039,0.6738)--(7.6602,0.6805) --cycle;
  \path[shade,top color=ovamber!95!white,bottom color=ovamber!65!black,shading angle=35,draw=ovamber!90!black,line width=.4pt]
    (6.1404,1.1076)--(6.2354,0.9825)--(6.3304,0.8735)--(6.4254,0.7777)--(6.5204,0.6934) --
    (6.6154,0.6195)--(6.7103,0.5551)--(6.8053,0.4997)--(6.9003,0.4531)--(6.9953,0.4152) --
    (7.0903,0.3861)--(7.1853,0.3660)--(7.2803,0.3555)--(7.3753,0.3553)--(7.4703,0.3666) --
    (7.5653,0.3912)--(7.6602,0.4318)--(7.6602,0.6805)--(7.5653,0.6143)--(7.4703,0.5731) --
    (7.3753,0.5503)--(7.2803,0.5424)--(7.1853,0.5473)--(7.0903,0.5636)--(6.9953,0.5906) --
    (6.9003,0.6278)--(6.8053,0.6751)--(6.7103,0.7326)--(6.6154,0.8007)--(6.5204,0.8803) --
    (6.4254,0.9726)--(6.3304,1.0799)--(6.2354,1.2056)--(6.1404,1.3563) --cycle;
  \path[shade,top color=ovamber!70!white,bottom color=ovamber!90!black,shading angle=35,draw=ovamber!90!black,line width=.45pt]
    (6.1404,3.0265)--(6.1527,3.0328)--(6.1651,3.0388)--(6.1774,3.0445)--(6.1897,3.0498) --
    (6.2021,3.0549)--(6.2144,3.0597)--(6.2267,3.0642)--(6.2390,3.0684)--(6.2514,3.0723) --
    (6.2637,3.0760)--(6.2760,3.0794)--(6.2884,3.0826)--(6.3007,3.0856)--(6.3130,3.0883) --
    (6.3253,3.0908)--(6.3377,3.0930)--(6.3500,3.0951)--(6.3623,3.0969)--(6.3747,3.0985) --
    (6.3870,3.0999)--(6.3993,3.1011)--(6.4116,3.1021)--(6.4240,3.1030)--(6.4363,3.1036) --
    (6.4486,3.1040)--(6.4610,3.1042)--(6.4733,3.1043)--(6.4856,3.1042)--(6.4979,3.1039) --
    (6.5103,3.1034)--(6.5226,3.1027)--(6.5349,3.1019)--(6.5473,3.1009)--(6.5596,3.0997) --
    (6.5719,3.0983)--(6.5842,3.0968)--(6.5966,3.0952)--(6.6089,3.0933)--(6.6212,3.0913) --
    (6.6336,3.0892)--(6.6459,3.0869)--(6.6582,3.0844)--(6.6705,3.0818)--(6.6829,3.0790) --
    (6.6952,3.0760)--(6.7075,3.0730)--(6.7199,3.0697)--(6.7322,3.0663)--(6.7445,3.0628) --
    (6.7568,3.0591)--(6.7692,3.0552)--(6.7815,3.0512)--(6.7938,3.0471)--(6.8062,3.0428) --
    (6.8185,3.0384)--(6.8308,3.0338)--(6.8398,3.0304)--(6.8431,3.0291)--(6.8555,3.0242) --
    (6.8678,3.0192)--(6.8801,3.0140)--(6.8925,3.0087)--(6.9048,3.0032)--(6.9171,2.9976) --
    (6.9295,2.9919)--(6.9418,2.9859)--(6.9541,2.9799)--(6.9664,2.9737)--(6.9788,2.9674) --
    (6.9911,2.9609)--(7.0034,2.9543)--(7.0158,2.9475)--(7.0281,2.9405)--(7.0404,2.9335) --
    (7.0527,2.9262)--(7.0651,2.9189)--(7.0774,2.9113)--(7.0897,2.9036)--(7.1021,2.8958) --
    (7.1144,2.8878)--(7.1267,2.8797)--(7.1390,2.8714)--(7.1514,2.8629)--(7.1637,2.8543) --
    (7.1760,2.8456)--(7.1884,2.8367)--(7.2007,2.8276)--(7.2130,2.8183)--(7.2253,2.8089) --
    (7.2377,2.7993)--(7.2500,2.7896)--(7.2623,2.7797)--(7.2747,2.7696)--(7.2870,2.7593) --
    (7.2993,2.7489)--(7.3116,2.7383)--(7.3240,2.7275)--(7.3363,2.7165)--(7.3486,2.7054) --
    (7.3610,2.6941)--(7.3733,2.6825)--(7.3856,2.6708)--(7.3979,2.6589)--(7.4103,2.6468) --
    (7.4226,2.6345)--(7.4349,2.6219)--(7.4473,2.6092)--(7.4596,2.5963)--(7.4719,2.5831) --
    (7.4842,2.5697)--(7.4966,2.5561)--(7.5089,2.5422)--(7.5212,2.5282)--(7.5336,2.5138) --
    (7.5459,2.4992)--(7.5582,2.4844)--(7.5705,2.4693)--(7.5829,2.4539)--(7.5952,2.4383) --
    (7.6075,2.4223)--(7.6199,2.4061)--(7.6322,2.3896)--(7.6445,2.3727)--(7.6569,2.3555) --
    (7.6602,2.3507)--(7.6692,2.3620)--(7.6815,2.3771)--(7.6938,2.3920)--(7.7062,2.4066) --
    (7.7185,2.4208)--(7.7308,2.4347)--(7.7432,2.4484)--(7.7555,2.4617)--(7.7678,2.4748) --
    (7.7801,2.4875)--(7.7925,2.5000)--(7.8048,2.5122)--(7.8171,2.5241)--(7.8295,2.5358) --
    (7.8418,2.5472)--(7.8541,2.5583)--(7.8664,2.5691)--(7.8788,2.5797)--(7.8911,2.5901) --
    (7.9034,2.6001)--(7.9158,2.6100)--(7.9281,2.6195)--(7.9404,2.6289)--(7.9527,2.6379) --
    (7.9651,2.6467)--(7.9774,2.6553)--(7.9897,2.6636)--(8.0021,2.6717)--(8.0144,2.6795) --
    (8.0267,2.6871)--(8.0390,2.6944)--(8.0514,2.7015)--(8.0637,2.7083)--(8.0760,2.7149) --
    (8.0884,2.7212)--(8.1007,2.7273)--(8.1130,2.7331)--(8.1253,2.7387)--(8.1377,2.7440) --
    (8.1500,2.7490)--(8.1623,2.7538)--(8.1747,2.7583)--(8.1870,2.7626)--(8.1993,2.7666) --
    (8.2116,2.7703)--(8.2240,2.7738)--(8.2363,2.7769)--(8.2486,2.7798)--(8.2610,2.7824) --
    (8.2733,2.7847)--(8.2856,2.7868)--(8.2979,2.7885)--(8.3103,2.7899)--(8.3226,2.7910) --
    (8.3349,2.7918)--(8.3473,2.7923)--(8.3596,2.7924)--(8.3596,2.7924)--(8.3473,2.8097) --
    (8.3349,2.8266)--(8.3226,2.8432)--(8.3103,2.8596)--(8.2979,2.8756)--(8.2856,2.8913) --
    (8.2733,2.9068)--(8.2610,2.9219)--(8.2486,2.9369)--(8.2363,2.9515)--(8.2240,2.9659) --
    (8.2116,2.9801)--(8.1993,2.9940)--(8.1870,3.0077)--(8.1747,3.0211)--(8.1623,3.0343) --
    (8.1500,3.0473)--(8.1377,3.0601)--(8.1253,3.0727)--(8.1130,3.0851)--(8.1007,3.0973) --
    (8.0884,3.1092)--(8.0760,3.1210)--(8.0637,3.1326)--(8.0514,3.1440)--(8.0390,3.1552) --
    (8.0267,3.1662)--(8.0144,3.1770)--(8.0021,3.1877)--(7.9897,3.1982)--(7.9774,3.2085) --
    (7.9651,3.2186)--(7.9527,3.2285)--(7.9404,3.2383)--(7.9281,3.2480)--(7.9158,3.2574) --
    (7.9034,3.2667)--(7.8911,3.2758)--(7.8788,3.2848)--(7.8664,3.2936)--(7.8541,3.3023) --
    (7.8418,3.3108)--(7.8295,3.3191)--(7.8171,3.3273)--(7.8048,3.3353)--(7.7925,3.3432) --
    (7.7801,3.3509)--(7.7678,3.3585)--(7.7555,3.3659)--(7.7432,3.3732)--(7.7308,3.3803) --
    (7.7185,3.3872)--(7.7062,3.3941)--(7.6938,3.4007)--(7.6815,3.4073)--(7.6692,3.4136) --
    (7.6602,3.4182)--(7.6569,3.4199)--(7.6445,3.4260)--(7.6322,3.4319)--(7.6199,3.4377) --
    (7.6075,3.4433)--(7.5952,3.4489)--(7.5829,3.4542)--(7.5705,3.4594)--(7.5582,3.4645) --
    (7.5459,3.4694)--(7.5336,3.4742)--(7.5212,3.4788)--(7.5089,3.4833)--(7.4966,3.4876) --
    (7.4842,3.4918)--(7.4719,3.4958)--(7.4596,3.4997)--(7.4473,3.5034)--(7.4349,3.5070) --
    (7.4226,3.5105)--(7.4103,3.5137)--(7.3979,3.5169)--(7.3856,3.5199)--(7.3733,3.5227) --
    (7.3610,3.5254)--(7.3486,3.5279)--(7.3363,3.5302)--(7.3240,3.5324)--(7.3116,3.5345) --
    (7.2993,3.5363)--(7.2870,3.5381)--(7.2747,3.5396)--(7.2623,3.5410)--(7.2500,3.5422) --
    (7.2377,3.5433)--(7.2253,3.5442)--(7.2130,3.5449)--(7.2007,3.5454)--(7.1884,3.5458) --
    (7.1760,3.5459)--(7.1637,3.5459)--(7.1514,3.5457)--(7.1390,3.5454)--(7.1267,3.5448) --
    (7.1144,3.5441)--(7.1021,3.5431)--(7.0897,3.5420)--(7.0774,3.5407)--(7.0651,3.5393) --
    (7.0527,3.5378)--(7.0404,3.5362)--(7.0281,3.5345)--(7.0158,3.5328)--(7.0034,3.5309) --
    (6.9911,3.5289)--(6.9788,3.5268)--(6.9664,3.5247)--(6.9541,3.5224)--(6.9418,3.5200) --
    (6.9295,3.5176)--(6.9171,3.5150)--(6.9048,3.5123)--(6.8925,3.5095)--(6.8801,3.5067) --
    (6.8678,3.5037)--(6.8555,3.5006)--(6.8431,3.4974)--(6.8398,3.4965)--(6.8308,3.4941) --
    (6.8185,3.4907)--(6.8062,3.4872)--(6.7938,3.4836)--(6.7815,3.4799)--(6.7692,3.4760) --
    (6.7568,3.4721)--(6.7445,3.4681)--(6.7322,3.4639)--(6.7199,3.4596)--(6.7075,3.4552) --
    (6.6952,3.4507)--(6.6829,3.4461)--(6.6705,3.4414)--(6.6582,3.4365)--(6.6459,3.4316) --
    (6.6336,3.4265)--(6.6212,3.4213)--(6.6089,3.4159)--(6.5966,3.4104)--(6.5842,3.4047) --
    (6.5719,3.3987)--(6.5596,3.3924)--(6.5473,3.3859)--(6.5349,3.3792)--(6.5226,3.3722) --
    (6.5103,3.3649)--(6.4979,3.3574)--(6.4856,3.3497)--(6.4733,3.3417)--(6.4610,3.3334) --
    (6.4486,3.3249)--(6.4363,3.3162)--(6.4240,3.3072)--(6.4116,3.2979)--(6.3993,3.2884) --
    (6.3870,3.2787)--(6.3747,3.2687)--(6.3623,3.2584)--(6.3500,3.2479)--(6.3377,3.2371) --
    (6.3253,3.2260)--(6.3130,3.2147)--(6.3007,3.2032)--(6.2884,3.1913)--(6.2760,3.1792) --
    (6.2637,3.1668)--(6.2514,3.1541)--(6.2390,3.1411)--(6.2267,3.1279)--(6.2144,3.1143) --
    (6.2021,3.1005)--(6.1897,3.0863)--(6.1774,3.0719)--(6.1651,3.0571)--(6.1527,3.0420) --
    (6.1404,3.0265) --cycle;
  \path[shade,top color=ovamber!90!white,bottom color=ovamber!85!black,shading angle=20,draw=ovamber!90!black,line width=.4pt]
    (7.6602,2.3507)--(7.7039,2.4040)--(7.7477,2.4533)--(7.7914,2.4989)--(7.8351,2.5410) --
    (7.8788,2.5797)--(7.9225,2.6152)--(7.9662,2.6475)--(8.0099,2.6767)--(8.0536,2.7027) --
    (8.0973,2.7256)--(8.1410,2.7454)--(8.1848,2.7618)--(8.2285,2.7749)--(8.2722,2.7845) --
    (8.3159,2.7904)--(8.3596,2.7924)--(8.3596,2.5437)--(8.3159,2.5504)--(8.2722,2.5512) --
    (8.2285,2.5469)--(8.1848,2.5379)--(8.1410,2.5244)--(8.0973,2.5068)--(8.0536,2.4851) --
    (8.0099,2.4595)--(7.9662,2.4299)--(7.9225,2.3964)--(7.8788,2.3588)--(7.8351,2.3171) --
    (7.7914,2.2709)--(7.7477,2.2200)--(7.7039,2.1639)--(7.6602,2.1020) --cycle;
  \path[shade,top color=ovamber!95!white,bottom color=ovamber!65!black,shading angle=35,draw=ovamber!90!black,line width=.4pt]
    (6.1404,3.0265)--(6.2354,3.0672)--(6.3304,3.0917)--(6.4254,3.1030)--(6.5204,3.1028) --
    (6.6154,3.0923)--(6.7103,3.0722)--(6.8053,3.0431)--(6.9003,3.0052)--(6.9953,2.9586) --
    (7.0903,2.9033)--(7.1853,2.8389)--(7.2803,2.7649)--(7.3753,2.6807)--(7.4703,2.5849) --
    (7.5653,2.4758)--(7.6602,2.3507)--(7.6602,2.1020)--(7.5653,2.2527)--(7.4703,2.3784) --
    (7.3753,2.4857)--(7.2803,2.5781)--(7.1853,2.6576)--(7.0903,2.7258)--(6.9953,2.7832) --
    (6.9003,2.8305)--(6.8053,2.8677)--(6.7103,2.8947)--(6.6154,2.9111)--(6.5204,2.9160) --
    (6.4254,2.9081)--(6.3304,2.8853)--(6.2354,2.8440)--(6.1404,2.7778) --cycle;
  \end{scope}
  \begin{scope}[xshift=.65cm]
  \fill[ovgray!90] (4.45,1.65) circle[radius=1.2pt];
  \draw[->,ovgray!90,line width=.7pt] (4.45,1.65)--(5.44,1.21);
  \node[right,inner sep=2pt] at (5.44,1.21) {$\boldsymbol{\ell}_1$};
  \draw[->,ovgray!90,line width=.7pt] (4.45,1.65)--(4.90,1.94);
  \node[above right,inner sep=2pt] at (4.90,1.94) {$\boldsymbol{\ell}_2$};
  \end{scope}
  \node[align=center] at (2.15,-.02) {$\caZ(X)\setminus(\caZ(X)\ominus D)$};
  \node[align=center] at (8.55,-.02)
    {$\pi_\Lambda(\caZ(X))\setminus\pi_\Lambda(\caZ(X)\ominus D)$};
\end{tikzpicture}\relax
}

\subsection{Geometry of subset sums: random models}

\label{subsec:overview-motivation}
From a structure-versus-randomness perspective, a natural heuristic question is whether
\emph{random} instances of \textsc{Subset Sum} are easy to solve.
The answer turns out to be yes, at least for the simplest model with independent uniformly random input integers $A=\{a_1,\ldots,a_n\}\subseteq [1,\Delta]$ (where $\Delta\le n^{\Oh(1)}$).
It is well known that, for some threshold $t_0=\tO(\Delta)$,
the middle interval $[t_0,(\sum_i a_i)-t_0]\cap\Z$ is completely filled by the 
attainable subset sums $\Sigma(A) \coloneqq  \{\sum_{b\in B} b: B\subseteq A\} $
with high probability (see e.g., \cite{BorgsCP01} for a similar statement); see \cref{fig:overview-random-one-dimensional} for an illustration.
Therefore, it suffices to compute the set of subset sums in the boundary interval 
$[0,t_0]\cap\Z$ (and 
its symmetric counterpart),
which can be done using Bringmann's algorithm \cite{Bringmann17} in $\tO(n+\Delta)$ time.

\begin{figure}[ht]
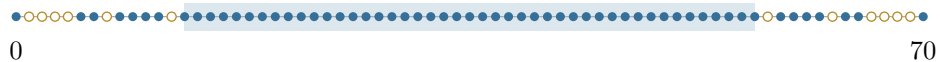

\centering
\OverviewRandomOneDimension
\caption{The subset sums $\Sigma(A)$ (shown as solid blue dots) of a random input $A = \{10,10,8,9,5,11,11,6\}$.}
\label{fig:overview-random-one-dimensional}
\end{figure}

\begin{figure}[ht]
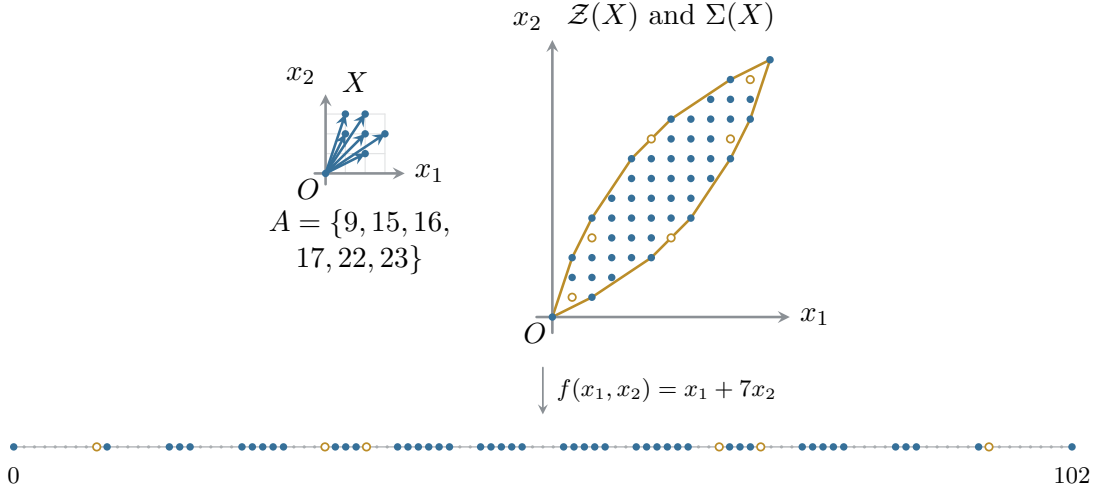

\centering
\OverviewPositiveGeometry
\par\medskip
\OverviewFlattenedGeometry
\caption{
 The subset sums (shown as solid blue dots) of an instance generated from a two-dimensional random model. (For clearer figures, the second coefficient ($7$) of the map $f$ is chosen to be smaller than what is defined in the text ($19$).)}
\label{fig:overview-positive}
\end{figure}

Clearly, such a filling phenomenon cannot always happen for worst-case input instances. 
It is easy to
systematically construct instances where the set of subset sums contains holes throughout.
One way to do so is to embed geometric restrictions: For example, for some small $d\ge 2$, take a size-$n$ subset $X$ of the box $[r]^d$ (with $r\le n^{\Oh(1)}$). For the linear map $f(\bfx)=\sum_{i=1}^d (nr+1)^{i-1} x_i$ injective on $[0,nr]^d$, define $A = f(X)$ as the input set.
Clearly, $\Sigma(A)\subseteq f(\caZ(X)\cap \Z^d)$, where $\caZ(X) \coloneqq \{\sum_{\bfx \in X} a_{\bfx} \bfx : a_{\bfx}\in [0,1]\}\subseteq [0,nr]^d$ is the \emph{zonotope} of $X$.
See an illustration in \cref{fig:overview-positive}.

Nevertheless, if the set $X$ above is randomly chosen, then a similar filling phenomenon still happens: all lattice points sufficiently deep inside $\caZ(X)$ are attainable subset sums of $X$, and are easy to deal with, as illustrated in \cref{fig:overview-positive}. Thus, it suffices to compute the remaining shallow points lying in a shell near the boundary.
Heuristically, $\Delta=\max A\approx n^{d-1} r^d$, and $\caZ(X)$ has surface area $\approx (nr)^{d-1}$. With a depth threshold $\tO_d(r)$, the shell volume is roughly $r(nr)^{d-1}\approx \Delta$, which fits within our time complexity budget.
Later we will sketch how to compute these shallow subset sums in time near-linear in the shell volume.
We remark that one can actually solve the all-target version of \textsc{Subset Sum} in this case: the entire set $\Sigma(X)$, as a compact representation, is computable in near-linear time in $\Delta$.

\subsection{Geometry of subset sums: worst-case input}
\label{subsec:overview-motivation-worstcase}
Remarkably, Conlon, Fox, and Pham
\cite{conlon2023homogeneous} showed that this phenomenon of random models actually occurs for \emph{arbitrary} polynomially large subsets of $[\Delta]$: almost all items in $A$ (except for very few garbage items) can be embedded into a \emph{generalized arithmetic progression (GAP)} $P=f(H\cap \Z^d)$ of constant rank $d$,
with a box $H=\prod_{i=1}^d[-m_i,m_i]$, $m_i\in \Pos$ and a linear map $f$ with integer coefficients,
 so that their subset sums have the same geometric filling phenomenon as illustrated earlier. We now state their result formally.

For $s\ge 1$, we say a GAP $P=f(H\cap \Z^d)$ is \emph{$s$-proper} if $f$ is injective on $sH\cap \Z^d$ (\emph{proper} means $1$-proper). We extract the following notion from the work of \cite{conlon2023homogeneous}:
\begin{definition}[$s$-generating; see \cref{def:generating}]
  \label{def:generating-specialized}
Let $P=f(H\cap \Z^d)$ be a proper GAP. For $s\ge 1$, if a set
$X'\subseteq H\cap \Z^d$ satisfies
\[ \Sigma(X')\supseteq (\boldsymbol{\tau}+sH)\cap \Z^d \]
for some $\boldsymbol{\tau}\in \Z^d$, then we say that $X'$ is
\emph{$s$-generating for} $H\cap \Z^d$ (and that $f(X')$ is
\emph{$s$-generating for} $P$).
See \cref{fig:overview-generating} for an illustration.
\end{definition}

\begin{figure}[ht]
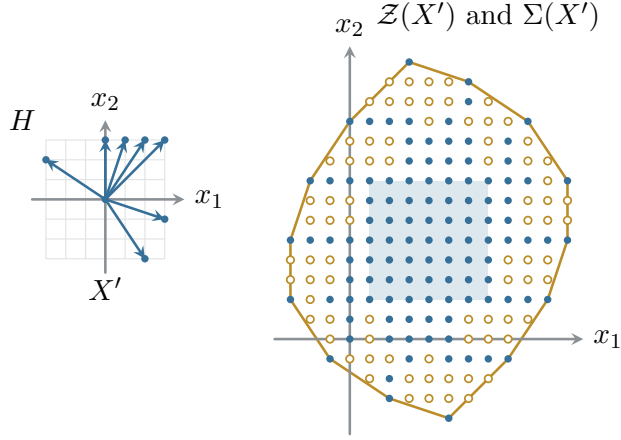

\centering
\OverviewGenerating
\caption{The vector set $X'\subseteq H\cap \Z^2$ is $1$-generating for $H\cap \Z^2$, because every lattice point in $\boldsymbol{\tau}+H$ (shaded) is a subset sum of $X'$.}
\label{fig:overview-generating}
\end{figure}
The following theorem is \cite[Theorem 1.5]{conlon2023homogeneous}:
\begin{theorem}[CFP theorem]
    \label{thm:cfp-intro}
    For any $\beta>1$, there are constants $C>1,d_{\mathrm{max}} \ge 1$ such that the following
    holds.

 Let $A$ be a set of integers in
 $[-\Delta,\Delta]\setminus \{0\}$,
 and $s \in [\Delta^{1/\beta}, C^{-1} \min \{\Delta,\lvert A\rvert/\log \Delta\}]\cap \Z$. 
There exist sets $A'\subsetneq \hat A\subseteq A$ and an $s$-proper origin-symmetric GAP $P\supseteq \hat A $ of rank $d\le d_{\mathrm{max}}$, such that $\lvert A\rvert-\lvert\hat A\rvert\le Cs \log \Delta$ and
 $\lvert A'\rvert\le Cs$, and $A'$ is $s$-generating for $P$.\footnote{We remark that these properties imply $\lvert P\rvert=\Oh_\beta(\Delta/s^{d-1})$ (see e.g., \cite[Proof of Theorem~2]{egyptian}).}
\end{theorem}
In this paper, we give an algorithmic version of \cref{thm:cfp-intro}, which computes the GAP $P$ and the sets $A'\subsetneq \hat A\subseteq A$ in $\tO_\beta(s^2\Delta)$ time. Along the way, we also simplify certain steps from the original proof in \cite{conlon2023homogeneous};
see \cref{sec:algorithmiccfp}.

The utility of \cref{thm:cfp-intro} stems from the following lemma, showing that a small generating set leads to the filling phenomenon; see \cref{fig:overview-filling} for an illustration.  Here, 
$K\ominus D=\{\bfx:\bfx+D\subseteq K\}$ is the \emph{erosion}.
\begin{lemma}[Special case of \cref{lem:innerpoints}]
    \label{lem:innerpoints-specialized}
Let $H\subseteq \R^d$ be an origin-symmetric convex body with $\spa(H\cap \Z^d)=\R^d$
and $X\subseteq H\cap \Z^d$. If a subset $X'\subseteq X$
is $d$-generating for $H\cap \Z^d$, then
$(\caZ(X)\ominus \lvert X'\rvert H)\cap \Z^d\subseteq \Sigma(X)$.
\end{lemma}
The proof of \cref{lem:innerpoints-specialized} follows from a standard rounding argument; see also \cite[Lemmas~30 and 31]{ChenHMZ26} and
\cite[Lemma~13]{pham2025sharp}.

Combining \cref{thm:cfp-intro} and \cref{lem:innerpoints-specialized} shows that, except for $|A|-|\hat A| = \tO_{\beta}(s)$ garbage items, the remaining input set $\hat A$ can be embedded into a GAP $P=f(H\cap \Z^d)$ as
\[
  \hat A=f(X),\,\, A'=f(X'),\,\, \text{ and } X'\subseteq X\subseteq H\cap \Z^d,
\]
with $X'$ $s$-generating for $H\cap \Z^d$,
so that all the deep lattice points in $\caZ(X) \ominus |X'|H$ are mapped by $f$ to attainable subset sums of $\hat A$ (provided $s\ge d$, which is true in our applications).
We remark that Chen, Huang, Mao, and Zhang \cite{ChenHMZ26} established a similar filling phenomenon for subset sums of dense vector sets under certain non-degeneracy conditions.

\begin{figure}[ht]
\centering
\begin{tikzpicture}[ov figure,font=\small,scale=.9]
  \begin{scope}[shift={(.2,1.05)},x=1.3cm,y=1.3cm]
    \fill[ovblue] (0.0000,0.0000) circle[radius=1.65pt];
    \fill[ovblue] (0.1000,0.6600) circle[radius=1.65pt];
    \fill[ovblue] (0.1200,-0.2400) circle[radius=1.65pt];
    \fill[ovblue] (0.2100,0.1800) circle[radius=1.65pt];
    \fill[ovblue] (0.2200,0.4200) circle[radius=1.65pt];
    \fill[ovblue] (0.3100,0.8400) circle[radius=1.65pt];
    \fill[ovblue] (0.3300,-0.0600) circle[radius=1.65pt];
    \fill[ovblue] (0.4200,0.3900) circle[radius=1.65pt];
    \fill[ovblue] (0.4300,0.6000) circle[radius=1.65pt];
    \fill[ovblue] (0.5200,0.0500) circle[radius=1.65pt];
    \fill[ovblue] (0.5200,1.0500) circle[radius=1.65pt];
    \fill[ovblue] (0.5400,0.1500) circle[radius=1.65pt];
    \fill[ovblue] (0.6200,0.7100) circle[radius=1.65pt];
    \fill[ovblue] (0.6300,0.5700) circle[radius=1.65pt];
    \fill[ovblue] (0.6400,-0.1900) circle[radius=1.65pt];
    \fill[ovblue] (0.6400,0.8100) circle[radius=1.65pt];
    \fill[ovblue] (0.7300,0.2300) circle[radius=1.65pt];
    \fill[ovblue] (0.7300,1.2300) circle[radius=1.65pt];
    \fill[ovblue] (0.7400,0.4700) circle[radius=1.65pt];
    \fill[ovblue] (0.7500,0.3300) circle[radius=1.65pt];
    \fill[ovblue] (0.8300,0.8900) circle[radius=1.65pt];
    \fill[ovblue] (0.8500,-0.0100) circle[radius=1.65pt];
    \fill[ovblue] (0.8500,0.9900) circle[radius=1.65pt];
    \fill[ovblue] (0.9400,0.4400) circle[radius=1.65pt];
    \fill[ovblue] (0.9500,0.6500) circle[radius=1.65pt];
    \fill[ovblue] (1.0400,1.1000) circle[radius=1.65pt];
    \fill[ovblue] (1.0600,0.2000) circle[radius=1.65pt];
    \fill[ovblue] (1.1500,0.6200) circle[radius=1.65pt];
    \fill[ovblue] (1.1600,0.8600) circle[radius=1.65pt];
    \fill[ovblue] (1.2500,1.2800) circle[radius=1.65pt];
    \fill[ovblue] (1.2700,0.3800) circle[radius=1.65pt];
    \fill[ovblue] (1.3700,1.0400) circle[radius=1.65pt];
  \end{scope}
  \node at (1.1,-.50) {$\Sigma(X\setminus X')$};
  \node[font=\large] at (2.5,1.6) {$+$};
  \begin{scope}[shift={(4.0,1.7)},x=1.3cm,y=1.3cm]
    \draw[ovblue,line width=.8pt,fill=ovblue!3] (-.6,-.6) rectangle (.6,.6);
    \filldraw[fill=ovblue!45,draw=ovblue,line width=.7pt]
      (-.51,-.16) rectangle (.19,.44);
  \end{scope}
  \node at (4,-.50) {$\Sigma(X')$};
  \draw[ovblue,line width=.5pt] (3.84,2.08)--(3.35,2.86);
  \node[above,inner sep=2pt] at (3.35,2.86) {$\boldsymbol\tau+dH$};
  \draw[ov arrow] (5.25,1.6)--(6.25,1.6);
  \begin{scope}[shift={(7.6,1.05)},x=1.3cm,y=1.3cm]
    \path[ov shell] (-0.6000,-0.6000)--(-0.4800,-0.8400)--(0.7200,-0.8400)--(1.2400,-0.7900)--(1.4500,-0.6100)--(1.8700,-0.2200)--(1.9700,0.4400)--(1.9700,1.6400)--(1.8500,1.8800)--(0.6500,1.8800)--(0.1300,1.8300)--(-0.0800,1.6500)--(-0.5000,1.2600)--(-0.6000,0.6000)--cycle;
    \path[fill=ovblue!45,draw=none] (-0.5100,-0.1600) rectangle (0.1900,0.4400);
    \path[fill=ovblue!45,draw=none] (-0.4100,0.5000) rectangle (0.2900,1.1000);
    \path[fill=ovblue!45,draw=none] (-0.3900,-0.4000) rectangle (0.3100,0.2000);
    \path[fill=ovblue!45,draw=none] (-0.3000,0.0200) rectangle (0.4000,0.6200);
    \path[fill=ovblue!45,draw=none] (-0.2900,0.2600) rectangle (0.4100,0.8600);
    \path[fill=ovblue!45,draw=none] (-0.2000,0.6800) rectangle (0.5000,1.2800);
    \path[fill=ovblue!45,draw=none] (-0.1800,-0.2200) rectangle (0.5200,0.3800);
    \path[fill=ovblue!45,draw=none] (-0.0900,0.2300) rectangle (0.6100,0.8300);
    \path[fill=ovblue!45,draw=none] (-0.0800,0.4400) rectangle (0.6200,1.0400);
    \path[fill=ovblue!45,draw=none] (0.0100,-0.1100) rectangle (0.7100,0.4900);
    \path[fill=ovblue!45,draw=none] (0.0100,0.8900) rectangle (0.7100,1.4900);
    \path[fill=ovblue!45,draw=none] (0.0300,-0.0100) rectangle (0.7300,0.5900);
    \path[fill=ovblue!45,draw=none] (0.1100,0.5500) rectangle (0.8100,1.1500);
    \path[fill=ovblue!45,draw=none] (0.1200,0.4100) rectangle (0.8200,1.0100);
    \path[fill=ovblue!45,draw=none] (0.1300,-0.3500) rectangle (0.8300,0.2500);
    \path[fill=ovblue!45,draw=none] (0.1300,0.6500) rectangle (0.8300,1.2500);
    \path[fill=ovblue!45,draw=none] (0.2200,0.0700) rectangle (0.9200,0.6700);
    \path[fill=ovblue!45,draw=none] (0.2200,1.0700) rectangle (0.9200,1.6700);
    \path[fill=ovblue!45,draw=none] (0.2300,0.3100) rectangle (0.9300,0.9100);
    \path[fill=ovblue!45,draw=none] (0.2400,0.1700) rectangle (0.9400,0.7700);
    \path[fill=ovblue!45,draw=none] (0.3200,0.7300) rectangle (1.0200,1.3300);
    \path[fill=ovblue!45,draw=none] (0.3400,-0.1700) rectangle (1.0400,0.4300);
    \path[fill=ovblue!45,draw=none] (0.3400,0.8300) rectangle (1.0400,1.4300);
    \path[fill=ovblue!45,draw=none] (0.4300,0.2800) rectangle (1.1300,0.8800);
    \path[fill=ovblue!45,draw=none] (0.4400,0.4900) rectangle (1.1400,1.0900);
    \path[fill=ovblue!45,draw=none] (0.5300,0.9400) rectangle (1.2300,1.5400);
    \path[fill=ovblue!45,draw=none] (0.5500,0.0400) rectangle (1.2500,0.6400);
    \path[fill=ovblue!45,draw=none] (0.6400,0.4600) rectangle (1.3400,1.0600);
    \path[fill=ovblue!45,draw=none] (0.6500,0.7000) rectangle (1.3500,1.3000);
    \path[fill=ovblue!45,draw=none] (0.7400,1.1200) rectangle (1.4400,1.7200);
    \path[fill=ovblue!45,draw=none] (0.7600,0.2200) rectangle (1.4600,0.8200);
    \path[fill=ovblue!45,draw=none] (0.8600,0.8800) rectangle (1.5600,1.4800);
    \node[above,inner sep=2pt] at (1.65,1.90) {$\caZ(X)$};
  \end{scope}
  \node at (8.5,-.50) {$\Sigma(X\setminus X')+\Sigma(X')$};
\end{tikzpicture}
\caption{An illustration of
    \cref{lem:innerpoints-specialized}.
When $X$ contains a small $d$-generating subset $X'$, sufficiently deep lattice points of $\caZ(X)$ are in $\Sigma(X)$, since convolving $\Sigma(X\setminus X')$ with $\Sigma(X')$ translates the solid region $(\boldsymbol\tau+dH)\cap\Z^d\subseteq\Sigma(X')$ to cover every lattice point in $\caZ(X)\ominus |X'|H$.}
\label{fig:overview-filling}
\end{figure}
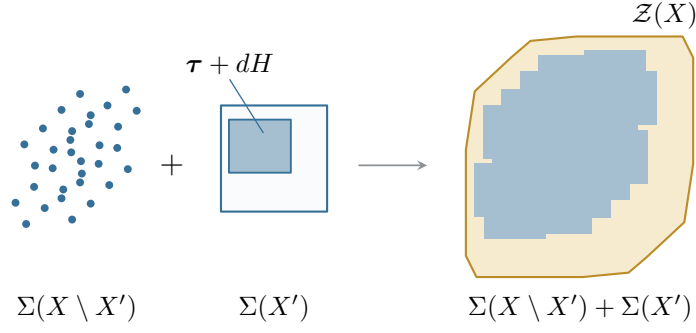

We apply the CFP theorem, with $s=\Delta^\delta$ for small fixed $\delta\in (0,1)$, so $d\le d_{\mathrm{max}} = \Oh_{\delta}(1)$, and the size of the garbage item set $A\setminus \hat A$ is thus $\tO(s) = \tO(\Delta^{\delta})$.
Our target time complexity for \textsc{Subset Sum} is $\Oh(n)+\tO_\delta(\Delta s^{\Oh(1)})$.

As a toy example, if the CFP theorem returns $d=1$, then we can immediately solve the instance as follows.  Let the returned 1D progression be $P = g\cdot \{-L,\dots,L\}$.
Since $A'$ is $s$-generating for $P$ whereas $|A'| = \Oh_\delta(s)$, we have $gL=\Oh_\delta(\Delta)$.
By \cref{lem:innerpoints-specialized}, every multiple of $g$ in the middle interval
$[\lvert A'\rvert gL,\ \sum_{a\in\hat A}a-\lvert A'\rvert gL]$ is a subset sum of $\hat A$.
The two remaining boundary intervals have total length at most
$2\lvert A'\rvert gL=\Oh_\delta(s \Delta)$, so the subset sums of $\hat A$ in them can be computed by Bringmann's algorithm \cite{Bringmann17} in $\tO_\delta(s\Delta)$ time. Finally compute the subset sums of the $\tO_\delta(s)$ garbage items $A\setminus \hat A$ by Bringmann's algorithm in $\tO_\delta(s\Delta)$ time, and merge the results. The overall runtime is $\tO_\delta(n+s \Delta)$ (provided the CFP output is given for free).

Unfortunately, if the CFP theorem returns $d\ge 2$, then the boundary part $\caZ(X)\setminus (\caZ(X)\ominus |X'|H)$ is a shell whose volume may no longer be bounded by $\Oh_{\delta}(s\Delta)$, unlike in the one-dimensional case. This is the main challenge that we have to address in our algorithm.

We briefly pause here to discuss some previous works.
Many pseudopolynomial-time \textsc{Subset Sum} and \textsc{Knapsack}
algorithms rely on long one-dimensional arithmetic progressions in subset
sums (e.g., \cite{GalilM91,BringmannW21,ChenLMZsoda24,clmzknapsackstoc24,Bringmann24,clmzfocs24,Jin24,linchenstoc25,weakapproxicalp,
linchensoda26}) in an analogous way: the utility of these progressions in these works is essentially the
$d=1$ case of \cref{lem:innerpoints-specialized}, sometimes combined with proximity arguments.
The existence of such progressions was usually proved via theorems of Freiman and S\'ark\"ozy
\cite{freiman93,Sarkozy2,fox}, or Szemer\'edi and Vu
\cite[Corollary~5.2]{szemeredivu-apsumsets} under certain density conditions;
for example, \cite{freiman93,Sarkozy2,fox} showed that every set $A\subseteq[\Delta]$ with
$\lvert A\rvert\ge C\sqrt{\Delta}$ (for some absolute constant $C>0$) has a homogeneous arithmetic progression of length $\Delta$ in $\Sigma(A)$.
The failure of this theorem for $|A|= o(\sqrt{\Delta})$ has been the main obstacle to improving many previous algorithms relying on one-dimensional progressions.

Chaimovich suggested algorithmically using higher dimensional structure,
\cite[pp.~350--351]{chaimovichsurvey}, and Szemer\'edi and Vu did extend
the aforementioned theorem to GAPs \cite[Theorem~7.1]{szemeredivu-apsumsets}, but its algorithmic utility is unclear. It turns out that the structural result of Conlon, Fox, and Pham
\cite{conlon2023homogeneous} provides the right higher dimensional tool for the algorithmic breakthrough.

\subsection{Shell volume under idealized assumptions}
\label{subsec:overview-boundary}
We now consider the case where the CFP theorem returns a GAP $f(H\cap \Z^d)$ of rank $d\ge 2$. Recall that $X\subseteq H\cap\Z^d$ encodes the input set $\hat A$ (excluding garbage items) by $f(X) = \hat A$, and the $s$-generating set $X'$ has size $\lvert X'\rvert= \Oh_{\delta}(s) = \Oh_\delta(\Delta^\delta)$. Write $D=|X'|H$.
We are interested in bounding the volume of the shell
\begin{equation}
\label{eqn:intro-shell}
\caZ(X)\setminus (\caZ(X)\ominus D).
\end{equation}
(Strictly speaking, we are interested in the number of its points contained in $\Z^d$, but they are roughly the same for our purposes.)
In general, this shell may have volume much larger than $s \Delta$, unlike in the one-dimensional case. 
Nevertheless, we can bound the volume by $\Oh_\delta(\lambda s\Delta)$ under the following two idealized assumptions.
Here, 
$\caZ_0(X)\coloneqq \{\sum_{\bfx \in X} a_{\bfx} \bfx : a_{\bfx}\in [-1/2,1/2]\}$ denotes the centered zonotope.
\begin{itemize}
  \item \emph{Relative fatness}: for some small fatness parameter $\lambda\ge 1$ (think of $\lambda = \Oh_\delta(1)$ or $\Oh_\delta(\log \Delta)$),
\[ \lvert X\rvert H\subseteq \lambda \caZ_0(X). \]
(Compare with this trivial containment in the other direction: $\caZ_0(X)\subseteq (\lvert X\rvert/2) H$.)
\item \emph{Injectivity}: $f$ is injective on $\caZ(X)\cap \Z^d$.
\end{itemize}

Writing $\fw(S)\coloneqq \sup f(S)-\inf f(S)$, injectivity implies 
\[|f(\caZ(X)\cap \Z^d)| \le 1 +  \fw(\caZ(X)) = 1 + \sum_{\bfx\in X}f(\bfx) \le 1+ |X|\Delta.\]  

Writing
\[\eps\coloneqq \lambda\lvert X'\rvert/\lvert X\rvert=\Oh_\delta(\lambda s/\lvert X\rvert),\]
relative fatness gives
\begin{equation}
  \label{eqn:intro-cazx}
 \caZ_0(X) \ominus D \supseteq (1-\eps)\caZ_0(X),
\end{equation}
so
\[ \vol_d\bigl(\caZ(X)\setminus (\caZ(X)\ominus D)\bigr)
 \le \Oh_\delta(\eps)\,\vol_d(\caZ(X)).\]
Chaining the two inequalities (and using volume to approximate lattice point counts) shows that  the shell 
\eqref{eqn:intro-shell}
has volume $\Oh_{\delta}(\lambda s\Delta)$, as desired (and matching the 1-dimensional toy case).
 
Both assumptions are necessary for the 
$\Oh_{\delta}(\lambda s\Delta)$
shell volume bound. See 
\cref{fig:overview-shell-obstructions} for illustrations of how the volume bound fails when exactly one of the two assumptions fails. 

In the next subsection, we explain the algorithm computing subset sums under the two idealized assumptions. 
Later, we explain the ideas needed to remove these assumptions. 
 
\begin{figure}[ht]
\centering
\OverviewShellObstructions
\caption{Left: relative fatness fails. The shell volume can be almost as large as the volume of the entire zonotope. 
\newline Right: Let $\Z^2\cap\ker f=\Z\boldsymbol{\ell}$, so points differing by integer multiples of $\boldsymbol{\ell}$ represent the same subset sums. Due to this redundancy, the volume of $\caZ(X)$ is much larger than $\lvert f(\Sigma(X))\rvert$. Consequently, the shell volume  is also too large.
Actually, points in the light yellow parts of the shell correspond to attainable subset sums, as they are $\boldsymbol{\ell}$-congruent to points in the deep part $\caZ(X)\ominus D$. So we actually only care about the dark yellow parts of the shell; we shall come back to this in \cref{subsec:overview-quotient-geometry}. }
\label{fig:overview-shell-obstructions}
\end{figure}

\subsection{Computing subset sums under idealized assumptions}
\label{subsec:overview-computing}
Assume both the relative fatness (with fatness parameter $\lambda = \Oh_{\delta}(1)$ for convenience) and the injectivity assumptions. The previous subsection showed that the shell \eqref{eqn:intro-shell} of the zonotope $\caZ(X)$ has volume $\Oh_\delta(s\Delta)$. In this subsection, we sketch how to compute all the subset sums in this shell \eqref{eqn:intro-shell}, in near-linear time in the shell volume.

Again, we can actually solve the all-target version of subset sum in this way, since the deep subset sums inside $\caZ(X)\ominus D$ are easy to handle; see \cref{subsec:overview-implementation}. The reason we would like to solve the all-target instead of single-target version is because in the end we have to merge the subset sums of $\hat A = f(X)$ with the subset sums of the garbage items $A\setminus \hat A$.

\paragraph*{Cap covering.}
We use the economic cap covering of Ewald, Larman, and Rogers, and
B\'ar\'any and Larman \cite{ewaldlarmanrogers1970,baranylarman1988}
(see also \cite[Lemma~3.1]{aryafonsecamount2017}). 
For any convex body $K\subseteq\R^d$ containing $\zero$ and $0<\eps<1$,
it provides a collection of caps $Q_i$ of $K$ which covers $K\setminus(1-\eps)K$, with total cap volume 
\[ \sum_i\vol_d(Q_i)\le\Oh_d(\eps)\vol_d(K). \]
Here, a \emph{cap} is the intersection of $K$ with a half space. 
Taking $K=\caZ_0(X)$, by \cref{eqn:intro-cazx}, we can therefore cover the shell $\caZ(X)\setminus(\caZ(X)\ominus D)$ with near-optimal total cap volume $\Oh_\delta(s\Delta)$; see
\cref{fig:overview-economic-covering-caps} for an illustration.

\begin{figure}[!htb]
\centering
\begin{minipage}[t]{.59\linewidth}
\vspace{0pt}
\centering
\OverviewEconomicCovering
\caption{Left: a cap covering of the shell $\caZ(X)\setminus (\caZ(X)\ominus D)$. Right: a covering of the shell using Macbeath regions.}
\label{fig:overview-economic-covering}
\refstepcounter{overviewpanel}\label{fig:overview-economic-covering-caps}
\end{minipage}\hfill
\begin{minipage}[t]{.39\linewidth}
\vspace{0pt}
\centering
\OverviewCapMacbeath
\caption{Macbeath regions at scales $1$ and $3/2$ centered at $\bfx$.}
\label{fig:overview-cap-macbeath}
\end{minipage}
\end{figure}
\FloatBarrier

We follow the literature to construct such cap coverings using \emph{Macbeath regions}~\cite{macbeath}, defined as follows (see \cref{fig:overview-cap-macbeath}):
\begin{restatable*}[Macbeath regions]{definition}{EuclideanMacbeathDefinition}
\label{def:macbeath-euc}
For a convex body $K\subseteq \R^d$, $\bfx\in K$, and $\lambda\ge 0$, the \emph{$\lambda$-scaled Macbeath region} at $\bfx$ is defined as
\[
 M_K^\lambda(\bfx)
 =\bfx+\lambda\bigl((K-\bfx)\cap (\bfx-K)\bigr).
\]
\end{restatable*}
Macbeath regions and cap coverings are standard tools in polytope approximation, with broad applications in computational geometry; see, e.g., \cite{ewaldlarmanrogers1970,BronnimannCP93,barany2000mregions,aryafonsecamount2017,AbdelkaderM18,AryaM25} and the references therein.
At a high level,  Macbeath regions satisfy a nice expansion-containment property (see \cref{sec:macbeath-regions}) which allows one to greedily construct a net of disjoint Macbeath regions (at scale $1/2$) centered in the shell, so that their constant-factor expansions cover the shell, and their total volume can be bounded using disjointness. Then, approximating Macbeath regions by caps gives the cap covering.

\paragraph*{Subset sums in caps.}
We now sketch how to compute the subset sums $\Sigma(X)\cap Q$ in a cap
$Q=\caZ(X)\cap\{\bfy:\langle\bfu,\bfy\rangle\le a\}$, in near-linear
$\tO_d(|Q\cap\Z^d|)$ time.
Assume $Q$ has sufficient thickness,
which can be guaranteed by the cap covering construction. 
Take an apex of $Q$, namely a vertex $\bfr\in\caZ(X)$ minimizing
$\langle\bfr,\bfu\rangle$ (by integrality, $\bfr\in\Z^d$).
By translation and possibly flipping the signs of items in $X$, we may
assume without loss of generality that $\bfr=\zero$, which implies that
$\langle\bfu,\bfx\rangle\ge 0$ for all $\bfx\in X$.
Then, we may assume without loss of generality that all $\bfx\in X$ are in $Q$ (otherwise, $\bfx$ is useless).

Applying a discrete John-type lemma to $Q-Q$ gives an
invertible linear map $\Phi$ with $\Phi\Z^d=\Z^d$ and a box
$B=\prod_{i=1}^d[-L_i,L_i]$, with $L_i\in\Pos$, such that
$\Phi(Q-Q)\subseteq B\subseteq\kappa_d\Phi(Q-Q)$, where $\kappa_d$
depends only on $d$ (see \cref{lem:latticeboxinandout}).
Since $\zero\in Q$, we have $\Phi(X)\subseteq\Phi(Q)\subseteq B$.
$|B\cap\Z^d|=\Oh_d(|Q\cap\Z^d|)$.

Partition $X=X_1\sqcup\cdots\sqcup X_{2^d}$ according to the orthants
containing $\Phi\bfx$, assigning zero coordinates consistently.
For each $j$, compute $S_j=\Sigma(\Phi X_j)\cap B$.
After reflecting the negative coordinates, this is the problem of computing
subset sums of nonnegative integer vectors inside $\prod_i[0,L_i]$,
which takes $\tO_d(|B\cap\Z^d|)$ time (\cref{lem:box-subset-sums}).
This generalizes Bringmann's algorithm \cite{Bringmann17} (derandomized by
Chan~\cite{chan26}), for computing $\Sigma(A)\cap[0,t]$ for a positive
integer input set $A$ in $\tO(|A|+t)$ time; the black-box reduction is
proved in \cref{app:box-subset-sums}.

If $Y\subseteq X$ has sum in $Q$, then each $Y\cap X_j$ also has sum
in $Q$, because all $\langle\bfu,\bfx\rangle$ are nonnegative.
Thus each transformed partial sum belongs to $S_j$.
Conversely, every element of $S_1+\cdots+S_{2^d}$ is a subset sum of
$\Phi(X)$. Hence keeping exactly those $\bfz\in S_1+\cdots+S_{2^d}$
with $\langle\bfu,\Phi^{-1}\bfz\rangle\le a$ gives
$\Phi(\Sigma(X)\cap Q)$.
We compute the sumset by convolution inside $2^dB$, which has
$\Oh_d(|B\cap\Z^d|)$ lattice points, and then apply $\Phi^{-1}$ to
the selected sums. The total time is $\tO_d(|Q\cap\Z^d|)$.

\subsection{Relaxing the injectivity assumption}
\label{subsec:overview-quotient-geometry}
We now discuss how to relax the assumption that $f$ is injective on $\caZ(X)\cap \Z^d$.
In this case, we can no longer bound the volume of the zonotope $\caZ(X)$ using the number of distinct integer values in $f(\Sigma(X))$, since different points in the zonotope may map to the same value under $f$.
But, in this case, many points $\bfx$ in the shell \eqref{eqn:intro-shell} are also known to map to attainable subset sums, as long as there is another deep point $\bfx'\in \caZ(X)\ominus D$ such that $f(\bfx) = f(\bfx')$. Therefore, we are actually interested in a subset of the shell.

To develop this idea, define the lattice $\Lambda=\Z^d\cap\ker f$, so that $f(\bfx) = f(\bfx')$ for two points $\bfx,\bfx' \in \caZ(X)\cap \Z^d$ precisely when $\bfx-\bfx'\in \Lambda$.
Write $\pi_\Lambda \colon \R^d\to \R^d/\Lambda$ for the quotient map, and let $\vol_d$ also denote the volume in $\R^d/\Lambda$. Therefore, we actually only care about the following \emph{quotient shell}:
\[
 \pi_\Lambda(\caZ(X))\setminus\pi_\Lambda(\caZ(X)\ominus D).
\]
See 
the right part of \cref{fig:overview-shell-obstructions} and \cref{fig:overview-ball-shell} for illustrations in 2D and 3D, respectively.

\begin{figure}[!htb]
\centering
\OverviewBallQuotientShell
\caption{The full shell $\caZ(X)\setminus(\caZ(X)\ominus D)$ (left) and representatives of
the quotient shell (right), modulo $\Lambda=\Z\boldsymbol{\ell}_1+\Z\boldsymbol{\ell}_2$.}
\label{fig:overview-ball-shell}
\end{figure}

The quotient shell satisfies a volume bound similar to that in
\cref{subsec:overview-boundary}: whenever $\caZ(X)$ contains a translate of $\eps^{-1}D$,
\begin{equation}
  \label{eqn:quotientvolbound}
 \vol_d\bigl(\pi_\Lambda(\caZ(X))\setminus\pi_\Lambda(\caZ(X)\ominus D)\bigr)
 \le\Oh_d(\eps)\,\vol_d(\pi_\Lambda \caZ(X)).
\end{equation}
We prove this in \cref{lem:quotient-shell-bound}.
Since $f$ has integer coefficients, the volume of $\pi_\Lambda(\caZ(X))$
is at most the length of the interval $f(\caZ(X))$, which is at most
$\lvert X\rvert\Delta$.
Taking $\eps=\lambda\lvert X'\rvert/\lvert X\rvert$ in \eqref{eqn:quotientvolbound}
gives a shell volume of $\Oh_\delta(\lambda s\Delta)$.

Then, to apply the idea from the previous subsection, we have to develop a quotient variant of cap covering, which only needs the caps to cover the useful part of $\caZ(X)$ that could map to the quotient shell under $\pi_{\Lambda}$, and requires that the total cap volume (measured in $\R^d/\Lambda$) is not much more than the quotient volume bound in \cref{eqn:quotientvolbound}. 
See \cref{sec:quotient-cap-covering}.

We also need to adjust the previous algorithm that computes subset sums in a cap, since the Euclidean volume of the cap $Q$ may be much greater than its volume in $\R^d/\Lambda$.
When this happens, 
we handle this issue by
reparameterizing the cap into a lower-dimensional box of smaller volume, in a way similar to the rank reduction of improper GAPs; see \cref{lem:local-properization}.

\subsection{Relaxing the fatness assumption}
\label{subsec:overview-classification}
\begin{figure}[!htb]
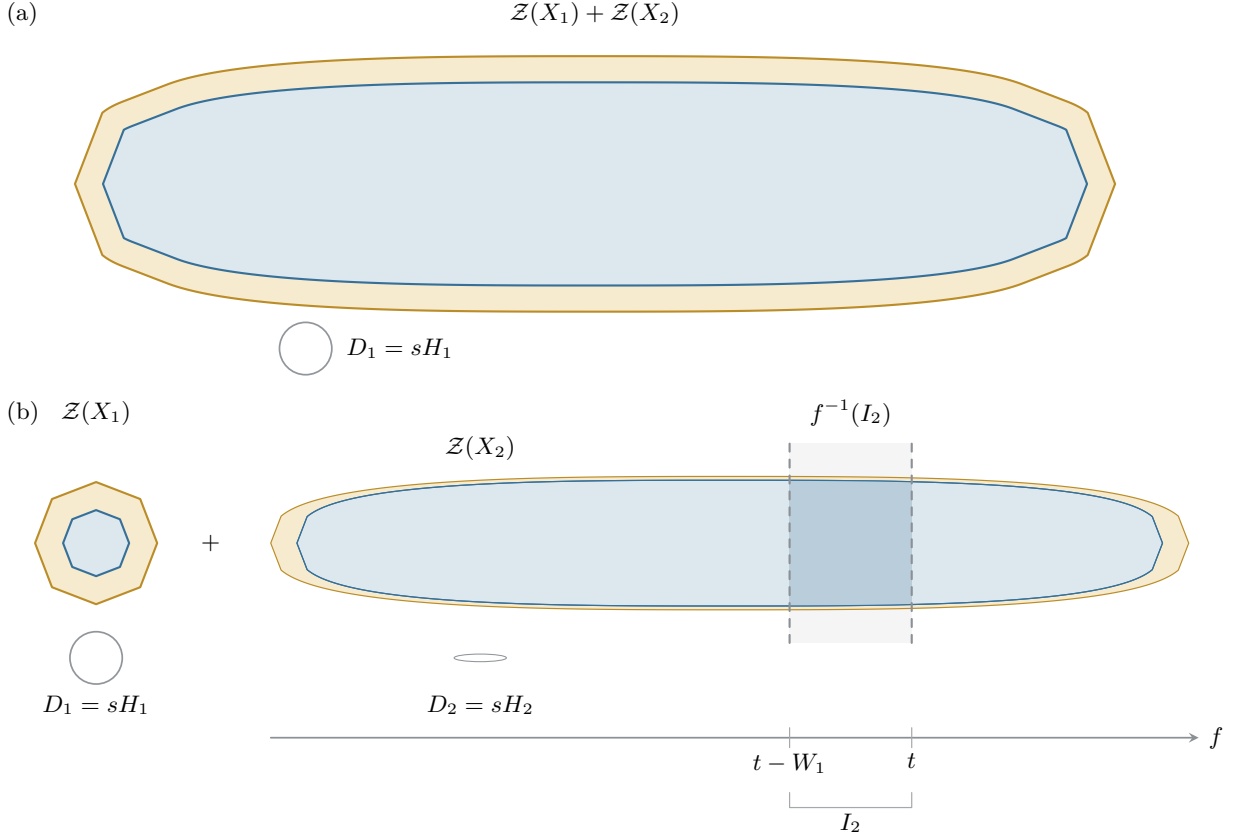

\centering
\OverviewClassClipping
\caption{(a) The shell of $\caZ(X_1)+\caZ(X_2)$ with thickness $sH_1$.
(b) The shells of $\caZ(X_i)$ with thickness $sH_i$. We restrict $\caZ(X_2)$ to
values in $I_2=[t-W_1,t]\cap f(\caZ(X_2))$.}
\label{fig:overview-class-clipping}
\refstepcounter{overviewpanel}\label{fig:overview-class-clipping-full}
\refstepcounter{overviewpanel}\label{fig:overview-class-clipping-restricted}
\end{figure}

We now discuss how to deal with the case where relative fatness fails.
To illustrate the idea, consider the toy case in \cref{fig:overview-class-clipping}, where $X$ is decomposed as
$X=X_1\sqcup X_2$, with $X_1,X_2\subseteq\Z^2$. The set $X_1$ has relatively few
vectors pointing in different directions, while $X_2$ has many more
vectors, all nearly horizontal. Write $H_i=\conv(\pm X_i)$, and
suppose $H_1$ is roughly round and $H_2\subseteq H_1$ is very thin
vertically. Assume each $X_i$ generates $\Z^2$, has an $s$-generating
subset for $H_i\cap\Z^2$ of size $\Oh_\delta(s)$, and satisfies relative
fatness with respect to $H_i$, with a small parameter $\lambda$.

Suppose $\caZ(X_2)$ is much longer horizontally than $\caZ(X_1)$
but no taller. Relative fatness for $X_1\sqcup X_2$ would
require growth in both directions proportional to
$\lvert X_1\rvert+\lvert X_2\rvert$, since
$\conv(\pm(X_1\sqcup X_2))=H_1$. But the height of
$\caZ(X_1)+\caZ(X_2)$ remains comparable to
that of $\caZ(X_1)$. Thus the union can fail relative fatness even though
$X_1$ and $X_2$ separately satisfy it, and the shell with thickness
$sH_1$ can be too large; see \cref{fig:overview-class-clipping-full}.

We first compute $f(\Sigma(X_1))$ and $f(\Sigma(X_2))$ separately (with shell 
thicknesses $sH_1$ and $sH_2$), using the algorithm discussed in \cref{subsec:overview-quotient-geometry}.
Then, we combine the results as follows. With the target sum $t$ in mind, we
only keep the sums of $X_2$ that could contribute to $t$. Indeed, set
$W_1=\sum_{\bfx\in X_1}f(\bfx)$. Every subset sum from $X_1$ has
value in $[0,W_1]$, so the relevant values from $X_2$ lie in
$I_2=[t-W_1,t]\cap f(\caZ(X_2))$; see
\cref{fig:overview-class-clipping-restricted}.

We apply the cap covering to
\[
 K=(\caZ(X_1)+\caZ(X_2))\cap
 f^{-1}\bigl(I_2+[0,W_1]\bigr).
\]
This contains $\caZ(X_1)+(\caZ(X_2)\cap f^{-1}(I_2))$, so no pair
of subset sums contributing to $t$ is lost. When $I_2$ is nonempty,
$K$ contains a translate of $\caZ(X_1)$, and its value width is at
most $2W_1\le 2\lvert X_1\rvert\Delta$. Relative fatness of $X_1$
lets us use $\eps=\lambda s/\lvert X_1\rvert$ in the quotient-shell
bound from \cref{subsec:overview-quotient-geometry}. Together with the
width bound, this gives a shell volume of $\Oh_\delta(\lambda s\Delta)$
for thickness $sH_1$. Restricting to $I_2$ is what removes
the dependence on the much larger $\lvert X_2\rvert$.

\paragraph*{Merging groups.}
We now explain how the cap covering helps combine the subset sums of
$X_1$ and $X_2$. Write $K_1=\caZ(X_1)$ and
$K_2=\caZ(X_2)\cap f^{-1}(I_2)$, so $K_1+K_2\subseteq K$.
For a cap $Q_i=K\cap\{\bfx:\langle\bfu_i,\bfx\rangle\ge a_i\}$
containing a lattice point $\bfq_i$, choose a lattice point
$\bfx_i\in K_1$ maximizing $\langle\bfu_i,\bfx_i\rangle$.
If $\bfz_1+\bfz_2\in Q_i$ with $\bfz_j\in K_j$, then
$\bfx_i+\bfz_2\in Q_i$ as well. Subtracting these two points of
$Q_i$, and then subtracting $\bfq_i$ from $\bfx_i+\bfz_2$, gives
\[
 \bfz_1\in\bfx_i+(Q_i-Q_i),
 \qquad
 \bfz_2\in\bfq_i-\bfx_i+(Q_i-Q_i).
\]
Thus, to find pairs of subset sums whose sum lies in $Q_i$, we only
need to consider these two portions of $K_1$ and $K_2$; see
\cref{fig:overview-cap-merge}.

\begin{figure}[!htb]
\centering
\begin{tikzpicture}[ov figure]
  \begin{scope}[shift={(1.1,1.9)}]
    \path[ov shell,fill=ovamber!23] (0,0) circle[radius=1];
    \path[fill=ovblue!17,draw=ovblue!65,line width=.45pt]
      (0,0) circle[radius=.82];
    \begin{scope}
      \clip (0,0) circle[radius=1];
      \path[fill=ovamber!55] (-1,.6) rectangle (1,1);
      \draw[ov cut] (-1,.6)--(1,.6);
    \end{scope}
    \draw[ovblue!65,line width=.45pt] (0,0) circle[radius=.82];
    \draw[ovamber!85!black,line width=.8pt] (0,0) circle[radius=1];
    \fill[black] (0,1) circle[radius=1.6pt];
    \node[above] at (0,1) {$\bfx_i$};
    \fill[ovamber!70!black] (-.42,.82) circle[radius=1.5pt];
    \node[above left,inner sep=2pt] at (-.42,.82) {$\bfz_1$};
    \node at (0,-.18) {$K_1\ominus D$};
    \node[below] at (0,-1) {$K_1$};
  \end{scope}
  \node at (2.45,1.9) {$+$};
  \begin{scope}[shift={(3.7,1.9)}]
    \path[ov shell,fill=ovamber!23] (-.8,-.9) rectangle (.8,.9);
    \path[fill=ovblue!17,draw=ovblue!65,line width=.45pt]
      (-.62,-.72) rectangle (.62,.72);
    \path[fill=ovamber!55,draw=none] (-.8,.5) rectangle (.8,.9);
    \draw[ov cut] (-.8,.5)--(.8,.5);
    \draw[ovblue!65,line width=.45pt] (-.62,-.72) rectangle (.62,.72);
    \draw[ovamber!85!black,line width=.8pt] (-.8,-.9) rectangle (.8,.9);
    \fill[black] (0,.9) circle[radius=1.6pt];
    \node[above] at (0,.9) {$\bfq_i-\bfx_i$};
    \fill[ovamber!70!black] (.42,.77) circle[radius=1.5pt];
    \node[below left,inner sep=2pt] at (.42,.86) {$\bfz_2$};
    \node at (0,-.18) {$K_2\ominus D$};
    \node[below] at (0,-.9) {$K_2$};
  \end{scope}
  \node at (4.95,1.9) {$=$};
  \begin{scope}[shift={(7.15,1.9)}]
    \def\ovMergeK{(-.8,-1.9)--(.8,-1.9)
      arc[start angle=-90,end angle=0,radius=1]
      --(1.8,.9) arc[start angle=0,end angle=90,radius=1]
      --(-.8,1.9) arc[start angle=90,end angle=180,radius=1]
      --(-1.8,-.9) arc[start angle=180,end angle=270,radius=1]--cycle}
    \def\ovMergeErosion{(-.8,-1.72)--(.8,-1.72)
      arc[start angle=-90,end angle=0,radius=.82]
      --(1.62,.9) arc[start angle=0,end angle=90,radius=.82]
      --(-.8,1.72) arc[start angle=90,end angle=180,radius=.82]
      --(-1.62,-.9) arc[start angle=180,end angle=270,radius=.82]--cycle}
    \path[ov shell,fill=ovamber!23] \ovMergeK;
    \path[fill=ovblue!17,draw=none] \ovMergeErosion;
    \begin{scope}
      \clip \ovMergeK;
      \path[fill=ovamber!55,draw=none] (-2,1.5) rectangle (2,2);
      \draw[ov cut] (-2,1.5)--(2,1.5);
    \end{scope}
    \draw[ovblue!65,line width=.45pt] \ovMergeErosion;
    \draw[ovamber!85!black,line width=.8pt] \ovMergeK;
    \fill[black] (0,1.9) circle[radius=1.6pt];
    \node[above] at (0,1.9) {$\bfq_i$};
    \fill[ovamber!70!black] (0,1.59) circle[radius=1.5pt];
    \node[below,inner sep=3pt] at (0,1.48) {$\bfz_1+\bfz_2$};
    \node at (1.05,1.69) {$Q_i$};
    \node at (0,0) {$K\ominus D$};
    \node[below] at (0,-1.9) {$K=K_1+K_2$};
  \end{scope}
\end{tikzpicture}
\caption{Merging $K_1$ and $K_2$ in the schematic case $K=K_1+K_2$.
The cap $Q_i$ of $K$ is highlighted in dark yellow.
The highlighted portions of $K_1$ and $K_2$ are
$K_1\cap(\bfx_i+Q_i-Q_i)$ and
$K_2\cap(\bfq_i-\bfx_i+Q_i-Q_i)$.
Every pair of points in $K_1,K_2$ whose sum lies in $Q_i$ must belong to these two portions respectively.}
\label{fig:overview-cap-merge}
\end{figure}
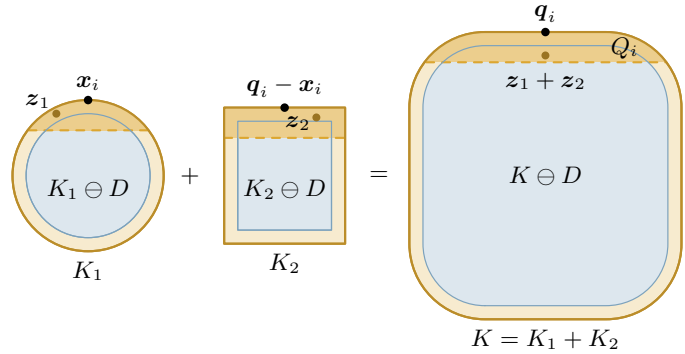

For each $Q_i$, list the $f$-values of the lattice points in these two
portions, keep those that occur among the already computed subset
sums of $X_1$ and $X_2$, and combine the two lists by sparse convolution.
Our idea of using cap covering for the merging is inspired by
Bringmann and Cassis's near-convex min-plus convolution algorithm
\cite{BringmannC23}.

This toy example has two groups. In general, our algorithm partitions the input items into
logarithmically many groups (together with the extra garbage), compute the subset sums for each of them, and finally merge
in reverse order using the same idea of interval restriction based on the target.
To construct this grouping,
we start from the initial CFP embedding, and repeatedly extract a relatively fat group so that the convex hull of the remaining vectors loses a constant
fraction of its volume, and then apply CFP again to the
remaining vectors.
Note that, unlike in the toy example (where the second group lies in the same dimension, except that the convex hull $H_2$ becomes thinner), the new CFP embedding may introduce new dimensions or become restricted to a lower dimensional subspace. We must also bound how often these changes occur;
see \cref{sec:classify}.

\subsection{Implementation details}
\label{subsec:overview-implementation}

We have discussed the main ideas leading to the algorithm, but there remain many implementation details to fill in.
In particular, our algorithmic ideas require performing certain geometric operations on polytopes and lattices, so we have to account for the time complexity to implement them. 
Notably, the most implementation-heavy part of the algorithm lies in constructing our quotient version of the cap covering by
greedily packing Macbeath regions. 

Fortunately, the polytopes that we encounter are zonotopes of the input vectors in $X$, possibly modified by a constant number of additional operations such as affine transformations, intersections, or Minkowski sums. 
Since the zonotope $\caZ(X)$ only has $\lvert\supp X\rvert$ many possibly weighted generators (our $X$ can actually be a multiset, so each generator is weighted by its multiplicity, and $\supp X$ denotes its support), we can answer the linear optimization query $\arg\max_{\bfy\in \caZ(X)}\langle \bfy,\bfu\rangle$ for any $\bfu$ in nearly $\lvert\supp X\rvert$ time by simply scanning the generators. By the equivalence of linear optimization and strong separation oracles \cite[Theorem~6.4.9]{GLS88}, we can also perform linear queries on its modified versions with roughly the same time complexity; see \cref{sec:geometry-oracles}.

Our geometric algorithm processes the quotient shell at a discretization scale at least as large as $\pi_\Lambda D$, so the number of basic geometric queries that we perform is bounded, up to logarithmic factors, by the quotient-shell volume bound divided by $\vol_d(\pi_{\Lambda} D)$. 
Since $\supp X\subseteq H\subseteq D$, and $f$ is injective on $H\cap\Z^d$, we actually have $ \lvert\supp X\rvert = \Oh_d(\vol_d(\pi_\Lambda D))$, so  the $\lvert\supp X\rvert$ factor in the geometric query time is compensated by the denominator $\vol_d(\pi_{\Lambda} D)$ in the query count. Thus the total time remains near-linear in the volume of the quotient shell bound.

The geometric operations we need to perform also include finding a $\Lambda$-lattice point inside a polytope. We use Lenstra's fixed-dimensional integer programming algorithm for this task, which works when the input polytope is provided as a linear optimization oracle  \cite[Section~2, Remark~(d)]{Lenstra83}; see \cref{lem:lattice-optimization}.

So far, we have focused on computing the shallow subset sums. 
The deep subset sums in $\caZ(X)\ominus D$ are easier to find. Intuitively, we can start at any point near the shell of $\caZ(X)$ and travel along a fixed vector $\bfx_0 \in D\cap\Z^d$ until it reaches the other side of the boundary, and the visited lattice points correspond to deep subset sums.  
This also allows us to compactly represent the subset sums (both deep and shallow) as a disjoint union of arithmetic progressions of the same common difference, which we call an \emph{AP decomposition} in the main text of the paper.
We also conveniently use AP decompositions (instead of geometric representations) as the input and output format for the two key lemmas in the main algorithm, in order to simplify the lemma interface; see \cref{sec:main-algo}.

When the answer to the \textsc{Subset Sum} instance is YES, it is not hard to also recover a subset that witnesses the target sum. 
In our algorithmic framework, witness recovery for shallow subset sums is straightforward. 
For deep subset sums, we solve an LP and apply the same rounding argument as in \cref{lem:innerpoints-specialized}; see \cref{lem:cap-cover-witness}.

    \section*{Roadmap}
\addcontentsline{toc}{section}{Roadmap}
\cref{sec:prelim} gives useful preliminaries, including frequently used notations and definitions.
\Cref{sec:main-algorithm} formally states the building blocks of our algorithm, and shows that they together imply \cref{thm:main-main}. These building blocks are proved in subsequent sections:
\begin{itemize}
  \item \Cref{sec:algorithmiccfp} proves the algorithmic Conlon--Fox--Pham theorem.
    \item
\Cref{sec:geometry-oracles} collects standard algorithmic tools for optimization and lattices.
\item
\Cref{sec:classify} partitions the input items into groups.
\item 
\Cref{sec:macbeath-regions} collects basic properties of Macbeath regions.
\item \Cref{sec:quotient-cap-covering} constructs the quotient variant of economic cap covering for a polytope.
  \item 
\Cref{sec:subset-sum-cap-covering}
applies the quotient cap covering of
\Cref{sec:quotient-cap-covering} to the zonotope of a group of input items, thereby giving a cap covering of the shallow subset sums together with a compact representation of all the deep subset sums.
\item 
\Cref{sec:one-group}
invokes the main theorem of \Cref{sec:subset-sum-cap-covering} to compute all the subset sums of a group of input items.
\item 
\cref{sec:merging-groups} invokes the main theorem of
\Cref{sec:subset-sum-cap-covering} to merge the results from different groups.
\end{itemize}

    \section*{Acknowledgments}
    \addcontentsline{toc}{section}{Acknowledgements}
    
    C. J. thanks Arturs Backurs and Christos Tzamos for introducing the main open problem to him during a meeting at MIT in summer 2018.
   Additionally, C. J. thanks 
   Karl Bringmann, Hongxun Wu, and Zixuan Xu 
   for discussions about this problem in the past, Timothy Chan for answering a question related to computational geometry, and Nick Fischer and Timothy Randolph for questions and comments that helped improve the exposition of this paper.
   C. J. also thanks Hongxun Wu for a gift subscription to ChatGPT.
    
   X. M. thanks Aviad Rubinstein for helpful discussions about this problem, and for supporting him throughout his Ph.D. journey. 

We thank the authors of \cite{chen2026subsetsumintegerprogramming} for coordinating concurrent conference submissions with us.
   
    \section*{AI Disclosure}
    \addcontentsline{toc}{section}{AI Disclosure}
The human authors started  this project in 2025, and had designed by June 2026 a high-level proof framework that resembles what is described in the overview (\cref{sec:overview}), with technical details complete in some components, but lacking in others (mostly in the geometric implementation).
Our initial inspiration was the discovery of the AP decomposition (\cref{lem:aprepresentation}), and the intuition that \cite[Theorem 1.5]{conlon2023homogeneous} can almost provide a ``convexified'' variant of \cref{lem:aprepresentation}.
AI assistance during this period was limited to literature searches: notably, the literature on economic cap covering and Macbeath regions was brought to our attention by Gemini in July 2025. 

Since July 2026, we have been interacting with ChatGPT to fill in the missing details. Along the way, the interactions also made us realize a few simplifications to our original framework. In more details:
\begin{itemize}
    \item The ideas described in \cref{subsec:overview-motivation,subsec:overview-motivation-worstcase,subsec:overview-boundary,subsec:overview-quotient-geometry,subsec:overview-classification} came from the human authors. 
    \item For computing the shallow subset sums in a single group, we originally considered a partition-and-convolve approach, but later switched to the simpler method sketched in \cref{subsec:overview-computing}.
    \item The implementation details described in \cref{subsec:overview-implementation} and  the main text received more AI input. We initially considered a more complicated approach in which most algorithmic operations would be performed intrinsically in $\R^d/ \Lambda$ with suitable discretization, but failed to work out the details. 
     GPT 5.6 Sol suggested that it is simpler to perform the operations in $\R^d$ by utilizing the fact that the zonotopes have fast linear optimization oracles.
\end{itemize}

  We also used Codex (GPT 5.6 Sol and GPT 6 Astra) for helping draft preliminary proofs, creating figures, and proofreading.

\section{Preliminaries}
\label{sec:prelim}

In this section, we collect frequently used notations and definitions. Less standard ones include the notations in displayed equations, and the definitions in 
\cref{subsec:progression,subsec:apdecomp}.

\subsection{Basic notation}

Let $[n] = \{1,2,\dots,n\}$. Let $\log x$ denote $\log_2 x$.
Assume $\Delta\ge 2$.

Use $\ll_d, \Oh_d(\cdot)$, etc., to hide factors depending on a parameter $d$. Such factors are also denoted by $C_d$ or $c_d$, etc. They are always computable functions of $d$.

Let $\tO(T)$ and $\tO_d(T)$ denote $\Oh(T \cdot (\log T)^{\Oh(1)})$ and $\Oh_d(T \cdot (\log T)^{\Oh_d(1)})$, respectively.

A (finite) \emph{multiset} $A$ is represented by a list of pairs of element $x$ together with its \emph{multiplicity} $\mu_A(x)$, over all $x\in \supp A \coloneqq \{x:\mu_A(x)>0\}$. Its \emph{size} is $\lvert A\rvert=\sum_{x} \mu_A(x)$. 
If $\mu_A(x)\le \mu_B(x)$ for all $x$,
 then write $A\subseteq B$, and define their \emph{multiset difference} $B\setminus A$ by the multiplicity function $\mu_{B\setminus A}\coloneqq \mu_B-\mu_A$.
The \emph{multiset union (sum)} $A\sqcup B$ is defined by 
$\mu_{A\sqcup B}\coloneqq \mu_A+\mu_B$.
For a multiset $A$, the notation $\{f(x): x\in A\}$ (also written as $f(A)$) builds a multiset that preserves the multiplicities.

The \emph{set of subset sums} of a (multi)set
$A$ is
\[\Sigma(A) = \Big\{\sum_{b\in B} b :B\subseteq A\Big\}. \]

The \emph{sumset} (or \emph{Minkowski sum}) of two sets $A,B$
is $A+B=\{a+b:a\in A,b\in B\}$. 
Define $A-B$ analogously.
For $k\in \Pos$, 
the \emph{$k$-fold sumset} of $A$ is $kA \coloneqq \underbrace{A+A+\dots + A}_{\text{$k$ times}}$. 

We denote the set dilation by $k\cdot A=\{ka:a\in A\}$ to distinguish it from the $k$-fold sumset $kA$, but \emph{with the following exceptions} to avoid notational clutter: For a lattice $L$ or a convex set $K$, we may write $rL,rK$ for their $r$-factor dilations.

For $r=p/q\in \Q$ in lowest terms, define \[\maxbitlen(r) = \max \{\log(2+\lvert p\rvert),\log(2+\lvert q\rvert)\},\] so that $r$ can be encoded as a fraction using $\Oh(\maxbitlen(r))$ bits.
We extend this notation to other objects in \Cref{par:rational-arithmetic}.

\subsection{Sparse convolution}
Generalizing
the fast Fourier transform (FFT), the following algorithm multiplies sparse polynomials in near-linear output-sensitive time.

\begin{lemma}[Nonnegative sparse convolution
{\cite[Theorem~1.1]{BringmannFN22}}]
    \label{lem:sparse-convolution}
       Given polynomials $f(x),g(x)$ (represented as lists of nonzero terms) with degrees at most $N$ and all coefficients in $ \Z \cap [0,U]$, their product $h(x) = f(x) g(x)$ can be computed in
    $\Oh(t \polylog (NU))$ time deterministically, where $t$ denotes the number of nonzero terms in $h(x)$.
\end{lemma}

In particular, the sumset of $A,B\subseteq [N]$ can be computed
in $\Oh(\lvert A+B\rvert\polylog N)$ time.

\subsection{Geometry and lattices}

Write $\vol_d(\cdot)$ for $d$-dimensional volume.
For $A\subseteq \R^d$, write $\conv(A)$ and $\spa(A)$ for the convex hull and linear span of $A$.
A \emph{(convex) polytope} is the convex hull of finitely many points.
A \emph{vertex} of a polytope $K$ is a point $\bfv\in K$ such that
$\bfv\notin \conv(K\setminus \{\bfv\})$.
A \emph{convex body} in $\R^d$ is a compact convex set in $\R^d$ with nonempty interior.

 Let $\langle \cdot,\cdot \rangle$ denote the Euclidean inner product.
For a linear function $f(\bfx) = \langle \bfx,\bfu\rangle$ and a bounded set $S$, write
$\fw(S)=\sup f(S)-\inf f(S)$.

Denote the \emph{Minkowski difference} (or \emph{erosion}) by 
$A\ominus B=\{\bfx:\bfx+B\subseteq A\}$.
If $A$ is convex, then $A\ominus B$ is convex.

Denote the line segment joining $\bfx,\bfy$ by
$[{\bfx},{\bfy}]=\{\eps{\bfx}+(1-\eps){\bfy}:0\le \eps\le 1\}$.
The \emph{zonotope} generated by a multiset $X=({\bfx}_i)_{i=1}^{\lvert X\rvert}$ is
\[
\caZ(X)=\sum_{i=1}^{\lvert X\rvert}\, [0,{\bfx}_i]=\Big \{\sum_{i=1}^{\lvert X\rvert} \eps_i {\bfx}_i:0\le \eps_i\le 1\Big \}.
\]
Similarly, define the \emph{centered zonotope} by
$\caZ_0(X)=\sum_{i=1}^{\lvert X\rvert} [-{\bfx}_i/2,{\bfx}_i/2]$.

A matrix $B=({\bfb}_1,\ldots,{\bfb}_{d'})\in \R^{d\times d'}$ with linearly
independent columns generates the \emph{lattice}
$\caL(B)= B\Z^{d'} = \{\sum_{i=1}^{d'} n_i {\bfb}_i:n_i\in \Z\}$ of \emph{rank} $d'$.
Its \emph{determinant} is $\det \caL(B)=\sqrt{\det(B^\top B)}$, which equals
$\lvert\det B\rvert$ when $d'=d$.
More generally,
 for arbitrary $X\subseteq \R^d$, denote
\[\caL(X)=\Big \{\sum_{i=1}^m n_i {\bfx}_i:m\ge 0,\ {\bfx}_i\in X,\ n_i\in \Z\Big \}.\] 
If $X\subseteq \Z^d$,  then $\caL(X)$ is a lattice.
In this paper, we only deal with rational lattices $L\subseteq \Q^d$, and our algorithm represents a rational lattice by its Hermite normal form (HNF) (\cref{sec:hermite}) by default.

For a lattice $\Lambda\subseteq \R^d$,
write $\pi_\Lambda$ for the quotient map \[\pi_\Lambda\colon \R^d\to \R^d/\Lambda.\] 
Let $\vol_d$ also denote the Haar measure on
$\R^d/\Lambda$, normalized so that $\vol_d(\pi_\Lambda(A))$ equals the Euclidean
volume $\vol_d(A)$ for every measurable $A\subseteq \R^d$ on which
$\pi_\Lambda$ is injective.

\begin{lemma}[cf.\ {\cite[Eq. (3.10)]{taovu06}}]
\label{lem:lattice-quotient-dilation}
Let $\Lambda,\Gamma\subseteq \R^d$ be lattices, $K\subseteq \R^d$ be an origin-symmetric compact convex set, and $s\ge 1$. Then,
$\lvert\pi_\Lambda(sK\cap \Gamma)\rvert\le (2s+1)^d\lvert\pi_\Lambda(K\cap \Gamma)\rvert$.

In particular, 
 $\lvert sK\cap \Gamma\rvert\le (2s+1)^d\lvert K\cap \Gamma\rvert$.
\end{lemma}

\begin{proof}
Let $q=\lceil2s\rceil$. Partition $sK\cap \Gamma$ into its at most
$q^d$ residue classes modulo $q \Gamma$.
Choose a representative $\bfa$ in each nonempty class.
For every $\bfx$ in the class of $\bfa$,
since $\bfx-\bfa\in 2sK\subseteq qK$, we have
$(\bfx-\bfa)/q\in K\cap \Gamma$.
Since $\pi_\Lambda((\bfx-\bfa)/q)$ determines $\pi_\Lambda(\bfx)$, summing over all $q^d$ residue classes gives
$\lvert\pi_\Lambda(sK\cap \Gamma)\rvert\le q^d\lvert\pi_\Lambda(K\cap \Gamma)\rvert$.
The claim then follows from $q\le 2s+1$.
\end{proof}

\subsection{Progressions and \texorpdfstring{$s$}{s}-generating sets}
\label{subsec:progression}

We use a slight variant of convex progressions (the standard definition corresponds to the $\Gamma= \Z^d$ case):
\begin{definition}[Progression $f(K\cap \Gamma)$]
\label{def:progression}
Let $\Gamma$ be a rank-$d$ lattice, $K\subseteq \spa(\Gamma)$ be an origin-symmetric compact convex set such that $\dim \spa(K\cap \Gamma) = d$,
 and $f\colon \spa(\Gamma) \to \R^k$ be
a linear map. The set $f(K\cap \Gamma)$ is called a \emph{progression} of \emph{rank} $d$.

By abuse of notation, a progression $f(K\cap \Gamma)$ is understood to include its representation $(K,\Gamma,f)$, since the same set may have different representations.

For $s\ge 1$, we say the progression $f(K\cap \Gamma)$ is \emph{$s$-proper} if $f$ is
injective on $sK\cap \Gamma$.  We abbreviate $1$-proper to \emph{proper}.
\end{definition}

The progression $f(K\cap \Gamma)$ with $f$ being the identity map is simply referred to as $K\cap \Gamma$.

An origin-symmetric generalized arithmetic progression (GAP) is a progression $f(K \cap \Z^d)$ where $K$ is a box
$\prod_{i=1}^d[-m_i,m_i]$ with $m_i\in \Pos$.
More details on GAPs are given in \cref{subsec:somelemmas} where they are needed.

We extract the following notion from \cite{conlon2023homogeneous} and call it \emph{$s$-generating}:

\begin{definition}[$s$-generating]
\label{def:generating}
Let  $f(K\cap \Gamma)$ be a proper progression (\cref{def:progression}).
For $s\ge 1$, if a multiset $X$ of vectors in $K\cap \Gamma$ satisfies
   $\Sigma(X)\supseteq (\boldsymbol{\tau}+sK)\cap \Gamma$ for some ${\boldsymbol{\tau}}\in \Gamma$, then we say $X$ is \emph{$s$-generating for} $K\cap \Gamma$, and we also say the multiset
$f(X)$ is \emph{$s$-generating for} $f(K\cap \Gamma)$.
\end{definition}

\begin{lemma}
    \label{lem:innerpoints}
Let $H\cap \Gamma$ be a progression of rank $d\ge 1$, and let $X$ be a multiset with $\supp(X)\subseteq H\cap \Gamma$.
If a submultiset $X'\subseteq X$ is $d$-generating for $H\cap \Gamma$, then
$(\caZ(X)\ominus \lvert X'\rvert H)\cap \Gamma\subseteq \Sigma(X)$.
\end{lemma}

\begin{proof}
    Without loss of generality, assume the ambient space is $\R^d$. 
Let $\bfy\in (\caZ(X)\ominus \lvert X'\rvert H)\cap \Gamma$. We will prove $\bfy\in \Sigma(X)$ by finding $\bfz'\in \Sigma(X\setminus X')$ such that $\bfy-\bfz'\in \Sigma(X')$.

Since $X'$ is $d$-generating for $H\cap \Gamma$ (\cref{def:generating}), there is $\boldsymbol{\tau}\in \Gamma$ with
$(\boldsymbol{\tau}+dH)\cap \Gamma\subseteq \Sigma(X')$.
Since $H$ is origin-symmetric (by \cref{def:progression}), $\zero\in H$, and thus $\boldsymbol{\tau}\in \Sigma(X')$.
Then, $\caZ(X')-\boldsymbol{\tau}$ is a Minkowski sum of $\lvert X'\rvert$
segments of the form $[0,\bfx]$ or $[-\bfx,0]$ where $\bfx \in X'$. Each segment is contained in $H$, so $\caZ(X')-\boldsymbol{\tau}\subseteq \lvert X'\rvert H$.  By the choice of $\bfy$, this implies $\bfy + (\caZ(X')-\boldsymbol{\tau}) \subseteq \caZ(X)$, so
\[
 \bfz\coloneqq \bfy-\boldsymbol{\tau}
 \in \caZ(X)\ominus \caZ(X')=\caZ(X\setminus X'),
\]
where the last equality follows from the standard fact $(K_1+K_2)\ominus K_2 = K_1$ for compact nonempty convex sets $K_1,K_2$ (see e.g., \cite[Lemma~3.1.11]{schneider2014convex}).

Write $X\setminus X'=(\bfx_i)_{i=1}^r$, so the linear program $\bfz = \sum_{i=1}^r \alpha_i \bfx_i, \boldsymbol{\alpha} \in [0,1]^r$ is feasible. Choose a basic feasible solution $\boldsymbol{\alpha}$, which can have at most $d$ fractional coordinates. Rounding $\boldsymbol{\alpha}$ down to $\boldsymbol \alpha' \in \{0,1\}^r$ gives
$\bfz' = \sum_{i=1}^r \alpha_i' \bfx_i \in \Sigma(X\setminus X')$ with $\bfz-\bfz'\in dH$. Thus, $\bfy-\bfz' = \boldsymbol{\tau} + \bfz -\bfz'\in (\boldsymbol{\tau}+dH)\cap \Gamma\subseteq \Sigma(X')$. This finishes the proof that $\bfy\in \Sigma(X)$.
\end{proof}

\subsection{AP decompositions}
\label{subsec:apdecomp}

\begin{definition}[AP decomposition]
\label{def:ap-decomposition}
For finite $A\subseteq \Z$ and
$r\in \Pos$, the \emph{AP decomposition of $A$ with common difference $r$} is the unique decomposition of $A$ into disjoint arithmetic progressions of common difference $r$, each maximal in $A$. 

The \emph{size} of the AP decomposition is the number of arithmetic progressions it contains, which equals $\lvert A\setminus (A+r)\rvert$ (since $A\setminus (A+r)$ consists of precisely the starting points of these arithmetic progressions).

Algorithmically, the AP decomposition is represented by the set $A \bmod r \subseteq \{0,\dots,r-1\}$ as a sorted list, together with a list $\caD_i$ for every $i\in A \bmod r$, which consists of the pairs $(a,\ell)$ representing the arithmetic progressions $\{a,a+r,\dots,a + (\ell-1) r\}$ with $a\equiv i \pmod r$, in increasing order of~$a$.
\end{definition}

For example, the AP decomposition of $A = \{4,5,6,7,9,10\}$ with common difference $r=2$ is $\{4,6\}\sqcup \{10\}\sqcup \{5,7,9\}$. It has size $\lvert A\setminus (A+r)\rvert = \lvert\{4,5,10\}\rvert=3$.

Using binary search, it is straightforward to implement the following operations on AP decompositions. We omit the proof here.

\begin{lemma}
\label{lem:ap-operations}
\begin{itemize}
  \item Initialization: Given integer arithmetic progressions $A_1,\dots,A_k$ of the same common difference~$r$, we can construct the AP decomposition of $\bigcup_{i=1}^k A_i$ of common difference $r$ in $\Oh(k \log k)$ time.
\end{itemize}
Given a size-$k$ AP decomposition of an integer set $A$, we can perform:
\begin{itemize}
\item Membership: given an integer $a$, decide whether $a\in A$ in $\Oh(\log k)$ time.
\item Restriction: given an interval $[\ell,u]$, compute the AP decomposition of $A\cap [\ell,u]$ in $\Oh(k)$ time.
\end{itemize}
\end{lemma}

    \section{Main Algorithm}
\label{sec:main-algorithm}

In this section, we formally state all the algorithmic ingredients, and use them to solve
\cref{prob:single-constraint-ilp} in $\Oh(n+\Delta^{1+o(1)})$ time, thereby 
proving \cref{thm:main-main}.

To solve \cref{prob:single-constraint-ilp}, we first use a standard proximity argument to reduce \cref{prob:single-constraint-ilp} to the \textsc{Subset Sum}  problem on an input multiset of size $\poly(\Delta)$. 
This proximity argument can be found in e.g., 
\cite[Lemma~2.1]{icalp21}, and was also used in \cite[Observation~6.1]{ChenLMZsoda24} for their $\tO(n+\Delta^{1.5})$-time algorithm. 
We omit the proof here.

\begin{lemma}
\label{lem:bounded-ilp-reduction}
In $\Oh(n+\Delta)$ time, an instance of \cref{prob:single-constraint-ilp}
can be reduced to a \textsc{Subset Sum} 
instance on an input multiset $A$ with  $\supp(A)\subseteq \{1,2,\dots,\Delta\}$ and $\lvert A\rvert \le \Oh(\Delta^2)$.
\end{lemma}

We henceforth work with the \textsc{Subset Sum} instance $A$ produced by \cref{lem:bounded-ilp-reduction}. 

\subsection{Grouping the input integers}
\label{sec:classification-interface}

We will partition the input integer multiset $A$ into a logarithmic number of groups $A_1,\dots,A_k$ and a small exceptional part $A^{\extr}$.
Each group $A_i$ is embedded into a low-dimensional progression $f(H_i \cap \Gamma_i)$, with a representing vector multiset $X_i$, where $\supp(X_i)\subseteq H_i \cap \Gamma_i$ and $A_i = f(X_i)$, so that $X_i$ satisfies a relative fatness condition and contains a small generating submultiset $X_i'$. 
These conditions are formally stated in 
\eqref{eq:oneclass-iii} below; they are essential for our algorithm that computes the subset sums of each group.
We also ensure certain extra conditions
between the groups, which are useful for merging the answers from the groups.

We now state the conditions required for each group. Let $\Gamma$ be a lattice, 
$H\subseteq \spa \Gamma$ be an origin-symmetric polytope,
$f\colon \spa \Gamma \to \R$ be a linear map,
and $X'\subseteq X$ be two nonempty multisets with $\supp(X)\subseteq H\cap \Gamma$.
For a fatness parameter $\lambda\ge 1$,
define the following conditions :
\begin{subequations}
\label{eq:oneclass-iii}
\begin{align}
 &f(\supp X)\subseteq [1,\Delta]\cap \Z, \text{ and }
 f(H\cap \Gamma) \text{ is a proper progression } && \text{(see \cref{def:progression})},
 \label{eq:oneclass-values}\\*
 &H\subseteq \frac{\lambda}{\lvert X\rvert} \caZ_0(X)
 && \hspace{-69pt}\text{(recall $\caZ_0(X)\coloneqq \sum_{\bfx \in X} [-\bfx/2,\bfx/2]$)},
 \label{eq:oneclass-balance}\\*
 &X' \text{ is $(\rk \Gamma)$-generating for }H\cap \Gamma  && \text{(see \cref{def:generating})}.
 \label{eq:oneclass-generating}
\end{align}
\end{subequations}

We now state the main partitioning theorem. Write $\Pi_{U}$ for the orthogonal projection onto the subspace $U$.

\begin{theorem}
\label{lem:main-classify}
Let $d\ge 1$ be given. 
Let $\bfe_1,\dots,\bfe_d$ be the standard basis of $\R^d$, and for $1\le d'\le d$ identify $\R^{d'}$ with the subspace $\spa \{\bfe_1,\dots,\bfe_{d'}\}$ of $\R^d$.

Given a nonempty multiset $A$
of $\Delta^{\Oh(1)}$ integers in $[1,\Delta]$, 
we can construct in $\tO_d(\Delta^{1+2/d})$ time
 a partition into submultisets, 
\begin{equation*}
 A=A^{\extr}\sqcup \bigsqcup_{i=1}^k A_i,
 \text{ where } \lvert A^{\extr}\rvert=\Oh_d(\Delta^{1/d} \log^2 \Delta),\, k=\Oh_d(\log \Delta),
\end{equation*}
integers
$1\le d_1\le d_2\le \dots \le d_k\le d$, 
and nonempty vector multisets $X_i',X_i$ with $X_i'\subseteq X_i$, such that
\[  \supp(X_i)\subseteq \Big (\Z \cap [-\Delta^{\Oh(1)}, \Delta^{\Oh(1)}]\Big )^{d_i} \text{ and } f(X_i)=A_i,\]
where $f\colon \R^d\to \R$ is given by $f(\bfx)=x_1$,
and, for $i\in [k]$, writing
\[
\Gamma_i\coloneqq \caL(X_i)\quad\text{and}\quad
H_i\coloneqq \conv \Big(\pm \bigcup_{j=i}^{k} \Pi_{\R^{d_i}}( \supp X_j)\Big),
\]
 the following properties hold:
\begin{enumerate}[label=\textup{(\roman*)}]
\item\label{item:main-classify-zonotope}
The conditions in \eqref{eq:oneclass-iii} hold for
$(\Gamma,H,f, X',X)=(\Gamma_i,H_i,f|_{\spa(\Gamma_i)}, X_i',X_i)$ with
$\lambda=\Oh_d(\log \Delta)$.
Moreover, the algorithm also returns $\boldsymbol{\tau}_i\in \Gamma_i$ witnessing \eqref{eq:oneclass-generating}, i.e.,
$\boldsymbol{\tau}_i+(\Gamma_i\cap (\rk \Gamma_i) H_i)\subseteq \Sigma(X_i')$.
\item\label{item:main-classify-generating}
$\lvert X_i'\rvert=\Oh_d(\Delta^{1/d})$. 
\item\label{item:main-classify-projection}
For each $i\in [k]$,
$\bigcup_{j=i}^k \Pi_{\R^{d_i}}(\supp X_j)\subseteq \Gamma_i$.
\item\label{item:main-classify-prefix}
For $1\le i\le k$,
 $\sum_{j=1}^i\lvert A_j\rvert=\Oh_d(\lvert A_i\rvert\log \Delta)$.
\end{enumerate}
\end{theorem}
We remark that the computed vectors $\boldsymbol{\tau}_i$ are needed only for the witness recovery steps.

We prove \cref{lem:main-classify} in \cref{sec:classify} by iteratively applying the algorithmic Conlon--Fox--Pham theorem (stated and proved in \cref{sec:algorithmiccfp}).

We give some intuition to help parse the statement of \cref{lem:main-classify} and motivate the next few subsections. Consider the projected vectors in a suffix, $\hat X_i \coloneqq \bigsqcup_{j=i}^k \Pi_{\R^{d_i}} X_j$. 
Then, $f(\hat X_i) = \bigsqcup_{j=i}^k A_j$,  $\supp(\hat X_i) \subseteq H_i \cap \Gamma_i$, and $X_i'\subseteq \hat X_i$ is $\rk(\Gamma_i)$-generating for $H_i\cap \Gamma_i$. 
(In particular, taking $i=1$, this embeds $\bigsqcup_{j=1}^k A_j = A\setminus A^{\extr}$ into a progression $f(H_1 \cap \Gamma_1)$ with a small generating $X'_1$, which is what the Conlon--Fox--Pham theorem achieves.)
Therefore, 
    \cref{lem:innerpoints} implies 
    \[ (\caZ(\hat X_i) \ominus \lvert X_i'\rvert H_i)\cap \Gamma_i \subseteq \Sigma(\hat X_i), \]
    characterizing the ``deep'' subset sums of $\hat X_i$. This, though, is insufficient for us to efficiently compute $\Sigma(\bigsqcup_{j=i}^k A_j) =f(\Sigma(\hat X_i))$. 
    Instead, we compute the single-group subset sums $\Sigma(A_i) = f(\Sigma(X_i))$, 
    which have the analogous ``deep'' characterization
    \[ (\caZ(X_i) \ominus \lvert X_i'\rvert H_i)\cap \Gamma_i \subseteq \Sigma(X_i), \]
    but also satisfy the crucial fatness condition \eqref{eq:oneclass-balance}, \[H_i \subseteq \frac{\Oh_d(\log \Delta)}{\lvert X_i\rvert} \caZ_0(X_i).\]
    In the next subsection we state the algorithms to efficiently compute $ f(\Sigma(X_i))$ using this fatness condition, and to compute the sumset of $f(\Sigma(X_i))$ and $f(\Sigma(\hat X_i\setminus X_i))$ restricted to an interval; here, the interval restriction is for controlling the time complexity.
    By iterating, we can thus compute the subset sums of $A_1\sqcup \dots \sqcup A_k$ restricted to the interval of interest, $[t-\sum_{a\in A^{\extr}} a,t]$. The condition~\ref{item:main-classify-prefix} of \cref{lem:main-classify} allows us to bound the total time complexity of this iterative procedure.

\subsection{Solving one group and merging groups}

We now state the algorithm for computing the subset sums of one group and compactly representing them as an AP decomposition (\cref{def:ap-decomposition}). 

\begin{lemma}
\label{lem:oneclass}
Let $X'\subseteq X$ be two multisets of vectors in $\Z^d$.
Suppose \eqref{eq:oneclass-iii} holds for  $X',X$, $\Gamma=\caL(X),H=\conv(\pm \supp X)$, a linear map $f\colon \R^d\to \R$ with integer coefficients, and an unknown $\lambda \ge 1$. Then, we have the following algorithms:
\begin{itemize}
\item $\mathsf{OneGroup}(X,X',f)$ computes an AP decomposition of $f(\Sigma(X))$ of size $\Oh_d(\lambda \lvert X'\rvert\Delta)$.
\item 
$\mathsf{OneGroupWitness}(X,X',f,\boldsymbol{\tau},t)$, where
$\boldsymbol{\tau}+(\Gamma\cap (\rk \Gamma) H)\subseteq \Sigma(X')$ and
$t\in f(\Sigma(X))$,
computes a submultiset $W\subseteq X$ such that $\sum_{\bfw\in W} f(\bfw)=t$.
\end{itemize}
Both algorithms run in
$\tO_d(\lambda \lvert X'\rvert\Delta)$ time, provided the vector coordinates and multiplicities in $X$ and coefficients of $f$ have $\Oh(\log \Delta)$ bits.
\end{lemma}

We prove
\cref{lem:oneclass} in \cref{sec:one-group}.

We now state the algorithm for merging the subset sums of one group $f(X)$ with an accumulated partial result that consists of subset sums of $f(Y)$ restricted to an interval $J$.
The algorithm returns the combined subset sums $f(\Sigma(X\sqcup Y))$ restricted to an interval $I$, defined so that any pair of $a\in f(\Sigma(X))$ and $b\in f(\Sigma(Y))$ with $a+b\in I$ must satisfy $b\in J$, and thus $b$ is included in the input. 
The algorithm is efficient if the interval length $\lvert J\rvert$ is not much larger than $\lvert X\rvert\cdot \Delta$.

\begin{lemma}
\label{lem:merge}
Let $X'\subseteq X$ be two multisets of vectors in $\Z^d$, and let
$Y$ be a multiset of vectors in $\caL(X)$.
Suppose \eqref{eq:oneclass-iii} holds for $X',X, \Gamma = \caL(X)$,  
$H=\conv(\pm \supp(X\sqcup Y))$, a linear map
$f\colon \R^d\to \R$ with integer coefficients, and an unknown $\lambda\ge 1$.
Assume $S_X = f(\Sigma(X))$ is implicitly given by access to its AP decomposition.

Let $J\subseteq f(\caZ(Y))$ be a nonempty closed interval (with integer endpoints), and $S_Y=f(\Sigma(Y))\cap J$ be implicitly given by access to its AP decomposition.
Let $I$ be a closed interval (with integer endpoints) satisfying
\begin{equation}
\label{eq:public-merge-restriction-condition}
 I\subseteq f(\caZ(X\sqcup Y)) \text{ and }
 \bigl(I-f(\caZ(X))\bigr)\cap f(\caZ(Y))\subseteq J.
\end{equation}
Then we have the following algorithms:
\begin{itemize}[beginpenalty=10000]
\item $\mathsf{Merge}(X,X',Y,f,S_X,J,S_Y,I)$ computes an AP decomposition of
\begin{equation*}
 (S_X+S_Y)\cap I=f(\Sigma(X\sqcup Y))\cap I,
\end{equation*}
where the equality follows from \eqref{eq:public-merge-restriction-condition}.
The decomposition has size
$\Oh_d\!\left(\lambda \lvert X'\rvert\left(\Delta+\frac{\lvert J\rvert}{\lvert X\rvert}\right)\right)$.
\item $\mathsf{MergeWitness}(X,X',Y,f,S_X, J,S_Y,I,\boldsymbol{\tau},t)$, where
$\boldsymbol{\tau}$ satisfies
$\boldsymbol{\tau}+(\Gamma\cap (\rk \Gamma) H)\subseteq \Sigma(X')$ and
$t\in (S_X+S_Y)\cap I$, computes $a\in S_X$ and
$b\in S_Y$ with $a+b=t$.
\end{itemize}
Both algorithms run in
$\tO_d\!\left(\lambda \lvert X'\rvert\left(\Delta+\frac{\lvert J\rvert}{\lvert X\rvert}\right)\right)$ time, provided the vector coordinates and multiplicities in $X, Y$,  and coefficients of $f$ have $\Oh(\log \Delta)$ bits.
\end{lemma}

We prove \cref{lem:merge} in \cref{sec:merging-groups}.

\subsection{Main algorithm}
\label{sec:main-algo}

\begin{algorithm}
\DontPrintSemicolon
\SetKwProg{Fn}{Function}{:}{}
\caption{Main algorithm for \textsc{Subset Sum} with target $t$ (decision version)}
\label{alg:subset-sum-main}
\Fn{$\mathsf{SubsetSum}_d(A,t)$}{
    Partition $A = A^{\extr}\sqcup \bigsqcup_{i=1}^k A_i $ by \cref{lem:main-classify}, and compute
    $\Sigma(A^{\extr})$ by standard DP\\
    \lIf{$k=0$}{\Return{$\one[t\in \Sigma(A^{\extr})]$}}
\lFor{$i\in [k]$}{compute the interval $I_i$  from \eqref{eq:assembly-windows}}
\lFor{$i\in [k]$}{find an AP decomposition of $\Sigma(A_i)$ by
    $\mathsf{OneGroup}(X_i,X_i',f)$ (\cref{lem:oneclass})}
    $S_k\gets \Sigma(A_k)\cap I_k$
    \textup{(\cref{lem:ap-operations})}\\
    \For{$i=k-1,\ldots,1$}{
        $Y_i \gets \bigsqcup_{j=i+1}^k \Pi_{\R^{d_i}} X_j$\\
        $S_i\gets \mathsf{Merge}\bigl(X_i,X_i',
        Y_i,
        f,\Sigma(A_i),I_{i+1},S_{i+1},I_i\bigr)$
        \textup{(\cref{lem:merge,lem:ap-operations})}
    }
    \Return{$\one[(t - \Sigma(A^{\extr})) \cap S_1 \neq \varnothing]$ \textup{(\cref{lem:ap-operations})}}
}
\end{algorithm}

We now describe the main algorithm for a given dimension parameter $d\ge 1$, with pseudocode given in 
\cref{alg:subset-sum-main}.
In words,  we first use \cref{lem:main-classify} to partition the input multiset $A$ (with $\supp(A) \subseteq \Z \cap [1,\Delta]$) into groups $A_1,A_2,\dots,A_k$ and a small exceptional part $A^{\extr}$.
Compute the subset sums of each group using $\mathsf{OneGroup}$ (\cref{lem:oneclass}), and then
merge their results by applying $\mathsf{Merge}$ (\cref{lem:merge}) successively in decreasing order $A_k,A_{k-1},\dots,A_1$. Finally, merge the result with the subset sums of $A^{\extr}$ computed by brute-force dynamic programming. 
After each merging step, we restrict the obtained partial subset sums to an interval that excludes those irrelevant values which can never contribute to the target subset sum $t$.

We now define the intervals for the restriction.
Given the partition in \cref{lem:main-classify}, for $i\in [k]$, write 
\[W_i=\sum_{a\in A_i} a,\,\, 
W_{\ge i}=\sum_{j =  i}^k W_j,\, \text{ and }
W^{\extr}=\sum_{a\in A^{\extr}} a.\]
We assume $0\le t\le \sum_{a\in A} a$ (otherwise, $t$ is not a subset sum of $A$).
Since $A^{\extr}$ and $A_1,\dots,A_{i-1}$ contribute total sum at most $W^{\extr}+\sum_{j<i} W_j$, the relevant interval for the subset sums of the suffix $A_{i}\sqcup \dots \sqcup A_{k}$ is
\begin{equation}
\label{eq:assembly-windows}
 I_i=[0,W_{\ge i}]\cap
 \Big [t-W^{\extr}-\sum_{j<i} W_j,\ t\Big ],
\end{equation}
which is nonempty due to $0\le t\le \sum_{a\in A} a$, and satisfies 
\begin{equation}
    \label{eq:assembly-window-recurrence}
 I_{i+1}\supseteq (I_i-[0,W_i])\cap [0,W_{\ge i+1}] \,\, \text{for all } i\in [k-1].
\end{equation}

\begin{lemma}
\label{lem:main-algorithm}
Given integer $d\ge 1$, a multiset $A$ of
integers in $[1,\Delta]$, and an integer $t$, 
$\mathsf{SubsetSum}_d(A,t)$ (\cref{alg:subset-sum-main}) correctly decides whether there is a submultiset
$B\subseteq A$ with $\sum_{a\in B} a=t$
in $\tO_d(\Delta^{1+2/d})$ time.
In the YES case, one can recover a witness $B$ in the same time complexity.
\end{lemma}

\begin{proof} ~
    \paragraph*{Correctness.}
    In \cref{alg:subset-sum-main}, if the partition step returns $k=0$, then the answer is clearly correct. We now assume $k\ge 1$.
For each $i\in [k]$, 
the partition step (\cref{lem:main-classify}) guarantees that 
the conditions \eqref{eq:oneclass-iii}
hold for $(X_i,X_i',H_i,\caL(X_i))$ where
$ H_i=\conv (\pm \bigcup_{j=i}^{k} \Pi_{\R^{d_i}}( \supp X_j))$ and $\lambda = \Oh_d(\log \Delta)$.
Clearly, the conditions \eqref{eq:oneclass-iii} remain valid 
with $H_i$ replaced by its subset $\conv(\pm \supp X_i)$. Thus,  \cref{lem:oneclass} applies, so
$\mathsf{OneGroup}(X_i,X_i',f)$ correctly computes an AP decomposition of $f(\Sigma(X_i))=\Sigma(A_i)$.
    
We now prove by backward induction that the computed sets $S_i$ (represented as AP decompositions) satisfy
\[
 S_i=\Sigma(A_i\sqcup \cdots\sqcup A_k)\cap I_i \text{ for all } i\in [k].
\]
The base case $i=k$ immediately follows from the definition $S_k = \Sigma(A_k)\cap I_k$. 
Now consider $1\le i<k$, and assume the displayed inductive hypothesis holds for $i+1$.
Since $Y_i=\bigsqcup_{j=i+1}^k \Pi_{\R^{d_i}} X_j$ by definition and $f(\bfx) = x_1$ by \cref{lem:main-classify}, we have
$f(Y_i)=\bigsqcup_{j=i+1}^k f(X_j) = A_{i+1}\sqcup \cdots\sqcup A_k$, and the induction hypothesis gives
$S_{i+1}=f(\Sigma(Y_i))\cap I_{i+1}$.
By \eqref{eq:assembly-windows} and \eqref{eq:assembly-window-recurrence}, 
the requirement \eqref{eq:public-merge-restriction-condition} is satisfied by $(I,J,X,Y) = (I_{i},I_{i+1},X_i,Y_i)$, so \cref{lem:merge} applies, and thus $\mathsf{Merge}$ correctly computes 
$f( \Sigma(X_i\sqcup Y_i))\cap I_i = \Sigma(A_i \sqcup \dots \sqcup A_{k})\cap I_i$, completing the induction step.

Setting $i=1$ in the previous displayed property gives \[S_1 = \Sigma(A_1\sqcup \cdots\sqcup A_k)\cap I_1 = \Sigma(A\setminus A^{\extr})\cap \Big [t - \sum_{a\in A^{\extr}} a,\,  t\Big ].\]
Thus, $t\in \Sigma(A)$ if and only if there exists $t'\in \Sigma(A^{\extr})$ such that $t-t'\in S_1$.  This proves the correctness of \cref{alg:subset-sum-main}.

\paragraph*{Running time.}
Write $s=\Delta^{1/d}$. By \cref{lem:main-classify}, the partitioning step takes $\tO_d(\Delta^{1+2/d})$ time,  and produces $k=\Oh_d(\log \Delta)$
groups and an exceptional part of size $\lvert A^{\extr}\rvert=\tO_d(s)$.
Hence, $W^{\extr} = \sum_{a\in A^{\extr}} a \le \tO_d(s \Delta)$, and
standard dynamic programming computes $\Sigma(A^{\extr})$ in $\Oh(\lvert A^{\extr}\rvert\cdot W^{\extr})= \tO_d(s^2 \Delta)$ time.

Recall $\lvert X_i'\rvert=\Oh_d(s)$, 
$\lambda=\Oh_d(\log \Delta)$, and $k =\Oh_d(\log \Delta)$ from \cref{lem:main-classify}.
By \cref{lem:oneclass}, each call to $\mathsf{OneGroup}$ takes
$\tO_d(\lambda \lvert X_i'\rvert\Delta) = \tO_d(s \Delta)$ time, with total time $\tO_d(s \Delta)$ over the $k$ calls.

Now we consider the $\mathsf{Merge}$ call at iteration $i$.
By \eqref{eq:assembly-windows}, 
\[\lvert I_{i+1}\rvert \le W^{\extr}+\sum_{j=1}^i W_j\le W^{\extr}+\Delta \sum_{j=1}^i \lvert A_j\rvert.\]
Thus, by \cref{lem:merge}, the $\mathsf{Merge}$ call takes time
\begin{align*}
\tO_d\left (\lambda \lvert X_i'\rvert\left(\Delta+\frac{\lvert I_{i+1}\rvert}{\lvert A_i\rvert}\right)\right )
 &\le \tO_d\left (\lambda \lvert X_i'\rvert\Delta+\lambda \frac{\lvert X_i'\rvert}{\lvert A_i\rvert} W^{\extr}
       +\lambda \lvert X_i'\rvert\Delta \frac{\sum_{j=1}^i \lvert A_j\rvert}{\lvert A_i\rvert}\right )\\
 &=\tO_d(s \Delta),
\end{align*}
where the last term is bounded using \cref{lem:main-classify}\,\ref{item:main-classify-prefix}. The total time
for $\mathsf{Merge}$ calls over all $k-1$ iterations is therefore $\tO_d(s \Delta)$.

The final step can be tested by looping over all $t'\in \Sigma(A^{\extr})$ and checking whether $t-t'\in S_1$ using 
\cref{lem:ap-operations}, in $\Oh((W^{\extr}+1) \log \Delta) = \tO_d(s \Delta)$ time. 

Thus, the total time of \cref{alg:subset-sum-main} is $\tO_d(\Delta^{1+2/d})$.

\paragraph{Witness recovery.}
The standard dynamic programming with backtracing can find a witness $B'\subseteq A^{\extr}$ with $\sum_{a\in B'} a=t'$. Then, we can find a witness $B''\subseteq A_1\sqcup \dots \sqcup A_k$ with $\sum_{a\in B''} a=t-t'$ by backtracing using the procedures $\mathsf{OneGroupWitness}$ and $\mathsf{MergeWitness}$ provided by \cref{lem:oneclass,lem:merge} with the vectors $\boldsymbol{\tau}_i$ computed by \cref{lem:main-classify}\,\ref{item:main-classify-zonotope}. The total time remains $\tO_d(\Delta^{1+2/d})$.
\end{proof}

\begin{proof}[Proof of \cref{thm:main-main}]
After applying the reduction of \cref{lem:bounded-ilp-reduction} in $\Oh(n+\Delta)$ time,
run the algorithm of \cref{lem:main-algorithm}
(which is uniform in the parameter $d$)
on the \textsc{Subset Sum} instance for
$d=1,\ldots,\lfloor\log \Delta\rfloor$ in parallel at equal speed, and return the answer of the first completed run.
Since the run with parameter $d$ takes time $C_d \Delta^{1 + 2/d} (\log \Delta)^{C_d}$ for some $C_d$, their overall time complexity is $\Oh(\log \Delta) \cdot \min_{d\le \lfloor\log \Delta\rfloor}\{C_d \Delta^{1 + 2/d}(\log \Delta)^{C_d}\}$, which is clearly in $\Oh(\Delta^{1+o(1)})$. Hence, the total time is $\Oh(n+\Delta^{1+o(1)})$ as claimed.
\end{proof}

\section{Algorithmic Conlon--Fox--Pham Theorem}
\label{sec:algorithmiccfp}
Conlon, Fox, and Pham  proved a structural  theorem that almost characterizes the set of subset sums using a generalized arithmetic progression (GAP)
\cite[Theorem~1.5]{conlon2023homogeneous}.
In this section, we prove the following algorithmic version of the Conlon--Fox--Pham theorem which efficiently finds such a GAP.

\begin{theorem}
    \label{thm:algorithmiccfp}
    For any $\beta>1$, there are constants $C>1,d \ge 1$ such that the following
    holds.

 Let $A$ be a multiset of integers in
 $[-\Delta,\Delta]\setminus \{0\}$,
 and $s \in [\Delta^{1/\beta}, C^{-1} \min \{\Delta,\lvert A\rvert/\log \Delta\}]\cap \Z$. 
We can compute multisets $A'\subsetneq \hat A\subseteq A$ and an $s$-proper origin-symmetric GAP $P\supseteq \supp(\hat A)$ of rank at most $d$, such that $\lvert A\rvert-\lvert\hat A\rvert\le Cs \log \Delta$ and
 $\lvert A'\rvert\le Cs$, and $A'$ is $s$-generating for $P$.

The algorithm runs in deterministic time 
$\tO_{\beta}\!\left(s\,\lvert s(\supp(A)\cup \{0\})\rvert\right)$, provided  $\lvert A\rvert\le \Delta^{\Oh(1)}$.
\end{theorem}

 Since
$s(\supp(A)\cup \{0\})\subseteq [-s \Delta,s \Delta]$, the time complexity can
also be bounded by $\tO_{\beta}(s^2 \Delta)$.

Our actual application requires a vector generalization of \cref{thm:algorithmiccfp}, which we state in \cref{subsec:vectorversion}. The same black-box generalization also appeared in \cite{pham2025sharp}.

Our proof of \cref{thm:algorithmiccfp} is a streamlined version  of
the original proof in \cite{conlon2023homogeneous}. We avoid some technical complications of their proof that are not algorithmically friendly, and make their central idea more transparent.
In particular, our adapted proof bypasses certain non-constructive definitions (e.g., \cite[Definitions 2.25, 2.28, 2.33]{conlon2023homogeneous}), and completely avoids the probabilistic steps in \cite[Section 2.5]{conlon2023homogeneous}.
We achieve this simplification by converting some of their indirect arguments into potential decrement steps in our algorithm.   
As in \cite{conlon2023homogeneous}, we use Freiman's theorem for low-doubling sets \cite{freiman93} (for us, its algorithmic version \cite{rw24}), but we do not require its strengthened version due to Bilu \cite{bilu} which was needed in \cite{conlon2023homogeneous}.

\subsection{Preliminaries}
\label{subsec:somelemmas}

A rank-$d$ \emph{generalized arithmetic progression (GAP) in $\R^k$}
is a set of the form
\[
 Q={\bfx}_Q+f_Q(B_Q\cap \Z^d)\subseteq \R^k,
\]
where ${\bfx}_Q\in \R^k$, $f_Q\colon \R^d\to \R^k$ is linear,
and $B_Q=\prod_{i=1}^d[a_i,b_i]$ is a box with integer endpoints $a_i<b_i$.
A GAP $Q$ is understood to include its representation $(\bfx_Q,f_Q,B_Q)$, since the same set may have different representations.
It has \emph{widths} $b_i-a_i+1 \ge 2$ and  \emph{volume} 
$\vol(Q)=\prod_{i=1}^d(b_i-a_i+1)\ge \lvert Q\rvert$.

We say $Q$ is \emph{origin-symmetric} if
$\bfx_Q=\zero$ and $a_i=-b_i$ for all $i$.
We say $Q$ is \emph{proper} if $f_Q$ is injective on $B_Q\cap \Z^d$.
For $r\in \Pos$, its $r$-fold sumset is the GAP
 $rQ=r \bfx_Q+f_Q((rB_Q)\cap \Z^d)$. Define the GAP $rQ-r' Q$ analogously.

We say $Q$ is an \emph{$r$-proper GAP} if the GAP $rQ$ is proper.

If $Q$ is proper and $\bfx_Q=\zero$, for $\bfq\in Q$,  let $f_Q^{-1}(\bfq)$ denote the unique
$\bfu\in B_Q\cap \Z^d$ such that $f_Q(\bfu)=\bfq$.

We need the following standard estimates:
\begin{lemma}
    \label{prop:growgap}
 Let $Q$ be a rank-$d$ GAP, and $p\in \Pos,r\in \Nonneg$. Then, $\lvert pQ-rQ\rvert\le (p+r)^d \lvert Q\rvert$, $\vol(pQ-rQ)\le (p+r)^d \vol(Q)$, and $\vol(pQ)\ge (p/2)^d \vol(Q)$.
\end{lemma}
\begin{proof}
    The first inequality follows from the observation that $pQ-rQ$ can be covered by the union of $(p+r)^d$ translates of $Q$.
    The second inequality follows similarly.

    Suppose $Q$ has widths $w_i$ for $i\in [d]$. 
 Then, $pQ$ has widths $p(w_i-1)+1\ge (p/2) w_i$ due to $w_i\ge 2$, giving the third inequality.
\end{proof}

We work with GAPs in $\Z$ until \cref{subsec:vectorversion} where we switch to GAPs in $\Z^k$. 

Randolph and W\k{e}grzycki \cite[Theorem~5]{rw24} proved
an algorithmic version of Freiman's theorem. We need the following strengthened version.

\begin{lemma}[Algorithmic Freiman's theorem]
\label{lem:rw24}
Let $K\ge 1$, and let
   $A \subseteq \Z \cap [-\Delta,\Delta]$ satisfy $\lvert A+A\rvert\le K\lvert A\rvert$ and $0\in A$.
   There is a deterministic algorithm that computes in $\Oh_K(\lvert A\rvert) \polylog \Delta$ time
   an origin-symmetric $2$-proper
   GAP $P$ of rank $d=\Oh_K(1)$ and volume $\Oh_K(\lvert A\rvert)$ which contains $A+A$, such that $\caL(f_P^{-1}(A))=\Z^d$, and the representation of $P$ has $\Oh_K(\log \Delta)$ bits.
\end{lemma}

Compared to the original randomized algorithm by \cite{rw24}, \cref{lem:rw24} is deterministic and additionally achieves $2$-properness and the $\caL(f_P^{-1}(A))=\Z^d$ condition (assuming $0\in A$).
In \cref{app:freiman} we describe how to strengthen \cite{rw24} to obtain \cref{lem:rw24}.

The following lemma of Conlon, Fox, and Pham
\cite{conlon2023homogeneous} shows that the sumset of many
dense subsets of an integer box contains a large box, provided each subset generates the integer lattice. It generalizes a result of Lev
\cite{lev} in the one-dimensional case.

\begin{lemma}[{\cite[Lemma 2.15]{conlon2023homogeneous}} rephrased]
\label{lem:fillboxreduced}
For any $c\in (0,1]$ and $d\in \Pos$, there are constants
$\gamma>0$ and $C\ge 1$ such that
the following holds.

Let $Q=\Z^d \cap \prod_{i=1}^d[-w_i,w_i]$ with $\min_i w_i\ge C$.
For $\ell\ge C$, let subsets
$A_1,\ldots,A_\ell\subseteq Q$ satisfy $\lvert A_j\rvert\ge c\lvert Q\rvert$, $\zero\in A_j$,
and $\caL(A_j)=\Z^d$ for all $j\in [\ell]$.
Then
$A_1+\cdots+A_\ell$ contains a translate of
$\lfloor\gamma \ell\rfloor Q$.
\end{lemma}

\subsection{Structure of \texorpdfstring{$s$}{s}-fold sumsets}

The following lemma is extracted and simplified from \cite[Lemmas 2.22 and 2.26]{conlon2023homogeneous}.
It embeds a given $\Delta$-bounded integer set into a rank-$d$ GAP $P$, so that its $s$-fold sumset, where $s\ge \Delta^{\Omega_d(1)}$, has size comparable to $\vol(sP)$.

\begin{lemma}
    \label{lem:hdim}
    For any $\beta \ge 1$, there is a constant $d_\beta\ge 1$
    such that the following holds.

 For any set $A \subseteq \Z\cap [-\Delta,\Delta]$ with $0\in A$ and any
 integer $s\ge \Delta^{1/\beta}$, there exists an
 origin-symmetric $2s$-proper GAP $P\supseteq A$ of rank $d\le d_{\beta}$ such that
 $\lvert sA\rvert\gg_\beta s^d\lvert P\rvert$ and $\caL(f_P^{-1}(A))=\Z^d$.

   Moreover, the GAP $P$ can be found by a deterministic algorithm in $\Oh_{\beta}(\lvert sA\rvert)\cdot \polylog(s \Delta)$ time.
\end{lemma}
The lower bound on $\lvert sA\rvert$ is nearly tight, as $\lvert sA\rvert\le \lvert sP\rvert\le s^d\lvert P\rvert$ by \cref{prop:growgap}.

To prove \cref{lem:hdim}, we use the following lemma:
\begin{lemma}
\label{lem:cfp-sumset-growth}
Given positive integers $K, q$, and a set $A\subseteq \Z\cap [-\Delta,\Delta]$
with $0\in A$ such that $\lvert2qA\rvert\le K\lvert qA\rvert$, we can 
deterministically compute an origin-symmetric
$2q$-proper GAP $P\supseteq A$ of rank $d=\Oh_K(1)$, such that
$q^d\lvert P\rvert\ll_K\lvert qA\rvert$ and $\caL(f_P^{-1}(A))=\Z^d$, in 
time $\Oh_K(\lvert qA\rvert) \polylog(q \Delta)$.
\end{lemma}
\begin{proof}
Assume $\lvert A\rvert\ge 2$; otherwise, $A=\{0\}$, and we return $P=\{0\}$.
We can compute $qA$ by repeated squaring using the sparse convolution algorithm (\cref{lem:sparse-convolution}), in total $ \Oh(\lvert qA\rvert) \polylog(q \Delta) $ time since each of the $\Oh(\log q)$ intermediate sumsets has size at most $\lvert qA\rvert$.
Then, apply the algorithmic Freiman's theorem (\cref{lem:rw24}) to $qA$ to obtain an origin-symmetric
$2$-proper GAP $Q=f_Q(B_Q\cap \Z^d)$ containing $2qA$, where
$B_Q=\prod_{i=1}^d[-w_i,w_i]$, of rank $d=\Oh_K(1)$ and size $\lvert Q\rvert\ll_K\lvert qA\rvert$, with $\caL(f_Q^{-1}(qA))=\Z^d$, in $\Oh_K(\lvert qA\rvert) \polylog(q \Delta)$ time.

We first claim that, for any 
 $g_1,\ldots,g_m\in A$ where $1\le m\le 2q$, it holds that
 \begin{equation}
     \label{eqn:fqg}
 f_Q^{-1}(g_1+\cdots+g_m)=\sum_{j=1}^m f_Q^{-1}(g_j).
 \end{equation}
To show this claim, denote
$x_j=\sum_{i=1}^j g_i$ for $0\le j\le m$.
Since $0\in A$, all $x_j$ belong to $2qA\subseteq Q$.
The equality $x_j+0=x_{j-1}+g_j$ and the $2$-properness of $Q$ imply
$f_Q^{-1}(x_j)=f_Q^{-1}(x_{j-1})+f_Q^{-1}(g_j)$ for $j\in [m]$. Then, \eqref{eqn:fqg} follows by induction on $j$.

     \Cref{eqn:fqg}
implies $f_Q^{-1}(qA)\subseteq qf_Q^{-1}(A)$, so
$\Z^d=\caL(f_Q^{-1}(qA))\subseteq \caL(f_Q^{-1}(A))$, and thus $\caL(f_Q^{-1}(A)) = \Z^d$.
For every $g\in A$, \eqref{eqn:fqg} also implies
$2qf_Q^{-1}(g)=f_Q^{-1}(2qg)\in B_Q\cap \Z^d$, and hence
$ f_Q^{-1}(g)\in (\frac{1}{2q} B_Q)\cap \Z^d$.
Combining these two properties  gives
$\Z^d = \caL(f_Q^{-1}(A)) \subseteq \caL((\frac{1}{2q} B_Q)\cap \Z^d)$, which implies the half-widths $w_i$ of the box $B_Q$ must satisfy $w_i\ge 2q$ for all $i\in [d]$.
Thus, the GAP
\[
 P=f_Q\left(B_P 
 \cap \Z^d\right),\,\, B_P = 
 \prod_{i=1}^d \left [-\left \lfloor \frac{w_i}{2q}\right \rfloor,\left \lfloor \frac{w_i}{2q}\right \rfloor\right ],
\]
contains $A$, with $f_P^{-1}(g)=f_Q^{-1}(g)$ for every $g\in A$.
Thus $\caL(f_P^{-1}(A))=\Z^d$ as claimed. Since $2qB_P\subseteq B_Q$
and $Q$ is proper, $P$ is $2q$-proper. Therefore,
\cref{prop:growgap} implies
$q^d\lvert P\rvert\le \vol(2qP)=\lvert2qP\rvert\le \lvert Q\rvert\ll_K\lvert qA\rvert$ as claimed.
\end{proof}

\begin{proof}[Proof of \cref{lem:hdim}]
Assume $\lvert A\rvert\ge 2$; otherwise, $A = \{0\}$, and we may return $P=\{0\}$.
We assume that $s$
is sufficiently large in terms of $\beta$; otherwise, $s=\Oh_\beta(1)$, and $\Delta\le s^\beta=\Oh_\beta(1)$, so we may return the one-dimensional arithmetic progression $P$ consisting of all multiples of $\gcd(A)$
in $[-\Delta,\Delta]$, which has size $\Oh_\beta(1)$ and satisfies $\lvert sA\rvert\ge s+1 \gg_{\beta} s\lvert P\rvert$ and $\caL(f_P^{-1}(A))=\Z$ as claimed.

Choose $h$ divisible by $4$ so that $2^h\le s<2^{h+4}$.
Then $h$ can be assumed to be sufficiently large in terms of $\beta$.
Since
$ \lvert2^h A\rvert \le 2\cdot 2^{h} \Delta +1 \le 3s^{\beta+1}$, by telescoping and averaging, there exists an integer $j\in [h/4,h/2-1]$ such that, writing $q=2^j \ge 2^{h/4}$,
\[
 \frac{\lvert2qA\rvert}{\lvert qA\rvert}
 \le \left (\frac{\lvert2^{h/2} A\rvert}{\lvert2^{h/4} A\rvert}\right )^{1/(h/2 - h/4)}\le (3s^{\beta+1})^{4/h}
 =(3s^{\beta+1})^{\Oh(1/\log s)}
 =2^{\Oh(\beta)}.
\]
Thus $\lvert2qA\rvert\le K_\beta\lvert qA\rvert$ for
$K_\beta=\Oh_\beta(1)$. By \cref{lem:cfp-sumset-growth}, there exists
an origin-symmetric $2q$-proper GAP $Q = f_Q(B_Q\cap \Z^{d'})\supseteq A$ of rank $d'=\Oh_\beta(1)$
such that $q^{d'}\lvert Q\rvert\ll_\beta\lvert qA\rvert$ and
$\caL(f_Q^{-1}(A))=\Z^{d'}$.

Let $T=qf_Q^{-1}(A)$ and $\ell=2^h/q \ge 2^{h/2}$. Then
$\zero\in T\subseteq qB_Q\cap \Z^{d'}$, $\caL(T)=\Z^{d'}$, and 
the density 
of $T$ in $qB_Q\cap \Z^{d'}$ can be bounded as 
\[
 \lvert qB_Q\cap \Z^{d'}\rvert=\vol(qQ) \mathrel{\underbrace{\le}_{\text{\cref{prop:growgap}}}} q^{d'}\lvert Q\rvert
 \ll_\beta\lvert qA\rvert=\lvert T\rvert.
\]
Since both $q\ge 2^{h/4}$ and $\ell \ge 2^{h/2}$ are sufficiently large in terms of $\beta$,
\cref{lem:fillboxreduced} applies to the box $qB_Q \cap \Z^{d'}$
and $\ell$ copies of $T$, showing that $\ell T = 2^h f^{-1}_Q(A)$ contains a translate of the box $rqB_Q\cap \Z^{d'}$ for some $r=\lfloor\gamma \ell\rfloor\gg_\beta \ell$.
Applying $f_Q$ shows that $2^h A$ contains a translate of $rqQ$.
Since $r\gg_\beta \ell$, $d'=\Oh_\beta(1)$, and $2^h\le s <2^{h+4} = 16 \ell q$, this implies
\begin{equation*}
 \lvert2sA\rvert\le \lvert32 \ell qQ\rvert\mathrel{\underbrace{\ll_\beta}_{\text{\cref{prop:growgap}}}}\lvert rqQ\rvert
 \le \lvert2^h A\rvert\le \lvert sA\rvert.
\end{equation*}
Therefore, \cref{lem:cfp-sumset-growth} applies to $A$ and $s$, and returns 
the required GAP $P$ in $\Oh_\beta(\lvert sA\rvert) \polylog(s \Delta)$ time.
\end{proof}

We record the following corollary:
\begin{corollary}
\label{lem:cfp-large-sumsets}
Let $\beta\ge 1$, $A\subseteq \Z\cap [-\Delta,\Delta]$, and integer
 $s\ge \Delta^{1/\beta}$. Then, for $m\in \Pos$,
$\lvert m s A- m sA\rvert\ll_{\beta,m}\lvert sA\rvert$.
\end{corollary}
\begin{proof}
We may assume $0\in A$ by translation invariance.
By \cref{lem:hdim}, there exists an origin-symmetric
GAP $P\supseteq A$ of rank $d=\Oh_\beta(1)$ satisfying
$s^d\lvert P\rvert\ll_\beta\lvert sA\rvert$. Thus \cref{prop:growgap} gives
\[
 \lvert msA-msA\rvert\le \lvert2msP\rvert
 \le (2ms)^d\lvert P\rvert\ll_{\beta,m}\lvert sA\rvert. \qedhere
\]
\end{proof}

\subsection{Growing subset sums}
In this section, we prove \cref{thm:algorithmiccfp} by adapting 
\cite[Sections 2.6 and 3]{conlon2023homogeneous}.
We will choose $C$ sufficiently large in terms of $\beta$.
The hypothesis in \cref{thm:algorithmiccfp} implies 
$\Delta\ge Cs$, so $\Delta$ is also sufficiently large in terms of $\beta$; we use this assumption later in the proof.

Let $d_\beta$ be the constant from \cref{lem:hdim}. Let $D_\beta\ge 1$ be the
implicit constant in \cref{lem:cfp-large-sumsets}
with $m=2$ and $2 \beta$ in place of $\beta$.
Let $\ell_\beta\in \Pos$, whose value will be determined later in the proof of
\cref{thm:algorithmiccfp}. Recall $s\ge \Delta^{1/\beta}$ from the hypothesis of \cref{thm:algorithmiccfp}.

Given the input multiset $A$, our algorithm is described in $\mathsf{Main}(A)$ of \cref{alg:main}. 
Starting from $\hat A = A$, each \textbf{while} iteration
constructs a submultiset $A'$ of the current remainder $\hat A$. If the condition at \cref{line:check} is violated, then we remove $A'$ from $\hat A$ and continue the loop. Otherwise, we return $\hat A$, $A'$, and a GAP computed from $\hat A$ as the desired answer. 
The submultiset $A'$ chosen at each iteration is $A'=G^{(1)}\sqcup \cdots\sqcup G^{(\ell_\beta)}$, where each $G^{(i)}$ is a size-$s$ submultiset chosen by \textsf{Greedy}.

\begin{algorithm}
\DontPrintSemicolon
\SetKwProg{Fn}{Function}{:}{}
\SetKwInput{KwInput}{Input}
\SetKwInput{KwOutput}{Output}
\SetKw{Break}{break}
\caption{Algorithm for \cref{thm:algorithmiccfp}}
\label{alg:main}
\Fn{$\mathsf{Main}(A)$}{
    $\hat A\gets A$\\
    \While{$\textup{\textsc{True}}$}{
        \lFor{$i\gets 1,2,\ldots,\ell_\beta$}{$G^{(i)}\gets \mathsf{Greedy}
            (\hat A\setminus (G^{(1)}\sqcup \cdots\sqcup G^{(i-1)}))$
            \label{line:maingreedy}}
        $A'\gets G^{(1)}\sqcup \cdots\sqcup G^{(\ell_\beta)}$\\
        \lIf{\label{line:check}$\lvert2s(\supp(\hat A)\cup \{0\})\rvert<
            (1+(2D_\beta)^{-1})\lvert2s(\supp(\hat A\setminus A')\cup \{0\})\rvert$}{\Break}
        \lElse{$\hat A\gets \hat A\setminus A'$}
    }
    Run \cref{lem:hdim} on 
$\supp(\hat A)\cup \{0\}$ and $s$ to compute a rank-$d$ GAP $P$ 
    \label{line:computegap}\\
    \Return{$\hat A,A',P$}
}
\Fn{$\mathsf{Greedy}(X)$}{
    \KwInput{a submultiset $X\subseteq A$ with $\lvert X\rvert\ge s$}
    \KwOutput{a submultiset $G\subseteq X$ with $\lvert G\rvert=s$}
    $G_0\gets \varnothing$\\
    \For{$j\gets 1,2,\ldots,s$}{
        Pick $g_j\in X\setminus G_{j-1}$ maximizing
        $\lvert\Sigma(G_{j-1}\sqcup \{g_j\})\rvert$\label{line:maximizergj}\\
        $G_j\gets G_{j-1}\sqcup \{g_j\}$
    }
    \Return{$G_s$}
}
\end{algorithm}

Viewing the sumset size $\lvert2s(\supp(\hat A)\cup \{0\})\rvert$ as a potential function, each failure of the check at \cref{line:check} causes a multiplicative potential drop, allowing us to bound the number of \textbf{while} iterations as follows:
\begin{lemma}
There is a constant $u_\beta\ge 2$ depending on $\beta$,
such that, if
\begin{equation}
  \label{eq:cfp-large-input}
 \lvert A\rvert \ge 2s \ell_\beta u_\beta \log \Delta,
\end{equation}
then $\mathsf{Main}(A)$ (\cref{alg:main}) is well-defined and terminates within
$u_{\beta} \log \Delta$ iterations of the {\rm \textbf{while}} loop, and returns submultisets 
$A'\subsetneq \hat A\subseteq A$ that satisfy $\lvert A'\rvert=s \ell_\beta$ and
$\lvert A\rvert-\lvert\hat A\rvert\le s \ell_\beta u_\beta \log \Delta$.
\label{lem:numiter}
\end{lemma}
\begin{proof}
The sumset size $\lvert 2s(\supp(\hat A)\cup \{0\})\rvert$ lies in $[1,4s \Delta +1]$ and is nonincreasing across the \textbf{while} iterations. Since each failure at \cref{line:check} decreases the size by at least a factor of $(1+(2D_\beta)^{-1})$, the total number of \textbf{while} iterations is at most
$1 + \log(4s \Delta+1)/\log(1+(2D_\beta)^{-1})\le u_\beta \log \Delta$ for some $u_\beta$ depending only on $\beta$.
Each \textbf{while} iteration makes $\ell_\beta$ calls to $\mathsf{Greedy}$, and a failed iteration removes a submultiset $A'$ of size $s \ell_\beta$ from $\hat A$. Therefore, the total number of removed elements is $\lvert A\rvert-\lvert\hat A\rvert \le s \ell_\beta u_\beta \log \Delta$, and
 there always remain enough elements in $\hat A$ assuming \eqref{eq:cfp-large-input}.
\end{proof}

For the rest of the analysis, assume \eqref{eq:cfp-large-input} holds,  and
consider the last \textbf{while}  iteration in $\mathsf{Main}(A)$. 
Let $G^{(1)},\ldots,G^{(\ell_\beta)}$ be the submultisets chosen at
\cref{line:maingreedy} in this iteration, and let the return values be
$\hat A,\, A' = G^{(1)}\sqcup \dots \sqcup G^{(\ell_\beta)}$, and $P=f_P(B_P\cap \Z^d)$, where
$B_P=\prod_{j=1}^d[-w_j,w_j]$ with $w_j\in \Pos$.
Define sets
\[Y_0=\supp(\hat A)\cup \{0\} \text{ and }
X_0=\supp(\hat A\setminus A')\cup \{0\}.\]
Since $A'\subsetneq \hat A$ and $0\notin \hat A$, we have $\lvert X_0\rvert\ge 2$.
Since the \textbf{if} condition at \cref{line:check} passes and $D_\beta \ge 1$, we have \begin{equation}
    \label{eqn:y0x0}
   \lvert2sY_0\rvert < 2\lvert2sX_0\rvert. 
\end{equation}

We now show that the multisets $G^{(i)}$ returned by $\mathsf{Greedy}$  have many distinct subset sums:
\begin{lemma}
    \label{lem:bilarge}
For every $i\in [\ell_\beta]$,
$\lvert\Sigma(G^{(i)})\rvert\gg_\beta\lvert sY_0\rvert$.
\end{lemma}
The proof of \cref{lem:bilarge} uses the following lemma from \cite{conlon2023homogeneous}:
\begin{lemma}[{\cite[Lemma 2.41]{conlon2023homogeneous}}]
    \label{lem:growelement}
   Let $S,A$ be finite nonempty sets of integers, and let $k\in \Pos$ be
   such that $\lvert kA\rvert\ge 2\lvert S\rvert$. Then there exists $a\in A$ such that
   $\lvert(S+a) \setminus S\rvert \ge \frac{\lvert S\rvert}{2k}$.
\end{lemma}

\begin{proof}[Proof of \cref{lem:bilarge}]
Define 
\[ t_0=1 \text{ and }t_h=\lvert2^h X_0\rvert/2 \text{ for } h \ge 1. \]
Since $\lvert X_0\rvert\ge 2$, $t_h\ge (2^h+1)/2$.

Fix $i\in [\ell_\beta]$ and consider the call to $\mathsf{Greedy}(X)$
with $X=\hat A\setminus (G^{(1)}\sqcup \cdots\sqcup G^{(i-1)})$, which returns $G^{(i)} = \{g_1,\dots,g_s\}$.
Write the prefix $G_j=\{g_1,\ldots,g_j\}$ for $0\le j\le s$. Then,
\[X_0\setminus \{0\} = \supp(\hat A \setminus A')\subseteq \supp(X \setminus G^{(i)} )\subseteq  \supp(X\setminus G_{j-1})\]
for every $j\in [s]$.
Write $\Sigma_j\coloneqq \Sigma(G_j)$.  
For $h\ge 1$, if $\lvert\Sigma_{j-1}\rvert\le t_h$, then \cref{lem:growelement} applied to $(S,A,k) = (\Sigma_{j-1}, X_0,2^h)$ shows that there exists $g \in X_0 \subseteq \supp(X \setminus G_{j-1})\cup \{0\}$ such that $\lvert(\Sigma_{j-1}+g)\setminus \Sigma_{j-1}\rvert \ge 2^{-h-1}\lvert\Sigma_{j-1}\rvert$, and thus $g\neq 0$. The greedy rule at \cref{line:maximizergj} then implies
\[
 \lvert\Sigma_j\rvert\ge (1+2^{-h-1})\lvert\Sigma_{j-1}\rvert.
\]
Thus, by telescoping, 
\begin{equation}
 \lvert\{j\in [s]:\lvert\Sigma_j\rvert\in [t_{h-1},t_h]\}\rvert
 \le 1+2^{h+1} \log \frac{t_h}{t_{h-1}}.
 \label{eqn:thth}
\end{equation}

Let $h^*$ be the largest integer such that $t_{h^*}\le \lvert\Sigma_s\rvert$.
Then $h^*\le \log \lvert\Sigma_s\rvert+1=\Oh(\log \Delta)$.
For $h>\frac{1}{2 \beta} \log \Delta$, \cref{lem:cfp-large-sumsets} gives $t_h/t_{h-1}=\Oh_\beta(1)$.
For the remaining $h$, the bound in \eqref{eqn:thth} is $\Oh(\Delta^{1/(2 \beta)} \log \Delta)$.
Thus, summing \eqref{eqn:thth} over all $h\in [1,h^*+1]$
gives
\[
 s\le \sum_{h=1}^{h^*+1}\left(1+2^{h+1} \log \frac{t_h}{t_{h-1}}\right)
 \le \Oh(\Delta^{1/(2 \beta)} \log^2 \Delta)+\Oh_\beta(2^{h^*}),
\]
where the first inequality is due to $t_0 \le \lvert\Sigma_1\rvert\le \dots \le \lvert\Sigma_{s}\rvert< t_{h^*+1}$.
Since $s\ge \Delta^{1/\beta}$, the inequality above implies $2^{h^*}\gg_\beta s$. In particular, $2^{h^*}\ge \Delta^{1/(2 \beta)}$ (for $\Delta$ sufficiently large in terms of $\beta$), so \cref{lem:cfp-large-sumsets} implies $\lvert2^{h^*} X_0\rvert\gg_\beta \lvert2sX_0\rvert$. Thus,
\[
 \lvert sY_0\rvert \le \lvert2s Y_0\rvert \mathrel{\underbrace{<}_{\text{\eqref{eqn:y0x0}}}} 2\lvert2sX_0\rvert
 \ll_\beta \lvert2^{h^*} X_0\rvert=2t_{h^*}\le 2\lvert\Sigma_s\rvert = 2\lvert\Sigma(G^{(i)})\rvert,
\]
which finishes the proof of \cref{lem:bilarge}.
\end{proof}

Now we analyze the rank-$d$ GAP $P$ computed at \cref{line:computegap}.
\begin{lemma}
\label{lem:gapgenerates}
$\caL(f_P^{-1}(X_0))=\Z^d$.
\end{lemma}
\begin{proof}
By \cref{lem:hdim}, $P$ is $2s$-proper and $\caL(f_P^{-1}(Y_0))=\Z^d$.
Suppose for contradiction that $\caL(f_P^{-1}(X_0))\ne \Z^d$.
Then there exists $y\in Y_0$ such that $f_P^{-1}(y)\notin \caL(f_P^{-1}(X_0))$, which implies 
\[
 2s f_P^{-1}(X_0)\cap \bigl(f_P^{-1}(y)+s f_P^{-1}(X_0)\bigr)=\varnothing.
\]
Since $2sP$ is proper, this implies $2sX_0\cap (y+sX_0)=\varnothing$. Since $0\in X_0\subseteq Y_0$, we have 
$2sX_0\subseteq 2sY_0$ and
$y+sX_0\subseteq 2sY_0$,
so $\lvert2sY_0\rvert\ge \lvert2sX_0\rvert+\lvert sX_0\rvert$.
By \cref{lem:cfp-large-sumsets} and the choice of $D_\beta$, $\lvert2sX_0\rvert\le D_\beta\lvert sX_0\rvert$. Hence,
$\lvert2sY_0\rvert\ge (1+1/D_\beta)\lvert2sX_0\rvert$,
which contradicts the \textbf{if} condition at \cref{line:check}, since $(2D_\beta)^{-1}<1/D_\beta$.
\end{proof}

\begin{lemma}
\label{lem:eachbigenerates}
For every $i\in [\ell_\beta]$, $\caL(f_P^{-1}(G^{(i)}))=\Z^d$.
\end{lemma}
\begin{proof}
Fix $i$ and suppose for contradiction that $L=\caL(f_P^{-1}(G^{(i)}))\ne \Z^d$.
Then, by \cref{lem:gapgenerates}, there exists $a\in \supp(\hat A\setminus A')$ such that $f_P^{-1}(a)\notin L$.
For every step $j$ in the call of $\mathsf{Greedy}$ returning $G^{(i)}$, $a$ is available at \cref{line:maximizergj}.
We claim $\Sigma(G_{j-1})\cap (a+\Sigma(G_{j-1}))=\varnothing$, and thus the greedy choice at \cref{line:maximizergj} must yield $\lvert\Sigma(G_j)\rvert=2\lvert\Sigma(G_{j-1})\rvert$. Indeed, if $a+\sum_{x\in X'} x=\sum_{y\in Y'} y$ for some submultisets $X',Y'\subseteq G_{j-1}\subseteq G^{(i)}$, the $s$-properness of $P$ would give
\[
 f_P^{-1}(a)=\sum_{y\in Y'} f_P^{-1}(y)-\sum_{x\in X'} f_P^{-1}(x)\in L,
\]
a contradiction.

Since $\lvert\Sigma(G_j)\rvert=2\lvert\Sigma(G_{j-1})\rvert$ for every $j\in [s]$, 
\[
 2^s=\lvert\Sigma(G_s)\rvert\le 2s \Delta+1\le 3 \Delta^2,
\]
which contradicts $s\ge \Delta^{1/\beta}$ for sufficiently large $\Delta$. Thus, $\caL(f_P^{-1}(G^{(i)}))=\Z^d$ as claimed.
\end{proof}

\begin{proof}[Proof of \cref{thm:algorithmiccfp}]
We first analyze the time complexity of $\textsf{Main}(A)$ (\cref{alg:main}).
Consider step $j$ of $\mathsf{Greedy}(X)$, and write $S=\Sigma(G_{j-1})$ and $U=\supp(X)\cup \{0\}$.
For all $g\in \Z$, we have $\lvert\Sigma(G_{j-1}\sqcup \{g\})\rvert=2\lvert S\rvert-r_g$, where $r_g=\lvert S\cap (S+g)\rvert$.
Using sparse convolution (\cref{lem:sparse-convolution}), we compute all nonzero $r_g$ in $\Oh(\lvert S-S\rvert\polylog \Delta)$ time.
Then, we find the desired maximizer $g_j$ by scanning $\supp(X\setminus G_{j-1})$, and compute $\Sigma(G_j)=S\cup (S+g_j)$.
Since $S\subseteq sU$,
\[
 \lvert S-S\rvert+\lvert\supp(X)\rvert\le 2\lvert sU-sU\rvert\mathrel{\underbrace{\ll_\beta}_{\text{\cref{lem:cfp-large-sumsets}}}}\lvert sU\rvert.
\]
Hence all $s$ steps in $\mathsf{Greedy}(X)$ take $\Oh_\beta(s\lvert sU\rvert\polylog \Delta)$ total time.
The two $2s$-fold sumsets at \cref{line:check} can be computed by repeated squaring with sparse convolution (\cref{lem:sparse-convolution}) in $\lvert2s(\supp(A)\cup \{0\})\rvert\polylog \Delta \ll_{\beta} \lvert s(\supp(A)\cup \{0\})\rvert\polylog \Delta$ time by \cref{lem:cfp-large-sumsets}.
Each \textbf{while} iteration makes $\ell_\beta=\Oh_\beta(1)$ greedy calls, and \cref{lem:numiter} bounds the number of iterations by $\Oh_\beta(\log \Delta)$.
Finally, \cref{lem:hdim} at \cref{line:computegap} takes $\Oh_\beta(\lvert s(\supp(\hat A)\cup \{0\})\rvert\polylog \Delta)$ time.
Therefore, the total time complexity is $\Oh_\beta(s\lvert s(\supp(A)\cup \{0\})\rvert\polylog \Delta)$, as claimed.

Now we prove the correctness of \cref{alg:main}.
By \cref{lem:numiter}, $A'\subseteq \hat A\subseteq A$, $\lvert A'\rvert=s \ell_\beta$, and $\lvert A\rvert-\lvert\hat A\rvert\le s \ell_\beta u_\beta \log \Delta$.
By \cref{lem:hdim}, $P = f_P(B_P\cap \Z^d)\supseteq \supp(\hat A)\cup \{0\}$ is an origin-symmetric $2s$-proper GAP of rank $d\le d_\beta$.
Next, it remains to prove that $A'= 
G^{(1)}\sqcup \cdots\sqcup G^{(\ell_\beta)}$ is $s$-generating for $P$.

For every $i\in [\ell_\beta]$, we have $\zero\in \Sigma(f_P^{-1}(G^{(i)}))\subseteq s B_P \cap \Z^d$ and
\[
\begin{aligned}
 \lvert\Sigma(f_P^{-1}(G^{(i)}))\rvert=\lvert\Sigma(G^{(i)})\rvert
 &\mathrel{\underbrace{\gg_\beta}_{\text{\cref{lem:bilarge}}}} \lvert s(\supp(\hat A)\cup \{0\})\rvert\mathrel{\underbrace{\gg_\beta}_{\text{\cref{lem:hdim}}}} s^d\lvert P\rvert
\ge \lvert sB_P\cap \Z^d\rvert.
\end{aligned}
\]
By \cref{lem:eachbigenerates}, $\caL(\Sigma(f_P^{-1}(G^{(i)})))=\Z^d$.
The box $sB_P\cap \Z^d$ has half-widths at least $s\ge \Delta^{1/\beta}$.
Thus, after choosing $\ell_\beta$ sufficiently large in terms of $\beta$, \cref{lem:fillboxreduced} applies to the dense sets $\Sigma(f_P^{-1}(G^{(i)}))$ ($i\in [\ell_\beta]$) in the box  $sB_P\cap \Z^d$, and gives $\boldsymbol{\tau}\in \Z^d$ such that
\[
 (\boldsymbol{\tau}+sB_P)\cap \Z^d
 \subseteq
 \sum_{i=1}^{\ell_\beta} \Sigma(f_P^{-1}(G^{(i)}))
 =
 \Sigma(f_P^{-1}(A')).
\]
Therefore, by \cref{def:generating}, $A'$ is $s$-generating for $P$, as desired.
The proof is finished
after choosing the constant $C$ in \cref{thm:algorithmiccfp} to be
 larger than $2 \ell_\beta u_\beta$ and sufficiently large in terms of $\beta$ as mentioned earlier, and setting the rank upper bound in \cref{thm:algorithmiccfp} to $d_\beta$.
\end{proof}

\subsection{Generalization to vectors}
\label{subsec:vectorversion}
Following Pham and Zakharov \cite{pham2025sharp}, we generalize
the algorithmic Conlon--Fox--Pham theorem (\cref{thm:algorithmiccfp}) from integers to vectors.
We also bound the widths of the obtained GAP and its $s$-fold sumset size for later applications in
\cref{sec:classify}.

\begin{theorem}
\label{thm:algorithmiccfp-highdim}
For any $\beta>1$ and $k\in \Pos$, there are constants $C>1,d\ge 1$
such that the following holds.

Let $A$ be a multiset of nonzero vectors in
$(\Z\cap [-\Delta,\Delta])^k$, and $s \in [\Delta^{1/\beta},\allowbreak C^{-1} \min \{\Delta,\lvert A\rvert/\log \Delta\}]\cap \Z$.
We can compute multisets $A'\subsetneq \hat A\subseteq A$ and an
$s$-proper origin-symmetric GAP $P\supseteq \supp(\hat A)$ of rank
$m\le d$, such that $\lvert A\rvert-\lvert\hat A\rvert\le Cs \log \Delta$ and
$\lvert A'\rvert\le Cs$, and $A'$ is $s$-generating for $P$.

Moreover, writing $P=f_P(B_P\cap \Z^m)$ with
$B_P=\prod_{j=1}^m[-w_j,w_j]$ and $w_j\in \Pos$, we have
$ w_j\le C \Delta$ and 
$\lvert sP\rvert=\vol(sP)\le C\lvert s(\supp(\hat A)\cup \{\zero\})\rvert$.

The algorithm runs in deterministic time
$\tO_{\beta,k}\!\left(s\,\lvert s(\supp(A)\cup \{\zero\})\rvert\right)$
provided $\lvert A\rvert\le \Delta^{\Oh(1)}$.
\end{theorem}

We prove
\cref{thm:algorithmiccfp-highdim} in
\cref{app:vector-cfp-proof} by a black-box application of 
\cref{thm:algorithmiccfp}.

\section{Basic Geometric Algorithms}
\label{sec:geometry-oracles}

In this section, we collect some standard algorithms for optimization and lattices.

\paragraph*{Rational arithmetic.}
\label{par:rational-arithmetic}
We use exact rational arithmetic in geometric algorithms. 
For a rational number $r=p/q$ in lowest terms, recall $\maxbitlen(r) \coloneqq \max \{\log(2+\lvert p\rvert),\log(2+\lvert q\rvert)\}$, so $r$ can be represented in $\Oh(\maxbitlen(r))$ bits.

All rational numbers $r$ encountered in our algorithm for \cref{thm:main-main} have $\maxbitlen(r)=(\log \Delta)^{\Oh_d(1)}$, so the time complexities of arithmetic operations are lower-order factors of the final time complexity. Nevertheless, we keep track of the input and output bit lengths in lemma statements for completeness.

We extend the 
$\maxbitlen(\cdot)$ notation to the following objects.

For a rational matrix $M$, define $\maxbitlen(M) = \max_{i,j} \maxbitlen(M_{i,j})$; define
$\maxbitlen(\bfv)$ for a rational vector $\bfv$ similarly.
For a rational linear map $f$, let $\maxbitlen(f)=\maxbitlen(M_f)$ where $M_f$ is its matrix in the standard bases.
For a rational polytope $K$, let
$\maxbitlen(K)= \max_{\bfv} \maxbitlen(\bfv)$ where $\bfv$ ranges over its vertices.

For a rational lattice $L$, let $\maxbitlen(L) = \maxbitlen(H_L)$ where the matrix $H_L$ is the HNF representation  of $L$ (see \cref{sec:hermite}).

For a multiset $A$, define $\maxbitlen(A) = \max \big \{\log(2+\lvert A\rvert),\max_{a\in \supp(A)} \maxbitlen(a)\big \}$.

We use the shorthand $\maxbitlen(x_1,\dots,x_k) = \max_{1\le i\le k} \maxbitlen(x_i)$.
In these definitions, a maximum over an empty set is taken to be $1$.

To avoid distraction, in the proofs we often omit routine arguments for bounding the output bit lengths. For example, 
we often implicitly use the following simple fact:
 For a rational polytope $K\subseteq \R^d$, a vector $\bfu\in \Q^d$,
and a lattice $L\subseteq \Q^d$, if 
$ \bfx\in K\cap (\bfu+L)$, then 
$ \maxbitlen(\bfx)=\Oh_d\bigl(\maxbitlen(K,\bfu,L)\bigr)$.

\subsection{Linear optimization oracles}
Recall the notion of a linear optimization oracle for a polytope.
Here we use the convention that the oracle always returns a vertex of the polytope.
\begin{definition}[$\Tquery(K)$]
\label{def:linear-optimization}
For a rational polytope $K\subseteq \R^d$, a \emph{linear optimization
oracle} receives a query $\bfc\in \Q^d$ and
returns a vertex $\bfx$ of $K$ that maximizes $\langle \bfc, \bfx\rangle$ (or reports $K=\varnothing$).

Let $\Tquery(K)$ be such that $K$ has a linear optimization oracle that answers a query
$\bfc$ in $\Tquery(K)\cdot (\maxbitlen(\bfc))^{\Oh_d(1)}$ time.
\end{definition}

One may also consider an \emph{arbitrary-maximizer} linear optimization oracle,
where the returned rational maximizer $\bfx\in K$ is not necessarily a vertex of $K$.
By \cite[Remark~6.5.2]{GLS88}, arbitrary-maximizer linear optimization oracles
are polynomially equivalent to (vertex-maximizer) linear optimization oracles.
Here and below, \emph{polynomial equivalence} between two oracles means the following
in both directions. Given a bound $b_0\ge \maxbitlen(K)$, a query $\bfq$ with
$b=\maxbitlen(\bfq)$ to one oracle for $K$ can be answered by making
$(b_0+b)^{\Oh_d(1)}$ queries $\bfq_i$ to the other oracle for $K$, with
$\maxbitlen(\bfq_i)\le (b_0+b)^{\Oh_d(1)}$, and spending
$(b_0+b)^{\Oh_d(1)}$ additional time.

\begin{lemma}
\label{lem:polytope-span}
Given a linear optimization oracle for a rational polytope $K\subseteq \R^d$ containing $\zero$,
one can compute $d'=\dim K$ and linearly independent vertices
$\bfp_1,\ldots,\bfp_{d'}$ of $K$ in $b_0^{\Oh_d(1)} \Tquery(K)$ time,
where $b_0=\maxbitlen(K)$.
\end{lemma}
\begin{proof}
Repeat the following for $i=0,1,\dots,d$: Write $V_i=\spa(\bfp_1,\dots,\bfp_i)$, and compute a basis $(\bfb_1,\dots,\bfb_{d-i})$ of $V_i^\perp$ by Gaussian elimination.
Query $\max_{\bfx\in K}\langle\pm \bfb_j,\bfx\rangle$ for all $1\le j\le d-i$. If any of these maxima is positive, then let $\bfp_{i+1}$ be the corresponding maximizing vertex of $K$, and continue to iteration $i+1$.
Otherwise, terminate the loop, and we have $d'=\dim K = i$ and $K\subseteq V_{d'}$, so we can return $\bfp_1,\dots,\bfp_{d'}$.
\end{proof}

Convex hulls and zonotopes have linear optimization oracles with query time near-linear in the number of input vectors:
\begin{lemma}
\label{lem:zonotope-algorithms}
Let $X$ be a multiset of vectors in $\Z^d$.
Then $\maxbitlen \big (\conv(\supp(X)),\caZ(X)\big )=\Oh_d(b_0)$ and
$\Tquery \big (\conv(\supp(X))\big),\Tquery \big (\caZ(X)\big )\le (\lvert\supp X\rvert+1)\cdot b_0^{\Oh_d(1)}$, where $b_0=\maxbitlen(X)$.
\end{lemma}

\begin{proof}
   Every vertex of $\conv(\supp(X))$ is a vector in $X$, and
every vertex of $\caZ(X)$ is a subset sum of $X$,
giving the claimed maximum bit length bounds.

Let $\bfc\in \Q^d$ be the given query.
The oracle for $\conv(\supp(X))$ returns the $\bfx\in \supp X$ that lexicographically maximizes $(\langle\bfc,\bfx\rangle,x_1,\ldots,x_d)$.
The oracle for $\caZ(X)$ returns the sum of $\mu_X(\bfx) \bfx$ over $\bfx\in \supp X$ for which
the first nonzero coordinate of
$(\langle\bfc,\bfx\rangle,x_1,\ldots,x_d)$ is positive.
This returns a point $\bfz \in \caZ(X)$ that lexicographically maximizes $(\langle \bfc,\bfz\rangle, z_1,\dots,z_d)$ and hence is a vertex of $\caZ(X)$ maximizing 
$\langle \bfc,\bfz\rangle$.
\end{proof}

We can build linear optimization oracles after taking intersection with a given polyhedron, and/or applying a linear map:
\begin{lemma}
\label{lem:constrained-optimization}
Let $K_0\subseteq \R^d$ be a rational polytope given by a linear optimization
oracle.
Let $A\in \Q^{m\times d}$, $\bfa\in \Q^m$, and $f\colon \R^d\to \R^{d'}$ be a rational linear map ($d'\le d$).
Define
$K=f\left(\left\{\bfx\in K_0:A \bfx\le \bfa\right\}\right)$.
Let $b_0=\maxbitlen(K_0,A,\bfa,f)$.
Then $\maxbitlen(K)=\Oh_d(b_0)$ and
$\Tquery(K)\le (\Tquery(K_0)+m) b_0^{\Oh_d(1)}$.
\end{lemma}

For the proof of \cref{lem:constrained-optimization}, recall that a \emph{strong separation oracle} for a rational polytope $K\subseteq \R^d$ receives a query point $\bfx\in \Q^d$, and either reports $\bfx\in K$, or returns
$\bfu\in \Q^d\setminus \{0\}$ and $a\in \Q$ satisfying
$\langle\bfu,\bfy\rangle\le a$ for all $\bfy \in K$
whereas 
$\langle\bfu,\bfx\rangle>a$.
By \cite[Theorem~6.4.9]{GLS88},
 linear optimization oracles and strong separation oracles are polynomially equivalent.

\begin{proof}[Proof of \cref{lem:constrained-optimization}]
By the aforementioned oracle equivalences, we have a 
strong separation oracle for $K_0$.
This immediately gives a strong separation oracle for $P
\coloneqq \{\bfx\in K_0:A \bfx\le \bfa\}$: on a query point $\bfx$,
 use the oracle for $K_0$ to certify $\bfx\in K_0$ (or return a separator of $\bfx$ from $K_0$), and then certify $A \bfx\le \bfa$
(or return the violated inequality as a separator).
By the aforementioned oracle equivalences,
we hence have a linear optimization oracle for $P$.
Given $\bfc\in \Q^{d'}$, use this oracle to find $\bfx\in P$ that maximizes $\langle \bfc, f(\bfx)\rangle$,  
and return $f(\bfx)$.
This gives an arbitrary-maximizer linear optimization oracle for $K = f(P)$.
By the aforementioned oracle equivalences, we get a linear optimization oracle for $K$.

Cramer's rule gives
$\maxbitlen(P)=\Oh_d(\maxbitlen(K_0,A,\bfa)) = \Oh_d(b_0)$.
Since $K=f(P)$, we have 
$\maxbitlen(K)=\Oh_d(\maxbitlen(f,P)) =\Oh_d(b_0) $.
The time bound follows from the polynomial equivalences and the additional time for checking each of the $m$ inequalities in $A \bfx \le \bfa$.
\end{proof}

We record some immediate corollaries of  \cref{lem:constrained-optimization}.
\begin{corollary}
\label{cor:polytope-oracles}
\leavevmode
\begin{enumerate}
\item Suppose a polytope $K$ is given by
$K=\{\bfx\in \R^d:A \bfx\le \bfa\}$, where $A\in \Q^{m\times d}$,
$\bfa\in \Q^m$, and $b_0=\maxbitlen(A,\bfa)$.
Then $\maxbitlen(K)=\Oh_d(b_0)$ and $\Tquery(K)=m b_0^{\Oh_d(1)}$.
\item Let $K_1,K_2\subseteq \R^d$ be rational polytopes given by linear
optimization oracles, and let $b_0=\maxbitlen(K_1,K_2)$.
Then $\maxbitlen(K_1\cap K_2)=\Oh_d(b_0)$ and
$\Tquery(K_1\cap K_2)\le (\Tquery(K_1)+\Tquery(K_2)) b_0^{\Oh_d(1)}$.
\end{enumerate}
\end{corollary}
\begin{proof}
For the first part, Cramer's rule gives $\maxbitlen(K)=\Oh_d(b_0)$.
Hence $K$ lies in a box $K_0$ with
$\maxbitlen(K_0)=\Oh_d(b_0)$ and $\Tquery(K_0)=b_0^{\Oh_d(1)}$.
Then apply \cref{lem:constrained-optimization} with $f$ being the identity map.

For the second part, the polytope $K_1\times K_2\subseteq \R^{2d}$ has a linear optimization oracle by separately optimizing
over $K_1$ and $K_2$ and returning the pair of vertices $(\bfx_1,\bfx_2)$.
Thus  
$\Tquery(K_1\times K_2)\le \Tquery(K_1)+\Tquery(K_2)$ and $\maxbitlen(K_1\times K_2)\le b_0$.
Then, apply \cref{lem:constrained-optimization} 
to $K_0=K_1\times K_2$, with the $2d$ inequalities expressing $\bfx_1=\bfx_2$,
and $f(\bfx_1,\bfx_2) = \bfx_1$.
\end{proof}

\begin{lemma}
\label{lem:polar-optimization}
Let $K\subseteq \R^d$ be a rational polytope containing $\zero$ in its interior,
given by a linear optimization oracle, and let $b_0=\maxbitlen(K)$.
Its polar
\[
K^\circ=\{\bfa\in \R^d:\langle\bfa,\bfx\rangle\le 1
\text{ for every }\bfx\in K\}
\]
is a rational polytope satisfying $\maxbitlen(K^\circ)=\Oh_d(b_0)$ and
$\Tquery(K^\circ)\le b_0^{\Oh_d(1)} \Tquery(K)$.
\end{lemma}
\begin{proof}
A strong separation oracle for 
$K^\circ$ can be directly implemented using a linear optimization oracle for $K$, so
the claimed time bound 
follows from the aforementioned oracle equivalences.
The bound on $\maxbitlen(K^\circ)$ directly follows from Cramer's rule (see also \cite[Lemma~6.2.4]{GLS88}).
\end{proof}

\subsection{Polytope approximation}

\begin{lemma}[cf.\ \cite{AwerbuchKleinberg08}]
\label{lem:bounded-coordinate-selection}
Let $K\subseteq \R^d$ be a nonempty rational polytope given by a linear optimization oracle, $d'=\dim K$, and $b_0=\maxbitlen(K)$.
Then, one can find
vertices $\bfp_0,\ldots,\bfp_{d'}$ of $K$ such that
\[
 K\subseteq \left\{\sum_{i=0}^{d'} a_i \bfp_i:\sum_{i=0}^{d'} a_i=1,\ \lvert a_i\rvert\le 2\right\},
\]
in $b_0^{\Oh_d(1)} \Tquery(K)$ time.
\end{lemma}

The set of vertices $\bfp_0,\ldots,\bfp_{d'}$ is called  (an affine version of) a barycentric spanner in \cite{AwerbuchKleinberg08}.
Their existence (for any compact convex set $K$) follows from   picking the vertices of a maximum-volume inscribed simplex of $K$, which gives the coefficient bound $\lvert a_i\rvert\le 1$.
The constructive proof below
achieving $\lvert a_i\rvert\le 2$
follows \cite{AwerbuchKleinberg08}.

\begin{proof}[Proof of \cref{lem:bounded-coordinate-selection}]
Obtain an arbitrary vertex $\bfp_0$ of $K$ by querying the oracle.
Apply \cref{lem:polytope-span} to $K-\bfp_0$ to compute $d' = \dim K$
and affinely independent vertices $\bfp_0,\bfp_1,\dots,\bfp_{d'} \in K$.

In the following, assume $d'=d$ for simplicity; the general case can be similarly proved by restricting to the affine hull of $K$.

Starting from $\bfp_0,\dots,\bfp_{d}$, repeat the following operation. For $\bfx\in \R^d$, define $(a_0(\bfx),\dots,a_d(\bfx)) \in \R^{d+1}$ as the unique tuple so that $\bfx=\sum_{i=0}^d a_i(\bfx) \bfp_i$ and $\sum_{i=0}^d a_i(\bfx)=1$.
Each $a_i(\bfx)$ is an affine function of $\bfx$.
Query $\max_{\bfx\in K}\pm a_i(\bfx)$ for all $0\le i\le d$.
If any of these maxima exceeds $2$, then replace $\bfp_i$ by the corresponding maximizing vertex of $K$ and repeat.
Otherwise, terminate the loop, and return the current $\bfp_0,\dots,\bfp_d$ as the answer. 
The termination condition immediately gives the claimed containment.

To bound the number of iterations, observe that 
the volume of the simplex $\conv \{\bfp_0,\dots,\bfp_d\}$ is multiplied by $\lvert a_i(\bfx)\rvert>2$ at each replacement.
Since $\maxbitlen(\bfp_i)\le b_0$, this volume lies between $2^{-\Oh_d(b_0)}$ and $2^{\Oh_d(b_0)}$ throughout.
Hence there are $\Oh_d(b_0)$ replacements, each requiring $\Oh_d(1)$ oracle calls.
Together, the total time is $b_0^{\Oh_d(1)} \Tquery(K)$.
\end{proof}

\begin{lemma}
\label{lem:macbeath-center}
Let $K\subseteq \R^d$ be a convex body.
There are an origin-symmetric parallelepiped $B$ and a point $\bfc$ such that
$\bfc+B\subseteq K\subseteq \bfc+\beta_d B$
for some $\beta_d\ge 1$ depending only on $d$.

If $K$ is a rational polytope given by a linear optimization oracle,
and $b_0=\maxbitlen(K)$, then one can compute such rational $B,\bfc$
with $\maxbitlen(B,\bfc)=\Oh_d(b_0)$
in $b_0^{\Oh_d(1)} \Tquery(K)$ time.
\end{lemma}

\begin{proof}
   By \cref{lem:bounded-coordinate-selection}, we can find $\bfv_0,\ldots,\bfv_d\in K$, such that every $\bfx\in K$ can be  written as 
$\bfx=\sum_{i=0}^d a_i(\bfx) \bfv_i$, $\sum_{i=0}^d a_i(\bfx)=1$, with $\lvert a_i(\bfx)\rvert\le 2$ for all $i$.

Define $\bfc=(\sum_{i=0}^d \bfv_i)/(d+1)$, $\bfq_i=\bfv_i-\bfv_0$ for $1\le i\le d$,
and $B=\frac{1}{d(d+1)} \sum_{i=1}^d[-\bfq_i,\bfq_i]$. We show that $\bfc$ and $B$ satisfy the claimed property.

For any $\bfx \in K$,  \[\bfx-\bfc=\sum_{i=1}^d\left (a_i(\bfx)-\frac{1}{d+1}\right ) \bfq_i\]
has coefficients of absolute value at most $3$, proving
$K\subseteq \bfc+3d(d+1) B$.

Conversely, for any $\bfy \in \bfc+B$, writing
$ \bfy = \bfc+\sum_{i=1}^d t_i \bfq_i $ with
 $\lvert t_i\rvert\le 1/(d(d+1))$, we have
\[
\bfy = \left(\frac{1}{d+1}-\sum_{i=1}^d t_i\right) \bfv_0
  +\sum_{i=1}^d\left(\frac{1}{d+1}+t_i\right) \bfv_i.
\]
The coefficients are nonnegative and sum to $1$, so 
$\bfy \in K$ by convexity of $K$, proving $\bfc+B\subseteq K$.

\end{proof}

\subsection{Lattice algorithms}
\label{sec:hermite}

A matrix $H\in \Z^{d\times d'}$ with linearly independent columns is in
\emph{(column) Hermite normal form} (HNF) if, writing $r_j = \max \{i: H_{i,j}\neq 0\}$, the following properties hold:
\begin{itemize}
   \item $r_1<\cdots<r_{d'}$,
   \item $H_{r_j,j}>0$ for all $1\le j\le d'$,
   \item $0\le H_{r_j,k}<H_{r_j,j}$ for all $1\le j<k\le d'$.
\end{itemize}
It is well-known that every integer lattice $L\subseteq \Z^d$ with $\rk L = d'$ has a unique HNF representation $H \in \Z^{d\times d'}$ such that $L=H\Z^{d'}$.

For a rational lattice $L\subseteq \Q^d$, writing $q$ for the least positive
integer with $qL\subseteq \Z^d$, and writing $H$ for the unique HNF representation of $qL$,
the unique HNF representation of $L$ is $H_L= H/q$.

We represent a rational lattice $L$ by its unique HNF representation $H_L$ by default, and denote $\maxbitlen(L) = \maxbitlen(H_L)$. 
The following lemma computes the HNF representation of a lattice generated by input vectors: 

\begin{lemma}[{\cite[Theorem~5]{StorjohannLabahn96}}]
\label{lem:lattice-basis-hnf}
Given $q\in \Pos$ and $\bfz_1,\ldots,\bfz_m\in \Z^d$ with
$b_0=\maxbitlen(q,\allowbreak \bfz_1,\ldots,\bfz_m)$,
one can compute the HNF representation $H_L$ of the lattice
$L=\caL(\bfz_1/q,\allowbreak \ldots,\bfz_m/q)$
 with
$\maxbitlen(H_L)=\Oh_d(b_0)$
in $\tO\!\left((m+1) \poly(d,b_0)\right)$ time.
\end{lemma}

The following lemma computes the kernel and image of a linear map on a lattice:
\begin{lemma}[{e.g., \cite[Section~14, ``The kernel and image algorithm'']{Lenstra08Lattices}}]
\label{lem:lattice-kernel-image}
Given a lattice $\Gamma\subseteq \Q^d$ and a rational linear map
$f\colon \R^d\to \R$ with
$b_0=\maxbitlen(\Gamma,f)$, one can compute 
the lattice $\Lambda=\Gamma\cap \ker f$ and $g\in \Q_{\ge 0}$ such that
$f(\Gamma)=g\Z$,
with $\maxbitlen(g,\Lambda)=\Oh_d(b_0)$  in $\poly(d,b_0)$ time.
\end{lemma}

The Lenstra--Lenstra--Lov\'asz algorithm computes a nearly orthogonal basis of a given lattice:
\begin{lemma}[{\cite{LLL}}]
\label{lem:latticereduct}
Every lattice $L\subseteq \R^d$ with $\rk L=d'$ has a basis
$(\bfb_i)_{i=1}^{d'}$ satisfying
$
 \prod_{i=1}^{d'}\|\bfb_i\|_2\le 2^{(d')^2} \det L
 $.

If $L$ is rational and $b_0=\maxbitlen(L)$, one can compute such a basis
with $\maxbitlen(\bfb_1,\ldots,\bfb_{d'})=\Oh_d(b_0)$
in $\poly(d,b_0)$ time.
\end{lemma}

Lenstra's algorithm solves integer linear programming (ILP) in fixed dimension.
The algorithm applies when the input polytope is given as a linear optimization oracle \cite[Section~2, Remark~(d)]{Lenstra83}.

\begin{lemma}[Fixed-dimension ILP \cite{Lenstra83}]
\label{lem:lattice-optimization}
Given a rational polytope $K\subseteq \R^d$ with a
linear optimization oracle, a vector $\bfc\in \Q^d$, and a
lattice $L\subseteq \Q^d$, writing $F=K\cap L$,
one can find $\bfx\in F$ that maximizes $\langle \bfc,\bfx\rangle$ (or report $F=\varnothing$) in $b_0^{\Oh_d(1)} \Tquery(K)$ time,
where $b_0=\maxbitlen(K,\bfc,L)$.
\end{lemma}
The original result in \cite{Lenstra83} was stated only for deciding whether $K\cap \Z^d = \varnothing$. The form stated in \cref{lem:lattice-optimization} follows from restricting to $\spa(L)$ and applying a linear map, and binary search 
for optimization, using \cref{lem:constrained-optimization}.

As a simple corollary, we can enumerate the values of a linear map on lattice
points in a polytope.

\begin{lemma}
\label{lem:quotient-lattice-enumeration}
Given a rational lattice $\Gamma\subseteq \Q^d$,
a rational linear map $f\colon \R^d\to \R$,
 and $A\in \Q^{m\times d}, \bfa \in \Q^{m}$ which define a polytope $K=\{\bfx\in \R^d:A \bfx\le \bfa\}$,  one can enumerate $f(K\cap \Gamma)$ in $(\lvert f(K\cap \Gamma)\rvert+1) m b_0^{\Oh_d(1)}$ time where $b_0=\maxbitlen(\Gamma,f,A,\bfa)$.
\end{lemma}
\begin{proof}
By clearing denominators in $\Gamma$ and $f$, we may assume without loss of generality that
$f(\Gamma) \subseteq \Z$. 
Using the linear optimization oracle for $K$ by \cref{cor:polytope-oracles},
we can minimize $f$ over $K\cap \Gamma$ by \cref{lem:lattice-optimization}.
After returning a value $a$, minimize $f$ over
$\{\bfx:f(\bfx)\ge a+1\}\cap K\cap \Gamma$,
and repeat until we reach an empty set.
Since we perform $\lvert f(K\cap \Gamma)\rvert+1$ optimizations, the claimed time complexity follows from 
\cref{cor:polytope-oracles} and
\cref{lem:lattice-optimization}.
\end{proof}

\begin{lemma}
\label{lem:lattice-span}
Given a linear optimization oracle for a rational polytope $K\subseteq \R^d$ containing $\zero$,
and a lattice $\Gamma\subseteq \Q^d$, one can compute
$d'=\dim \spa(K\cap \Gamma)$ and linearly independent vectors
$\bfp_1,\ldots,\bfp_{d'}\in K\cap \Gamma$ in $b_0^{\Oh_d(1)} \Tquery(K)$ time,
where $b_0=\maxbitlen(K,\Gamma)$.
\end{lemma}
The proof of
\cref{lem:lattice-span}
is similar to that of \cref{lem:polytope-span}, using \cref{lem:lattice-optimization}. We omit the proof here.

\subsection{Reparameterizing lattices}

We record the following standard discrete John-type lemma based on lattice reduction; see e.g., \cite[Section~3.5]{taovu06}. 

\begin{lemma}
\label{lem:latticeboxinandout}
Let $K\subseteq \R^d$ be an origin-symmetric convex body and
$\Gamma\subseteq \R^d$ a lattice with
$\spa(K\cap \Gamma)=\R^d$.
There are a basis $\bfb_1,\ldots,\bfb_d$ of $\Gamma$ and positive integers
$m_1,\ldots,m_d$ such that
\[
 K\subseteq \left\{\sum_{i=1}^{d} t_i \bfb_i:\lvert t_i\rvert\le m_i\right\}
 \subseteq \kappa_d K,
\]
where $\kappa_d\ge 1$ depends only on $d$.

If $K$ is a rational polytope given by a linear optimization oracle and
$\Gamma$ is rational, one can compute $\bfb_i,m_i$ in
$b_0^{\Oh_d(1)} \Tquery(K)$ time, where $b_0=\maxbitlen(K,\Gamma)$.
\end{lemma}
\begin{proof}
Apply \cref{lem:bounded-coordinate-selection} to $K$ to obtain
$\bfp_0,\ldots,\bfp_d$, and write
$\zero=\sum_{i=0}^d \beta_i \bfp_i$ with $\sum_{i=0}^d \beta_i=1$.
 Without loss of generality, assume $\lvert\beta_0\rvert=\max_{0\le i\le d}\lvert\beta_i\rvert \neq 0$, so $\bfp_0 = -\sum_{i=1}^d \frac{\beta_i}{\beta_0} \bfp_i$.
Then, since $(\bfp_i-\bfp_0)_{i=1}^d$ form a linear basis, $(\bfp_i)_{i=1}^d $ also form a linear basis.
Let $T\colon \R^d\to \R^d$ send $\bfp_1,\ldots,\bfp_d$ to the standard basis.

For $\bfx \in K$,
writing $\bfx=\sum_{i=0}^d a_i \bfp_i$, where $\sum_{i=0}^d a_i=1$ and $\lvert a_i\rvert\le 2$ (by \cref{lem:bounded-coordinate-selection}), we have
\[
 \bfx=\sum_{i=1}^d\left(a_i-a_0 \frac{\beta_i}{\beta_0}\right) \bfp_i,
 \text{ where } 
 \left\lvert a_i-a_0 \frac{\beta_i}{\beta_0}\right\rvert\le 4.
\]
Thus, 
\[TK\subseteq [-4,4]^d.\]
On the other hand,
for any $\bfy \in [-1/d,1/d]^d$,
symmetry and convexity of $K$ give $dy_i \bfp_i\in K$ for all $1\le i\le d$, so $\sum_{i=1}^d y_i \bfp_i \in K$ holds by convexity of $K$. This proves \[[-1/d,1/d]^d\subseteq TK.\]

Let 
$\bfv_1,\ldots,\bfv_d$ be the LLL basis (\cref{lem:latticereduct}) of
$T \Gamma$.
For all $\bfy \in TK$, writing $\bfy = \sum_{j=1}^{d} t_j \bfv_j$,
 Cramer's rule gives for all $i\in [d]$ that
\[ \lvert t_i\rvert  =\frac{\lvert\det(\bfv_1,\dots,\bfv_{i-1}, \bfy ,\bfv_{i+1},\dots,\bfv_d)\rvert}{\lvert\det(\bfv_1,\dots,\bfv_d)\rvert} \le \frac{\left\|\bfy \right\|_2 \prod_{j\in [d], j\neq i}\|\bfv_j\|_2}{\det(T \Gamma)},\]
so \cref{lem:latticereduct} then gives
   \begin{equation}
 \lvert t_i\rvert\,\|\bfv_i\|_2
 \le  2^{d^2}\left\|\bfy \right\|_2 \le 2^{d^2} 4 \sqrt{d} \eqqcolon A_d,\label{eqn:tempad}
   \end{equation}
where the last inequality uses $\bfy \in TK\subseteq [-4,4]^d$.

Return $\bfb_i=T^{-1} \bfv_i$ and
$m_i=\lceil A_d/\|\bfv_i\|_2\rceil$.
Then, by the inequality \eqref{eqn:tempad} for all $\bfy \in TK$, the first claimed containment $K\subseteq \{\sum_{i=1}^d t_i \bfb_i : \lvert t_i\rvert\le m_i\}$ holds. It remains to prove the second containment.

We claim that $\|\bfv_i\|_2 \le A_d$ for all $i\in [d]$. Indeed, since $\spa(K\cap \Gamma)=\R^d$,   $T \Gamma$ contains $d$ linearly independent vectors $Y = \{\bfy_1,\dots,\bfy_d\}$ in $TK$. By linear independence, there exists $\bfy \in Y$ so that writing 
$\bfy = \sum_{j=1}^{d} t_j \bfv_j$ gives a nonzero $t_i$, and $t_i$ is integer due to $\bfy \in Y\subseteq T \Gamma$. 
So the inequality~\eqref{eqn:tempad} gives 
$\|\bfv_i\|_2\le A_d$ as claimed.

Therefore, 
$m_i\|\bfv_i\|_2 = \lceil A_d/\|\bfv_i\|_2 \rceil \cdot \|\bfv_i\|_2\le 2A_d$ for all $i\in [d]$. Hence, 
\[\left \{\sum_{i=1}^d t_i \bfv_i : \lvert t_i\rvert\le m_i\right  \} \subseteq [-2dA_d,2dA_d]^d \subseteq (2d^2 A_d) TK.\]
Applying $T^{-1}$ on both sides then proves the second claimed containment, for $\kappa_d\coloneqq 2d^2 A_d$.

The claimed running time follows from
\cref{lem:bounded-coordinate-selection,lem:latticereduct}.
\end{proof}

The following lemma approximates an improper progression $f(K\cap \Gamma)$ by a proper one $g(\Phi K \cap \Phi \Gamma ), f = g \circ \Phi$, and approximates $\Phi K$ by a box.
The proof is similar to the rank-reduction steps in \cite[Section 4]{taovu08john}.

\begin{lemma}
\label{lem:local-properization}
Given an origin-symmetric rational polytope $K\subseteq \R^d$ with a linear optimization oracle, a lattice $\Gamma\subseteq \Q^d$ such that $\spa(K\cap \Gamma)=\R^d$, and a nonzero rational linear map $f\colon \R^d\to \R$,
 one can compute $1\le d'\le d$, rational linear maps
$\Phi\colon \R^d\to \R^{d'}$ and $g\colon \R^{d'}\to \R$, and a box
$B=\prod_{i=1}^{d'}[-m_i,m_i]$, $m_i\in \Pos$,
such that
\begin{itemize}
      \item $f=g \circ \Phi$,
   \item $\Phi \Gamma=\Z^{d'}$,
         \item $\Phi K\subseteq B$,
            \item $\lvert B\cap \Z^{d'}\rvert\le C_d\lvert f(K\cap \Gamma)\rvert$, where $C_d$ depends only on $d$,
               \item and $g$ is injective on $4B\cap \Z^{d'}$,
\end{itemize}
 with $\maxbitlen(\Phi,g,B)=\Oh_d(b_0)$ in $b_0^{\Oh_d(1)} \Tquery(K)$ time, where $b_0=\maxbitlen(K,\Gamma,f)$.
\end{lemma}
\begin{proof}
Compute $\Lambda=\Gamma\cap \ker f$ by
\cref{lem:lattice-kernel-image}, with $\rk \Lambda=d-1$.
Choose an integer $a =\Oh_d(1)$ so that
$a>8 \kappa_{i}+2d$ for all $1\le i\le d$, where $\kappa_{i}$ is as in  
\cref{lem:latticeboxinandout}.
Use \cref{lem:lattice-span} to compute
\[
 V_j=\spa(\Lambda\cap a^j K) \text{ for } 0\le j\le d.
\]
These nested spaces have dimension at most $d-1$, so
$V_{j^*}=V_{j^*+1}$ for some $0\le j^*\le d-1$.
Denote $q=a^{j^*} = \Oh_d(1)$ and $V=V_{j^*}$.  
Use \cref{lem:lattice-span} to compute linearly independent $\bflambda_1,\ldots,\bflambda_{\dim V}
\in \Lambda\cap qK$, and let $\Lambda_0 = \caL(\bflambda_1,\ldots,\bflambda_{\dim V})$. Then,
\[V = \spa(\Lambda_0) = \spa(\Lambda\cap aqK),\]
 and
$\Lambda_0$ has a fundamental parallelepiped in $(\dim V) qK\subseteq dqK$.

Define $d'=d-\dim V$, and choose a rational surjective linear map $\Phi_0\colon \R^d\to \R^{d'}$ with kernel $V$. 
Since $V\subseteq \ker f$, we can find
a unique linear map $g_0$ such that $f=g_0 \circ \Phi_0$.
Define $\Gamma_0=\Phi_0 \Gamma$, $K_0=\Phi_0 K$.
Apply \cref{lem:latticeboxinandout} to $K_0$ and $\Gamma_0$ inside the $d'$-dimensional subspace $\mathrm{im}(\Phi_0) = \spa(K_0\cap \Gamma_0)$.
It returns a basis matrix $U$ of $\Gamma_0$ and a box
$B = \prod_{i=1}^{d'}[-m_i,m_i],m_i \in \Pos$ with
\[K_0\subseteq UB\subseteq \kappa_{d'} K_0.\]
Define $\Phi=U^{-1} \Phi_0$ and $g=g_0 \circ U$. We return $\Phi,g,B$ as our answer.
Clearly, $g \circ \Phi = g_0 \circ \Phi_0 = f$, $\Phi \Gamma = U^{-1} \Phi_0 \Gamma = U^{-1} \Gamma_0 =\Z^{d'}$, and $\Phi K = U^{-1} \Phi_0 K = U^{-1} K_0 \subseteq B$ as claimed.
We next verify the remaining two claimed properties.

We claim that for every $\bfz\in 4 \kappa_{d'} K_0\cap \Gamma_0$,
there exists
\begin{equation}
\label{eq:local-properization-lift}
 \bfx\in (4 \kappa_{d'}+dq) K\cap \Gamma \text{ with }
 \Phi_0 \bfx=\bfz \text{ and } f(\bfx)=g_0(\bfz).
\end{equation}
Indeed,  we can take $\bfx'\in \Gamma$ and $\bfy\in 4 \kappa_{d'} K$ with $\Phi_0 \bfx'=\Phi_0 \bfy=\bfz$,  reduce $\bfx'-\bfy\in V$ modulo $\Lambda_0$ to get $\bfw \in \bfx' - \bfy +\Lambda_0\subseteq V$ that lies inside the fundamental parallelepiped of $\Lambda_0$ contained in $dqK$, and let $\bfx = \bfy + \bfw$. 
One can verify that 
\eqref{eq:local-properization-lift} holds for $\bfx$.

We now prove $g_0$ is injective on $4 \kappa_{d'} K_0\cap \Gamma_0$.
Indeed, if $\bfz_1,\bfz_2\in 4 \kappa_{d'} K_0\cap \Gamma_0$, $g_0(\bfz_1)=g_0(\bfz_2)$, then
\eqref{eq:local-properization-lift} gives
$\bfx_1,\bfx_2\in (4 \kappa_{d'}+dq) K\cap \Gamma$ with
$\Phi_0 \bfx_1=\bfz_1$,  $\Phi_0 \bfx_2=\bfz_2$, and $f(\bfx_1)=f(\bfx_2)$. Thus, 
\[
\bfx_1 - \bfx_2 \in  \Lambda\cap (8 \kappa_{d'}+2dq) K
 \subseteq \Lambda\cap aqK\subseteq V = \ker \Phi_0.
\]
Hence $\bfz_1=\Phi_0 \bfx_1 = \Phi_0 \bfx_2 = \bfz_2$, proving that $g_0$ is injective on
$4 \kappa_{d'} K_0\cap \Gamma_0$ as claimed.

Therefore, $g = g_0 \circ U$ is injective on $U^{-1}
(4 \kappa_{d'} K_0\cap \Gamma_0) \supseteq 4B\cap \Z^{d'}$, as claimed.

We bound  $\lvert B\cap \Z^{d'}\rvert$ by
\begin{align*}
   \lvert B\cap \Z^{d'}\rvert = \lvert UB\cap \Gamma_0\rvert & \le \lvert\kappa_{d'} K_0\cap \Gamma_0\rvert\\
   & = \lvert g_0(\kappa_{d'} K_0\cap \Gamma_0)\rvert
\hspace{10pt} &&
\text{(by injectivity of $g_0$)}\\
& \le \lvert f((4 \kappa_{d'}+dq) K\cap \Gamma)\rvert
\hspace{10pt} &&
\text{(by \eqref{eq:local-properization-lift})}\\
& \ll_d \lvert f(K\cap \Gamma)\rvert 
\hspace{10pt} &&
\text{(by \cref{lem:lattice-quotient-dilation})}
\end{align*}
as claimed.

The claimed time and maximum bit length bounds follow from 
\cref{lem:constrained-optimization,lem:lattice-kernel-image,lem:latticeboxinandout}.
\end{proof}

    \section{Grouping the Input Integers}
\label{sec:classify}

In this section, we prove \cref{lem:main-classify} by 
repeatedly applying 
\cref{thm:algorithmiccfp-highdim}.

\subsection{Preliminaries}

We need the following standard estimate for the number of lattice points in a section.
\begin{lemma}
\label{lem:classify-lattice-section}
Let $K\subseteq \R^d$ be an origin-symmetric compact convex set,
$\Gamma\subseteq \R^d$ a lattice with $\spa(K\cap \Gamma)=\R^d$,
and $U\subseteq \R^d$ a linear subspace of dimension $d-k$.
Then, for every $s>0$,
$\lvert sK \cap \Gamma \cap U\rvert \le
(5d)^d \,s^{-k}\,\lvert sK \cap \Gamma\rvert$.
\end{lemma}
\begin{proof}
Choose $\bfv_1,\dots,\bfv_k\in K\cap \Gamma$ so that $U, \bfv_1,\dots,\bfv_k$ together span $\R^d$. 
Then, the translated sets 
 $(\Gamma\cap sK\cap U) + \sum_{i=1}^k a_i \bfv_i$, for integers
$0\le a_i\le \lfloor s/d\rfloor$, are disjoint subsets of $\Gamma\cap 2sK$. Hence,
\[
 (s/d)^k
 \lvert\Gamma\cap sK\cap U\rvert\le \lvert\Gamma\cap 2sK\rvert.
\]
Applying \cref{lem:lattice-quotient-dilation} to $sK$ gives
$ \lvert\Gamma\cap 2sK\rvert\le 5^d\lvert\Gamma\cap sK\rvert$. This proves the claim since $5^d d^k\le (5d)^d$.
\end{proof}

Given a multiset of vectors, the following lemma peels off a submultiset satisfying relative fatness, so that the remaining vectors have a smaller convex hull volume.

\begin{lemma}
\label{lem:one-step-peeling}
Let $Y$ be a nonempty multiset of nonzero vectors in $\Z^d$,
$H_Y=\conv(\pm Y)$, and $b_0=\maxbitlen(Y)$.
In $\lvert\supp Y\rvert b_0^{\Oh_d(1)}$ time, one can partition
$Y=F \sqcup W$ so that
\begin{itemize}
   \item
      $ (c_{\mathrm{peel},d}\lvert F\rvert) H_Y \subseteq \caZ_0(F)$, where $c_{\mathrm{peel},d}\in (0,1)$ only depends on $d$,
      \item and
         $\vol_{\dim H_Y}(\conv(\pm W)) \le \tfrac{1}{2} \vol_{\dim H_Y}(H_Y)$.
\end{itemize}
\end{lemma}

We prove \cref{lem:one-step-peeling} in \cref{subsec:one-step-peeling}.

Recall the following conventions from
\cref{sec:classification-interface}:
Let $\bfe_1,\dots,\bfe_d$ be the standard basis of $\R^d$, and for every $1\le d'\le d$ identify $\R^{d'}$ with the subspace $\spa \{\bfe_1,\dots,\bfe_{d'}\}$ of $\R^d$.
 Write $\Pi_{U}$ for the orthogonal projection onto the subspace $U$.

\subsection{Proof of \texorpdfstring{\Cref*{lem:main-classify}}{Theorem \getrefnumber{lem:main-classify}}}

We begin with an outline of the proof. The proof is split into a main stage and a post-processing stage:
\begin{itemize}
    \item The main stage partitions $A$ into
    $\Oh_d(\log \Delta)$
    groups with an extra part $A^{\extr}$, and represents the grouped integers as vectors.
     \item The post-processing stage merges contiguous groups into larger groups.
\end{itemize}
The only purpose of the post-processing stage is to ensure the bound in \cref{lem:main-classify}\,\ref{item:main-classify-prefix}, namely that the $i$-th group constitutes an $\Omega_d(1/\log \Delta)$ fraction of the total size of the first $i$ groups.
The other main properties, including relative fatness and $s$-generation, are established in the main stage. The fatness coefficient $\lambda$ degrades from $\Oh_d(1)$ to $\Oh_d(\log \Delta)$ after post-processing.

The main stage consists of iterations $i=1,2,\dots$, where iteration $i$ selects the $i$-th group of integers $f(E_i)$ from the remaining integers $f(Y_i)$, where $E_i,Y_i$ are multisets of vectors in $\Z^{d_i}$, and $f$ is projection to the first coordinate. Initially, $d_0 = 1$ and $Y_0$ includes all integers of $A$ as $1$-dimensional vectors.

In iteration $i$,
first run the algorithmic Conlon--Fox--Pham theorem (\cref{thm:algorithmiccfp-highdim}) on the remaining vectors $Y_{i-1}\setminus E_{i-1} \subseteq \Z^{d_{i-1}}$ to embed them into a GAP, except for $\Oh_d(s \log \Delta)$ of them which are discarded and put into the extra part.
By incorporating their coordinates from the GAP embedding, each remaining
vector in $Y_{i-1}\setminus E_{i-1}$
 receives a new representation as a  possibly higher dimensional vector, while its image under $f$ remains unchanged.
These new vectors form the multiset $Y_i$, which contains a small $s$-generating submultiset $E_i'$ due to the Conlon--Fox--Pham theorem.
We then use \cref{lem:one-step-peeling} to extract a submultiset of $Y_i$
 that satisfies relative fatness. Together with $E_i'$, they form the $i$-th group $E_i$.

The main technical details are to bound the number of iterations, and to show that the dimension of the representing vectors does not exceed $d$. This requires keeping track of the lattice counts of the GAPs and changes in their dimensions.

Now we proceed to the formal proof of \cref{lem:main-classify}.

If $\Delta = \Oh_d(1)$, then $\lvert A\rvert \le \Delta^{\Oh(1)}=\Oh_d(1)$, and we can 
return no groups and $A^{\extr}=A$ to satisfy all the requirements.
We henceforth assume
$\Delta$ is sufficiently large in terms of $d$.

\paragraph*{Basic setup.}
Let $C_{\mathrm{CFP},d}$ and $r_{\mathrm{CFP},d}$ be the upper bounds for the constant $C$ and the returned rank in
\cref{thm:algorithmiccfp-highdim} for $\beta=d+1$ and all $1\le k\le d$.
Let $C_{\mathrm{count},d}=(5d)^d$. Choose $C_{\mathrm{range},d}>d$
sufficiently large in terms of 
$C_{\mathrm{CFP},d},r_{\mathrm{CFP},d},
C_{\mathrm{count},d}$, and $ d$.
Define \[s=\lceil C_{\mathrm{range},d} \Delta^{1/d}\rceil.\]

Recall $\supp(A) \subseteq \Z \cap [1,\Delta]$, and $f(\bfx)=x_1$.
For $i\ge 0$, we will define $d_i\ge 1$ nondecreasing in $i$, and multisets $E_i\subseteq Y_i$ with $\supp(Y_i)\subseteq \Z^{d_i}$. The initial values are defined by
\[d_0=1, \,\, Y_0=\{a \bfe_1:a\in A\}, \, \text{ and } E_0=\varnothing.\]
For $i\ge 0$, write \[
\Gamma_i =\caL(Y_i) \text{ and }
H_i = \conv(\{\zero\}\cup \pm Y_i).
\]
We regard $H_i\cap \Gamma_i$ as a progression of rank $\dim H_i$
in $\spa \Gamma_i$, embedded in $\R^{d_i}$.
We maintain the following four invariants for $i\ge 0$:
\begin{subequations}
\begin{align}
&f(H_i\cap \Gamma_i) \text{ is an }s \text{-proper progression},
\label{eq:classify-proper-bound}\\
&\supp(Y_i)\subseteq (\Z\cap [-\Delta_i,\Delta_i])^{d_i} \setminus \{\zero\},
\,\text{where }\Delta_i\coloneqq C_{\mathrm{range},d}^{\,d_i+1} \Delta,
\label{eq:classify-support-bound}\\
&N_i\coloneqq \lvert(sH_i)\cap \Gamma_i\rvert
\le 3s \Delta\left(\frac{C_{\mathrm{count},d}}{s}\right)^{q_i},
\,\text{where }q_i\coloneqq d_i-\dim(H_i),
\label{eq:classify-count-bound}\\
&1\le d_i\le d.
\label{eq:classify-dimension-bound}
\end{align}
\end{subequations}
The invariants clearly hold at $i=0$.

In the following, 
we describe iteration $i\ge 1$ of the main stage, assuming
the invariants hold at $i-1$.

\paragraph*{Applying Conlon--Fox--Pham.}
Let $R\coloneqq Y_{i-1}\setminus E_{i-1}$.
If $\lvert R\rvert<C_{\mathrm{range},d} s \log \Delta$, then we stop the iterations and pass the produced $k_0\coloneqq i-1$ groups to the post-processing stage described at the end of the proof.
Otherwise, we apply \cref{thm:algorithmiccfp-highdim} to $R$ with $\Delta_{i-1}$ in place of $\Delta$ and $d+1$ in place of $\beta$.
We first verify the requirements of \cref{thm:algorithmiccfp-highdim}:
By \eqref{eq:classify-support-bound},
$\supp(R)\subseteq ([-\Delta_{i-1},\Delta_{i-1}]^{d_{i-1}}\cap \Z^{d_{i-1}})\setminus \{\zero\}$.
The definition of $s$ and $\Delta_{i-1}$ gives $s\ge \Delta_{i-1}^{1/(d+1)}$ for $\Delta$ sufficiently large in terms of $d$.
Since $\lvert R\rvert\ge C_{\mathrm{range},d} s \log \Delta$ and
$\log \Delta_{i-1}=\log \Delta+\Oh_d(1)$, the choice of $C_{\mathrm{range},d}$ gives
$ s\le
 C_{\mathrm{CFP},d}^{-1} \min \{\Delta_{i-1},\lvert R\rvert/\log \Delta_{i-1}\}$ for $\Delta$ sufficiently large in terms of $d$.
Thus, \cref{thm:algorithmiccfp-highdim} applies. 

\cref{thm:algorithmiccfp-highdim} returns $R'\subsetneq \hat R\subseteq R$ and an $s$-proper GAP containing $\supp(\hat R)$,
\begin{equation}
 g(B\cap \Z^m), \, \text{ with } g\colon \R^m \to \R^{d_{i-1}} \text{ linear,}\,\, 
B=\prod_{j=1}^m[-w_j,w_j],
\text{ and } w_j\le C_{\mathrm{CFP},d} \Delta_{i-1},
\label{eqn:gbzm}
\end{equation}
such that $R'$ is $s$-generating for the GAP,
$m\le r_{\mathrm{CFP},d}$, $\lvert R'\rvert\le C_{\mathrm{CFP},d} s$, and
$\lvert R\rvert-\lvert\hat R\rvert\le C_{\mathrm{CFP},d} s \log \Delta_{i-1}$.
Let $g^{-1}$ denote the unique inverting map from $g(B\cap \Z^m)$ to $B\cap \Z^m$.

Since $R'$ is $s$-generating for the GAP, there is
$\bfz\in \Z^m$ such that
\begin{equation}\label{eq:classify-cfp-generating-box}
\Sigma(g^{-1}(R'))\supseteq (\bfz+sB)\cap \Z^m.
\end{equation}
In particular,  we have
\begin{equation}
    \label{eqn:calg}
\caL(g^{-1}(R'))=\Z^m
\end{equation}
 due to $sB\cap \Z^m \supseteq \{0,1\}^m$.

\paragraph*{Computing the translate and pre-images.}
We compute $g^{-1}(\bfa)$ for every $\bfa\in \supp(\hat R)$ by a brute-force enumeration of all points in $B\cap \Z^m$.
Then, we compute a $\bfz$ that satisfies \eqref{eq:classify-cfp-generating-box} as follows.
By \cref{thm:algorithmiccfp-highdim} and
$\lvert R'\rvert\le C_{\mathrm{CFP},d} s$,
\begin{equation*}
 \big \lvert \lvert R'\rvert B \cap \Z^m\big \rvert
 \ll_d \lvert sB \cap \Z^m\rvert
 \ll_d \lvert s(\supp(R)\cup \{\zero\})\rvert\le N_{i-1},
\end{equation*}
where the last inequality follows from $R\subseteq Y_{i-1}$, \eqref{eq:classify-count-bound}, and the definitions of $H_{i-1},\Gamma_{i-1}$.
We use standard dynamic programming to compute $S\coloneqq \Sigma(g^{-1}(R'))\subseteq \lvert R'\rvert B \cap \Z^m$
in $\tO_d(\lvert R'\rvert N_{i-1})=\tO_d(sN_{i-1})$ time. 
Since $sB$ and $ \lvert R'\rvert B$ are axis-aligned boxes, we can compute the multidimensional prefix sums of the indicator function of $S$, and then
compute $\lvert S\cap (\bfz+sB)\rvert$ for all $\bfz\in \lvert R'\rvert B \cap \Z^m$, in $\tO_d(N_{i-1})$ total time.
We choose $\bfz$ for which
$\lvert S\cap (\bfz+sB)\rvert=\lvert sB\cap \Z^m\rvert$,
so that \eqref{eq:classify-cfp-generating-box} holds.
The overall time complexity of these steps is $\tO_d(sN_{i-1})$.

\paragraph*{New vector representations.}
Let $r_g=\dim g(\R^m)$. 
We may assume
$g(\bfe_1),\dots,g(\bfe_{r_g})$ are linearly independent, by relabeling
the coordinates of the GAP in \eqref{eqn:gbzm} if necessary. Let
$d_i\coloneqq d_{i-1}+(m-r_g)\ge d_{i-1}\ge 1$, and
define $\phi\colon \R^m\to \R^{d_i}$ by
\[ \phi(\bfu)\coloneqq \bigl(g(\bfu),u_{r_g+1},\ldots,u_m\bigr), \]
which is an injective linear map  due to 
the linear independence of $g(\bfe_1),\ldots,g(\bfe_{r_g})$.
Clearly,
\begin{equation}\label{eq:classify-pi-phi}
\Pi_{\R^{d_{i-1}}}(\phi(\bfu))=g(\bfu)  \text{ for all }\bfu\in \R^m,
\end{equation}
and
   $\phi(\Z^m)\subseteq \Z^{d_i}$.
Define 
\begin{align*}
Y_i & \coloneqq \phi(g^{-1}(\hat R)),\\
E_i' &\coloneqq \phi(g^{-1}(R')) \subseteq Y_i,
\end{align*}
 with $\supp(Y_i)\subseteq \Z^{d_i}\setminus \{\zero\}$.
\Cref{eq:classify-pi-phi} gives
\begin{equation}
    \label{eq:classify-pi-y}
\Pi_{\R^{d_{i-1}}}(Y_i)=\hat R\subseteq R = Y_{i-1}\setminus E_{i-1}.
\end{equation}

\paragraph*{Coordinate bounds.}
Take arbitrary $\bfy\in Y_i =\phi(g^{-1}(\hat R))$.
For $1\le j\le d_{i-1}$,
 by
\eqref{eq:classify-pi-y}, we have
$\lvert y_j\rvert\le \max_{\bfy'\in Y_{i-1}}\|\bfy'\|_{\infty} \le \Delta_{i-1}$.
For $d_{i-1}<j\le d_i$,
$\lvert y_j\rvert\le \max_{\bfu \in g^{-1}(\hat R)}\|\bfu\|_{\infty}\le \max_{\bfu \in B\cap \Z^m}\|\bfu\|_{\infty} $, so
\eqref{eqn:gbzm} and the choice of
$C_{\mathrm{range},d}$ give $\lvert y_j\rvert\le 
C_{\mathrm{CFP},d} \Delta_{i-1}\le 
C_{\mathrm{range},d} \Delta_{i-1}$.
In the latter case, $d_i\ge d_{i-1}+1$, so
$C_{\mathrm{range},d} \Delta_{i-1}\le \Delta_i$ by the definition of $\Delta_i$. 
Thus
$\|\bfy\|_\infty\le \Delta_i$ in either case, proving
\eqref{eq:classify-support-bound}.

\paragraph*{Generation.}
By definition of $E_i'$ and \eqref{eqn:calg}, we have $\caL(E_i')=\phi(\caL(g^{-1}(R')))=\phi(\Z^m)$.
By definition of $Y_i$, we have $\supp(Y_i) \subseteq \phi(\Z^m)$.
Since $E_i'\subseteq Y_i$, and $\Gamma_i =\caL(Y_i)$,
this means
\begin{equation}
    \label{eq:classify-step-spans-lattice}
 \caL(E_i')=\Gamma_i=\phi(\Z^m).
\end{equation}
Taking real spans gives $\spa(Y_i) = \spa(\Gamma_i) = \phi(\R^m)$.
Since $H_i = \conv(\{\zero\}\cup \pm Y_i)$ and $\Gamma_i = \caL(Y_i)$, we get
\begin{equation}
    \label{eq:classify-step-spans}
 \spa(H_i\cap \Gamma_i)=\spa(H_i)=\phi(\R^m), \, \text{ and } \dim(H_i)=m,
\end{equation}
where the last claim follows from the injectivity of $\phi$.

By definition of $Y_i$, $\supp(Y_i) \subseteq \phi(B\cap \Z^m)$, so $H_i\subseteq \phi(B)$.
Thus,
\begin{equation}
    \label{eqn:sbm}
\phi\big (  sB \cap \Z^m\big )  \mathrel{\underbrace{=}_{\text{$\phi$ injective}}} \phi ( sB) \cap \phi(\Z^m)
\mathrel{\underbrace{=}_{\text{\eqref{eq:classify-step-spans-lattice}}}}
(s \phi(B))\cap \Gamma_i
 \supseteq (sH_i)\cap \Gamma_i.
\end{equation}
Let $\boldsymbol{\tau}_i\coloneqq \phi(\bfz)\in \Gamma_i$ for $\bfz$ defined in \eqref{eq:classify-cfp-generating-box}.
Then, applying $\phi$ to \eqref{eq:classify-cfp-generating-box} gives
\[ \Sigma(E_i')\supseteq
\phi\big ( (\bfz + sB) \cap \Z^m\big )   = \boldsymbol{\tau}_i + \phi(sB \cap \Z^m)
\mathrel{\underbrace{\supseteq}_{\text{\eqref{eqn:sbm}}}}
 \boldsymbol{\tau}_i + (sH_i \cap \Gamma_i).\]
Hence $E_i'$ is $s$-generating for
$H_i\cap \Gamma_i$
with translate $\boldsymbol{\tau}_{i}$. We also have $\lvert E_i'\rvert = \lvert R'\rvert\le C_{\mathrm{CFP},d} s$.

Since the vertices of $\bfz+sB$ lie in $\Z^m$, taking convex hulls in
\eqref{eq:classify-cfp-generating-box} and symmetrizing gives
$sB\subseteq \caZ_0(g^{-1}(R'))$. Applying $\phi$ gives
\begin{equation}
\label{eq:classify-generating-recentering}
 sH_i\subseteq \phi(sB)\subseteq \caZ_0(E_i').
\end{equation}

\paragraph*{Properness.}
Applying $g$ to \eqref{eqn:calg} gives $g(\Z^m)=\caL(R')\subseteq \Gamma_{i-1}$, so 
\[\Pi_{\R^{d_{i-1}}}(\Gamma_i) \mathrel{\underbrace{=}_{\text{\eqref{eq:classify-step-spans-lattice}}}} \Pi_{\R^{d_{i-1}}}(\phi(\Z^m)) \mathrel{\underbrace{=}_{\text{\eqref{eq:classify-pi-phi}}}} g(\Z^m) \subseteq \Gamma_{i-1}.  \]
By \eqref{eq:classify-pi-y} and the definition of $H_i,H_{i-1}$, we have $\Pi_{\R^{d_{i-1}}}(H_i)\subseteq H_{i-1}$.  Thus, $\Pi_{\R^{d_{i-1}}}(sH_i \cap \Gamma_i) \subseteq sH_{i-1}\cap \Gamma_{i-1}$.

By \eqref{eq:classify-pi-phi} and the $s$-properness of the GAP, the map 
$\Pi_{\R^{d_{i-1}}}(\phi(\cdot ))$ is injective on $sB \cap \Z^m$.
Hence, $\Pi_{\R^{d_{i-1}}}$ is injective on $\phi(sB\cap \Z^m)\supseteq sH_i \cap \Gamma_i$.
Thus,
$\Pi_{\R^{d_{i-1}}}\colon (sH_i)\cap \Gamma_i
 \to (sH_{i-1})\cap \Gamma_{i-1}$
is an injective map. 

By \eqref{eq:classify-proper-bound}, $f$ is injective on
$(sH_{i-1})\cap \Gamma_{i-1}$.
By definition of $f$,  $f \circ \Pi_{\R^{d_{i-1}}}=f$.
Thus, $f$ is injective on $(sH_i)\cap \Gamma_i$.
Together with \eqref{eq:classify-step-spans}, this shows $f(H_i\cap \Gamma_i)$ is an $s$-proper progression, proving
\eqref{eq:classify-proper-bound} for $i$.

\paragraph*{Lattice count and dimension bounds.}
Recall $N_i \coloneqq \lvert sH_i \cap \Gamma_i\rvert$ from
\eqref{eq:classify-count-bound}.
Since the map $\Pi_{\R^{d_{i-1}}}$
is injective from $(sH_i)\cap \Gamma_i$ to
$(sH_{i-1})\cap \Gamma_{i-1}$, we have $N_i\le N_{i-1}$.

Let $U\coloneqq g(\R^m)$, which equals
$\Pi_{\R^{d_{i-1}}}(\spa(H_i))$ by
\eqref{eq:classify-pi-phi} and \eqref{eq:classify-step-spans}.
Recall $r_g\coloneqq \dim U$. 
Since $\Pi_{\R^{d_{i-1}}}(H_i)\subseteq H_{i-1}$, we have
$U\subseteq \spa(H_{i-1})$ and hence $r_g\le \dim(H_{i-1})$.
By applying \cref{lem:classify-lattice-section}
in $\spa(H_{i-1})$,
whose dimension is $\dim(H_{i-1})\le d$, we get
\begin{equation*}
\lvert(sH_{i-1})\cap \Gamma_{i-1}\cap U\rvert
 \le C_{\mathrm{count},d} s^{-(\dim(H_{i-1})-r_g)} N_{i-1}.
\end{equation*}
Since $\Pi_{\R^{d_{i-1}}}$ maps
injectively from $(sH_i)\cap \Gamma_i$ to
$(sH_{i-1})\cap \Gamma_{i-1}\cap U$, this means
\begin{equation}
\label{eq:classify-Pi-count}
N_i
 \le C_{\mathrm{count},d} s^{-(\dim(H_{i-1})-r_g)} N_{i-1}.
\end{equation}

Recall $d_i \coloneqq d_{i-1}+(m-r_g) = d_{i-1}+(\dim(H_i)-r_g)$ (by \eqref{eq:classify-step-spans}), and $q_i\coloneqq d_i - \dim(H_i) = d_{i-1}-r_g, q_{i-1}\coloneqq d_{i-1} - \dim(H_{i-1}) $ from
\eqref{eq:classify-count-bound}. Putting them together gives
\begin{equation}
\label{eq:classify-codimension-increment}
 q_i-q_{i-1}=\dim(H_{i-1})-r_g\ge 0.
\end{equation}
 If $q_i>q_{i-1}$, then
\eqref{eq:classify-Pi-count} gives
\[
 N_i\le (C_{\mathrm{count},d}/s)^{q_i-q_{i-1}} N_{i-1}.
\]
If $q_i=q_{i-1}$, then use the initial bound $N_i\le N_{i-1}$.
In either case,
 $N_i\le (C_{\mathrm{count},d}/s)^{q_i-q_{i-1}} N_{i-1}$ holds, proving
\eqref{eq:classify-count-bound}  for $i$.

By \cref{thm:algorithmiccfp-highdim},
\[
 s^m\le \lvert sB\cap \Z^m\rvert\ll_d \lvert s(\supp(\hat R)\cup \{\zero\})\rvert
 \le \lvert s(\supp(Y_i)\cup \{\zero\})\rvert\le N_i,
\]
where the third inequality uses $\Pi_{\R^{d_{i-1}}}(Y_i)=\hat R$ by \eqref{eq:classify-pi-y}, and the last inequality uses
$\supp(Y_i)\subseteq H_i\cap \Gamma_i$.
Combining this with \eqref{eq:classify-count-bound} and $q_i=d_i-m$ gives
\[ s^{d_i-1}\ll_d \Delta C_{\mathrm{count},d}^{q_i}=\Oh_d(\Delta), \]
where the last step uses $q_i=d_{i-1}-r_g\le d$ by \eqref{eq:classify-dimension-bound}.
If $d_i\ge d+1$, this contradicts
$s^d\ge C_{\mathrm{range},d}^d \Delta$ (by definition of $s$) when $C_{\mathrm{range},d}$ is
sufficiently large in terms of $d$.
Thus $d_i\le d$, proving \eqref{eq:classify-dimension-bound} for $i$.

\paragraph*{Extracting the group.}
Apply \Cref{lem:one-step-peeling} to $Y_i$, obtaining $Y_i=F_i\sqcup W_i$ such that
\begin{align}
(c_{\mathrm{peel},d}\lvert F_i\rvert) H_i & \subseteq \caZ_0(F_i),
\label{eq:classify-peeling-balance}\\
\vol_{\dim(H_i)}(\conv(\pm W_i)) & \le \tfrac{1}{2} \vol_{\dim(H_i)}(H_i).
\notag
\end{align}
Define the multiset $E_i\coloneqq F_i\cup E'_i\subseteq Y_i$, where the multiplicities of the union are defined by the maximum (instead of sum) of 
the corresponding multiplicities in the two multisets.
Since $Y_i\setminus E_i\subseteq W_i$,
\begin{equation}
\label{eq:classify-volume-decrease}
 \vol_{\dim(H_i)}(\conv(\pm (Y_i\setminus E_i)))
 \le \tfrac{1}{2} \vol_{\dim(H_i)}(H_i).
\end{equation}

By $\lvert E_i\rvert\le  \lvert F_i\rvert+\lvert E_i'\rvert\le  \lvert F_i\rvert+C_{\mathrm{CFP},d} s$, we have
\begin{equation}
\label{eq:classify-raw-zonotope-comparison}
 \lvert E_i\rvert H_i\subseteq \lvert F_i\rvert H_i+C_{\mathrm{CFP},d} sH_i
 \mathrel{\underbrace{\subseteq}_{\text{\eqref{eq:classify-peeling-balance}, \eqref{eq:classify-generating-recentering}}}}
 c_{\mathrm{peel},d}^{-1} \caZ_0(F_i)+C_{\mathrm{CFP},d} \caZ_0(E'_i)
 \subseteq 
 (c_{\mathrm{peel},d}^{-1}+C_{\mathrm{CFP},d}) \caZ_0(E_i).
\end{equation}
Set $A_i^{\extr}=f(R\setminus \hat R)$. Then
$\lvert A_i^{\extr}\rvert\le C_{\mathrm{CFP},d} s \log \Delta_{i-1}$ and, by
\eqref{eq:classify-pi-y} and the definition of $Y_i$,
\begin{equation}
\label{eq:classify-value-partition}
 f(Y_{i-1}\setminus E_{i-1})
 =A_i^{\extr}\sqcup f(E_i)\sqcup f(Y_i\setminus E_i).
\end{equation}

\paragraph*{Iteration count.}
Since all vertices of $H_i$ are integer vectors with
coordinates at most $\Oh_d(\Delta)$, we have
\begin{equation}
\label{eq:volumepirange}
 1/d!\le \vol_{\dim(H_i)}(H_i)\le \Oh_d(\Delta^d).
\end{equation}

By \eqref{eq:classify-codimension-increment} and
\eqref{eq:classify-dimension-bound}, 
both $d_i$ and $q_i$ are nondecreasing integers
in $[0,d]$, so there are
$\Oh(d)$ maximal intervals of consecutive iterations $i$ with $(d_i,q_i)$
being constant. 
On each such interval, $\dim(H_i) = d_i-q_i$ also stays constant.
For consecutive iterations $i-1, i$ in one such interval,
 \eqref{eq:classify-pi-y} gives
$Y_i\subseteq Y_{i-1}\setminus E_{i-1}$, so 
\eqref{eq:classify-volume-decrease} (applied to $i-1$) gives
\[
 \vol_{\dim(H_i)}(H_i)\le \tfrac{1}{2} \vol_{\dim(H_{i-1})}(H_{i-1}).
\]
Thus, by \eqref{eq:volumepirange}, each such interval has
length $\Oh_d(\log \Delta)$.  Thus the main stage terminates after $k_0$ iterations for some 
$k_0=\Oh_d(\log \Delta)$.

Iterating \eqref{eq:classify-value-partition} gives $A=A^{\extr}\sqcup \bigsqcup_{i=1}^{k_0} f(E_i)$ where
\begin{equation*}
 A^{\extr}\coloneqq f(Y_{k_0}\setminus E_{k_0})
      \sqcup \bigsqcup_{i=1}^{k_0} A_i^{\extr}.
\end{equation*}
The termination rule gives $\lvert Y_{k_0}\setminus E_{k_0}\rvert \le C_{\mathrm{range},d} s \log \Delta$.
Since
$\lvert A_i^{\extr}\rvert=\Oh_d(s \log \Delta_{i-1})$, we obtain
$
 \lvert A^{\extr}\rvert=\Oh_d(s \log^2 \Delta)$
as required.

\paragraph*{Time complexity.}
The application of the algorithmic Conlon--Fox--Pham theorem (\cref{thm:algorithmiccfp-highdim})
takes 
$\tO_d(s\lvert s(\supp(R)\cup \{\zero\})\rvert) = \tO_d(sN_{i-1})$ time because
$s(\supp(R)\cup \{\zero\})\subseteq (sH_{i-1})\cap \Gamma_{i-1}$.
Earlier, we argued that computing the translate $\bfz$ and $g^{-1}(\bfa)$ for all $\bfa \in \supp(\hat R)$ also takes
$\tO_d(sN_{i-1})$ time.

The application of
\Cref{lem:one-step-peeling}
takes
$\Oh_d(\lvert\supp R\rvert)(\log \Delta)^{\Oh_d(1)}$ time, which is dominated by the time complexity of \cref{thm:algorithmiccfp-highdim}.
By \eqref{eq:classify-count-bound},  we have $N_{i-1}\le 3s \Delta$ provided $s\ge C_{\mathrm{count},d}$ (which holds when $\Delta$ is sufficiently large in terms of $d$).
Therefore, each iteration takes $\tO_d(s^2 \Delta)$ time.
The total time complexity over all $k_0 = \Oh_d(\log \Delta)$ iterations is therefore $\tO_d(s^2 \Delta)$ as claimed.

\paragraph*{Post-processing stage.}
Now we merge some of the contiguous groups among the $k_0$ obtained groups based on their sizes.
Consider the prefix maxima, 
\[
 \{i_1<\cdots<i_k\}
 =\big \{r\in [k_0]:\lvert E_r\rvert=\max_{j\in [r]}\lvert E_j\rvert\big \},
\]
and let $i_{k+1}=k_0+1$. For $u\in [k]$, define
\[
 X_u=\bigsqcup_{r=i_u}^{i_{u+1}-1} \Pi_{\R^{d_{i_u}}}(E_r), \,\,
 A_u=f(X_u),
\]
and define $X_u'=E'_{i_u}$. Then, $\lvert X_u'\rvert\le \lvert E'_{i_u}\rvert\le C_{\mathrm{CFP},d} s$, proving
\cref{item:main-classify-generating}.
Since $f \circ \Pi_{\R^{d_{i_u}}} = f$ and $ A=A^{\extr}\sqcup \bigsqcup_{i=1}^{k_0} f(E_i)$, we have
$ A=A^{\extr}\sqcup \bigsqcup_{u=1}^{k} f(X_u)$ as required.
The returned dimensions are $d_{i_1}\le \cdots\le d_{i_k}$, which are in $[1,d]$, and $k\le k_0=\Oh_d(\log \Delta)$, as required.

By definition of the prefix maxima $i_{u}$ and $i_{u+1}$, we have $\lvert E_i\rvert\le \lvert E_{i_u}\rvert$ for $1\le i<i_{u+1}$. Therefore,
\[
 \lvert E_{i_u}\rvert\le \lvert X_u\rvert = \lvert A_u\rvert\le \sum_{v=1}^u \lvert A_v\rvert
 =\sum_{i=1}^{i_{u+1}-1}\lvert E_i\rvert\le k_0\lvert E_{i_u}\rvert,
 \]
and in particular, $\sum_{v=1}^u \lvert A_v\rvert \le k_0\lvert A_u\rvert$.
Together with $k_0=\Oh_d(\log \Delta)$, this proves
\cref{lem:main-classify}\,\ref{item:main-classify-prefix}.

By \eqref{eq:classify-pi-y},
$E_i\sqcup \Pi_{\R^{d_i}}(Y_{i+1})\subseteq Y_i$ for $1\le i<k_0$;
 starting from $E_{k_0}\subseteq Y_{k_0}$ and iteratively applying this containment 
for $i=k_0-1,\ldots,i_u$, we obtain
\[
 \bigsqcup_{v=u}^k \Pi_{\R^{d_{i_u}}}(X_v)
 =\bigsqcup_{i=i_u}^{k_0} \Pi_{\R^{d_{i_u}}}(E_i)
 \subseteq Y_{i_u}.
\]
Thus, 
\[
 \tilde H_u\coloneqq \conv \!\left(\pm \bigcup_{v=u}^k
       \Pi_{\R^{d_{i_u}}}(\supp X_v)\right)\subseteq H_{i_u}.
\]
Since $\lvert X_u\rvert\le k_0\lvert E_{i_u}\rvert$ and 
$E_{i_u}\subseteq X_u$,
\[
 \lvert X_u\rvert H_{i_u}\subseteq k_0\lvert E_{i_u}\rvert H_{i_u}
\mathrel{
\underbrace{\subseteq}_{\text{
 \eqref{eq:classify-raw-zonotope-comparison}
}}
} 
  k_0(c_{\mathrm{peel},d}^{-1}+C_{\mathrm{CFP},d}) \caZ_0(E_{i_u})
 \subseteq k_0(c_{\mathrm{peel},d}^{-1}+C_{\mathrm{CFP},d}) \caZ_0(X_u).
\]
Since $\supp X_u\subseteq \tilde H_u \subseteq H_{i_u}$, this proves
\eqref{eq:oneclass-balance} with
$\lambda=k_0(c_{\mathrm{peel},d}^{-1}+C_{\mathrm{CFP},d})=\Oh_d(\log \Delta)$.

Since $E'_{i_u}\subseteq E_{i_u}\subseteq X_u$ by definition,
\eqref{eq:classify-step-spans-lattice} gives
\[
 \Gamma_{i_u}=\caL(E'_{i_u})\subseteq \caL(X_u)
 \subseteq \caL(Y_{i_u})=\Gamma_{i_u},
\]
so all containments are equalities.
Thus, for all $v\ge u$, 
$\Pi_{\R^{d_{i_u}}}(\supp X_v) \subseteq \supp Y_{i_u} \subseteq \Gamma_{i_u} = \caL(X_u)$, 
proving \cref{item:main-classify-projection}.

Since $\supp X_u\subseteq \tilde H_u\subseteq \spa \Gamma_{i_u}$ and
$\caL(X_u)=\Gamma_{i_u}$, we have
$\spa(\tilde H_u\cap \Gamma_{i_u})=\spa \Gamma_{i_u}$.
Together with $\tilde H_u\subseteq H_{i_u}$,
\eqref{eq:classify-proper-bound} implies that
$f(\tilde H_u\cap \caL(X_u))$ is an $s$-proper progression,
proving \eqref{eq:oneclass-values}.
Moreover, $X_u'=E'_{i_u}$ is $s$-generating for
$H_{i_u}\cap \Gamma_{i_u}$ with translate $\boldsymbol{\tau}_{i_u}$, and $\rk \caL(X_u)\le d\le s$,
so $X_u'$ is $(\rk \caL(X_u))$-generating for
$\tilde H_u\cap \caL(X_u)$, proving
\eqref{eq:oneclass-generating}.
We return $\boldsymbol{\tau}_{i_u}$ as the translate for $X_u'$.

By \eqref{eq:classify-support-bound},
$\supp X_u\subseteq \supp Y_{i_u}\subseteq (\Z\cap [-\Delta_{i_u},\Delta_{i_u}])^{d_{i_u}}$,
where $\Delta_{i_u}\le C_{\mathrm{range},d}^{d+1} \Delta=\Oh_d(\Delta)\le \Delta^{\Oh(1)}$ for $\Delta$ sufficiently large in terms of $d$, as required.

The time complexity of the post-processing stage is dominated by the time complexity for the main stage.

\subsection{Proof of \texorpdfstring{\Cref*{lem:one-step-peeling}}{Lemma \getrefnumber{lem:one-step-peeling}}}
\label{subsec:one-step-peeling}

The proof proceeds by repeatedly extracting vertices from the remaining input points  to form a simplex of sufficiently large volume,  until no more such simplex can be formed, at which point we terminate and conclude that the convex hull of the remaining input points has shrunken considerably. 
We appropriately discretize the space and bucket the input points for better running time.

\begin{proof}[Proof of \Cref{lem:one-step-peeling}]
Given the input multiset $Y$ of nonzero vectors in $\Z^d$, let
$d'=\dim \spa Y$.
Apply \cref{lem:bounded-coordinate-selection} to $\conv(\{\zero\}\cup \supp Y)$
(which has a trivial near-linear-time linear optimization oracle due to \cref{lem:zonotope-algorithms})
to obtain affinely independent 
$\bfp_0,\ldots,\bfp_{d'}$, and write
$\zero=\sum_{i=0}^{d'} \beta_i \bfp_i$ with $\sum_{i=0}^{d'} \beta_i=1$.
Without loss of generality, assume $\lvert\beta_0\rvert=\max_{0\le i\le d'}\lvert\beta_i\rvert>0$, so $\bfp_0 = -\sum_{i=1}^{d'} \frac{\beta_i}{\beta_0} \bfp_i$. 
Thus, $\bfp_1,\ldots,\bfp_{d'}$  form a basis of
$\spa Y$ and belong to $\supp Y$.
Define the bijective linear map $M\colon \R^{d'}\to \spa Y$ by $M(\bfe_i) =\bfp_i$ for $1\le i\le d'$.

For every $\bfx\in \supp Y$, 
\cref{lem:bounded-coordinate-selection} gives
$\bfx=\sum_{i=0}^{d'} a_i \bfp_i$ with $\sum_{i=0}^{d'} a_i=1$ and $\lvert a_i\rvert\le 2$.
Hence
\[
 \bfx=\sum_{i=1}^{d'}\left(a_i-a_0 \frac{\beta_i}{\beta_0}\right) \bfp_i,
\]
where each coefficient has absolute value at most $4$.
Writing $\hat Y=M^{-1} Y$ and $\hat H=\conv(\pm \hat Y)$,
this gives $\{\bfe_1,\ldots,\bfe_{d'}\}\subseteq \supp \hat Y\subseteq [-4,4]^{d'}$. Then, by symmetry and convexity of $\hat H$, 
\[
 (d')^{-1}[-1,1]^{d'}\subseteq \hat H\subseteq [-4,4]^{d'}.
\]

Choose $\eps_d>0$ sufficiently small in terms of $d$.
By multilinearity of the determinant and applications of the triangle inequality, we can partition $[-4,4]^{d'}$
into a family $\caQ$ of $\Oh_d(1)$ boxes such that
\begin{equation}
 \bigl\lvert\det(\bfx_1,\ldots,\bfx_{d'})
       -\det(\bfy_1,\ldots,\bfy_{d'})\bigr\rvert\le \eps_d
 \label{eq:classify-determinant-rounding}
\end{equation}
whenever $\bfx_j,\bfy_j$ belong to the same box for all $j\in [d']$.
Writing $\bfm(Q)$ for the center of a box $Q\in \caQ$, define the good tuples of boxes,
\[
 \caT=\left\{(Q_1,\ldots,Q_{d'})\in \caQ^{d'}:
 \bigl\lvert\det(\bfm(Q_1),\ldots,\bfm(Q_{d'}))\bigr\rvert
 \ge 3 \eps_d\right\}.
\]
Then, for $(\bfx_1,\dots,\bfx_{d'}) \in  Q_1\times \dots \times Q_{d'}$, \eqref{eq:classify-determinant-rounding}
shows that $\lvert\det(\bfx_1,\ldots,\bfx_{d'})\rvert$ is at least $2 \eps_d$
if $(Q_1,\ldots,Q_{d'})\in \caT$, and less than $4 \eps_d$ otherwise.

Initialize $\hat F=\varnothing$ and $\hat W=\hat Y$.
While there is a tuple $(Q_1,\ldots,Q_{d'})\in \caT$ such that
$Q_j\cap \supp \hat W\ne \varnothing$ for all $j\in [d']$, choose
$\bfw_j\in Q_j\cap \supp \hat W$ for each $j\in [d']$, and extract $k$ occurrences
of each $\bfw_j$ from $\hat W$ to $\hat F$, where $k$ is the minimum multiplicity of $\bfw_j$ in $\hat W$ over all $j\in [d']$.
(The boxes $Q_1,\ldots,Q_{d'}$ must be distinct, since their centers have nonzero determinant.)
Finally return $F = M \hat F$ and $W = M \hat W$. Since each iteration decreases $\lvert\supp \hat W\rvert$ by at least one, this process can be implemented in near-linear time. It remains to show the claimed properties on $F$ and $W$.

Consider the tuple $(\bfw_1,\ldots,\bfw_{d'})$ extracted in an arbitrary iteration of the process, which satisfies
$\lvert\det(\bfw_1,\ldots,\bfw_{d'})\rvert\ge 2 \eps_d$.
Since $\bfw_j\in \hat H\subseteq [-4,4]^{d'}$  for all $j\in [d']$,
by Cramer's rule, every $\bfx\in \hat H$ can be written as $\bfx=\sum_{j=1}^{d'} t_j \bfw_j$ with $\lvert t_j\rvert\le \frac{(d')!4^{d'}}{2 \eps_d} = \Oh_d(1)$ for all $j\in [d']$.
Thus, for some sufficiently small $c>0$ depending on $d$,
\[
 c \hat H
 \subseteq \sum_{j=1}^{d'}[-\bfw_j/2,\bfw_j/2].
\]
Taking Minkowski sums over the $\lvert\hat F\rvert/d'$ extracted tuples gives
$(c/d')\lvert\hat F\rvert\hat H \subseteq \caZ_0(\hat F)$. Applying $M$ on both sides then proves the first claimed property for $c_{\mathrm{peel},d} = c/d$.

Now we prove the second claimed property.
If $\dim \spa \hat W<d'$, then $\vol_{d'}(\conv(\pm \hat W))=0$, and there is nothing to prove.
Assume 
$\dim \spa \hat W=d'$, and choose $\bfv_1,\ldots,\bfv_{d'}\in \supp \hat W$ maximizing
$\lvert\det(\bfv_1,\ldots,\bfv_{d'})\rvert$.
By maximality and Cramer's rule,
$ \conv(\pm \hat W)\subseteq
 \left\{\sum_{j=1}^{d'} t_j \bfv_j:\lvert t_j\rvert\le 1\right\}
 $.
 By the termination rule, the boxes containing $\bfv_1,\ldots,\bfv_{d'}$ form a tuple outside $\caT$,
so $\lvert\det(\bfv_1,\ldots,\bfv_{d'})\rvert<4 \eps_d$.
Consequently,
\[
 \vol_{d'}(\conv(\pm \hat W))
 \le 2^{d'}\lvert\det(\bfv_1,\ldots,\bfv_{d'})\rvert
 <4\cdot 2^{d'} \eps_d
 \le \tfrac{1}{2} \vol_{d'}(\hat H),
\]
where the last inequality follows from $\vol_{d'}(\hat H)\ge (2/d')^{d'}$
and the choice of $\eps_d$.
Applying $M$ on both sides then proves the second claimed property. 
\end{proof}

\section{Macbeath Regions}
\label{sec:macbeath-regions}

Recall the definition of Macbeath regions~\cite{macbeath}.

\begin{restatable}[Macbeath regions]{definition}{EuclideanMacbeathDefinition}
\label{def:macbeath-euc}
For a convex body $K\subseteq \R^d$, $\bfx\in K$, and $\lambda\ge 0$, the \emph{$\lambda$-scaled Macbeath region} at $\bfx$ is defined as
\[
 M_K^\lambda(\bfx)
 =\bfx+\lambda\bigl((K-\bfx)\cap (\bfx-K)\bigr).
\]
\end{restatable}
Macbeath regions are convex and centrally symmetric.

\subsection{Basic properties}
Ewald, Larman, and Rogers~\cite{ewaldlarmanrogers1970} proved the
\emph{expansion-containment} property of Macbeath regions.
In its parameterized form (Abdelkader and Mount~\cite[Lemma~1]{AbdelkaderM18}), for $\bfx,\bfy\in K$ and
$0<\lambda<1$,
if $M_K^\lambda(\bfx)\cap M_K^\lambda(\bfy)\ne \varnothing$, then for
$\alpha\ge 0$ and $\beta = \frac{2+\alpha(1+\lambda)}{1-\lambda}$, $ M_K^{\alpha \lambda}(\bfx)\subseteq M_K^{\beta \lambda}(\bfy)$.

We present two variants of the expansion-containment property. The proofs are similar to the original proof in \cite{ewaldlarmanrogers1970}.

The first variant relaxes the $\lambda<1$ restriction, but additionally requires the center of one Macbeath region to be contained in the other.

\begin{lemma}
\label{lem:center-containment-general}
Let $K\subseteq \R^d$ be a convex body, $\bfx,\bfy\in K$, and $\lambda\ge 0$. If $\bfx\in M_K^\lambda(\bfy)$, then, for $\alpha\ge 1$,
$ M_K^\alpha(\bfx)\subseteq M_K^{(2 \lambda+1) \alpha}(\bfy)$.
\end{lemma}
\begin{proof}
Since $\bfx \in M_K^\lambda(\bfy)$, we have $\bfx=\bfy+\lambda \bfu$ for some
$\bfu$ with $\bfy\pm \bfu\in K$.
For every $\bfv$ with $\bfx\pm \bfv\in K$, convexity of $K$ gives
\[
 \bfy\pm \frac{\bfv}{\lambda+1}
 =\frac{\lambda(\bfy-\bfu)+(\bfx\pm \bfv)}{\lambda+1}\in K.
\]
Thus both $\bfy+\bfu$ and $\bfy+\bfv/(\lambda+1)$ lie in $M_K^1(\bfy)$.
Then, convexity of $M_K^1(\bfy)$ gives
\[
 \bfy+\frac{\bfx+\alpha \bfv-\bfy}{\lambda+\alpha(\lambda+1)}
 =\frac{\lambda}{\lambda+\alpha(\lambda+1)}(\bfy+\bfu)
 +\frac{\alpha(\lambda+1)}{\lambda+\alpha(\lambda+1)}
   \left(\bfy+\frac{\bfv}{\lambda+1}\right)
 \in M_K^1(\bfy).
\]
Since every point of $M_K^\alpha(\bfx)$ has the form $\bfx+\alpha \bfv$
with $\bfx\pm \bfv\in K$, this gives
\[
 M_K^\alpha(\bfx)
 \subseteq M_K^{\lambda+\alpha(\lambda+1)}(\bfy)
 \subseteq M_K^{(2 \lambda+1) \alpha}(\bfy). \qedhere
\]
\end{proof}

The second variant involves two shrunken Macbeath regions whose images modulo a lattice $\Lambda$ have nonempty intersection.
In general, this condition alone does not imply expansion containment, but an analogous property holds under a balance condition on one of the centers, defined below.

\begin{definition}[Balancedness]
\label{def:balanced-center}
Let $K\subseteq \R^d$ be a convex body, $\Lambda\subseteq \R^d$ be a
lattice, and $\gamma\ge 1$. A point $\bfy\in K$ is
\emph{$\gamma$-balanced in $K$ modulo $\Lambda$} if
$K\cap (\bfy+\Lambda)\subseteq M_K^{\gamma}(\bfy)$.
\end{definition}

\begin{lemma}
\label{lem:quotient-expansion-containment}
Let $K\subseteq \R^d$ be a convex body, $\Lambda\subseteq \R^d$ be a lattice,
$\bfx,\bfy\in K$, and $\gamma\ge 1$. If $\bfy$ is $\gamma$-balanced modulo $\Lambda/2$
and $\pi_\Lambda(M_K^{1/2}(\bfx))\cap \pi_\Lambda(M_K^{1/2}(\bfy))\ne \varnothing$, then, for $\alpha\ge 1$,
$M_K^\alpha(\bfx)\subseteq M_K^{13 \gamma \alpha}(\bfy)$.
\end{lemma}

\begin{proof}
The intersection hypothesis gives
$\bfy+\bfh/2=\bfx+\bfh'/2+\boldsymbol{\ell}$ with
$\bfy\pm \bfh,\bfx\pm \bfh'\in K$ and $\boldsymbol{\ell}\in \Lambda$.
Then, by convexity of $K$,
\[
 \bfc\coloneqq \bfy-\boldsymbol{\ell}/2
 =\frac{1}{2}\left(\bfx+\frac{\bfh'}{2}\right)
  +\frac{1}{2}\left(\bfy-\frac{\bfh}{2}\right)\in K.
\]
Since $\bfc\in K\cap (\bfy+\Lambda/2)$ and $\bfy$ is $\gamma$-balanced modulo $\Lambda/2$, we have $\bfc \in M_K^{\gamma}(\bfy)$, that is,
$\bfy\pm (\bfc-\bfy)/\gamma\in K$. Hence, by convexity of $K$,
\[
 \bfy-\frac{\bfx-\bfy}{4 \gamma+2}
 =\frac{4 \gamma\bigl(\bfy-(\bfc-\bfy)/\gamma\bigr)+(\bfy-\bfh)+(\bfx+\bfh')}{4 \gamma+2}\in K.
\]
Since $\bfy,\bfx\in K$, this gives $\bfx\in M_K^{4 \gamma+2}(\bfy)$.
Then, by \cref{lem:center-containment-general},
\[
 M_K^\alpha(\bfx)\subseteq M_K^{(8 \gamma+5) \alpha}(\bfy)
 \subseteq M_K^{13 \gamma \alpha}(\bfy). \qedhere
\]
\end{proof}

As an immediate corollary, we obtain the
covering property of the quotient variant of Macbeath nets:
\begin{corollary}
\label{cor:quotient-macbeath-cover}
Let $K\subseteq \R^d$ be a convex body, $\Lambda\subseteq \R^d$ be a lattice,
and $\gamma\ge 1$. Let $X\subseteq K$ be such that every $\bfx\in X$ is
$\gamma$-balanced in $K$ modulo $\Lambda/2$. Let $Y\subseteq X$ be
inclusion-maximal such that $\pi_\Lambda(M_K^{1/2}(\bfy))$, $\bfy\in Y$,
are pairwise disjoint. Then, for every $\alpha\ge 1$,
\[
 \bigcup_{\bfx\in X} M_K^\alpha(\bfx)
 \subseteq \bigcup_{\bfy\in Y} M_K^{13 \gamma \alpha}(\bfy).
\]
\end{corollary}
\begin{proof}
By maximality, for every $\bfx\in X$ there exists $\bfy\in Y$ such that
$\pi_\Lambda(M_K^{1/2}(\bfx))\cap \pi_\Lambda(M_K^{1/2}(\bfy))\ne \varnothing$.
Since $\bfy$ is $\gamma$-balanced in $K$ modulo $\Lambda/2$,
\cref{lem:quotient-expansion-containment} gives
$M_K^\alpha(\bfx)\subseteq M_K^{13 \gamma \alpha}(\bfy)$.
\end{proof}

The following lemma will be useful later for packing arguments; it states that,
under a mild condition, 
the Macbeath region of a shallow point remains shallow in a modulo-$\Lambda$ sense.
\begin{lemma}
\label{lem:shallow-macbeath-region}
Let $K,D\subseteq \R^d$ be convex bodies, with $D$ origin-symmetric,
let $\Lambda\subseteq \R^d$ be a lattice, and let $\bfy\in K$.
If $(\bfy+\Lambda/2)\cap (K\ominus D)=\varnothing$, then
$\pi_\Lambda(M_K^1(\bfy))
 \subseteq \pi_\Lambda K\setminus \pi_\Lambda(K\ominus 2D)$.
\end{lemma}
\begin{proof}
Let $\bfx\in M_K^1(\bfy)$, so $\bfx,2 \bfy-\bfx\in K$.
If $\bfx+\boldsymbol{\ell}\in K\ominus 2D$ for some $\boldsymbol{\ell}\in \Lambda$,
then convexity of $K$ gives
\[
 \bfy+\frac{\boldsymbol{\ell}}{2}+D
 =\frac{1}{2}\bigl((\bfx+\boldsymbol{\ell}+2D)+(2 \bfy-\bfx)\bigr)
 \subseteq K,
\]
contradicting $(\bfy+\Lambda/2)\cap (K\ominus D)=\varnothing$.
\end{proof}

\subsection{Finding a balanced center}

Later applications of \cref{cor:quotient-macbeath-cover}  require finding a balanced center to start with.
The following lemma achieves this.

\begin{lemma}
\label{lem:balanced-center}
For any $d\ge 1$, there are $q_d\in \Z_{\ge 1}$ and $\gamma_d\ge 1$ such that the following holds.

Let $K\subseteq \R^d$ be a convex body, $\Lambda\subseteq \R^d$ be a lattice, and $\bfx_0\in K$.
There exists $\bfy\in K\cap (\bfx_0+q_d^{-1} \Lambda)$ that is $\gamma_d$-balanced in $K$ modulo $\Lambda$, and satisfies
$ K\cap (\bfx_0+\Lambda)\subseteq M_K^{\gamma_d}(\bfy)$.

If $K$ is a rational polytope given by a linear optimization oracle and
$\Lambda,\bfx_0$ are rational, then one can compute such $\bfy$ in
$b_0^{\Oh_d(1)} \Tquery(K)$ time, where $b_0=\maxbitlen(K,\Lambda,\bfx_0)$.
\end{lemma}

The proof proceeds as follows.
Start with the affine space $V=\bfx_0+\spa((K-\bfx_0)\cap \Lambda)$. 
Pick a point $\bfy\in \bfx_0+q_d^{-1} \Lambda$ near the center of $K\cap V$. 
If $K\cap (\bfy+\Lambda)$ contains a point outside $V$, we enlarge $V$ to include that point and repeat; otherwise, $\bfy$ has the required properties.
\begin{proof}[Proof of \cref{lem:balanced-center}]
By translation, we may assume $\bfx_0=\zero\in K$ without loss of generality.

Choose an integer $\beta = \Oh_d(1)$ so that $\beta \ge \max \{\beta_1,\ldots,\beta_d\}$
from \cref{lem:macbeath-center}. Define
$q_d=2d \beta$ and $\gamma_d=2 \beta+1$.
Use \cref{lem:lattice-span}
to compute $\bfv_1,\ldots,\bfv_{d'}\in K\cap \Lambda$ that form a  basis of $\spa(K\cap \Lambda)$.

We proceed in iterations $i=d',d'+1,\ldots,\rk \Lambda$.
At the start of iteration $i$, we have linearly independent vectors
$\bfv_1,\ldots,\bfv_i$,  and write
$V_i=\spa(\bfv_1,\ldots,\bfv_i)$ and $F_i=K\cap V_i$.
We maintain 
\begin{equation}
\label{eq:balanced-center-section}
 \bfv_j\in \Lambda\cap (F_i-F_i) \text{ for all } 1\le j\le i.
\end{equation}
Initially, $i=d'$, and 
\eqref{eq:balanced-center-section}
holds since $\zero\in K$.

In iteration $i$, apply \cref{lem:macbeath-center} to $F_i\subseteq V_i$, obtaining
$\bfc\in F_i$ and an origin-symmetric parallelepiped $B$ with
$\bfc+B\subseteq F_i\subseteq \bfc+\beta B$.
Write $\bfc=\sum_{j=1}^{i} t_j \bfv_j$, and let $n_j$ be the nearest integer to $q_d t_j$, for each $1\le j\le i$.
 Define $\bfy_i=q_d^{-1} \sum_{j=1}^{i} n_j \bfv_j \in V_i \cap q_d^{-1} \Lambda$.
By $ F_i \subseteq \bfc + \beta B$ and \eqref{eq:balanced-center-section}, $\bfv_j\in 2 \beta B$ for $1\le j\le i$, so
\[
 \bfy_i-\bfc = q_d^{-1} \sum_{j=1}^i (n_j - q_d t_j) \bfv_j\in q_d^{-1} i \beta B\subseteq \frac{1}{2} B.
\]
Combining with $\bfc + B\subseteq F_i \subseteq \bfc+\beta B$, this gives
$\bfy_i+B/2\subseteq F_i\subseteq \bfy_i+(\beta+1/2) B$. Since $F_i \subseteq K$ and $\gamma_d = 2 (\beta +1/2)$, this gives
\begin{equation}
\label{eq:balanced-center-rounding1}
 F_i\subseteq M_K^{\gamma_d}(\bfy_i).
\end{equation}
From $\bfy_i +B/2\subseteq F_i$ we also get
\begin{equation}
\label{eq:balanced-center-rounding2}
 \bfy_i\in F_i\cap q_d^{-1} \Lambda \subseteq K\cap q_d^{-1} \Lambda.
\end{equation}

Use \cref{lem:lattice-span} to
search for $\bfx\in K\cap (\bfy_i+\Lambda)$ outside $V_i$.
\begin{itemize}
  \item 
If such $\bfx$ exists, append $\bfv_{i+1}=\bfx-\bfy_i \in \Lambda$ and continue
to iteration $i+1$.
Since $\bfv_{i+1}\notin  V_i$, appending it preserves the linear independence of $\bfv_1,\dots,\bfv_{i+1}$.
We have 
$\bfx \in F_{i+1}$, and $\bfy_i \in F_i \subseteq F_{i+1}$,
 so $\bfv_{i+1}=\bfx-\bfy_i\in F_{i+1}-F_{i+1}$, and
\eqref{eq:balanced-center-section} holds at iteration $i+1$.
\item If no such point $\bfx$ exists, return $\bfy=\bfy_i$, which satisfies the congruence requirement due to 
\eqref{eq:balanced-center-rounding2}.
In this case, $K\cap (\bfy_i+\Lambda)\subseteq K\cap V_i = F_i$, so
\eqref{eq:balanced-center-rounding1} proves that $\bfy_i$ is $\gamma_d$-balanced in $K$ modulo $\Lambda$ as claimed. 
\Cref{eq:balanced-center-rounding1} and $K\cap \Lambda\subseteq F_{d'}\subseteq F_i$ yield the other claimed containment 
$ K\cap \Lambda \subseteq M_K^{\gamma_d}(\bfy_i)$.
\end{itemize}

The time and maximum bit length bounds follow from 
\cref{lem:constrained-optimization,lem:macbeath-center,lem:lattice-span}.
\end{proof}

\subsection{Constructing a cap}

For a convex body $K\subseteq \R^d$,  a \emph{cap} of $K$ is the intersection of $K$ with a halfspace,  $Q=K\cap \{\bfx\in \R^d:\langle\bfu,\bfx\rangle\ge a\}$, and is represented by its defining halfspace $(\bfu,a)$.

It is standard to use caps to approximate Macbeath regions, and vice versa; see e.g., \cite{ewaldlarmanrogers1970,BronnimannCP93,aryafonsecamount2017}.
A classic lemma \cite[Lemma 2.3]{aryafonsecamount2017} states that a sufficiently shallow cap is contained in a dilated Macbeath region centered at the centroid of the cap base.

Here we use a slight variant that is more convenient for our purpose. Given a prescribed center, the lemma constructs a cap that approximates the Macbeath region at this center.

\begin{lemma}
\label{lem:macbeath-cap}
Let $K\subseteq \R^d$ be a convex body and let
$\bfy$ be an interior point of $K$.
There exists $\bfu\ne \zero$ such that, for every $t\ge 1$, the cap
$Q_t=\{\bfx\in K:\langle\bfu,\bfx-\bfy\rangle\ge -t\}$ satisfies
\[
 K\cap M_K^t(\bfy)\subseteq Q_t\subseteq M_K^{(4d+3) t}(\bfy).
\]

If $K$ is a rational polytope given by a linear optimization oracle and
$\bfy$ is rational, one can compute $\bfu$ with
$\maxbitlen(\bfu)=\Oh_d(b_0)$ in $b_0^{\Oh_d(1)} \Tquery(K)$ time,
where $b_0=\maxbitlen(K,\bfy)$.
\end{lemma}
\begin{proof}
By translation, assume $\bfy=\zero$ without loss of generality.
Consider the polar
$K^\circ=\{\bfa\in \R^d:\langle\bfa,\bfx\rangle\le 1 \text{ for all }\bfx\in K\}$.
Use \cref{lem:bounded-coordinate-selection} to find 
$\bfp_0,\ldots,\bfp_d\in K^\circ$, and 
define $\bfc=(d+1)^{-1} \sum_{i=0}^d \bfp_i$.
If $\bfc\ne \zero$, return $\bfu=\bfc$; otherwise, return any nonzero $\bfp_i$ as $\bfu$. 
In both cases, $\bfu\ne \zero$ and $\bfu\in K^\circ$. We now prove the correctness.

To prove the first containment, take any $\bfx \in K \cap M_K^t(\zero) = K\cap (-tK)$, and it suffices to show $\langle \bfu, \bfx \rangle \ge -t$. This directly follows from the fact that $\bfu\in K^{\circ}$ and $\bfx \in -tK$.

To prove the second containment, fix any $\bfx\in Q_t = \{ \bfx\in K : \langle \bfu,\bfx\rangle \ge -t\}$.
The definition of $\bfu$ implies $\langle \bfc,\bfx \rangle \ge -t$ in both cases. 
Let $\bfa \in K^\circ$ be arbitrary. Then, \cref{lem:bounded-coordinate-selection} gives 
$\bfa=\sum_{i=0}^d \beta_i \bfp_i$ with $\sum_{i=0}^d \beta_i=1$ and $\lvert\beta_i\rvert\le 2$, so
\[
 \bfq\coloneqq \frac{(2d+2) \bfc - \bfa}{2d+1} = \frac{(2 \sum_{i=0}^d \bfp_i)-\bfa}{2d+1}
 =\sum_{i=0}^d \frac{2-\beta_i}{2d+1} \bfp_i\in K^\circ,
\]
where we used the convexity of $K^\circ$. 
Since $\bfx \in K$, this implies $\langle\bfq,\bfx\rangle\le 1$.
Then,
\[
\begin{aligned}
 \langle\bfa,\bfx\rangle
 &=(2d+2)\langle\bfc,\bfx\rangle-(2d+1)\langle\bfq,\bfx\rangle\\
 &\ge -(2d+2) t-(2d+1)\ge -(4d+3) t,
\end{aligned}
\]
where the last inequality uses $t\ge 1$.
Since $\bfa \in K^\circ$ is arbitrary, this gives $-\bfx\in (4d+3) t(K^\circ)^\circ = (4d+3) tK$.
Combining with $\bfx \in K$, this gives 
$\bfx \in (4d+3) t(K \cap (-K)) = M_K^{(4d+3) t}(\zero)$ as desired.

The time and maximum bit length bounds follow from 
\cref{lem:polar-optimization} and \cref{lem:bounded-coordinate-selection}.
\end{proof}

\section{Quotient Cap Covering}
\label{sec:quotient-cap-covering}

In this section, we state and prove
our quotient variant of economic cap covering.
It is a collection of caps $Q_i$ of a convex body $K$ covering the ``quotient shell'' of $K$, namely the set of $\bfx \in K$ for which no lattice-congruent point $\bfx'\in \bfx + \Lambda$ belongs to the deep inner part $K\ominus D$. The total volume of the caps $Q_i$ measured in $\R^d/\Lambda$ is small.

The formal statement is given in
\cref{thm:geometric-cap-cover} below.  For computational efficiency,  we approximate each returned cap $Q_i$ by a parallelepiped $P_i$.
Additionally, we return a set $R$ of points that forms a $D$-scale net (in a modulo-$\Lambda$ sense) of the inner neighborhoods of the caps. 
In later applications, points in $R$ serve as starting points to
 explore the deep part $(K\ominus D) \bmod \Lambda$.

\begin{theorem}
\label{thm:geometric-cap-cover}
Given a lattice $\Lambda\subseteq \R^d$ of rank $d'<d$, a convex body $K\subseteq \R^d$, and $0<\rho<1$, write
$D=\rho(K-K)/2$ and
\[r=\rho \vol_d(\pi_\Lambda K)/\vol_d(\pi_\Lambda D).\]

Then, there exist
a collection of $\Oh_d(r)$ caps $Q_i=K\cap \{\bfx: \langle\bfu_i,\bfx\rangle\ge a_i\}$ with enclosing parallelepipeds $P_i\supseteq Q_i$, and a point set $R\subseteq \R^d$ of size $\lvert R\rvert=\Oh_d(r)$, such that the following properties hold:
\begin{enumerate}[label=\textup{(\roman*)},leftmargin=2.5em,beginpenalty=10000]
   \item\label{item:cap-cover-volume} $\sum_i \vol_d(\pi_\Lambda P_i) \le \Oh_d(\rho) \vol_d(\pi_\Lambda K)$,
\item\label{item:cap-cover-thickness} each $Q_i$ contains a translate of $c_d D$, for some $0<c_d\le 1$ depending only on $d$,
   \item\label{item:cap-cover-deep} $K\setminus \bigcup_i Q_i\subseteq (K\ominus D)+\Lambda$,
   \item\label{item:cap-cover-centers} $R+4d \beta D\subseteq K$, where $\beta=\lceil\beta_d\rceil$ for the constant $\beta_d$ in \cref{lem:macbeath-center},
   \item\label{item:cap-cover-local} and $\bigl(K\setminus \bigcup_i Q_i\bigr)\cap (\bigcup_i P_i+D+\Lambda)\subseteq R+D+\Lambda$.
\end{enumerate}

If $K$ is a rational polytope given by a linear optimization oracle and
$\Lambda,\rho$ are rational,
then, writing $b_0=\maxbitlen(K,\Lambda,\rho)$,
one can compute  all $\bfu_i,a_i,P_i$
and $R$  
in time \[r\cdot b_0^{\Oh_d(1)} \Tquery(K),\] with
$\maxbitlen(\bfu_i,a_i,P_i)=\Oh_d(b_0)$ for every $i$
and
$\maxbitlen(\bfr)=\Oh_d(b_0)$ for every $\bfr\in R$.
\end{theorem}

We will show in \cref{lem:geometric-boundary-representatives} that $r
= \rho \vol_d(\pi_\Lambda K)/\vol_d(\pi_\Lambda D) \gg_d 1$.

In the extreme case where $\rho\gg_d 1$,
\cref{thm:geometric-cap-cover} follows easily:
we can return the single trivial cap $K$, with an enclosing parallelepiped $B\supseteq K-\bfc \supseteq  \Omega_d(1)\cdot B$
computed from \cref{lem:macbeath-center}, and $R=\varnothing$.
\cref{item:cap-cover-thickness} directly follows, 
 \cref{item:cap-cover-volume} 
holds due to \cref{lem:volumeineq}, 
 and
\cref{item:cap-cover-deep,item:cap-cover-centers,item:cap-cover-local}
are vacuous.

The rest of this section  proves \cref{thm:geometric-cap-cover} when $0<\rho< \eps_d$, where $\eps_d\in (0,1)$ only depends on $d$.

\subsection{Preliminaries}
\label{subsec:quotient-cover-preliminaries}
We begin with two simple geometric lemmas.

The volume in $\R^d/\Lambda$ satisfies the following subhomogeneity property for convex bodies under dilation.
\begin{lemma}
\label{lem:volumeineq}
Let $\Lambda\subseteq \R^d$ be a lattice and $K\subseteq \R^d$ be a convex body. For $0<a\le b$,
\begin{equation*}
 \frac{\vol_d(\pi_\Lambda(aK))}{a^d}
 \ge \frac{\vol_d(\pi_\Lambda(bK))}{b^d}.
\end{equation*}
\end{lemma}

\begin{proof}
For a compact convex set $E$, let $n_E(\bfx)=\lvert(\bfx+\Lambda)\cap E\rvert$.
By partitioning $E$ into lattice translates of a fundamental domain, we get
\[\vol_d(\pi_\Lambda(E))=\int_E n_E(\bfx)^{-1}\,d \bfx.\]

For $\bfx\in K$, convexity gives
$a(K-\bfx)\subseteq b(K-\bfx)$, so
$n_{aK}(a \bfx) = \lvert\Lambda \cap (aK-a \bfx)\rvert \le \lvert\Lambda \cap (bK-b \bfx)\rvert =  n_{bK}(b \bfx)$. Hence
\[
 \frac{\vol_d(\pi_\Lambda(aK))}{a^d}
 =\int_K \frac{d \bfx}{n_{aK}(a \bfx)}
 \ge \int_K \frac{d \bfx}{n_{bK}(b \bfx)}
 =\frac{\vol_d(\pi_\Lambda(bK))}{b^d}.\qedhere
\]
\end{proof}

We  use
\cref{lem:volumeineq}
to bound the volume of the quotient shell as follows.
\begin{lemma}
\label{lem:quotient-shell-bound}
Let $\Lambda\subseteq \R^d$ be a lattice, and
$K,D\subseteq \R^d$ be convex bodies, where $D$ is origin-symmetric, and $K$ contains a translate of $\rho^{-1} D$ for some $\rho>0$.
Then, for all $a,b\ge 0$,
\[
 \vol_d\!\big(\pi_\Lambda(K+aD)\setminus \pi_\Lambda(K\ominus bD)\big)
 \le d(a+b) \rho(1+a \rho)^{d-1} \vol_d(\pi_\Lambda K).
\]
\end{lemma}

\begin{proof}
   By translation, assume $\rho^{-1} D\subseteq K$ without loss of generality.
Convexity gives $K+aD\subseteq (1+a \rho) K$ and, when $b \rho<1$,
$(1-b \rho) K\subseteq K\ominus bD$.
Writing $v=\vol_d(\pi_\Lambda K)$,
\cref{lem:volumeineq}  gives
\begin{align*}
 \vol_d\!\big(\pi_\Lambda(K+aD)\setminus \pi_\Lambda(K\ominus bD)\big)
 &\le \bigl((1+a \rho)^d-\max \{1-b \rho,0\}^d\bigr) v\\
 &\le d(a+b) \rho(1+a \rho)^{d-1} v,
\end{align*}
where the last step used the mean value theorem.
\end{proof}

The following lemma finds a deep lattice translate of a point in a convex polytope.
\begin{lemma}
\label{lem:lattice-translate-depth}
In $\R^d$,
given a polytope $K$ with a linear optimization oracle,
a lattice $\Lambda $, an origin-symmetric $d$-dimensional parallelepiped $B $,
and a point $\bfx$, one can find the largest $t\in \Nonneg$
(or report that none exists)
such that
\[
 \bfx+\boldsymbol{\ell}+tB\subseteq K
 \text{ for some }\boldsymbol{\ell}\in \Lambda,
\]
together with such a $\boldsymbol{\ell}$, 
in $\maxbitlen(K,\Lambda,B,\bfx)^{\Oh_d(1)} \Tquery(K)$ time.
\end{lemma}

\begin{proof}
By convexity, 
the condition 
$\bfx+\boldsymbol{\ell}+tB\subseteq K$ is equivalent to 
$\bfx+\boldsymbol{\ell}+t \bfv \in K$ for all $2^{d}$ vertices $\bfv$ of $B$.
The task then reduces to solving a $(d+1)$-dimensional integer linear program (ILP) over the lattice 
$(\boldsymbol{\ell},t) \in \Lambda \times \Z$, 
and a feasibility polytope with a linear optimization oracle that can be built from $K$ via \Cref{lem:constrained-optimization}.
The ILP can be solved using
\cref{lem:lattice-optimization}.
The claimed time complexity then follows from
\Cref{lem:constrained-optimization} and
\cref{lem:lattice-optimization}.
\end{proof}

\subsection{Discretization into cells}
\label{sec:quotient-discretization}

We now proceed to the proof of \cref{thm:geometric-cap-cover}. Write $d'=\rk \Lambda < d$.

Recall $D \coloneqq \rho(K-K)/2$, where $0<\rho<1$. For computational efficiency, we first approximate $D$ by a parallelepiped $B$ as follows.
Use \cref{lem:macbeath-center} to  choose an origin-symmetric parallelepiped $B$
and a point $\bfc$ such that
\begin{equation}
   \label{eqn:defnb}
\bfc+\rho^{-1} B\subseteq K\subseteq \bfc+\beta \rho^{-1} B.
\end{equation}
where $\beta=\lceil\beta_d\rceil = \Oh_d(1)$  from \cref{lem:macbeath-center}.
Then, taking differences gives
\begin{equation}
B\subseteq D\subseteq \beta B.
\label{eqn:dbd}
\end{equation}
We sometimes also work with $B_0\coloneqq 2d \beta^2 B$.

For convenience, we discretize 
$\Lambda$ into 
 a superlattice $\caL(\bfv_1,\dots,\bfv_{d})$ whose fundamental parallelepiped approximates $\pi_{\Lambda} B$:
\begin{lemma}[$\bfv_i,k_i$, and $V$]
\label{lem:quotient-box-discretization}
We can compute positive integers $k_1,\ldots,k_{d'}$ and a basis
$\bfv_1,\ldots,\bfv_d$ of $\R^d$, such that
$k_1 \bfv_1,\ldots,k_{d'} \bfv_{d'}$ form a basis of $\Lambda$, and the centered zonotope
\begin{equation*}
 V\coloneqq \caZ_0((\bfv_i)_{i=1}^{d})
\end{equation*}
 satisfies
\begin{equation}
\label{eq:cap-cover-box-scale}
 2V\subseteq B \text{ and }
 \vol_d(V)=\Theta_d(\vol_d(\pi_\Lambda B)).
\end{equation}

The computation takes time
$b_0^{\Oh_d(1)}$, and the return values satisfy
$\maxbitlen(k_1,\ldots,k_{d'},
\bfv_1,\ldots,\bfv_d)=\Oh_d(b_0)$.
\end{lemma}

Combining with \eqref{eqn:dbd} gives 
$\vol_d(V)=\Theta_d(\vol_d(\pi_\Lambda D))$.

\begin{proof}[Proof of 
\cref{lem:quotient-box-discretization}]
By applying a linear transformation, we may assume $B=[-1,1]^d$
without loss of generality.
Apply \cref{lem:latticereduct} to compute an LLL basis $\bfu_1,\ldots,\bfu_{d'}$ of $\Lambda$.
Complete $\bfu_1,\ldots,\bfu_{d'}$ to a basis
$\bfu_1,\ldots,\bfu_d$ of $\R^d$ by adjoining
an orthogonal basis of $\spa(\Lambda)^\perp$.
Write $\delta_i=\|\bfu_i\|_\infty$ for $1\le i\le d$, and define the matrix
\[A=(\bfu_1/\delta_1,\cdots, \bfu_d/\delta_d),\] in which each entry has absolute value at most one, so
\[A[-1/d,1/d]^d\subseteq B.\]
By \cref{lem:latticereduct},
\[\lvert\det(A)\rvert = \frac{\lvert\det(\bfu_1,\dots,\bfu_d)\rvert}{\prod_{i=1}^d \|\bfu_i\|_\infty}\ge \frac{\lvert\det(\bfu_1,\dots,\bfu_d)\rvert}{\prod_{i=1}^d \|\bfu_i\|_2} =\frac{\det \Lambda}{\prod_{i=1}^{d'} \|\bfu_i\|_2}  \ge \Omega_d(1).\] 
 Then, Cramer's rule gives
\[B \subseteq A[-C_d,C_d]^d\]
for a constant $C_d$ depending only on $d$.

Since $A^{-1} \Lambda=\caL\big ((\delta_i \bfe_i)_{i=1}^{d'}\big)$,
for $s>0$, 
\[
 \vol_d\bigl(\pi_\Lambda(A[-s,s]^d)\bigr) = \lvert\det A\rvert \vol_d\bigl(\pi_{A^{-1} \Lambda}([-s,s]^d)\bigr)  =\lvert\det A\rvert(2s)^{d-d'}
   \prod_{i=1}^{d'} \min \{2s,\delta_i\}.
\]
Thus, plugging in $s=1/d$ and $s=C_d$, the two containments above imply
\[\vol_d(\pi_\Lambda(B))
=\Theta_d(\lvert\det A\rvert\prod_{i=1}^{d'} \min \{1,\delta_i\}).\]

Define $k_i=\lceil2d \delta_i\rceil$ and $\bfv_i=\bfu_i/k_i$ for all $1\le i\le d'$, so that $\delta_i/k_i = \Theta_d( \min \{1,\delta_i\})$.
For $d'<i\le d$, define $\bfv_i = \bfu_i/(2d \delta_i)$. Then,
\[
 \vol_d(V)= \lvert \det (\bfv_1,\dots,\bfv_d)\rvert = \lvert\det A\rvert \left (\prod_{i=1}^{d'} \frac{\delta_i}{k_i}\right ) \cdot (2d)^{-(d-d')}
 =\Theta_d(\vol_d(\pi_\Lambda(B))),
\]
as claimed.
Also, since $\|\bfv_i\|_\infty \le \frac{1}{2d}$, we have 
$2V\subseteq B$ as claimed. 
\end{proof}

We apply \cref{lem:quotient-box-discretization} to $\Lambda$ and $B$, and obtain $k_i,\bfv_i,V = \sum_{i=1}^d [-1/2,1/2]\cdot \bfv_i$.  
Define the \emph{cell centers}
\begin{equation*}
 \caG =\sum_{i=1}^{d'}\{0,\ldots,k_i-1\} \bfv_i
 +\sum_{i=d'+1}^{d}\Z  \bfv_i.
\end{equation*}
For $\bfg \in \caG$, we call $\bfg+V$ or 
$\pi_\Lambda(\bfg+V)$ a \emph{cell}, with volume $\vol_d(\pi_\Lambda(V)) = \vol_d(V)  \in \Theta_d(\vol_d(\pi_\Lambda(B)))$.
The cells have disjoint interior, and the collection of half-open cells 
$\{\bfg + \sum_{i=1}^d [-1/2,1/2)\cdot \bfv_i : \bfg\in \caG\}$ forms a tiling of a fundamental domain of $\Lambda$ in $\R^d$.

We index 
$\bfg\in \caG$ by its coordinate $(\hat g_1,\dots,\hat g_d)$ in the 
 basis $(\bfv_i)_i$, where $\bfg = \sum_{i=1}^d \hat g_i \bfv_i$, 
and order points in $\caG$ lexicographically by the coordinates $(\hat g_1,\dots,\hat g_d)$.

Given a polytope $K'\subseteq \R^d$, define the \emph{occupied} centers
\begin{equation}
\caP_{K'}=\{\bfg\in \caG:(\bfg+V)\cap (K'+\Lambda)\ne \varnothing\}.\label{eqn:occupiedp}
\end{equation}
Since $(\bfg+V)_{\bfg \in \caG}$ form a tiling, 
$\caP_{K'}$ is nonempty whenever $K'$ is nonempty.

\begin{lemma}
\label{lem:quotient-cell-search}
Given $\bfh\in \caG$, one can find its lexicographical successor in $\caP_{K'}$ (or report that no successor exists), and find the 
lexicographically first point of 
$\caP_{K'}$, 
in $(b_0+\maxbitlen(\bfh,K'))^{\Oh_d(1)} \Tquery(K')$ time.
\end{lemma}

\begin{proof}
To find the successor (or the first) in lexicographical order, it suffices to solve the task of deciding whether there exists $\bfg = \sum_{i=1}^d \hat g_i \bfv_i \in \caP_{K'}$ subject to individual inequality constraints on each coordinate $\hat g_i$.
Since
 $\bfg\in \caP_{K'}$ if and only if 
$\bfg+\boldsymbol{\ell}\in K'-V$ for some $\boldsymbol{\ell}\in \Lambda$,
this task reduces to solving ILPs over the lattice
$(\bfg,\boldsymbol{\ell})\in \caL(\bfv_1,\ldots,\bfv_d)\times \Lambda$, and a 
feasibility polytope in $\R^d\times \spa(\Lambda)$ with a linear optimization oracle that can be built from $K'$ and $V$ via
\Cref{lem:constrained-optimization}. The ILPs can be solved using
\cref{lem:lattice-optimization}.
The claimed time complexity follows from
\Cref{lem:constrained-optimization} and \cref{lem:lattice-optimization}.
\end{proof}

\subsection{Enumerating shallow cells}
Recall the definition of the parallelepiped
$B$
in  \eqref{eqn:defnb}.
We scale it up to \[B_0\coloneqq 2d \beta^2 B.\]
Then, $K$ contains a translate of $(2d \beta^2 \rho)^{-1} B_0$.
Define the \emph{shallow} cell centers, 
\begin{equation}
   \label{eqn:shallows}
\caS =\caP_K\setminus \bigl((K\ominus 2B_0)+\Lambda\bigr),
\end{equation}
where $\caP_K$ is defined in 
\eqref{eqn:occupiedp}.

\begin{lemma}
\label{lem:shallow-cell-shell}
For every $\bfg\in \caS$,
\[
 \pi_\Lambda(\bfg+V)\subseteq
 \pi_\Lambda(K+B_0)\setminus \pi_\Lambda(K\ominus 3B_0).
\]
\end{lemma}

\begin{proof}
Since $\bfg\in \caP_K$, the definition of $\caP_K$ yields 
$\pi_\Lambda(\bfg+V)\subseteq
\pi_\Lambda(K-V+V)\subseteq
\pi_\Lambda(K+B_0)$,
where the last step uses
$V-V=2V\subseteq B\subseteq B_0$
by \eqref{eq:cap-cover-box-scale}.

If $\bfg + \bfv +\boldsymbol \ell \in K\ominus 3B_0$, $\bfv\in V$, $\boldsymbol \ell \in \Lambda$, then 
 $\bfg + \boldsymbol \ell \in (K\ominus 3B_0) - V\subseteq K\ominus 2B_0$, 
contradicting $\bfg\in \caS$.
This proves the claim.
\end{proof}

We now bound the number of shallow centers and describe how to enumerate all of them.
The enumeration algorithm works by 
 skipping long intervals of deep centers.
\begin{lemma}
\label{lem:geometric-boundary-representatives}
We have $1\le \lvert\caS\rvert=\Oh_d\bigl(\rho \vol_d(\pi_\Lambda K)/\vol_d(\pi_{\Lambda} D)\bigr)$.
Moreover, one can compute $\caS$ in
$\Oh_d\bigl(\rho \vol_d(\pi_\Lambda K)/\vol_d(\pi_{\Lambda} D)\bigr) \cdot 
b_0^{\Oh_d(1)} \Tquery(K)$ time.
\end{lemma}

\begin{proof}~
   \paragraph*{Bounding $\lvert\caS\rvert$.}
To see that $\caS\ne \varnothing$, choose the lexicographically first $\bfg\in \caP_K$.
If $\bfg\notin \caS$, then
$\bfg + \boldsymbol \ell \in K\ominus 2B_0$ for some $\boldsymbol \ell \in \Lambda$, so $\bfg-\bfv_d+
\boldsymbol \ell
\in (K\ominus 2B_0) + 2V \subseteq K$, and hence $\bfg-\bfv_d \in \caP_{K}$, contradicting the choice of $\bfg$.

Since $K$ contains a translate of $(2d \beta^2 \rho)^{-1} B_0$, where $\rho<1$ and $\beta=\Oh_d(1)$,
\cref{lem:quotient-shell-bound} shows that the shell
$\pi_\Lambda(K+B_0)\setminus \pi_\Lambda(K\ominus 3B_0)$ from
\cref{lem:shallow-cell-shell} has volume at most $\Oh_d(\rho) \vol_d(\pi_\Lambda K)$.
Then, \cref{lem:shallow-cell-shell} and the disjointness of the interiors of $\pi_{\Lambda}(\bfg+V)$ give
$\lvert\caS\rvert\vol_d(V)\le \Oh_d(\rho) \vol_d(\pi_\Lambda K)$.
Since $\vol_d(V)=\Theta_d(\vol_d(\pi_\Lambda D))$, this proves the claimed bound on $\lvert\caS\rvert$.

Analogously,
\cref{lem:quotient-shell-bound} shows
 \begin{equation}
    \label{eqn:tempboundvol}
 \vol_d\bigl(\pi_\Lambda K\setminus
       \pi_\Lambda(K\ominus sB_0)\bigr)
 \le \Oh_d(s \rho) \vol_d(\pi_\Lambda K)
 \text{ for every }s>0.
 \end{equation}

\paragraph*{Enumeration algorithm.}

Use \cref{lem:quotient-cell-search}  to find the lexicographically
first $\bfg_0\in \caP_K$.
Starting from $\bfg\gets \bfg_0$, we perform the following iterations:
for the current $\bfg$, apply \cref{lem:lattice-translate-depth}
to find the largest $t\in \Nonneg$ such that 
there exists $\boldsymbol \ell \in \Lambda$ with 
$\bfg + \boldsymbol \ell + tB_0 \subseteq K$.
Then,  by definition of $\caS$ in
   \eqref{eqn:shallows},
$\bfg \in \caS$ if and only if $t<2$ or $t$ does not exist.
\begin{itemize}
   \item If
$\bfg \in \caS$, add $\bfg$ to the answer set, and use \cref{lem:quotient-cell-search} to find the next $\bfg' \in \caP_K$ after $\bfg$ in  lexicographical order, and continue the next iteration with $\bfg\gets \bfg'$.
\item Otherwise, $t\ge 2$ and $\bfg \notin \caS$.
In this case we skip the next candidates $\bfg + j \bfv_d$ for all $0\le j\le t-2$:  
find the next
$\bfg'' \in \caP_K$
after $\bfg+(t-2) \bfv_d$ in lexicographical order
via \cref{lem:quotient-cell-search}, and continue to
$\bfg\gets \bfg''$.
\end{itemize}
Terminate this algorithm once no lexicographical successor in $\caP_K$ exists.

We first show correctness.
Clearly, 
$\bfg \in \caP_K$ throughout the algorithm and its lexicographical order is increasing.
To prove that no elements of $\caS$ are missed, it suffices to prove that all skipped elements in the second case are not in $\caS$.
Indeed, since $\bfv_d \in B_0$,  the definition of $t$ implies, for all $0\le j\le t-2$, that
\[
 \bfg+j \bfv_d+\boldsymbol{\ell}+2B_0
 \subseteq \bfg+\boldsymbol{\ell}+tB_0\subseteq K,
\]
that is, $\bfg + j \bfv_d  \in (K\ominus 2B_0)+\Lambda$, so $\bfg + j \bfv_d \notin \caS$.
This proves that the algorithm correctly finds all elements of $\caS$.

Now we bound the number of iterations in this algorithm.
The iterations that find $\bfg \in \caS$ can be bounded by $\lvert\caS\rvert$; it remains to bound the second type of iterations which have $t\ge 2$.
Consider any iteration with $t\ge 2$, and let
\[J\coloneqq \bigcup_{j=0}^{t-2}(\bfg+j \bfv_d+V) \subseteq \bfg + (t-1) B_0 \subseteq K-\boldsymbol \ell\] denote the union of cells of the skipped centers.
We claim that $(\pi_{\Lambda} J) \cap \pi_{\Lambda}(K\ominus 2tB_0) = \varnothing$.
Indeed, if $\bfu+\bfm\in K\ominus 2tB_0$ for some $\bfu\in J$ and 
$\bfm\in \Lambda$, then,
since
$\bfu \in \bfg +(t-1) B_0$, we have
\[
 \bfg+\bfm+(t+1) B_0
 \subseteq \bfu+\bfm+2tB_0\subseteq K,
\]
contradicting the maximality of $t$.

Let $s\ge 2$ be a power of two. Consider all iterations with $t\in [s,2s)$. The previous paragraph shows that the corresponding set $J$ satisfies
$\pi_{\Lambda} J\subseteq \pi_\Lambda K\setminus \pi_\Lambda(K\ominus 4sB_0)$; by \eqref{eqn:tempboundvol}, this set has volume $\Oh_d(s \rho \vol_d(\pi_{\Lambda} K))$.
Since $\vol_d(\pi_\Lambda J) = \vol_d(J) = (t-1) \vol_d(V) \ge (s/2) \vol_d (V)$, and 
the sets $\pi_\Lambda J$ across the iterations have disjoint interiors, we can thus bound the number of iterations with $t\in [s,2s)$ by the volume ratio, $\Oh_d(\rho \vol_d(\pi_\Lambda K)/\vol_d(V))$.
Using $\maxbitlen(\rho, K)\le b_0$, we have $\maxbitlen(B_0) = \Oh_d(b_0)$, and hence the range of possible $t$ is 
$t\le 2^{\Oh_d(b_0)}$. 
Summing over $\Oh_d(b_0)$ relevant powers of two $s\ge 2$, and combining the proved bound on $\lvert\caS\rvert$, the total number of iterations is
$\Oh_d(\rho \vol_d(\pi_\Lambda K)/\vol_d(V))\cdot b_0$.
Finally, combining the time bounds in \cref{lem:quotient-cell-search,lem:lattice-translate-depth}
gives the claimed time complexity.
\end{proof}

\subsection{Constructing a quotient Macbeath net}

Recall from  \eqref{eqn:defnb} that the origin-symmetric parallelepiped $B_0$ satisfies 
$\bfc+
(2d \beta^2 \rho)^{-1} B_0
\subseteq K\subseteq \bfc+
(2d \beta \rho)^{-1} B_0$, where $1\le \beta \le \Oh_d(1)$.
Recall from \eqref{eq:cap-cover-box-scale} that $V$ is an origin-symmetric parallelepiped with $2V\subseteq B\subseteq B_0$.
By the discussion preceding
\cref{subsec:quotient-cover-preliminaries}, we may assume 
\[\theta\coloneqq 2d \beta^2 \rho < 1.\] 
For convenience, we will work with Macbeath regions in the following enlarged version of $K$ to provide thickness near the boundary:
\[
 K^+\coloneqq \bfc+(1+2 \theta)(K-\bfc),
 \]
 so 
 \[
 K+2B_0\subseteq K^+\subseteq K+2 \beta B_0.
\]
In particular, 
 for all  $\bfp \in K+B_0$,
\begin{equation}
   \label{eqn:boundthick}
\bfp + B_0 \subseteq M_{K^+}^1(\bfp).
   \end{equation}

Our goal in this subsection is to select a set of Macbeath regions, with small total volume in $\R^d/\Lambda$, that forms a  covering of all the shallow cells $(\bfg+V)_{\bfg\in \caS}$ in $K$ modulo $\Lambda$.
To achieve this, we first select a subset of balanced centers (see \cref{def:balanced-center}) whose Macbeath regions already form a covering. 
From them, we greedily select an inclusion-maximal pairwise disjoint subcollection. Their total volume is bounded by packing, and the covering property (after enlargement) follows from  
\cref{cor:quotient-macbeath-cover}
due to balancedness of the centers.

\paragraph*{Candidate balanced centers.}

Take $q=q_d\in \Pos$ and $\gamma=\lceil\gamma_d\rceil\ge 1$ from \cref{lem:balanced-center}.
Define \[\alpha=2 \gamma+1,\] and for $\bfy\in K^+$ use the shorthand
\[M_\bfy\coloneqq (K^+-\bfy)\cap (\bfy-K^+) = M_{K^+}^1(\bfy) - \bfy.\]

For each shallow cell, the following lemma chooses a constant number of balanced centers modulo $q \Lambda$, whose Macbeath regions cover all its $\Lambda$-lattice translates. The proof proceeds by enumerating the cosets of $q \Lambda$ in $\Lambda$, and applying 
\cref{lem:balanced-center} to the lattice $q \Lambda$.
\begin{lemma}
\label{lem:cap-cover-balanced-centers}
Given $\bfg\in \caS$,
in $b_0^{\Oh_d(1)} \Tquery(K)$ time, 
one can compute $Y_\bfg\subseteq K^+\cap (\bfg+\Lambda)$ with
$\lvert Y_\bfg\rvert=\Oh_d(1)$, such that each
$\bfy\in Y_\bfg$ is $\gamma$-balanced in $K^+$ modulo $q \Lambda$ and
satisfies $B_0\subseteq \alpha M_\bfy$. Moreover,
\begin{equation}
\label{eq:cap-cover-candidates}
 K\cap (\bfg+V+\Lambda)\subseteq
 \bigcup_{\bfy\in Y_\bfg} M_{K^+}^\alpha(\bfy).
\end{equation}
\end{lemma}
\begin{proof}
   The set $\bfg + \Lambda$ is partitioned into 
$q^{\rk \Lambda} = \Oh_d(1)$ sets
$\bfg + \boldsymbol{\ell}_0+q \Lambda$, where
$\boldsymbol \ell_0+q \Lambda$ is a coset of $q \Lambda$ in $\Lambda$.
We iterate over all these cosets $\boldsymbol \ell_0+q \Lambda$, and for each of them, 
use  \cref{lem:lattice-optimization}  
 to find a point
$\bfg_0\in (K+B_0)\cap (\bfg+\boldsymbol{\ell}_0+q \Lambda)$ (if exists). 
Then, by applying \cref{lem:balanced-center} to $K^+,q \Lambda$, and $\bfg_0$, we obtain a
center $\bfy\in K^+$ such that $\bfy \in \bfg_0+\Lambda\subseteq \bfg+\Lambda$, and
$\bfy$ is $\gamma$-balanced in $K^+$ modulo $q \Lambda$, and 
\begin{equation}
   \label{eqn:tempk}
M_{K^+}^\gamma(\bfy)\supseteq K^+\cap (\bfg_0+ q \Lambda) = 
K^+\cap (\bfg+\boldsymbol \ell_0 + q \Lambda).
\end{equation}
Return $Y_\bfg$ as the set of these $\Oh_d(1)$ found centers $\bfy$.
The time bound follows from \cref{lem:lattice-optimization,lem:balanced-center,lem:constrained-optimization,cor:polytope-oracles}.
It remains to prove \eqref{eq:cap-cover-candidates} and $B_0\subseteq \alpha M_\bfy$.

In each iteration, since $\bfg_0\in K+B_0 \subseteq K^+$,
\eqref{eqn:boundthick} gives
\[\bfg_0+B_0\subseteq M_{K^+}^1(\bfg_0).\]
By \eqref{eqn:tempk}, 
the center $\bfy\in Y_\bfg$ chosen in this iteration satisfies $\bfg_0\in M_{K^+}^\gamma(\bfy)$. Applying 
\cref{lem:center-containment-general} to $\bfg_0$ and $\bfy$ then gives \[
M_{K^+}^1(\bfg_0)\subseteq M_{K^+}^\alpha(\bfy).\]
Chaining the two displayed inequalities above, taking differences, and dividing by two, we get $B_0\subseteq \alpha M_\bfy$ as claimed.

To prove \eqref{eq:cap-cover-candidates}, take $\bfx\in K\cap (\bfg+V+\Lambda)$ and pick $\bfg'\in (\bfg+\Lambda)\cap (\bfx+V)$.
Since $\bfg'\in K+B_0$,
\eqref{eqn:boundthick} gives $\bfg' + B_0 \subseteq M_{K^+}^1(\bfg')$.
In the iteration for the coset $(\bfg'-\bfg)+q \Lambda$,
the chosen center $\bfy$ satisfies $\bfg' \in M_{K^+}^\gamma(\bfy)$. Thus, the same argument as the previous paragraph gives
\[ \bfx\in \bfg'+B_0\subseteq
M_{K^+}^1(\bfg') \subseteq  M_{K^+}^\alpha(\bfy),
\]
proving 
\eqref{eq:cap-cover-candidates}.
\end{proof}

We use \cref{lem:geometric-boundary-representatives} to compute the shallow centers $\caS$.
Recall $r\coloneqq \rho \vol_d(\pi_\Lambda K)/\vol_d(\pi_\Lambda D)$, and $\lvert\caS\rvert=\Oh_d(r)$ by \cref{lem:geometric-boundary-representatives}.
We then compute $Y_\bfg$ 
 for all $\bfg\in \caS$ using \cref{lem:cap-cover-balanced-centers} in
$r\cdot b_0^{\Oh_d(1)} \Tquery(K)$ total time.

\paragraph*{A maximal packing.}
After computing the balanced centers $\bigcup_{\bfg\in \caS} Y_{\bfg}$, the following lemma selects an inclusion-maximal family $Y$ whose $1/2$-scaled Macbeath regions are disjoint modulo $2q \Lambda$: 
\begin{lemma}
\label{lem:cap-cover-packing}
One can compute an inclusion-maximal
$Y\subseteq \bigcup_{\bfg\in \caS} Y_\bfg$ such that 
$\pi_{2q \Lambda} M_{K^+}^{1/2}(\bfy)$, $\bfy\in Y$, are pairwise disjoint, in
time $r\cdot b_0^{\Oh_d(1)} \Tquery(K)$.
\end{lemma}

Before describing the algorithm for \cref{lem:cap-cover-packing}, we first state its consequences. The first is the covering property:
By \cref{cor:quotient-macbeath-cover,lem:cap-cover-balanced-centers}, for any inclusion-maximal
$Y\subseteq \bigcup_{\bfg\in \caS} Y_\bfg$
such that $\pi_{2q \Lambda} M_{K^+}^{1/2}(\bfy)$, $\bfy\in Y$,
are pairwise disjoint, it holds that 
\begin{equation}
   \label{eqn:disjoint-conseq}
 K\cap (\caS+V+\Lambda)\subseteq
 \bigcup_{\bfy\in Y} M_{K^+}^{13 \gamma \alpha}(\bfy).
\end{equation}
The second is the volume bound due to the packing property:
\begin{lemma}
\label{lem:shallow-macbeath-packing}
Let $Y\subseteq \bigcup_{\bfg\in \caS} Y_\bfg$ be such that 
$\pi_{2q \Lambda} M_{K^+}^{1/2}(\bfy)$, $\bfy\in Y$, are pairwise disjoint.
Then,
\[
 \sum_{\bfy\in Y} \vol_d(\pi_\Lambda M_\bfy)
 =\Oh_d\bigl(\rho \vol_d(\pi_\Lambda K)\bigr).
\]
In particular, $\lvert Y\rvert \le \Oh_d(r)$.
\end{lemma}

\begin{proof}
Let $\bfy\in Y_\bfg$. The equality $\bfy+\Lambda=\bfg+\Lambda$ and
the definition of $\caS$ give
$(\bfy+\Lambda)\cap (K\ominus 2B_0)=\varnothing$.
Since $K^+\subseteq K+2 \beta B_0$, and $(K+2 \beta B_0)\ominus (2 \beta B_0) = K$
(see e.g., \cite[Lemma~3.1.11]{schneider2014convex}), we have 
$K^+\ominus (2 \beta+2) B_0\subseteq K\ominus 2B_0$.
Together with $q \Lambda\subseteq \Lambda$, this then implies
\[
 (\bfy+q \Lambda)\cap \bigl(K^+\ominus (2 \beta+2) B_0\bigr)
 =\varnothing.
\]
Then, applying \cref{lem:shallow-macbeath-region} to $K^+$, $(2 \beta+2) B_0$, and $2q \Lambda$ gives
   \begin{equation}
      \label{eqn:containmenttemp}
 \pi_{2q \Lambda} M_{K^+}^{1}(\bfy)
 \subseteq \pi_{2q \Lambda} K^+
 \setminus \pi_{2q \Lambda}\bigl(K^+\ominus (4 \beta+4) B_0\bigr).
   \end{equation}

Then, we have
\begin{align*}
 &\sum_{\bfy\in Y} \vol_d(\pi_\Lambda M_\bfy)\\
 \le \ & \sum_{\bfy\in Y} \vol_d(\pi_{2q \Lambda} M_\bfy)
 \hspace{10pt} && \text{(since $2q \Lambda\subseteq \Lambda$)}\\
 \le \ & 2^d \sum_{\bfy\in Y} \vol_d\bigl(\pi_{2q \Lambda} M_{K^+}^{1/2}(\bfy)\bigr)
 \hspace{10pt} && \text{(by \cref{lem:volumeineq})}\\
 \le \ & 2^d \vol_d\!\left(\pi_{2q \Lambda} K^+
       \setminus \pi_{2q \Lambda}\bigl(K^+\ominus (4 \beta+4) B_0\bigr)\right)
 \hspace{10pt} && \text{(by \eqref{eqn:containmenttemp} and disjointness)}\\
 = \ & \Oh_d\bigl(\rho \vol_d(\pi_{2q \Lambda} K^+)\bigr)
 \hspace{10pt} && \text{(by \cref{lem:quotient-shell-bound},
$\bfc+\theta^{-1} B_0\subseteq K^+$, and $\theta=\Oh_d(\rho)$)}\\
 = \ & \Oh_d\bigl(\rho \vol_d(\pi_{\Lambda} K^+)\bigr)
 \hspace{10pt} && \text{(since $[\Lambda \colon 2q \Lambda]= \Oh_d(1)$)}.
\end{align*}
By \cref{lem:volumeineq} and  $K^+-\bfc = (1+2 \theta)(K-\bfc)$,
\[
 \vol_d(\pi_{\Lambda} K^+)
 \le (1+2 \theta)^d \vol_d(\pi_\Lambda K)
 =\Oh_d\bigl(\vol_d(\pi_\Lambda K)\bigr),
\]
finishing the proof of the volume bound.

Combined with 
the property that $B_0 \subseteq \alpha M_{\bfy}$ for all $\bfy \in Y$ (see \cref{lem:cap-cover-balanced-centers}), the volume bound implies
\[\lvert Y\rvert \le \Oh_d(\frac{\rho \vol_d(\pi_{\Lambda} K)}{\vol_d(\pi_\Lambda (B_0/\alpha))}) = \Oh_d(\frac{\rho \vol_d(\pi_{\Lambda} K)}{\vol_d(\pi_\Lambda D)}) = \Oh_d(r)\]
as claimed.
\end{proof}

We now prove \cref{lem:cap-cover-packing} by a greedy algorithm. The volume bound from \cref{lem:shallow-macbeath-packing} will bound the work of this algorithm.
\begin{proof}[Proof of \cref{lem:cap-cover-packing}]
Initialize $Y=\varnothing$. Initialize all elements in 
$\bigcup_{\bfg\in \caS} Y_\bfg$ as candidate centers. Repeat the following steps until no candidate remains:
\begin{enumerate}[beginpenalty=10000]
\item Choose a remaining candidate center $\bfy$, add it to $Y$, and remove it from the remaining candidates.
\item Use \cref{lem:quotient-cell-search} to enumerate
$\bfg\in \caP_{M_{K^+}^{13 \gamma \alpha}(\bfy)}
= \{\bfg\in \caG:(\bfg+V)\cap (M_{K^+}^{13 \gamma \alpha}(\bfy)+\Lambda)\ne \varnothing\}$
(defined in \eqref{eqn:occupiedp}).
\item For each enumerated $\bfg$ that also belongs to $\caS$, and each remaining candidate $\bfy'\in Y_\bfg$,
test whether
\[
 \bigl(M_{K^+}^{1/2}(\bfy)-M_{K^+}^{1/2}(\bfy')\bigr)
 \cap 2q \Lambda\ne \varnothing.
\]
If so, remove $\bfy'$ from the remaining candidates.
\end{enumerate}
Finally, return $Y$.

\paragraph*{Correctness.}
The test in step~3  
decides whether $\pi_{2q \Lambda} M_{K^+}^{1/2}(\bfy)$ intersects
$\pi_{2q \Lambda} M_{K^+}^{1/2}(\bfy')$.
To show that our greedy algorithm indeed constructs an inclusion-maximal disjoint collection, it suffices to show that all 
remaining candidates $\bfy'\in \bigcup_{\bfg\in \caS} Y_{\bfg}$ with
$\pi_{2q \Lambda} M_{K^+}^{1/2}(\bfy')$ intersecting 
$\pi_{2q \Lambda} M_{K^+}^{1/2}(\bfy)$ are included in the enumeration of step~2 and step~3.

Fix any such candidate $\bfy'$.
Since $\bfy$ is $\gamma$-balanced modulo $q \Lambda$,
\cref{lem:quotient-expansion-containment} gives
\[
\bfy' \in  M_{K^+}^{\alpha}(\bfy')\subseteq M_{K^+}^{13 \gamma \alpha}(\bfy).
\]
Let $\bfg \in \caS$ be such that $\bfy'\in Y_\bfg\subseteq \bfg +\Lambda$ (see \cref{lem:cap-cover-balanced-centers}). Hence 
$\bfg\in \caP_{M_{K^+}^{13 \gamma \alpha}(\bfy)}$.
Thus $\bfg$ is enumerated in step~2, and $\bfy'$ is enumerated in step~3.

\paragraph*{Running time.}
For a center $\bfy$ chosen in step~1, write
$N_\bfy=\lvert\caP_{M_{K^+}^{13 \gamma \alpha}(\bfy)}\rvert$ for the number of enumerated cell centers $\bfg$.
 Since $\lvert Y_{\bfg}\rvert = \Oh_d(1)$ (see \cref{lem:cap-cover-balanced-centers}), there
are $\Oh_d(N_\bfy)$ intersection tests in this iteration.
For every enumerated cell $\bfg +V$, 
\[
\pi_{\Lambda}(\bfg +V) \subseteq \pi_\Lambda(M_{K^+}^{13 \gamma \alpha} (\bfy)+2V)
 \subseteq \pi_\Lambda(\bfy+14 \gamma \alpha M_\bfy),
\]
where the last step follows from $2V\subseteq B_0\subseteq \alpha M_\bfy\subseteq \gamma \alpha M_\bfy$
(see \cref{lem:cap-cover-balanced-centers}).
By disjointness of the interiors of $\pi_{\Lambda}(\bfg+V),\bfg\in \caG$, we get 
\[
 N_\bfy \vol_d(\pi_{\Lambda} V)
 \le \vol_d\bigl(\pi_\Lambda(14 \gamma \alpha M_\bfy)\bigr)
  = \Oh_d(\vol_d(\pi_\Lambda M_\bfy)),
\]
where the last step used \cref{lem:volumeineq} and $\gamma,\alpha = \Oh_d(1)$.
Thus, due to disjointness of our constructed answers, the volume bound in
\Cref{lem:shallow-macbeath-packing} gives that the total number of intersection tests performed is
\[
 \sum_{\bfy\in Y} N_\bfy
 =\Oh_d\!\left(
 \frac{\rho \vol_d(\pi_\Lambda K)}{\vol_d(\pi_\Lambda V)}\right)
 =\Oh_d\!\left(
 \frac{\rho \vol_d(\pi_\Lambda K)}{\vol_d(\pi_\Lambda D)}\right)
 =\Oh_d(r).
\]
Also recall from 
\cref{lem:shallow-macbeath-packing} that $\lvert Y\rvert = \Oh_d(r)$.

Thus, to prove the claimed running time bound, it remains to show that each intersection test in step~3 can be implemented in
$b_0^{\Oh_d(1)} \Tquery(K)$ time.
To see this, note that \cref{cor:polytope-oracles} gives linear optimization
oracles for $M_{K^+}^{1/2}(\bfy)$ and $M_{K^+}^{1/2}(\bfy')$, which immediately give a linear optimization oracle for their difference set. Testing its intersection with
the lattice $2q \Lambda$ is done using \cref{lem:lattice-optimization}.  
\end{proof}

\subsection{Completing the proof}

We now complete the proof of \cref{thm:geometric-cap-cover}.

\paragraph*{Constructing caps.}
Let $Y$ be the set of Macbeath centers constructed by \cref{lem:cap-cover-packing}.
We first move the centers from $K^+$ into $K$ by shrinking with respect to $\bfc$:
for each $\bfy\in Y$, define
$\bfz_\bfy=\bfc+(1+2 \theta)^{-1}(\bfy-\bfc)\in K$; we have $M_K^1(\bfz_\bfy)-\bfz_\bfy=(1+2 \theta)^{-1} M_\bfy$.
Then, apply \cref{lem:macbeath-cap} to $K^+,\bfy$ to obtain a direction $\bfu_\bfy\ne \zero$ giving the following cap approximation:  for every $t\ge 1$,
\[
 K^+\cap (\bfy+tM_\bfy)
 \subseteq K^+\cap \{\bfx:\langle\bfu_\bfy,\bfx-\bfy\rangle\ge -t\}
 \subseteq \bfy+(4d+3) tM_\bfy,
\]
which, by shrinking with respect to $\bfc$ and substituting $t=m(1+2 \theta)$, gives
\begin{align*}
 K\cap \left(\bfz_\bfy+mM_\bfy\right)
 \subseteq Q_{\bfy} \coloneqq K\cap \left\{\bfx:
 \langle\bfu_\bfy,\bfx-\bfz_\bfy\rangle\ge -m\right\} \subseteq \bfz_\bfy+(4d+3) mM_\bfy.
\end{align*}
Set $m\coloneqq 13 \gamma \alpha+2 \beta \alpha$ above, and return the caps
$ Q_\bfy$ for all $\bfy\in Y$.
      We now prove that these caps satisfy the requirements.

We first prove \cref{item:cap-cover-thickness}.
Plugging in $t=1$ in the first displayed containment from the previous paragraph gives $K^+\cap (\bfy + M_{\bfy}) = (\bfy + M_{\bfy}) \subseteq K^+ \cap  \{\bfx:\langle\bfu_\bfy,\bfx-\bfy\rangle\ge -1\}$; shrinking with respect to $\bfc$ and using $(1+2 \theta)^{-1}\le m$ yield $\bfz_{\bfy} + (1+2 \theta)^{-1} M_{\bfy} \subseteq  Q_{\bfy}$.
Since $D\subseteq B_0\subseteq \alpha M_\bfy$ (by \cref{lem:cap-cover-balanced-centers}), this gives
$\bfz_\bfy+(\alpha(1+2 \theta))^{-1} D
 \subseteq Q_\bfy$, proving 
 \cref{item:cap-cover-thickness}.

      We now prove the covering property \cref{item:cap-cover-deep}.
      Recall from \eqref{eqn:disjoint-conseq} that $K\cap (\caS+V+\Lambda)\subseteq \bigcup_{\bfy \in Y} M_{K^+}^{13 \gamma \alpha}(\bfy)$.
By definition of $\bfz_{\bfy}$, we have $\bfy-\bfz_\bfy\in 2 \theta(K-\bfc)\subseteq 2 \beta B_0
\subseteq 2 \beta \alpha M_\bfy$, where the last two steps follow from \eqref{eqn:defnb} and \cref{lem:cap-cover-balanced-centers}.
Hence
\[
 M_{K^+}^{13 \gamma \alpha}(\bfy)
 \subseteq \bfz_\bfy+(13 \gamma \alpha+2 \beta \alpha) M_\bfy
 =\bfz_\bfy+mM_\bfy.
\]
Therefore, combining with \eqref{eqn:disjoint-conseq}
yields 
\[K\cap (\caS +V+\Lambda) \subseteq K \cap \bigcup_{\bfy \in Y}(\bfz_{\bfy} + mM_{\bfy}) \subseteq  \bigcup_{\bfy \in Y} Q_{\bfy}.\]
Thus, 
\[ K \setminus \bigcup_{\bfy \in Y} Q_{\bfy}\subseteq K\setminus (\caS + V+\Lambda) \subseteq K\cap \big(\bigcup_{\bfg\in \caP_{K}\setminus \caS} \bfg+V+\Lambda\big ), \]
where the last containment follows from the definition of $\caP_{K}$.
For $\bfg\in \caP_K\setminus \caS$, the definition of $\caS$ and $V\subseteq B_0$
give $\bfg+V\subseteq (K\ominus B_0)+\Lambda$.
Since $D\subseteq B_0$, this shows
\begin{equation*}
 K\setminus \bigcup_{\bfy\in Y} Q_\bfy
 \subseteq (K\ominus D)+\Lambda,
\end{equation*}
proving  \cref{item:cap-cover-deep}.

We now construct the enclosing parallelepipeds $P_{\bfy}\supseteq Q_{\bfy}$ and prove
\cref{item:cap-cover-volume}.
We apply \cref{lem:macbeath-center} to
each $M_\bfy$, and obtain  
 an origin-symmetric parallelepiped
$P'_\bfy$ with $P'_\bfy\subseteq M_\bfy\subseteq \beta P'_\bfy$.
We return
\[
 P_\bfy=\bfz_\bfy+(4d+3) m \beta P'_\bfy\supseteq Q_\bfy.
\]
Since $P_\bfy-\bfz_\bfy\subseteq (4d+3) m \beta M_\bfy$,
\cref{lem:volumeineq} gives
$\vol_d(\pi_{\Lambda} P_{\bfy})\le \Oh_{d}(\vol_d(\pi_{\Lambda} M_{\bfy}))$, and then
\cref{lem:shallow-macbeath-packing} gives \cref{item:cap-cover-volume}:
\[
\sum_{\bfy\in Y} \vol_d(\pi_\Lambda P_\bfy)=\Oh_d(\rho \vol_d(\pi_\Lambda K)).
\]

\paragraph*{Constructing $R$.}
Start with $R=\varnothing$.
For each returned parallelepiped $P_\bfy$, use \cref{lem:quotient-cell-search}
to enumerate $\bfg\in \caP_{P_\bfy+B_0}$.
 For each enumerated $\bfg$,
use \cref{lem:lattice-translate-depth} to find the largest $t\in \Nonneg$
with a $\boldsymbol{\ell}\in \Lambda$ such that
$\bfg+\boldsymbol{\ell}+tB_0\subseteq K$ (or report that no such $t$ exists).
If it returns $t\ge 2$ and a corresponding $\boldsymbol \ell$, add $\bfr=\bfg+\boldsymbol{\ell}$ to $R$.

We first show the time complexity of the enumeration is within the claimed time bound
$r\cdot b_0^{\Oh_d(1)} \Tquery(K)$. It suffices to show that the number of enumerated points $\bfg$ is $\Oh_d(r)$. 
Note that each enumerated $\bfg$ satisfies
$\bfg+V\subseteq P_\bfy+2B_0+\Lambda$ due to $2V\subseteq B_0$.
The containments $B_0\subseteq \alpha M_\bfy\subseteq \alpha \beta P'_\bfy$
show that $P_\bfy+2B_0\subseteq \bfz_{\bfy} + \Oh_d(1)\cdot (P_{\bfy} - \bfz_{\bfy})$.
 Thus, due to disjointness of the interiors of $\pi_{\Lambda}(\bfg + V)$, the total number of 
enumerated points $\bfg$ can be bounded by the volume ratio $\Oh_d(\sum_\bfy \vol_d(\pi_\Lambda P_\bfy)/\vol_d(\pi_\Lambda V))=\Oh_d(r)$ as desired.

By the algorithm, 
every $\bfr\in R$ satisfies $\bfr + 2B_0\subseteq K$;
 since $4d \beta D\subseteq 4d \beta^2 B=2B_0$, 
this proves \cref{item:cap-cover-centers}.
To prove \cref{item:cap-cover-local}, let
$\bfx\in \bigl(K\setminus \bigcup_{\bfy\in Y} Q_\bfy\bigr)
\cap \bigl(\bigcup_{\bfy\in Y} P_\bfy+D+\Lambda\bigr)$.
Since $\bfx \in \bigcup_{\bfy \in Y}(P_{\bfy}+D+\Lambda)$, there is an enumerated $\bfg$ such that $\bfx\in \bfg+V+\Lambda$.
Then,  
since $\bigcup_{\bfy \in Y} Q_\bfy\supseteq K\cap (\caS+V+\Lambda)$, we must have
$\bfg\in \caP_K\setminus \caS$.
Hence, by the definition of $\caS$,
the test at $\bfg$ returns $t\ge 2$, adding a 
 center $\bfr\in \bfg+\Lambda \subseteq \bfx + V +\Lambda \subseteq \bfx + D +\Lambda$. This proves 
\cref{item:cap-cover-local}.

This completes the proof of \cref{thm:geometric-cap-cover}.

\section{Cap Covering for Subset Sums}
\label{sec:subset-sum-cap-covering}

In this section, we apply the quotient cap covering
(\cref{thm:geometric-cap-cover})
 to the zonotope of vectors  representing input integers in a proper progression $f(H\cap \Gamma)$.
We assume that the input multiset contains a small generating submultiset, which will be supplied by \cref{lem:main-classify} in our applications. Then, all sufficiently deep lattice points in the zonotope are attainable subset sums due to 
\cref{lem:innerpoints}. 
We compute the deep subset sums, and also construct caps that cover the remaining shallow subset sums.

The main theorem is given below. Here,
we further allow restricting the zonotope to an interval $[\ell,u]$ of $f$-values.
The theorem returns the caps $Q_i$ together with their enclosing parallelepipeds $P_i$, and 
an AP decomposition that compactly represents
the deep subset sums. 

\begin{theorem}
\label{thm:subset-sum-cap-covering}
Let two multisets $X'\subseteq X$ with $\supp(X)\subseteq  \Z^d$ and $\spa(X) =\R^d$,
a linear map $f\colon \R^d\to \R$ with integer coefficients,
and integers $\ell\le u$ be given.
Write $\Gamma=\caL(X)$, $H=\conv(\pm \supp X)$, $D=\lvert X'\rvert H$, $\Lambda=\Gamma\cap \ker f$, and $K=\caZ(X)\cap f^{-1}([\ell,u])$.

Suppose $f(H\cap \Gamma)$ is a proper progression, $X'$ is $d$-generating for $H\cap \Gamma$, and $K$ contains translates of $H$ and $\rho^{-1} D$ for some (unknown) $\rho>0$.
Then, there is an algorithm that constructs:
\begin{enumerate}[label=\textup{(\roman*)},leftmargin=2.5em,beginpenalty=10000]
\item An integer set $G$, represented as an AP decomposition of size $\Oh_d(\rho(u-\ell))$, such that 
\begin{equation}
\label{eq:cap-cover-deep-values}
 G\subseteq f((K\ominus D)\cap \Gamma)
 \subseteq f(\Sigma(X))\cap [\ell,u],
\end{equation}
where the second containment follows from \cref{lem:innerpoints}.
\item Caps
$Q_i=K\cap \{\bfx:\langle\bfu_i,\bfx\rangle\ge a_i\}$ 
with enclosing parallelepipeds $P_i\supseteq Q_i$ such that
\begin{gather}
\label{eq:cap-cover-shallow-points}
 \{\bfx\in K\cap \Gamma:f(\bfx)\notin G\}
 \subseteq \bigcup_{i} Q_i,\\*
\label{eq:subset-sum-cap-volume}
 \sum_{i}
 \frac{\vol_d(\pi_\Lambda P_i)}{\det \Gamma}
 \le \Oh_d(\rho(u-\ell)),
\end{gather}
and every $Q_i$ contains a translate of
$c_d H$, where $0<c_d\le 1$ depends only on $d$.
\end{enumerate}
Writing $b_0=\maxbitlen(X,X',f,\ell,u)$,
the construction algorithm takes time
\begin{equation*}
 b_0^{\Oh_d(1)} \rho(u-\ell),
\end{equation*}
 and $\maxbitlen(\bfu_i,a_i,P_i)=\Oh_d(b_0)$ for all $i$.
\end{theorem}

The following lemma bounds lattice point count by volume.
\begin{lemma}
\label{lem:quotient-lattice-point-bound}
Let $\Lambda\subseteq \Gamma\subseteq \R^d$ be two lattices and $K\subseteq \R^d$ be an origin-symmetric convex body with
$\spa(K\cap \Gamma)=\R^d$. Then
$\lvert\pi_\Lambda(K\cap \Gamma)\rvert\det \Gamma
 \le \Oh_d\bigl(\vol_d(\pi_\Lambda K)\bigr)$.
\end{lemma}

\begin{proof}
Choose linearly independent $\bfu_1,\ldots,\bfu_d\in K\cap \Gamma$ and
let $\Gamma_0=\caL(\bfu_1,\dots,\bfu_d)$.
Then $P_0=\sum_{i=1}^d[0,1)\cdot \bfu_i\subseteq dK$ is a
fundamental parallelepiped of $\Gamma_0$.
Since $\Gamma_0\subseteq \Gamma$,
$P_0$ contains a 
measurable fundamental domain $P$ of $\Gamma$.
Pick $\lvert \pi_{\Lambda}(K\cap \Gamma)\rvert$ points $\bfx_i\in K\cap \Gamma$, one per class modulo $\Lambda$.
Then $\pi_\Lambda(\bfx_i+P)$ are pairwise disjoint subsets of $\pi_\Lambda((d+1) K)$, each of volume
$\det \Gamma$.
Counting volumes gives 
\begin{align*}
 \lvert\pi_\Lambda(K\cap \Gamma)\rvert\det \Gamma
 &\le \vol_d(\pi_\Lambda((d+1) K))\\
 &\le (d+1)^d \vol_d(\pi_\Lambda K),
\end{align*}
where the last step used \cref{lem:volumeineq}.
\end{proof}

\subsection{Proof of \texorpdfstring{\Cref*{thm:subset-sum-cap-covering}}{Theorem \getrefnumber{thm:subset-sum-cap-covering}}}
We now proceed to the proof of \cref{thm:subset-sum-cap-covering}.

By \cref{lem:zonotope-algorithms,lem:constrained-optimization},
$K = \caZ(X) \cap f^{-1}([\ell,u])$ has a linear optimization oracle with
$\Tquery(K)\le \lvert\supp X\rvert b_0^{\Oh_d(1)}$.

Given the input $X$ and $f$,  first compute $\Gamma = \caL(X)$ and $\Lambda = \Gamma \cap \ker f$ by
\cref{lem:lattice-basis-hnf,lem:lattice-kernel-image}.
Recall $K =  \caZ(X) \cap f^{-1}([\ell,u])$.
Write $w=\fw(K)\le u-\ell$.
Since $f(\Gamma)=g\Z$ for some integer $g\ge1$, a slab of
$f$-width $w$ has quotient volume $(w/g)\det\Gamma\le w\det\Gamma$, so
\begin{equation}
\vol_d(\pi_\Lambda K)\le w\det\Gamma \le (u-\ell)\det \Gamma.
\label{eqn:volpik}
\end{equation}

Write $\beta=\lceil\beta_d\rceil$, where $\beta_d$ is the constant in
\cref{lem:macbeath-center}.
We first use the following lemma to approximate the unknown $\rho$ for which $K$ contains a translate of $\rho^{-1} D$, and to compute an approximating parallelepiped. \begin{lemma}
\label{lem:intrinsic-scale}
In the setup of \cref{thm:subset-sum-cap-covering}, 
in $b_0^{\Oh_d(1)}\lvert\supp X\rvert$ time,
one can compute 
$\rho'>0$ and an origin-symmetric parallelepiped $B$ with $\maxbitlen(\rho',B)=\Oh_d(b_0)$,
such that
$\rho'\le \beta \rho$, and
writing $D'=\rho'(K-K)/2$,
\[
 D\subseteq D'\subseteq B\subseteq \beta D'.
\]
\end{lemma}

\begin{proof}
Write $A=(K-K)/2$, which has a linear optimization oracle based on the oracle for $K$. 
Apply \cref{lem:macbeath-center} to 
$A$ to obtain an origin-symmetric parallelepiped $P$ with
$P\subseteq A\subseteq \beta P$.

Since $D = \lvert X'\rvert\conv(\pm \supp X)$, we can compute the smallest $\rho'>0$ such that $D\subseteq \rho' P$, by scanning the points in $X$ and comparing against the defining inequalities of the parallelepiped $P$.
Define $B=\beta \rho' P$. Then
\[
 D\subseteq \rho' P\subseteq \rho' A\subseteq B\subseteq \beta \rho' A.
\]
Since $K$ contains a translate of $\rho^{-1} D$, 
we have $D\subseteq \rho A\subseteq \beta \rho P$, so $\rho'\le \beta \rho$ by the minimality of $\rho'$.
\end{proof}

Let $\rho',B$ be computed from \cref{lem:intrinsic-scale},
and write $D'=\rho'(K-K)/2$. 

We first quickly handle the corner case with $\rho'\ge 1$. 
In this case, it suffices to return $G=\varnothing$ with the trivial cap $Q_1=K$, and use \cref{lem:macbeath-center} to compute the required approximating parallelepiped: 
$\bfc+\beta^{-1} P_1\subseteq K\subseteq \bfc+ P_1$, which clearly satisfies 
$\vol_d(\pi_\Lambda P_1)=\Oh_d(\vol_d(\pi_\Lambda K))$.
The hypothesis on $K$ gives the required thickness of $Q_1$.
Since $1\le \rho'\le \beta \rho$, 
and $\vol_d(\pi_\Lambda K)\le (u-\ell)\det\Gamma$, the required volume bound in \eqref{eq:subset-sum-cap-volume} is satisfied as well.

Assume henceforth that $\rho'<1$, so that \cref{thm:geometric-cap-cover} becomes applicable with $\rho'$.

\paragraph*{Cap covering.}

Apply \cref{thm:geometric-cap-cover} to $\Lambda = \Gamma \cap \ker f, K,\rho'$ and $D'$ to obtain the caps $Q_i$ with enclosing parallelepipeds $P_i$, and the point set $R$.
The theorem gives
\[
 R+4d \beta D'\subseteq K\]
 and
 \[
 \sum_i \vol_d(\pi_\Lambda P_i)
 =\Oh_d(\rho' \vol_d(\pi_\Lambda K)).
\]
Thus  the required volume bound \eqref{eq:subset-sum-cap-volume}
holds due to $\rho'\le \beta \rho$ and \eqref{eqn:volpik}.
Each $Q_i$ contains a translate of $c_d D'\supseteq c_d D \supseteq c_d H$, as required.

\paragraph*{Constructing $G$.}
 We now compute the deep subset sums $G$. Intuitively, we do so by reporting line segments (along a fixed direction $\bfx_0$) contained in a deep part of $K$ with both end points near the boundary, in a modulo-$\Lambda$ sense.  The actual implementation of this intuition directly works with the $f$-images of the points for convenience.

We enumerate 
\[ E\coloneqq \bigcup_i f(P_i\cap \Gamma) \subseteq \Z\]
using \cref{lem:quotient-lattice-enumeration}.
By applying $f$ to both sides of \cref{thm:geometric-cap-cover}~\ref{item:cap-cover-deep} intersected with $\Gamma$, we get 
\[ f(K\cap \Gamma)\setminus E \subseteq f((K\ominus D')\cap \Gamma).\]

We first describe how to compute 
$f(K\cap \Gamma)\setminus E$ as an AP decomposition.
We choose $\bfx_0\in \pm \supp X$ with $p\coloneqq f(\bfx_0)>0$.
In a residue class $r+p\Z$, we say a \emph{gap} is a
maximal nonempty arithmetic progression of common difference $p$ in the set $(r+p\Z)\setminus E$; gaps may be unbounded.

\begin{claim}
    \label{claimtemp1}
A gap is either contained in $f(K\cap \Gamma)$ or disjoint from
$f(K\cap \Gamma)$.
\end{claim}
\begin{proof}
Take any gap, and suppose it contains some element in $f(K\cap \Gamma)$; let it be $t\in f(K\cap \Gamma)\setminus E$.  Then $t = f(\bfx)$ for some $\bfx \in (K\cap \Gamma) \setminus \bigcup_i P_i \subseteq (K\cap \Gamma)\setminus \bigcup_i Q_i$. By \cref{thm:geometric-cap-cover}~\ref{item:cap-cover-deep},
$\bfx+\Lambda$ intersects $(K\ominus D')\cap \Gamma$.
Since $\bfx_0\in \conv(\pm \supp X)\subseteq D\subseteq D'$, $\bfx+\bfx_0+\Lambda$ and $\bfx-\bfx_0+\Lambda$ both intersect $K\cap \Gamma$, so $t-p,t+p\in f(K\cap \Gamma)$. Induction then shows that all elements of the gap are contained in $f(K\cap \Gamma)$, proving the claim.
\end{proof} 
In particular, every unbounded gap is disjoint from $f(K\cap \Gamma)$, since $f(K\cap \Gamma)$ is bounded.

Let $B$ be the parallelepiped computed from \cref{lem:intrinsic-scale}.
Enumerate
\[
 S\coloneqq \bigcup_{\bfa\in R} f((\bfa+B)\cap \Gamma)
\]
using \cref{lem:quotient-lattice-enumeration}.
Since $B\subseteq \beta D'$ and $R+4d \beta D'\subseteq K$, we have
\[
 S\subseteq f((K\ominus D')\cap \Gamma).
\]
\begin{claim}
    \label{claimtemp2}
    A gap $e,e+p,\ldots,e'$
belongs to $f(K\cap \Gamma)$ if and only if $e\in S$.
\end{claim}
\begin{proof}
If the gap belongs to $f(K\cap \Gamma)$, choose $\bfx\in K\cap \Gamma$ with
$f(\bfx)=e$ and $\bfz\in P_i\cap \Gamma$ with $f(\bfz)=e-p$.
Then $\bfx-\bfz-\bfx_0\in \Lambda$, so $\bfx\in P_i+D'+\Lambda$.
Since $e\notin E$, $\bfx$ lies in no $Q_i$.
The local covering in \cref{item:cap-cover-local} of \cref{thm:geometric-cap-cover} gives
$\bfx\in R+D'+\Lambda\subseteq R+B+\Lambda$, hence $e\in S$.
Conversely, $e\in S\subseteq f(K\cap \Gamma)$ implies that the whole gap belongs to $f(K\cap \Gamma)$.
\end{proof}

Therefore, by \cref{claimtemp1,claimtemp2}, we can compute an AP decomposition of $f(K\cap \Gamma)\setminus E$ with common difference $p$ as follows: partition $E$ into residue classes modulo $p$.
In each residue class, include every bounded gap whose first value belongs to $S$. 

\begin{claim}
    \label{claimtemp3}
If $\bfx\in K\cap \Gamma$ and $f(\bfx)   \notin f(K\cap \Gamma)\setminus E$, then either $f(\bfx)\in S$ or $\bfx \in \bigcup_i Q_i$.
\end{claim}
\begin{proof}
By the hypothesis, we must have $f(\bfx)\in E$.
Then the definitions of $E$ and $\Lambda$ give $\bfz\in P_i\cap \Gamma$ with $\bfx-\bfz\in \Lambda$.
If $\bfx$ lies in no $Q_i$, then \cref{thm:geometric-cap-cover}~\ref{item:cap-cover-local}
gives $\bfx\in R+D'+\Lambda\subseteq R+B+\Lambda$, so $f(\bfx)\in S$ by definition of $S$. 
\end{proof}

Thus, we can return the answer set $G$ (represented as an AP decomposition) defined by
\[
 G=(f(K\cap \Gamma)\setminus E)\cup S
 \subseteq f((K\ominus D')\cap \Gamma)
 \subseteq f((K\ominus D)\cap \Gamma),
\]
proving the first containment in \eqref{eq:cap-cover-deep-values} (the second containment in \eqref{eq:cap-cover-deep-values} follows from \cref{lem:innerpoints}).

To prove \eqref{eq:cap-cover-shallow-points}, note that for every $\bfx \in K\cap \Gamma$ with $f(\bfx)\notin G$, 
\cref{claimtemp3} immediately shows that $\bfx$ belongs to some cap $ Q_i$ as desired.

The size of the resulting AP decomposition of $G$ is $\Oh(|E|+|S|)$, and the construction time is  near-linear in $|E|$ and $|S|$, using the basic
operations in \cref{lem:ap-operations}.

\paragraph*{Time and output size.}
Write $r=\rho' \vol_d(\pi_\Lambda K)/\vol_d(\pi_\Lambda D')$ as in
\cref{thm:geometric-cap-cover}.
Applying \cref{thm:geometric-cap-cover} takes
$r\lvert\supp X\rvert b_0^{\Oh_d(1)}$ time.
Computing $\Gamma,\Lambda,\rho',B$ takes
$\lvert\supp X\rvert b_0^{\Oh_d(1)}$ time by \cref{lem:lattice-basis-hnf,lem:lattice-kernel-image} and \cref{lem:intrinsic-scale}, which is not the dominating term.

Since $\supp X\subseteq D\cap \Gamma$, we have 
$\lvert\supp X\rvert\le \lvert f(D\cap \Gamma)\rvert$ by properness.
Since $D\cap \Gamma$ spans $\R^d$, \cref{lem:quotient-lattice-point-bound}
and $D\subseteq D'$ give
$ \lvert f(D\cap \Gamma)\rvert\det \Gamma
 \le \Oh_d(\vol_d(\pi_\Lambda D'))$.
Then, chaining the two inequalities and using the definition of $r$ gives
\[
 r\lvert\supp X\rvert
 \le \Oh_d\left(\frac{\rho' \vol_d(\pi_\Lambda K)}{\det \Gamma}\right)
 \le \Oh_d(\rho' w) = \Oh_d(\rho w),
\]
where the last inequality used \eqref{eqn:volpik}.
Thus the time complexity for
\cref{thm:geometric-cap-cover} is 
$ \rho w b_0^{\Oh_d(1)}$.

We next bound
the time for enumerating $E$ and $S$, and the sizes $|E|$ and $|S|$, which also control the size of the AP decomposition of the answer set $G$ and its construction time. 
For $P=P_i$ or $P=\bfa+B$, cap thickness (\cref{thm:geometric-cap-cover}~\ref{item:cap-cover-thickness}) and $D'\subseteq B$ give
$D'\subseteq c_d^{-1}(P-P)$ in both cases.
Applying
\cref{lem:quotient-lattice-point-bound} to $c_d^{-1}(P-P)$ and using the definition of $\Lambda$,
 gives
\[
 \lvert f(P\cap \Gamma)\rvert
\le \Oh_d\left(\frac{\vol_d(\pi_\Lambda c_d^{-1}(P-P))}{\det \Gamma}\right)
 \le \Oh_d\left(\frac{\vol_d(\pi_\Lambda P)}{\det \Gamma}\right),
\]
where the last step used \cref{lem:volumeineq}.
For $P= P_i$, we use the volume bound in \cref{thm:geometric-cap-cover}~\ref{item:cap-cover-volume};
for $P = \bfa+B$, we use $\lvert R\rvert=\Oh_d(r)$ and
$\vol_d(\pi_\Lambda B)=\Oh_d(\vol_d(\pi_\Lambda D'))$ (due to $B\subseteq \beta D'$ and \cref{lem:volumeineq}).
Together these bounds give
\[
 \sum_i \big(1 + \lvert f(P_i\cap \Gamma)\rvert\big )
 +\sum_{\bfa\in R}\big ( 1 + \lvert f((\bfa+B)\cap \Gamma)\rvert\big )
 \le \Oh_d\left(\frac{\rho' \vol_d(\pi_\Lambda K)}{\det \Gamma}\right)
 \le \Oh_d(\rho' w),
\]
where the $+1$ terms are bounded using the fact that the number of caps returned and the size of $R$ are $\Oh_d(r) = \Oh_d(r\lvert\supp X\rvert)=\Oh_d(\rho' w)$ by
\cref{thm:geometric-cap-cover}.
Hence, 
\cref{lem:quotient-lattice-enumeration} enumerates all elements of $E$ and $S$ in
$b_0^{\Oh_d(1)} \rho' w$ time, and $\lvert E\rvert+\lvert S\rvert=\Oh_d(\rho' w)$.

Finally, $\rho'\le \beta \rho$ and $w\le u-\ell$ give the claimed
output size and running time bounds.
This completes the proof of 
\cref{thm:subset-sum-cap-covering}.

\subsection{Witness recovery}
\label{subsec:cap-cover-witness-recovery}

\begin{lemma}
\label{lem:cap-cover-witness}
In the setup of \cref{thm:subset-sum-cap-covering}, given $t\in G$ and $\boldsymbol{\tau}\in \Gamma$ with
$\boldsymbol{\tau}+(dH\cap \Gamma)\subseteq \Sigma(X')$,
one can recover a submultiset $Y\subseteq X$ with $\sum_{\bfy\in Y} f(\bfy)=t$ in $b_0^{\Oh_d(1)}\bigl(\lvert\supp X\rvert+\sum_{\bfv\in X'}\lvert f(\bfv)\rvert\bigr)$
time.
\end{lemma}

We recover the witness by
performing the same rounding step as in the proof of \cref{lem:innerpoints}. It involves solving a linear program, which can be done via divide and conquer with the optimization tools from \cref{sec:geometry-oracles}:
\begin{lemma}
\label{lem:zonotope-coefficients}
Given $\bfv_1,\ldots,\bfv_n\in \Z^d$ and
$\bfy\in \caZ((\bfv_i)_{i=1}^n)\cap \Q^d$,
one can find rational $\boldsymbol{\alpha}\in [0,1]^n$
such that
$\bfy=\sum_{i=1}^n \alpha_i \bfv_i$ and $\lvert\{i\in [n] :0<\alpha_i<1\}\rvert\le d$
 in $\tO_d(n\cdot b_0^{\Oh_d(1)})$ time,
 with $\maxbitlen(\alpha_i)=\Oh_d(b_0)$ for each $i\in [n]$,
 where $b_0=\maxbitlen(\bfv_1,\ldots,\bfv_n,\bfy)$.
\end{lemma}
\begin{proof}
A vertex solution of the LP
$\{\boldsymbol{\alpha}:\sum_i \alpha_i \bfv_i=\bfy,\ 0\le \alpha_i\le 1\}$ can only have at most $d$ fractional coordinates $0<\alpha_i<1$.
We find such a vertex by divide and conquer.

For $I\subseteq [n]$, write $Z_I=\caZ((\bfv_i)_{i\in I})$.
Build a balanced binary tree with $n$ leaves representing $1,2,\dots,n$, and identify each tree node by the set $I$ consisting of its leaf elements. 
Each tree node receives a target vector $\bfy_I\in \Q^d \cap \caZ((\bfv_i)_{i\in I})$, with the root $[n]$ receiving $\bfy_{[n]} = \bfy$.
Starting from the root $I=[n]$, we recursively find a vertex $\boldsymbol \alpha_{I}\in [0,1]^{I}$ of 
$\{\boldsymbol{\alpha}\in [0,1]^I:\sum_{i\in I} \alpha_i \bfv_i=\bfy_I\}$ as follows: let $I_1,I_2$ be its two children, so $I=I_1\sqcup I_2$,  and use \cref{lem:zonotope-algorithms,cor:polytope-oracles} to find a vertex
\[
 \bfy_{I_1} \in Z_{I_1}\cap (\bfy_I-Z_{I_2}),
\]
recurse with children's targets $\bfy_{I_1}$ and $\bfy_{I_2}\coloneqq \bfy_{I}-\bfy_{I_1}$, and return the concatenation of the two children's solutions $\boldsymbol \alpha_{I_1},\boldsymbol \alpha_{I_2}$.
A singleton $I=\{i\}$ is solved trivially.

We prove by induction that $\boldsymbol \alpha_{I}$ is indeed a vertex of 
$\{\boldsymbol{\alpha}\in [0,1]^I:\sum_{i\in I} \alpha_i \bfv_i=\bfy_I\}$.
The singleton case is immediate.
If $\boldsymbol \alpha_{I}$ is the midpoint of two feasible vectors $\boldsymbol \alpha_{I}' ,\boldsymbol \alpha_{I}''$, then their restrictions to $I_1$,  $\boldsymbol \alpha_{I_1}' ,\boldsymbol \alpha_{I_1}''$, satisfy $\frac{1}{2}\sum_{i\in I_1}(\alpha'_i\bfv_i+\alpha''_i\bfv_i) = \sum_{i\in I_1}\alpha_i\bfv_i = \bfy_{I_1}$ (and analogously for $I_2$). Since $\bfy_{I_1}$ is a vertex, this implies
$\sum_{i\in I_1}\alpha'_i\bfv_i = \sum_{i\in I_1}\alpha''_i\bfv_i = \bfy_{I_1}$ (and analogously for $I_2$). 
Then, since $\boldsymbol \alpha_{I_1}$ is a vertex by induction, and is the midpoint of $\boldsymbol \alpha_{I_1}'$ and
$\boldsymbol \alpha_{I_1}''$, both of which are feasible for the child $I_1$, we have 
$\boldsymbol \alpha_{I_1}' = \boldsymbol \alpha_{I_1}''$. An analogous argument holds for $I_2$. Therefore $\boldsymbol{\alpha}_I' = \boldsymbol{\alpha}_I''$, proving that $\boldsymbol{\alpha}_I$ is a vertex.

Since the computed $\boldsymbol{\alpha}$ is a vertex solution, one can bound the maximum bit length of all intermediate targets $\boldsymbol{\bfy_I}$ by $\Oh_d(b_0+\log n)$.
Thus a node $I$ takes
$\tO_d(\lvert I\rvert)(b_0+\log n)^{\Oh_d(1)}$ time by
\cref{lem:zonotope-algorithms,cor:polytope-oracles}.
Summing the time complexities across $\Oh(\log n)$ levels gives the stated time complexity.
\end{proof}

\begin{proof}[Proof sketch of \cref{lem:cap-cover-witness}]
We follow the proof of \cref{lem:innerpoints} with a small adaptation (namely, applying $f$ in appropriate places), and
make some steps algorithmic.
First, since $t\in G$, by \cref{thm:subset-sum-cap-covering}, there is 
$\bfy\in(K\ominus\lvert X'\rvert H)\cap\Gamma$ with $f(\bfy)=t$, where 
$K\subseteq\caZ(X)$; 
the proof of \cref{lem:innerpoints} shows 
 $\bfz=\bfy-\boldsymbol{\tau}\in\caZ(X\setminus X')\cap\Gamma$, and hence $f(\bfz)=t-f(\boldsymbol \tau)$.
Here, it suffices to find any
$\bfz\in\caZ(X\setminus X')\cap\Gamma$ with
$f(\bfz)=t-f(\boldsymbol{\tau})$; we do so using
\cref{lem:zonotope-algorithms,lem:constrained-optimization,lem:lattice-optimization}
in $b_0^{\Oh_d(1)}\lvert\supp X\rvert$ time.

Second, the proof of \cref{lem:innerpoints} chooses a basic feasible solution $\boldsymbol \alpha$
of the LP $\bfz = \sum_{i=1}^{|X\setminus X'|}\alpha_{i}\bfx_i, \alpha_{i}\in [0,1]$, where $\{\bfx_i\}_i = X\setminus X'$.
We compute such $\boldsymbol \alpha$ using \cref{lem:zonotope-coefficients} (taking the multiplicities of $X\setminus X'$ into account), and round it to $\boldsymbol \alpha'$ in the same way as in 
\cref{lem:innerpoints}, giving a $\bfz' \in \Sigma(X\setminus X')$ together with a witnessing submultiset of $X\setminus X'$, such that $f(\bfy - \bfz')\in \Sigma(f(X'))$.

Finally, it remains to find a witness of
$f(\bfy - \bfz')\in \Sigma(f(X'))$. 
Writing $R=\sum_{\bfv\in X'}\lvert f(\bfv)\rvert$, this follows from a standard divide-and-conquer algorithm with 
sparse convolution (\cref{lem:sparse-convolution}) in $\tO(R)$ time. This proves the claimed time bound and finishes the proof.
\end{proof}

\section{Solving One Group}
\label{sec:one-group}

In this section, we prove \cref{lem:oneclass}.
The proof proceeds by applying the cap covering from 
\cref{sec:subset-sum-cap-covering}, and then computing integer sets containing the $f$-values of the subset sums in the caps.

\subsection{Computing subset sums in a cap}
\label{subsec:subset-sums-in-cap}

We use the following lemma to compute a set $S\subseteq f(\Sigma(X))$ containing $f(\Sigma(X)\cap Q)$ for a cap $Q$ of the zonotope. For convenience, we assume a parallelepiped enclosing the cap is given as part of the input.
\begin{lemma}
\label{lem:cap-subset-sums}
Given a multiset $X$ of vectors in $\Z^d$, and a linear map
$f\colon \R^d\to \R$ with integer coefficients,
write $\Gamma=\caL(X)$ and $K=\caZ(X)$.
Let $\bfu\in \Q^d\setminus \{\zero\}$, $a\in \Q$,
 and a rational parallelepiped $P$ such that
$P\supseteq Q\coloneqq K\cap \{\bfx:\langle\bfu,\bfx\rangle\le a\}$ be given.

Then, one can compute an integer set $S$ satisfying
\[
 f(\Sigma(X)\cap Q)\subseteq S\subseteq f(\Sigma(X))
\]
in
$ \left(\lvert\supp X\rvert+\lvert f((P-P)\cap \Gamma)\rvert\right)
 b_0^{\Oh_d(1)}$
time, where $b_0=\maxbitlen(X,f,\bfu,a,P)$.

Moreover, given $t\in S$, one can recover a submultiset $Y\subseteq X$ with $\sum_{\bfy\in Y} f(\bfy)=t$ in the same time complexity.
\end{lemma}

The proof proceeds by reducing the task to the following multidimensional \textsc{Subset Sum} lemma. 

\begin{lemma}
\label{lem:box-subset-sums}
Given $\bft\in \Nonneg^d$, write $B=\prod_{j=1}^d[0,t_j]$. For a multiset $X$ with $\supp X\subseteq B\cap \Z^d$, one can compute $\Sigma(X)\cap B$
in $\tO_{d}(\lvert B\cap \Z^d\rvert)$ time.

Moreover,  if $\bft\in \Sigma(X)$, one can recover a witness $Y\subseteq X$ with $\sum_{\bfy \in Y} \bfy = \bft$ in the same time complexity.
\end{lemma}
\cref{lem:box-subset-sums} follows from a direct generalization of Bringmann's algorithm \cite{Bringmann17} with its derandomization by Chan~\cite{chan26}; we include a proof of \cref{lem:box-subset-sums}  in \cref{app:box-subset-sums} by a black-box reduction to \cite{Bringmann17,chan26}.

We now prove \cref{lem:cap-subset-sums}.
\begin{proof} [Proof of \cref{lem:cap-subset-sums}]
First, we may assume that $\langle\bfu,\bfx\rangle\ge 0$ for all $\bfx\in X$ without loss of generality.
This can be justified by negating those $\bfx\in X$ with $\langle\bfu,\bfx\rangle<0$,
and translating $K,Q,P,\Sigma(X)$ accordingly, leaving $\bfu,\Gamma$, and $P-P$ unchanged. 
Then assume the cap $Q = K \cap \{\bfx:\langle \bfu,\bfx\rangle \le a\}$ satisfies $a\ge 0$ (otherwise, return $S=\varnothing$).

Then, we can assume without loss of generality that all $\bfx \in X$ satisfy $\bfx \in Q$. 
To justify this, note that if $Y\subseteq X$ contains an $\bfx$ such that $\langle \bfu,\bfx\rangle >a$, then $\sum_{\bfy\in Y}\langle \bfu, \bfy\rangle >a $ because $\langle \bfu,\bfy\rangle \ge 0$ for all $\bfy$. Thus we can simply remove all $\bfx \notin Q$ from the input multiset $X$.

We next transform the instance into a boxed form in order to apply 
\cref{lem:box-subset-sums}.
Since $\zero\in Q\subseteq P$, we have $Q\subseteq P-P$.
Let $E=\spa(X)$, $\Gamma_0=\caL(X)\subseteq \Gamma$, and assume 
$f|_E$ is nonzero (otherwise,  return $S=\{0\}$).
Apply \cref{lem:local-properization} to
$(P-P)\cap E$, $\Gamma_0$, and $f|_E$.
It gives maps $\Phi\colon E\to \R^{d'}$, $g\colon \R^{d'}\to \R$,
and a box $B=\prod_{j=1}^{d'}[-U_j,U_j]$ such that
$ \Phi \Gamma_0=\Z^{d'}, f|_E=g \circ \Phi$, and
  $\Phi((P-P)\cap E)\subseteq B$,
  with
\[
 V\coloneqq \lvert B\cap \Z^{d'}\rvert
 \le \Oh_d\bigl(\lvert f((P-P)\cap \Gamma)\rvert\bigr).
\]

Partition all $\bfx\in X$ into $X=X_1\sqcup \cdots\sqcup X_{2^{d'}}$ according to the sign pattern of the coordinates
of $\Phi \bfx$.
For each $X_j$,
we compute
 \[
 S_j=g\bigl(\Sigma(\Phi X_j)\cap B\bigr)
\]
in $b_0^{\Oh_d(1)} V$ time
using \cref{lem:box-subset-sums} (which is applicable after negating the negative coordinates).
Compute and return $S=S_1+\cdots+S_{2^{d'}}$ using sparse convolution (\cref{lem:sparse-convolution}),
since each intermediate sumset $S_1+\cdots+S_j$ is contained in
$g(2^{d'} B\cap \Z^{d'})$, which has size $\Oh_d(V)$.
Including the scan of the original input $X$ and the construction of the multisets $\Phi X_j$,
the overall time complexity is
$b_0^{\Oh_d(1)}(\lvert\supp X\rvert+V)$.

We now show correctness of $S$ using the guarantees of \cref{lem:local-properization}.
Since  $f|_E=g \circ \Phi$,
$g(\Sigma(\Phi X_j))=f(\Sigma(X_j))$, so $S_j\subseteq f(\Sigma(X_j))$, and consequently $S\subseteq f(\Sigma(X))$.
To prove the other direction, note that the guarantee $\Phi((P-P)\cap E)\subseteq B$ implies
\[f\big(\Sigma(X_j)\cap Q\big )\subseteq f\big(\Sigma(X_j)\cap (P-P)\big )= g\big (\Phi\big (\Sigma(X_j)\cap (P-P)\big )\big )
\subseteq g \big (\Sigma(\Phi X_j)\cap B\big ) =  S_j.\]
If $Y\subseteq X$ has sum in $Q$, then each $Y\cap X_j$ also has sum in
$Q$, since $\langle\bfu,\bfx\rangle\ge 0$ for all $\bfx\in X$.
Hence $f(\sum_{\bfy\in Y\cap X_j} \bfy)\in S_j$ for every $j$, proving
$f(\Sigma(X)\cap Q)\subseteq S$.

The witness recovery is straightforward, using 
the witness-recovery step of \cref{lem:box-subset-sums}.
\end{proof}

\subsection{Proof of \texorpdfstring{\cref*{lem:oneclass}}{Lemma \getrefnumber{lem:oneclass}} }
\label{sec:recursive}

We now prove \cref{lem:oneclass} by combining the cap covering theorem (\cref{thm:subset-sum-cap-covering}) and the lemma that handles caps (\cref{lem:cap-subset-sums}).

Fix $X,X',f$ as in \cref{lem:oneclass}.
Write $\Gamma=\caL(X)$, $n=\lvert X\rvert$, and $s=\lvert X'\rvert$.
For simplicity, assume $\rk \Gamma=d$ without loss of generality (otherwise, restrict to the subspace spanned by $\Gamma$).
Let $H=\conv(\pm \supp X)$, $D=sH$, and $K=\caZ(X)$.
Define $\rho=s \lambda/n$, and
$w_X=\sum_{\bfx\in X} f(\bfx)=\fw(K)$.

We apply \cref{thm:subset-sum-cap-covering} to $X,X',f$ and interval $[0,w_X]$. We first check the conditions required by
\cref{thm:subset-sum-cap-covering} are satisfied:
By properness of $f(H\cap \Gamma)$, $f$ is injective on $\supp X$.
By definition of $D$ and $\rho$, \eqref{eq:oneclass-balance} gives
$D\subseteq \rho \caZ_0(X)$,
so $K$ contains a translate of $\rho^{-1} D$.
By \eqref{eq:oneclass-generating}, $X'$ is $d$-generating for
$H\cap \Gamma$, so $\boldsymbol{\tau}+(dH\cap \Gamma)\subseteq \Sigma(X')$
for some $\boldsymbol{\tau}\in \Gamma$;
 taking convex hulls gives
$\boldsymbol{\tau}+dH\subseteq \caZ(X')$.
In particular, $K$ contains a translate of $H$.
Therefore, 
 \cref{thm:subset-sum-cap-covering} is applicable;
write $G$ for the returned set (represented as an AP decomposition of size $\Oh_d(\rho w_X)$), and $(Q_i,P_i)_{i\in \caI}$ for the returned caps
and enclosing parallelepipeds. 

For each $i\in \caI$, apply \cref{lem:cap-subset-sums} to
the cap $Q_i$ with enclosing parallelepiped $P_i$, and obtain a set $S_i$ such that 
\[f(\Sigma(X)\cap Q_i)\subseteq S_i \subseteq f(\Sigma(X)).\]
Then, by \cref{thm:subset-sum-cap-covering}, we get
$G\subseteq f(\Sigma(X))$ and
\[
 f(\Sigma(X))\setminus G
 \subseteq \bigcup_{i\in \caI} S_i.
\]
Therefore, we can return $f(\Sigma(X)) = G \cup \bigcup_{i\in \caI} S_i$ as the answer set.
Since $G$ is given as an AP decomposition, while $S_i$ are given explicitly, this answer set can be represented as an AP decomposition of the same common difference using the basic operations in 
\cref{lem:ap-operations}. 

\paragraph*{Time and output size.}
\cref{thm:subset-sum-cap-covering} runs in $\rho w_X(\log \Delta)^{\Oh_d(1)}$ time.

Each application of
\cref{lem:cap-subset-sums} takes $(\lvert\supp X\rvert + \lvert f((P_i-P_i)\cap \Gamma)\rvert)(\log \Delta)^{\Oh_d(1)}$ time.
We now bound their total time.
Let $B_i\coloneqq c_d^{-1}(P_i-P_i)$,
 where $c_d \in (0,1]$ is as in \cref{thm:subset-sum-cap-covering},
so
$\lvert f((P_i-P_i)\cap \Gamma)\rvert\le \lvert f(B_i\cap \Gamma)\rvert$.
Since $Q_i$ contains a translate of $c_d H$ by \cref{thm:subset-sum-cap-covering},
$2c_d H\subseteq Q_i-Q_i\subseteq P_i-P_i$, so $H\subseteq B_i$.
Then, since $f$ is injective on $\supp X$, 
$\lvert\supp X\rvert=\lvert f(\supp X)\rvert
\le \lvert f(B_i\cap \Gamma)\rvert$.
Thus each application of \cref{lem:cap-subset-sums} takes
$\lvert f(B_i\cap \Gamma)\rvert(\log \Delta)^{\Oh_d(1)}$ time.

Since $\Lambda=\Gamma\cap \ker f$,
$\lvert f(B_i\cap \Gamma)\rvert=\lvert\pi_\Lambda(B_i\cap \Gamma)\rvert$.
\cref{lem:quotient-lattice-point-bound} then gives
\[
 \lvert f(B_i\cap \Gamma)\rvert
 \le \Oh_d\left(\frac{\vol_d(\pi_\Lambda B_i)}{\det \Gamma}\right).
\]
$B_i$ is a translate of $(2/c_d) P_i$, so
\cref{lem:volumeineq} gives
$\vol_d(\pi_\Lambda B_i)\le \Oh_d(\vol_d(\pi_\Lambda P_i))$.
Summing and using the volume bound \eqref{eq:subset-sum-cap-volume} in \cref{thm:subset-sum-cap-covering}, we obtain
\begin{equation*}
 \sum_{i\in \caI}
 \lvert f(B_i\cap \Gamma)\rvert
 =\Oh_d(\rho w_X).
\end{equation*}
Thus, the applications of \cref{lem:cap-subset-sums} take total time $\rho w_X(\log \Delta)^{\Oh_d(1)}$.

We next bound the size of the returned AP decomposition of the answer set $f(\Sigma(X))$.
Since $S_i\subseteq f(\Sigma(X))$, \cref{thm:subset-sum-cap-covering} gives
$\bigl(\bigcup_i S_i\bigr)\setminus G\subseteq \bigcup_i f(Q_i\cap \Gamma)$.
Then, using $\lvert f(Q_i\cap \Gamma)\rvert\le \lvert f((P_i-P_i)\cap \Gamma)\rvert$ and the same bounds above, we similarly get $\sum_i\lvert f(Q_i\cap \Gamma)\rvert = \Oh_d(\rho w_X)$.
Since $G$ has an AP decomposition of size $\Oh_d(\rho w_X)$, this shows that the answer set $f(\Sigma(X)) = G \cup \bigcup_{i} S_i$ also has
an AP decomposition of size $\Oh_d(\rho w_X)$ of the same common difference, which can be efficiently computed using 
\cref{lem:ap-operations}.

Finally, since $f(\bfx)\in [1,\Delta]$ for $\bfx \in X$, the definition of $\rho$ and $w_X$ gives  $\rho w_X=(s \lambda/n) w_X\le s \lambda \Delta$.
The construction therefore takes $\tO_d(s \lambda \Delta)$ time and
returns an AP decomposition of size $\Oh_d(s \lambda \Delta)$.

\paragraph*{Witness recovery.}
Given $t\in f(\Sigma(X))$, if $t\in G$ (deep case), apply \cref{lem:cap-cover-witness} to recover a witness.
If $t\notin G$ (shallow case), then find $i\in \caI$ with $t\in S_i$ and use the witness-recovery step of \cref{lem:cap-subset-sums} to recover a witness.

\section{Merging Groups}
\label{sec:merging-groups}
In this section, we prove \cref{lem:merge}.
As in \cref{sec:one-group}, we apply the cap covering theorem (\cref{thm:subset-sum-cap-covering}) to the sumset, and obtain the deep sums together with caps that cover the shallow sums.
We use the caps to identify the relevant input pairs that can contribute to the shallow sums,  and then compute these sums by sparse convolution
(\cref{lem:sparse-convolution}).

\begin{proof}[Proof of \cref{lem:merge}]
Recall from the statement of
\cref{lem:merge} that $X', X$, and $Y$ are input multisets of vectors in $ \Gamma =  \caL(X) \subseteq \Z^d$.
For simplicity, assume $\rk \Gamma=d$ without loss of generality (otherwise, restrict to the subspace spanned by $\Gamma$).
Our task is to compute the sumset of 
$S_Y = f(\Sigma(Y)) \cap J$ and $S_X = f(\Sigma(X))$ restricted to the given interval $I$.

Write $w_X=\sum_{\bfx\in X} f(\bfx)$,  and write $J^+=J+[0,w_X]$ for the interval that contains the possible sums. 
Define the clipped zonotope
\[K=\caZ(V)\cap f^{-1}(J^+), \text{ where }V\coloneqq X\sqcup Y.\]
Then, clearly,
\[
 \caZ(X)+\bigl(\caZ(Y)\cap f^{-1}(J)\bigr)
 \subseteq K\subseteq \caZ(V). 
\]
Since $J\subseteq f(\caZ(Y))$ by assumption, $K$ contains a translate of $\caZ(X)$.

Write 
$H=\conv(\pm \supp V)$,
$s=\lvert X'\rvert$, and $D=sH$. 
We apply 
the cap covering \cref{thm:subset-sum-cap-covering}
to the multisets $X'\subseteq V$, the clipped zonotope $K$ and the interval $J^+$.
We first check the requirements of \cref{thm:subset-sum-cap-covering} are met:
\begin{itemize}
      \item The properness of $f(H\cap \Gamma)$ indeed holds by assumption. 
\item 
\Cref{eq:oneclass-generating} guarantees that $X'$ is $d$-generating for $H\cap \Gamma$.
  \item 
        Recall 
from \eqref{eq:oneclass-balance} that $H \subseteq \frac{\lambda}{\lvert X\rvert} \caZ_0(X)$.
Set \[\rho=s \lambda/\lvert X\rvert.\]
Then, $\caZ(X)$ contains a translate of
$\rho^{-1} D = \rho^{-1} sH$.
Since $K$ contains a translate of $\caZ(X)$, $K$ also contains a translate of $\rho^{-1} D$.
\item 
Taking convex hulls of the $d$-generating condition
 $\boldsymbol{\tau}+(dH\cap \Gamma)\subseteq \Sigma(X')$
 gives
$\boldsymbol{\tau}+dH\subseteq \caZ(X')\subseteq \caZ(X)$. So $K$ also contains a translate of $dH$.
\end{itemize}
Therefore  \cref{thm:subset-sum-cap-covering}
is applicable;
write $G$ for the returned set of deep subset sums, and
$Q_i=K\cap \{\bfx:\langle\bfu_i,\bfx\rangle\ge a_i\}$,
$i\in \caI$, for the returned caps, with enclosing parallelepipeds $P_i$.

\paragraph*{Useful summand pairs}
We now use the caps $Q_i$ to cover the useful summand pairs.

For each $i$,
compute $\bfq_i\in Q_i\cap \Gamma$  using
\cref{lem:constrained-optimization,lem:lattice-optimization}
(discard $i$ if such $\bfq_i$ does not exist).
For each $i$, 
pick $\bfx_i\in \caZ(X)$
that maximizes
$\langle\bfu_i,\bfx\rangle$, which can be computed by \cref{lem:zonotope-algorithms} and satisfies $\bfx_i \in \caZ(X) \cap \Gamma$.
Define $B_i=P_i-P_i$, and enumerate $U_i\coloneqq f(B_i\cap \Gamma)$ by
\cref{lem:quotient-lattice-enumeration}, and define the candidate summand lists
\begin{equation}
\label{eq:merge-candidate-lists}
\begin{aligned}
 C_{X,i}&=f(\bfx_i)+U_i,\\
 C_{Y,i}&=f(\bfq_i-\bfx_i)+U_i.
\end{aligned}
\end{equation}

\begin{claim}
\label{clm:merge-pair-coverage}
Let $\bfz_X\in \caZ(X)\cap \Gamma$ and
$\bfz_Y\in \caZ(Y)\cap \Gamma \cap f^{-1}(J)$.
If $f(\bfz_X+\bfz_Y)\notin G$, then for some $i$, it holds that
$ (f(\bfz_X),f(\bfz_Y))\in C_{X,i}\times C_{Y,i}$.
\end{claim}

\begin{proof}
Write $\bfr=\bfz_X+\bfz_Y \in K\cap \Gamma$. 
Then, since $f(\bfr)\notin G$,  \cref{thm:subset-sum-cap-covering}  implies
$\bfr\in Q_i\subseteq P_i$ for some $i$.
Define $\bfp = \bfx_i+\bfz_Y\in \caZ(V)$. Since $f(\bfz_Y)\in J$, we have $\bfp\in K$ and $f(\bfp)\in J^+$ by definitions of $J^+$ and $K$.
Moreover,
$\langle\bfu_i,\bfp\rangle = \langle\bfu_i,\bfr\rangle + 
\langle\bfu_i,\bfx_i-\bfz_X\rangle
\ge \langle\bfu_i,\bfr\rangle$ by definition of $\bfx_i$.
Thus, by definition of $Q_i$, $\bfp\in Q_i\subseteq P_i$.
Since $P_i-P_i=B_i$, we obtain
\begin{align*}
 \bfz_Y-(\bfq_i-\bfx_i)&=\bfp-\bfq_i\in B_i,\\*
 \bfz_X-\bfx_i&=\bfr-\bfp\in B_i.
\end{align*}
Since both these vectors lie in $\Gamma$, applying $f$ gives
$f(\bfz_Y) -f(\bfq_i-\bfx_i)\in U_i$ and $f(\bfz_X)-f(\bfx_i)\in U_i$, proving
$(f(\bfz_X),f(\bfz_Y))\in C_{X,i}\times C_{Y,i}$.
\end{proof}

For each $i$, after computing the candidate summand lists \eqref{eq:merge-candidate-lists},
 compute 
$S_{X,i}\coloneqq C_{X,i}\cap S_X$ and
$S_{Y,i}\coloneqq C_{Y,i}\cap S_Y$ using the membership queries (\cref{lem:ap-operations}) supported by the given AP decompositions of $S_X$ and $S_Y$.
Then, compute $T_i=S_{X,i}+S_{Y,i}$ using 
sparse convolution (\cref{lem:sparse-convolution}).
Finally, let $S^+\coloneqq G\cup \bigcup_{i} T_i$, and return the answer set $S^+\cap I$ represented as an AP decomposition.
Since $G$ is given as an AP decomposition and the $T_i$ are explicit sets, the AP decomposition of the answer set can be efficiently 
constructed using the basic operations in \cref{lem:ap-operations}.

\paragraph*{Correctness.}
By construction, $T_i\subseteq S_X+S_Y\subseteq f(\Sigma(V))$.
Since $G\subseteq f(\Sigma(V))\cap J^+$
by \eqref{eq:cap-cover-deep-values}, this gives
$S^+\subseteq f(\Sigma(V))$.

Conversely, let $t\in f(\Sigma(V))\cap I$, 
and suppose $t$ is witnessed by $t = f(\bfz_X)+f(\bfz_Y)$ with $\bfz_X\in \Sigma(X)$ and $\bfz_Y\in \Sigma(Y)$.
The interval condition \eqref{eq:public-merge-restriction-condition}
gives $f(\bfz_Y)\in J$, so
$(f(\bfz_X),f(\bfz_Y))\in S_X\times S_Y$.
If $t\notin G$, \cref{clm:merge-pair-coverage} gives $i$ with
$(f(\bfz_X),f(\bfz_Y))\in S_{X,i}\times S_{Y,i}$, so $t\in T_i$.

Thus $S^+\cap I=(S_X+S_Y)\cap I=f(\Sigma(V))\cap I$.

\paragraph*{Time and output size.}
In \Cref{thm:subset-sum-cap-covering}, since $b_0=\Oh_d(\log \Delta)$, the time complexity of
\Cref{thm:subset-sum-cap-covering} is 
$\rho\lvert J^+\rvert(\log \Delta)^{\Oh_d(1)}$.

For each $i$, computing $\bfq_i,\bfx_i$ takes
$\lvert\supp V\rvert(\log \Delta)^{\Oh_d(1)}$ time by
\cref{lem:zonotope-algorithms,lem:constrained-optimization,lem:lattice-optimization}. 
Enumerating $U_i = f(B_i\cap \Gamma), C_{X,i},C_{Y,i}$ and  computing $T_i = S_{X,i}+S_{Y,i}$ take
$\lvert C_{X,i}+C_{Y,i}\rvert(\log \Delta)^{\Oh_d(1)}$ time by
\cref{lem:quotient-lattice-enumeration,lem:sparse-convolution,lem:ap-operations}.
Since both $C_{X,i}$ and $C_{Y,i}$ are translates of $U_i = f(B_i\cap \Gamma)$, we have
$\lvert C_{X,i}+C_{Y,i}\rvert \le \lvert f(2B_i\cap \Gamma)\rvert$.
The cap thickness guarantee in \cref{thm:subset-sum-cap-covering} gives
$2c_d H\subseteq Q_i-Q_i\subseteq B_i$, where $c_d\in (0,1]$,
so $\supp V\subseteq H\cap \Gamma\subseteq c_d^{-1} B_i\cap \Gamma$;
the properness of $f(H\cap \Gamma)$ and
 \cref{lem:lattice-quotient-dilation} then imply
\[
 \lvert\supp V\rvert
 \le \lvert f(c_d^{-1} B_i\cap \Gamma)\rvert
 =\Oh_d(\lvert f(B_i\cap \Gamma)\rvert).
\]
Thus, the total time for each $i$ is
$\lvert f(2B_i\cap \Gamma)\rvert(\log \Delta)^{\Oh_d(1)}$.

Let $\Lambda=\Gamma\cap \ker f$ as in \cref{thm:subset-sum-cap-covering}. Then, 
$\lvert f(2B_i\cap \Gamma)\rvert=\lvert\pi_\Lambda(2B_i\cap \Gamma)\rvert$.
Moreover, $\supp V$ spans $\R^d$, since $X\subseteq V$ and $\Gamma=\caL(X)$ has rank $d$.
By monotonicity and \cref{lem:quotient-lattice-point-bound} applied to $(2/c_d) B_i$,
\[
 \lvert f(2B_i\cap \Gamma)\rvert
 \le \lvert f((2/c_d) B_i\cap \Gamma)\rvert
 \le \Oh_d\left(\frac{\vol_d(\pi_\Lambda((2/c_d) B_i))}{\det \Gamma}\right) \le 
 \Oh_d\left(\frac{\vol_d(\pi_\Lambda P_i)}{\det \Gamma}\right),
\]
where the last step follows from $B_i = P_i-P_i$ and 
\cref{lem:volumeineq}.
Summing and using the volume bound \eqref{eq:subset-sum-cap-volume} in
\cref{thm:subset-sum-cap-covering}, we obtain
\begin{equation*}
 \sum_{i\in \caI}\lvert f(2B_i\cap \Gamma)\rvert
 =\Oh_d(\rho\lvert J^+\rvert).
\end{equation*}
Hence the total time complexity is $\rho\lvert J^+\rvert(\log \Delta)^{\Oh_d(1)}$ time.

By \cref{thm:subset-sum-cap-covering}, the returned set $G$ has an AP decomposition of size 
$\Oh_d(\rho\lvert J^+\rvert)$.
Since $T_i\subseteq C_{X,i}+C_{Y,i}$, the same bound from the previous paragraphs gives
$\sum_{i}\lvert T_i\rvert=\Oh_d(\rho\lvert J^+\rvert)$.
Therefore, the answer set
$\big (G\cup \bigcup_{i} T_i\big )\cap I$ also has 
an AP decomposition of the same common difference of size  
$\Oh_d(\rho\lvert J^+\rvert)$, which is 
computable within the same time bound by \cref{lem:ap-operations}.

Finally, since $f(\bfx)\in [1,\Delta]$ for $\bfx\in X$,
we have $w_X = \sum_{\bfx\in X} f(\bfx) \le \lvert X\rvert\Delta$, and therefore
\[
 \rho\lvert J^+\rvert=\rho(w_X+\lvert J\rvert)
 =\Oh_d\!\left(s \lambda
       \left(\Delta+\frac{\lvert J\rvert}{\lvert X\rvert}\right)\right).
\]
This gives the stated construction time and output size.

\paragraph*{Witness recovery.}

Given $t\in (S_X+S_Y)\cap I$,  if $t\notin G$ (shallow case), find $i$ and $a\in C_{X,i}$
such that $a\in S_X$ and $t-a\in S_Y$, and return the witness pair $(a,t-a)$.
If $t\in G$ (deep case), apply \cref{lem:cap-cover-witness}
to recover a witnessing submultiset $T\subseteq V$ with $\sum_{\bfv\in T} f(\bfv)=t$ (which also immediately gives the claimed witness pair).
\end{proof}

    \appendix
    \section{Proof of \texorpdfstring{\Cref*{lem:rw24}}{Lemma \getrefnumber{lem:rw24}}}
\label{app:freiman}

For a GAP $P = x_P + f_P(B_P\cap \Z^d)\subseteq \Z$, let $\maxbitlen(P)=\maxbitlen(x_P,f_P,B_P)$.

We prove \cref{lem:rw24} by extending
Randolph and W\k{e}grzycki's randomized algorithmic Freiman theorem \cite[Theorem~5 and full-version Appendix~A]{rw24}.
Besides derandomization, 
our extended version establishes two extra properties, namely $2$-properness, and that the vectors representing the input integers generate $\Z^d$. We first state and prove the two lemmas aimed for these properties.

\begin{lemma}
\label{lem:rankreduction}
Given a rank-$d$ GAP $P\subseteq \Z$ with $b_0=\maxbitlen(P)$,
in $\vol(P) b_0^{\Oh_d(1)}$ time,
one can either certify the 
$2$-properness of $P$, or compute a GAP $Q\supseteq P$ of rank $d'<d$ with
$\vol(Q)=\Oh_d(\vol(P))$ and ${\maxbitlen(Q)=\Oh_d(b_0)}$.
\end{lemma}

\begin{lemma}
 \label{lem:reducegap}
Let $P\subseteq \Z$ be a rank-$d$ GAP, $0\in A\subseteq P$, and
$b_0=\maxbitlen(P)$.
Given $A$ and $P$, one can compute in $\vol(P) b_0^{\Oh_d(1)}$ time
an origin-symmetric $4$-proper GAP
$Q\supseteq A$ of rank $d'\le d$ with $\vol(Q)=\Oh_d(\vol(P))$,
$\maxbitlen(Q)=\Oh_d(b_0)$, and $\caL(f_Q^{-1}(A))=\Z^{d'}$.
\end{lemma}

\begin{proof}[Proof of \cref{lem:rankreduction}]
Write $P=x_P+f_P(\Z^d\cap \prod_i[0,w_i])$ and let
$B=\prod_i[-w_i,w_i]$. We may assume that $f_P$ is nonzero.
Apply \cref{lem:local-properization} to $B$, $\Z^d$, and $f_P$.
This gives rational linear maps $\Phi$ and $g$ and an integer box
$C=\prod_{i=1}^{d'}[-m_i,m_i]$, where $1\le d'\le d$ and $m_i\in\Pos$.
We have $f_P=g\circ\Phi$, $\Phi\Z^d=\Z^{d'}$, and $\Phi B\subseteq C$,
and $g$ is injective on $4C\cap\Z^{d'}$.

If $d'=d$, then $\Phi$ is invertible, so $f_P$ is injective on
$B\cap\Z^d$. Since $B$ is an integer translate of
$\prod_i[0,2w_i]$, this proves that $P$ is $2$-proper.
If $d'<d$, return $Q=x_P+g(C\cap\Z^{d'})$.
The containments $\prod_i[0,w_i]\subseteq B$ and $\Phi B\subseteq C$
give $P\subseteq Q$. By the bound on $C$ in \cref{lem:local-properization},
\[
 \vol(Q)
 =\lvert C\cap\Z^{d'}\rvert
 \le \Oh_d\bigl(\lvert f_P(B\cap\Z^d)\rvert\bigr)
 \le \Oh_d(\vol(P)),
\]
where the last inequality uses $\prod_i(2w_i+1)\le 2^d\vol(P)$.
The running time and bit length bounds follow from
\cref{lem:local-properization}.
\end{proof}

\begin{proof}[Proof of \cref{lem:reducegap}]
If $A=\{0\}$, return the rank-zero GAP $\{0\}$.
Otherwise, write $P=x_P+f_P(\Z^d\cap\prod_i[0,w_i])$.
For each $a\in A$, choose $\bfu_a\in\Z^d\cap\prod_i[0,w_i]$
with $x_P+f_P(\bfu_a)=a$.
Let $\tilde A=\{\bfu_a-\bfu_0:a\in A\}$, so that $f_P(\tilde A)=A$,
and let $\Gamma=\caL(\tilde A)$.
Then $B=\prod_i[-w_i,w_i]\cap\spa(\Gamma)$ contains $\tilde A$,
so $B\cap\Gamma$ spans $\spa(\Gamma)$.
Since $A\ne\{0\}$, $f_P$ is nonzero on this subspace.

Choose a basis of $\Gamma$ and apply \cref{lem:local-properization}
in this basis to obtain $\Phi$, $g$, and an integer box $C\subseteq\R^{d'}$.
Set $Q=g(C\cap\Z^{d'})$.
The containments $\Phi\tilde A\subseteq\Phi B\subseteq C$ and
$f_P|_{\spa(\Gamma)}=g\circ\Phi$ give $A\subseteq Q$.
Since $C$ is origin-symmetric and $g$ is injective on $4C\cap\Z^{d'}$,
$Q$ is origin-symmetric and $4$-proper. Moreover,
\[
 \caL(f_Q^{-1}(A))
 =\caL(\Phi\tilde A)=\Phi\Gamma=\Z^{d'}.
\]
The bound on $C$ in \cref{lem:local-properization} gives
\[
 \vol(Q)
 =\lvert C\cap\Z^{d'}\rvert
 \le \Oh_d\bigl(\lvert f_P(B\cap\Gamma)\rvert\bigr)
 \le \Oh_d(\vol(P)),
\]
since $\Gamma\subseteq\Z^d$ and $\prod_i(2w_i+1)\le 2^d\vol(P)$.

Enumerating the coefficient box of $P$ gives $\bfu_a$ for each $a\in A$
in $\vol(P)b_0^{\Oh_d(1)}$ time.
\Cref{lem:local-properization}
then gives $\maxbitlen(Q)=\Oh_d(b_0)$ within the same time bound.
\end{proof}

\begin{proof}[Proof of \cref{lem:rw24}]
The singleton case is immediate, so assume $\lvert A\rvert\ge 2$.
The Pl\"unnecke--Ruzsa inequality \cite[Corollaries~6.28--6.29]{taovu06} gives
$\lvert rA-sA\rvert\le K^{r+s}\lvert A\rvert$ for nonnegative integers $r,s$
with $r+s>0$.
For fixed $r,s$, we can compute $rA-sA$ in
$\Oh_K(\lvert A\rvert)\polylog\Delta$ time by sparse nonnegative convolution
(\cref{lem:sparse-convolution}).

We derandomize the modeling step of \cite[Lemma~A.2]{rw24}
by conditional expectations.
Let $A_{\mathrm{diff}}=8A-8A$.
 Choose a prime $m$ with
$8\lvert A_{\mathrm{diff}}\rvert<m<16\lvert A_{\mathrm{diff}}\rvert$,
and let $q=2^L$ be the least power of two exceeding
$\max_{x\in A_{\mathrm{diff}}}\lvert x\rvert$.
For a positive integer $n$, write $[z]_n$ for the representative of
$z\bmod n$ in $[-n/2,n/2)$, and define $\phi_\lambda(x)=[\lambda x]_q$.
We choose an odd $\lambda\bmod q$ such that, for every
$x\in A_{\mathrm{diff}}\setminus\{0\}$,
\begin{equation}
 m\nmid\phi_\lambda(x).
 \label{eq:rw-deterministic-avoidance}
\end{equation}
Write $x=2^vu$ with $u$ odd, and let $q'=2^{L-v}$.
Since $\phi_\lambda(x)=2^v[\lambda u]_{q'}$ and $m$ is odd,
$m\mid\phi_\lambda(x)$ if and only if $m\mid[\lambda u]_{q'}$.
For uniformly random odd $\lambda\bmod q$, the residue $[\lambda u]_{q'}$
is uniform among the odd residues modulo $q'$, so
$\Pr[m\mid\phi_\lambda(x)]<2/m$.
Let $N(\lambda)$ count the nonzero $x\in A_{\mathrm{diff}}$ for which
$m\mid\phi_\lambda(x)$. Then
\[
 \mathbb E_\lambda N(\lambda)
 <\frac{2\lvert A_{\mathrm{diff}}\rvert}{m}<\frac14.
\]

We fix the bits of $\lambda$ from the lowest to the highest.
At each step, choose the next bit to minimize the conditional expectation
of $N(\lambda)$.
After the lowest $t\ge 1$ bits have been fixed to $\lambda_0$, let
$h=\min\{t,L-v\}$ for each $x=2^vu$. Then
\[
 \Pr\bigl[m\mid\phi_\lambda(x)\bigm|\lambda\equiv\lambda_0\pmod{2^t}\bigr]
 =\frac{\lvert\{z\in\Z\cap[-q'/2,q'/2):
       m\mid z,\ z\equiv\lambda_0u\pmod{2^h}\}\rvert}{q'/2^h}.
\]
To count the numerator, write $z=mj$.
Since $m$ is odd, the congruence fixes one residue class of $j$ modulo
$2^h$, whose representatives in $-q'/2\le mj<q'/2$ can be counted by
integer division.
Thus we can compute the conditional expectations for the two choices of
the next bit in $\Oh(\lvert A_{\mathrm{diff}}\rvert)\polylog\Delta$ time.
Throughout, the expectation never increases, so after all bits are fixed,
$N(\lambda)=0$ and \eqref{eq:rw-deterministic-avoidance} holds.

Partition $[-q/2,q/2)$ into $16$ equal half-open intervals.
Choose an interval containing the most values $\phi_\lambda(a)$ for $a\in A$,
and let $A'$ consist of the corresponding elements of $A$.
Then $\lvert A'\rvert\ge\lvert A\rvert/16$.
Define $\psi\colon\Z\to\Z_m$ by $\psi(x)=\phi_\lambda(x)\bmod m$.
We claim that $\psi$ is a Freiman $8$-isomorphism on $A'$.
For two such sums, let $x\in A_{\mathrm{diff}}$ be their difference.
The $16$ values $\phi_\lambda(a)$ lie in one interval of length $q/16$,
so the difference of their sums has absolute value below $q/2$ and equals
$\phi_\lambda(x)$.
The sums of their $\psi$-images therefore agree exactly when
$m\mid\phi_\lambda(x)$, which by \eqref{eq:rw-deterministic-avoidance}
is equivalent to $x=0$.

Let $B=\psi(A')$ and $\alpha=\lvert B\rvert/m=\Omega_K(1)$.
For $R\subseteq\Z_m$ and $\varepsilon>0$, define
\[
 \operatorname{Bohr}_m(R,\varepsilon)
 =\{x\in\Z_m:\lVert rx/m\rVert_{\R/\Z}\le\varepsilon
       \text{ for every }r\in R\},
\]
where $\lVert\cdot\rVert_{\R/\Z}$ denotes distance to the nearest integer.
By \cite[Lemma~A.3]{rw24}, we can compute in $\tO(m)$ time a set
$R\subseteq\Z_m$ with $\lvert R\rvert<\alpha^{-2}=\Oh_K(1)$ such that
\[
 \operatorname{Bohr}_m(R,1/4)\subseteq 2B-2B.
\]

If $R\ne\varnothing$, \cite[Lemma~A.4]{rw24} computes a proper GAP
$P\subseteq\operatorname{Bohr}_m(R,1/4)$ of rank $\Oh_K(1)$ and size
$\Omega_K(m)$ in $\tO_K(m)$ time.
If $R=\varnothing$, take $P$ to be a centered interval of $m$ residues in $\Z_m$.
The Freiman $8$-isomorphism on $A'$ induces a Freiman $2$-isomorphism
$\psi\colon 2A'-2A'\to 2B-2B$.
Thus $Q=\psi^{-1}(P)\subseteq 2A'-2A'$ is a proper GAP with the same
rank and size as $P$.
Enumerating $2A'-2A'$ and tabulating $\psi$ computes $Q$ in
$\Oh_K(\lvert A\rvert)\polylog\Delta$ time.

Since
\[
 \lvert Q+A\rvert\le\lvert3A-2A\rvert
 =\Oh_K(\lvert A\rvert)=\Oh_K(\lvert Q\rvert),
\]
\cite[Lemma~A.5]{rw24} gives $X\subseteq A$ with $\lvert X\rvert=\Oh_K(1)$
and $A\subseteq Q-Q+X$.
Its running time is $\Oh_K(\lvert A\rvert)\polylog\Delta$, since
$Q-Q+A\subseteq5A-4A$.
Adjoin each $x\in X$ to $Q-Q$ as a direction with coefficient in $\{0,1\}$.
This gives a GAP $P_0\supseteq A$ of rank $\Oh_K(1)$,
volume $\Oh_K(\lvert A\rvert)$, and
$\maxbitlen(P_0)=\Oh_K(\log\Delta)$.
Indeed, doubling the coefficient widths of $Q$ and adjoining
$\lvert X\rvert=\Oh_K(1)$ binary directions multiply its volume by $\Oh_K(1)$.

Apply \cref{lem:reducegap} to $(P_0,A)$ to obtain an origin-symmetric
$4$-proper GAP $Q\supseteq A$ of rank $d=\Oh_K(1)$, and return $P=2Q$.
Then $P$ is origin-symmetric and $2$-proper and contains $A+A$.
Since $Q\subseteq P$ and $f_P=f_Q$, properness of $P$ gives
$f_P^{-1}(a)=f_Q^{-1}(a)$ for every $a\in A$.
Thus $\caL(f_P^{-1}(A))=\Z^d$.
By \cref{prop:growgap,lem:reducegap},
$\vol(P)\le 2^d\vol(Q)=\Oh_K(\lvert A\rvert)$.
The total running time is $\Oh_K(\lvert A\rvert)\polylog\Delta$.
\end{proof}

    \section{Proof of \texorpdfstring{\Cref*{thm:algorithmiccfp-highdim}}{Theorem \getrefnumber{thm:algorithmiccfp-highdim}}}
\label{app:vector-cfp-proof}

We deduce \cref{thm:algorithmiccfp-highdim} from
\cref{thm:algorithmiccfp} by encoding vectors as integers. 

\begin{proof}[Proof of \cref{thm:algorithmiccfp-highdim}]
We follow the reduction in \cite[Appendix]{pham2025sharp}.
We will apply \Cref{thm:algorithmiccfp} with exponent
$\beta_0\coloneqq 2 \beta k$ in place of $\beta$; let positive integers $C_0,d = \Oh_{\beta_0}(1) = \Oh_{\beta,k}(1)$ be the corresponding constants in \Cref{thm:algorithmiccfp}.

Let $H=2\lceil\Delta^2\rceil+1$,  and define
\[
\psi(a_1,\dots,a_k)=\sum_{j=1}^{k} H^{j-1} a_j,
\]
which is a linear map injective on $[-\Delta^2,\Delta^2]^k\cap \Z^k$.
Write $\psi^{-1}$ for the inverse of $\psi$ on the domain
$\psi([-\Delta^2,\Delta^2]^k\cap \Z^k)$, which 
can be easily computed in $\Oh_{\beta,k}(\polylog \Delta)$ time.

Then, $\psi(A)$ is a multiset of nonzero integers in
$[-\Delta_0,\Delta_0]$, where
\[
\Delta_0\coloneqq kH^{k-1} \Delta\ll_k \Delta^{2k-1}.
\]
By setting the constant
$C=C(\beta,k)$ in
\cref{thm:algorithmiccfp-highdim}
sufficiently large, we have
$\Delta_0\le \Delta^{2k}$,  and $2C_0 s \Delta\le 2C_0 \Delta^2/C\le \Delta^2$.
Hence $\Delta_0^{1/\beta_0} \le \Delta^{2k/\beta_0}=\Delta^{1/\beta}\le s$, and
\Cref{thm:algorithmiccfp} applies to $\psi(A)$ with parameters $\beta_0,s$, and $\Delta_0$.

\Cref{thm:algorithmiccfp} returns $A_0'\subseteq \hat A_0\subseteq \psi(A)$
with $\lvert A_0'\rvert\le C_0 s$, and an $s$-proper GAP $P_0\subseteq \Z$ of
rank $m\le d$. Write
$P_0=f_0(B\cap \Z^m)\subseteq \Z$, with $f_0$ linear 
and
$B=\prod_{i=1}^m[-w_i,w_i]$ with $w_i\in \Pos$.
Moreover, $\supp(\hat A_0)\subseteq P_0$ and $A_0'$ is $s$-generating
for $P_0$.

By injectivity, we have
$\hat A\coloneqq \psi^{-1}(\hat A_0)\subseteq A$ and $ A'\coloneqq \psi^{-1}(A_0')\subseteq \hat A$.

\paragraph*{Defining the lifted GAP.}
Since $A'_0$ is $s$-generating for $P_0$  (\cref{def:generating}), there is
$\boldsymbol{\tau}\in \Z^m$ such that, with $Q\coloneqq (\boldsymbol{\tau}+sB)\cap \Z^m$,
$f_0(Q)\subseteq \Sigma(A_0')=\psi(\Sigma(A'))$.
Recall $\psi$ is injective on $[-\Delta^2,\Delta^2]^k\cap \Z^k$ and $\Delta^2 \ge 2C_0 s \Delta$.
Since $\Sigma(A')\subseteq [-C_0 s \Delta,C_0 s \Delta]^k$, for each
$\bfx\in Q$ we have
$F(\bfx)\coloneqq \psi^{-1}(f_0(\bfx))\in \Sigma(A')$.
Then, for $\bfx,\bfx+\bfe_i\in Q$, 
$\psi\bigl(F(\bfx+\bfe_i)-F(\bfx)\bigr)=f_0(\bfe_i)$.
Since $\|F(\bfx+\bfe_i)-F(\bfx)\|_\infty\le 2C_0 s \Delta\le \Delta^2$,
injectivity of $\psi$ gives
\[F(\bfx+\bfe_i)-F(\bfx)=\psi^{-1}(f_0(\bfe_i)) \eqqcolon \bfg_i\]
for $\bfx,\bfx+\bfe_i\in Q$.
Compute $\bfg_i$ for $i\in [m]$ and
define $f\colon \R^m\to \R^k$ by $f(\bfz)=\sum_i z_i \bfg_i$.
By definition, we have 
 $f_0=\psi \circ f$.
We output the origin-symmetric GAP $P\coloneqq f(B\cap \Z^m)$. It remains to verify that $P$ satisfies the claimed properties.

By $s$-properness of $P_0$, 
$f_0$ is injective on $(sB)\cap \Z^m$.
Since $f_0=\psi \circ f$, this means 
$P$ is $s$-proper.

For any $\bfx\in (sB)\cap \Z^m$,
inductively applying the displayed equality above yields \[f(\bfx) = F(\boldsymbol{\tau}+\bfx)-F(\boldsymbol{\tau}).\] In particular,  $\|f(\bfx)\|_{\infty} \le 2C_0 s \Delta$.
Fix $a_0\in \hat A_0$. Let $\bfa=\psi^{-1}(a_0)$ and
$\bfz=f_0^{-1}(a_0)\in B\cap \Z^m$. Then
$\psi(f(\bfz))=a_0=\psi(\bfa)$, and both
$f(\bfz)$ and $\bfa$ lie in
$[-\Delta^2,\Delta^2]^k$. Injectivity gives $f(\bfz)=\bfa$, so
$\supp(\hat A)\subseteq P$.

Recall $A'_0$ is $s$-generating for $P_0$, i.e., $f_0^{-1}(A_0')$ is $s$-generating for
$B\cap \Z^m$. 
Applying $f$ therefore shows that
$f(f_0^{-1}(A_0'))=\psi^{-1}(A_0')=A'$ is $s$-generating for $P$.

Since $\lvert A'\rvert=\lvert A_0'\rvert$ and $\lvert\hat A\rvert=\lvert\hat A_0\rvert$, the bounds in 
    \cref{thm:algorithmiccfp} give $\lvert A'\rvert \le C_0 s$ and $\lvert A\rvert- \lvert \hat A\rvert \le C_0 s \log \Delta_0$.
 Since $\log \Delta_0\ll_{\beta,k} \log \Delta$, enlarging $C=C(\beta,k)$ gives 
$\lvert A'\rvert \le Cs$ and $\lvert A\rvert- \lvert \hat A\rvert \le C s \log \Delta$ as required.

\paragraph*{Time and coordinate bounds.}
Since
$s(\supp(\psi(A))\cup \{0\})
 =\psi\bigl(s(\supp(A)\cup \{\zero\})\bigr)$,
the application of \cref{thm:algorithmiccfp} and the remaining operations take
$\tO_{\beta,k}(s\lvert s(\supp(A)\cup \{\zero\})\rvert)$ time.

For every $i\in [m]$, we have
$2sw_i \bfg_i=F(\boldsymbol{\tau}+sw_i \bfe_i)-F(\boldsymbol{\tau}-sw_i \bfe_i)$ by the previous displayed equality.
Thus $w_i\|\bfg_i\|_\infty\le C_0 \Delta$.
Hence, $w_i\le C_0 \Delta$ since $\bfg_i\neq \zero$.

Denote $U_0=\supp(\hat A_0)\cup \{0\}$. 
By $s$-properness, $s$-generation, $\lvert A_0'\rvert\le C_0 s$, and
\cref{lem:cfp-large-sumsets},
\[
 \lvert sP\rvert=\lvert sP_0\rvert
 \le \lvert\Sigma(A_0')\rvert\le \lvert C_0 sU_0\rvert
 \ll_{\beta,k}\lvert sU_0\rvert = \lvert \psi(s(\supp(\hat A)\cup \{\zero\}))\rvert,
\]
proving the required bound on $\lvert sP\rvert$ (after enlarging $C = C(\beta,k)$).
\end{proof}

    \section{Proof of \texorpdfstring{\Cref*{lem:box-subset-sums}}{Lemma \getrefnumber{lem:box-subset-sums}}}
\label{app:box-subset-sums}

\begin{proof}[Proof of \cref{lem:box-subset-sums}]
  Let $X$ be the input multiset with $\supp(X)\subseteq B \cap \Z^d$ where $B = \prod_{j=1}^d[0,t_j]$. We may assume $X$ does not contain $\zero$, and has maximum multiplicity at most $\max_{j=1}^d t_j$.
 First, by a standard binary bundling argument (see e.g., \cite{KoiliarisX19}), we can reduce to an equivalent instance with size $n = \lvert X\rvert = \tO_d(\lvert B\cap \Z^d\rvert)$. This is done by going over each input item $\bfx \in \supp(X)$ of multiplicity $\mu(\bfx)$,  choosing a logarithmic number of multiplicities $m_1,\dots,m_k$ such that $\Sigma(\{m_1,\dots,m_k\})= \{0,1,\dots,\mu(\bfx)\}$, and creating new items $(m_i \bfx)_i$.
 Denote the new equivalent input instance still by $X$, and write $X = \{\bfx_1,\dots,\bfx_n\}$.

 We now encode the input vectors as integers.
Let $p_1,\ldots,p_d$ be distinct fixed primes, and let $q_j$ be the
smallest power of $p_j$ at least $d(t_j+1)$.
Define
\[
 P=\prod_{j=1}^d q_j,
\,\,\, E(\bfa)=\sum_{j=1}^d a_j \frac{P}{q_j},\,\,\, \text{ and } 
 H=E(\bft).
\]
Then $P=\Oh_d(\prod_j(t_j+1))$ and $H/P = \sum_{j=1}^d \frac{t_j}{q_j}<1$.

For any $\bfa \in B\cap \Z^d$, we have $0\le E(\bfa) \le E(\bft) = H$, where $E(X) =\{ E(\bfx_1),\dots,E(\bfx_n)\}$.
The encoding $E(\cdot)$ is linear: $E(\bfx+\bfy) = E(\bfx)+E(\bfy)$, so for any $\bfa \in \Sigma(X)\cap B$ it holds that $E(\bfa) \in \Sigma(E(X)) \cap [0,H]$.

Conversely, we show that, 
for any $h\in [0,H]$, there can be at most one $\bfa\in \Nonneg^d$ satisfying $E(\bfa) = h$. 
To see this, let $\bfa$ be such a vector.
Then, from $a_j P/q_j\le E(\bfa) \le H<P$, we get $0\le a_j<q_j$ for every $j$.
By definition of $E(\cdot)$, we have $E(\bfa)\equiv a_j \frac{P}{q_j} \pmod{q_j}$ for every $j$.
Since the $q_j$ are pairwise coprime, we thus have 
$a_j=h(P/q_j)^{-1} \bmod q_j$,
uniquely determining $\bfa$.

Thus, we can apply Chan's deterministic
all-targets
algorithm for 
 \textsc{Subset Sum}
\cite{chan26} to input integers $E(X)$ with target
$H$, and obtain $\Sigma(E(X))\cap [0,H]$ in $\tO(n+H)$ time.
Then, using the inversion above, we can obtain $\Sigma(X) \cap B$, thereby solving our vector \textsc{Subset Sum} instance.
By a standard divide-and-conquer approach, Chan's algorithm can support finding a witnessing submultiset attaining the target sum. This handles the witness-recovery part, proving
\cref{lem:box-subset-sums}.
\end{proof}

    \section{AP Decompositions of All Subset Sums}
\label{app:subset-sum-representation}

Recall that the AP decomposition (\cref{def:ap-decomposition}) of an integer set $S$ with common difference
$s$ has size $\lvert S\setminus (S+s)\rvert$.
We prove the following lemma. Here $\ln x = \log_e x$.
\begin{lemma}
\label{lem:aprepresentation}
Let $A$ be a nonempty multiset of nonzero integers. For some
$a\in \supp A$, the AP decomposition of $\Sigma(A)$ with common difference
$\lvert a\rvert$ has size at most $\frac{\lvert\Sigma(A)\rvert\ln \lvert\Sigma(A)\rvert}{\lvert A\rvert}$.
\end{lemma}

In particular, if $\lvert A\rvert=n$ and $\supp (A)\subseteq \{1,\dots,\Delta\}$, then $\Sigma(A)$ has an AP
decomposition of size $\Oh(\Delta \log(n \Delta))$.

It remains open whether such an AP decomposition can be computed in
$\Oh((n+\Delta)^{1+o(1)})$ time.
In our \textsc{Subset Sum} algorithm (\cref{thm:main-main}), we compute an AP decomposition for each group (\cref{lem:oneclass}), but when we combine the answers from all the groups in \cref{alg:subset-sum-main}, we restrict attention to shorter intervals, and consequently lose information about the full $\Sigma(A)$.

To prove \cref{lem:aprepresentation}, we use the submultiplicativity theorem of Gyarmati,
Matolcsi, and Ruzsa \cite{submult}.
The same theorem was also used by Bringmann, Fischer, and Nakos \cite{BringmannFN25} to design an improved output-sensitive algorithm for computing $\Sigma(A) \cap [0,t]$ given a multiset $A$ of positive integers.

\begin{theorem}[{\cite[Theorem~1.2]{submult}}]
\label{thm:submult}
Let $n\ge 2$ and let $A_1,\ldots,A_n\subseteq \Z$ be finite and nonempty.
Define $S=\sum_{j=1}^n A_j$ and
$S_i=\sum_{j\in [n]\setminus \{i\}} A_j$.
Then
\[
 \lvert S\rvert\le \left(\prod_{i=1}^n\lvert S_i\rvert\right)^{1/(n-1)}.
\]
\end{theorem}

\begin{proof}
[Proof of \cref{lem:aprepresentation}]
Let $n=\lvert A\rvert$. The case $n=1$ is immediate. For $n\ge 2$, write
$A=\{a_1,\ldots,a_n\}$, $S=\Sigma(A)$, and
$S_i=\Sigma(A\setminus \{a_i\})$.
Applying \cref{thm:submult} to $A_i=\{0,a_i\}$ gives
$\prod_{i=1}^n\lvert S_i\rvert\ge \lvert S\rvert^{n-1}$.
Taking logarithms and using
$\ln \frac{1}{x} \ge 1-x $ for $0<x\le 1$ gives
\[
 \ln \lvert S\rvert
 \ge \sum_{i=1}^n \ln \frac{\lvert S\rvert}{\lvert S_i\rvert}
 \ge \sum_{i=1}^n\left(1-\frac{\lvert S_i\rvert}{\lvert S\rvert}\right).
\]
By averaging, there exists $i$ such that $\frac{1}{n} {\ln \lvert S\rvert} \ge 1 - \lvert S_i\rvert/\lvert S\rvert$, that is, $\lvert S\rvert-\lvert S_i\rvert \le \frac{1}{n}{\lvert S\rvert\ln \lvert S\rvert}$.
Since $S=S_i\cup (S_i+a_i)$, we have $\lvert S\cap (S+a_i)\rvert \ge \lvert S_i+a_i\rvert = \lvert S_i\rvert$, and therefore
\[\lvert S\setminus (S+a_i)\rvert = \lvert S\rvert - \lvert S \cap (S+a_i)\rvert \le \lvert S\rvert-\lvert S_i\rvert \le \frac{1}{n}{\lvert S\rvert\ln \lvert S\rvert},\] as desired.
\end{proof}

We now show that an efficient algorithm for computing an AP decomposition of $\Sigma(A)$ would imply an efficient algorithm for the following problem posed by W\k{e}grzycki:

\defproblem{prob:bounded-support-subset-sum}{\textsc{Bounded Support Subset Sum}}{
Positive integers $a_1,\ldots,a_n$ and a target integer $t$.
}{
Compute $\max \{\lvert I\rvert:I\subseteq [n],\ \sum_{i\in I} a_i=t\}$,
or report that no such $I$ exists.
}

W\k{e}grzycki \cite[Problem~2]{AlgUWOpenProblems25} asked whether this problem
can be solved in $\Oh(n^{1-\delta} t+n^{\Oh(1)})$ time for some absolute
constant $\delta>0$.

\begin{proposition}
Suppose there is
a $T(n,\Delta)$-time algorithm for computing
an AP decomposition of $\Sigma(A)$ given any  multiset $A$ of at most
$n$ integers in $[1,\Delta]$.
Then, {\rm\textsc{Bounded Support Subset Sum}}
can be solved in
$\tO\!\left( T\!\left(n,3t \sqrt n\right)
+ n+t \sqrt n \right)$
time.
\end{proposition}

In particular, if $T(n,\Delta) \le \tO(n+\Delta)$,
then \textsc{Bounded Support Subset Sum} can be solved in
 $\tO(n+t \sqrt n)$ time.

\begin{proof}
Let $K=\lceil\sqrt n\rceil$.
Call $a_i \in [1, t/K]$ the small items, and $a_i \in (t/K, t]$ the large items.
Let $L[s]$ denote the cardinality of the largest subset of large items with sum exactly $s$ (where $L[s]=-\infty$ if no such subset exists).
Similarly define $S[s]$ for small items.
We will compute these two arrays. Then, the answer is clearly $\max_{0\le s\le t}\bigl\{L[s]+S[t-s]\bigr\}$ which can be computed in $\Oh(t)$ additional time.

\paragraph*{Large items.}
By definition, a subset of sum at most $t$ must contain fewer than $K$ large items.
Define the multiset $B = \{Ka_i + 1: i\in [n] \text{ is a large item}\}$ and $t' = tK+K$. 
Note that any $B'\subseteq B$ of size $\lvert B'\rvert\ge K$ has sum $\sum_{b\in B'} b > \lvert B'\rvert(K\cdot (t/K) + 1) \ge  t'$. Use Bringmann's \textsc{Subset Sum} algorithm \cite{Bringmann17,chan26} to compute $\Sigma(B) \cap [0,t']$ in $\tO(n+t')=\tO(n+tK) $ time.
Then, for every $s\in \Z\cap [0,t]$ and $k\in \Z\cap [0,K)$, we have $Ks+k \in \Sigma(B) \cap [0,t']$ if and only if there is a subset of exactly $k$ large items that sums to $s$.
Thus we immediately obtain the array $L[s]$ for all $s\in [0,t]$.

\paragraph*{Small items.}
Define $\Delta = 3t \sqrt{n}$,  and the multiset $E = \{2na_i + 1: i\in [n] \text{ is a small item}\}$ which contains positive integers bounded by $2n\cdot t/K +1\le \Delta$.
Then, for every $s\in \Z\cap [0,t]$ and $k\in \Z\cap [0,n]$, $2ns+k \in \Sigma(E)$ if and only if there is a subset of exactly $k$ small items that sums to $s$.
Thus, for each possible sum $s\in \Z\cap [0,t]$, the desired maximum cardinality is $S[s] = \max \{v - 2ns : v\in \Sigma(E) \cap [2ns,2ns + n]\}$. To compute all $S[s]$, 
it suffices to answer the predecessor queries $\max (\Sigma(E)\cap [0,u])$ for all $u\in Q\coloneqq  \{2ns+n: s\in [0,t]\cap \Z\}$.

Compute an AP decomposition $\caD$ of $\Sigma(E)$ in $T(n,\Delta) = T(n,3t \sqrt n)$ time, and let $p>0$ denote the common difference.
We answer the predecessor queries offline in time near-linear in $\lvert Q\rvert = t+1$ and $\lvert\caD\rvert$, as follows. 
Sort all endpoints in $\caD$ and all query points in $Q$ together in increasing order, and treat them as time steps. 
Iterate over the time steps in increasing order, and maintain a dynamic set $P_u$ (using a balanced search tree) such that at time step $u$, \[P_u = \{ a \bmod p :  \{a,a+p,\dots,b\} \in \caD \text{ and } a\le u < b \}.\] Also maintain \[r_u = \max \{b :  \{a,a+p,\dots,b\} \in \caD \text{ and } b\le  u \}.\]
At query time $u\in Q$, the predecessor of $u$ in $\Sigma(E)$ is precisely
\[ \max \big\{ r_u, \max \{ u-((u-a) \bmod p) : a\in P_u \}\big \}, \]
which can be computed in logarithmic time by making $\Oh(1)$ predecessor queries to $P_u$.
The dynamic set $P_u$ can be maintained when the time step $u$ reaches an endpoint of an arithmetic progression in $\caD$, by inserting or deleting the residue of that progression modulo $p$.

The total time is therefore
$
 \tO\bigl(T(n,\Delta)+ n+tK \bigr) = \tO( T(n,3t \sqrt{n})+n+t \sqrt{n})$.
\end{proof}

  \bibliographystyle{alphaurl}
  \bibliography{main}

@inproceedings{BringmannW21,
  author    = {Karl Bringmann and
               Philip Wellnitz},
  title     = {On Near-Linear-Time Algorithms for Dense Subset Sum},
  booktitle = {Proceedings of the 2021 {ACM-SIAM} Symposium on Discrete Algorithms, {SODA} 2021},
  pages     = {1777--1796},
  publisher = {{SIAM}},
  year      = {2021},
  doi       = {10.1137/1.9781611976465.107},
  bibsource = {dblp computer science bibliography, https://dblp.org}
}

@inproceedings{Bringmann17,
  author    = {Karl Bringmann},
  title     = {A Near-Linear Pseudopolynomial Time Algorithm for {Subset Sum}},
  booktitle = {Proceedings of the Twenty-Eighth Annual {ACM-SIAM} Symposium on Discrete
               Algorithms, {SODA} 2017},
  pages     = {1073--1084},
  publisher = {{SIAM}},
  year      = {2017},
  doi       = {10.1137/1.9781611974782.69},
  bibsource = {dblp computer science bibliography, https://dblp.org}
}

@article{fox,
  title={Subset sums, completeness and colorings},
  author={Conlon, David and Fox, Jacob and Pham, Huy Tuan},
  journal={arXiv preprint arXiv:2104.14766},
  year={2021},
    eprinttype    = {arXiv},
  eprint       = {2104.14766},
}

@article {Sarkozy2,
    AUTHOR = {S{\'{a}}rk\"{o}zy, A.},
     TITLE = {Finite addition theorems. {II}},
   JOURNAL = {J. Number Theory},
  FJOURNAL = {Journal of Number Theory},
    VOLUME = {48},
      YEAR = {1994},
    NUMBER = {2},
     PAGES = {197--218},
      ISSN = {0022-314X},
   MRCLASS = {11B13},
  MRNUMBER = {1285539},
MRREVIEWER = {Georges Grekos},
       DOI = {10.1006/jnth.1994.1062},
       URL = {https://doi.org/10.1006/jnth.1994.1062},
}

@article {freiman93,
    AUTHOR = {Freiman, Gregory A.},
     TITLE = {New analytical results in subset-sum problem},
      NOTE = {Combinatorics and algorithms (Jerusalem, 1988)},
   JOURNAL = {Discrete Math.},
  FJOURNAL = {Discrete Mathematics},
    VOLUME = {114},
      YEAR = {1993},
    NUMBER = {1-3},
     PAGES = {205--217},
      ISSN = {0012-365X},
   MRCLASS = {11B25 (11N25)},
  MRNUMBER = {1217753},
       DOI = {10.1016/0012-365X(93)90367-3},
       URL = {https://doi.org/10.1016/0012-365X(93)90367-3},
}

@article{GalilM91,
  author    = {Zvi Galil and
               Oded Margalit},
  title     = {An Almost Linear-Time Algorithm for the Dense Subset-Sum Problem},
  journal   = {{SIAM} J. Comput.},
  volume    = {20},
  number    = {6},
  pages     = {1157--1189},
  year      = {1991},
  doi       = {10.1137/0220072},
  bibsource = {dblp computer science bibliography, https://dblp.org}
}

@article{Pisinger99,
  author    = {David Pisinger},
  title     = {Linear Time Algorithms for Knapsack Problems with Bounded Weights},
  journal   = {J. Algorithms},
  volume    = {33},
  number    = {1},
  pages     = {1--14},
  year      = {1999},
  doi       = {10.1006/jagm.1999.1034},
  bibsource = {dblp computer science bibliography, https://dblp.org}
}

@inproceedings{icalp21,
  author    = {Adam Polak and
               Lars Rohwedder and
               Karol W\k{e}grzycki},
  title     = {Knapsack and Subset Sum with Small Items},
  booktitle = {48th International Colloquium on Automata, Languages, and Programming, {ICALP} 2021},
  series    = {LIPIcs},
  volume    = {198},
  pages     = {106:1--106:19},
  publisher = {Schloss Dagstuhl - Leibniz-Zentrum f{\"{u}}r Informatik},
  year      = {2021},
  doi       = {10.4230/LIPIcs.ICALP.2021.106},
  bibsource = {dblp computer science bibliography, https://dblp.org}
}

@article{KoiliarisX19,
  author    = {Konstantinos Koiliaris and
               Chao Xu},
  title     = {Faster Pseudopolynomial Time Algorithms for Subset Sum},
  journal   = {{ACM} Trans. Algorithms},
  volume    = {15},
  number    = {3},
  pages     = {40:1--40:20},
  year      = {2019},
  doi       = {10.1145/3329863},
  bibsource = {dblp computer science bibliography, https://dblp.org}
}

@article{EisenbrandW20,
  author    = {Friedrich Eisenbrand and
               Robert Weismantel},
  title     = {Proximity Results and Faster Algorithms for Integer Programming Using
               the {Steinitz} Lemma},
  journal   = {{ACM} Trans. Algorithms},
  volume    = {16},
  number    = {1},
  pages     = {5:1--5:14},
  year      = {2020},
  doi       = {10.1145/3340322},
  bibsource = {dblp computer science bibliography, https://dblp.org}
}

@inproceedings{BringmannC22,
  author    = {Karl Bringmann and
               Alejandro Cassis},
  title     = {Faster Knapsack Algorithms via Bounded Monotone Min-Plus-Convolution},
  booktitle = {49th International Colloquium on Automata, Languages, and Programming, {ICALP} 2022},
  series    = {LIPIcs},
  volume    = {229},
  pages     = {31:1--31:21},
  publisher = {Schloss Dagstuhl - Leibniz-Zentrum f{\"{u}}r Informatik},
  year      = {2022},
  doi       = {10.4230/LIPIcs.ICALP.2022.31},
  bibsource = {dblp computer science bibliography, https://dblp.org}
}

@inproceedings{AxiotisBJTW19,
  author    = {Kyriakos Axiotis and
               Arturs Backurs and
               Ce Jin and
               Christos Tzamos and
               Hongxun Wu},
  title     = {{F}ast Modular Subset Sum using Linear Sketching},
  booktitle = {Proceedings of the Thirtieth Annual {ACM-SIAM} Symposium on Discrete Algorithms, {SODA} 2019},
  pages     = {58--69},
  publisher = {{SIAM}},
  year      = {2019},
  doi       = {10.1137/1.9781611975482.4},
  bibsource = {dblp computer science bibliography, https://dblp.org}
}

@article{AbboudBHS22jcss,
  author    = {Amir Abboud and
               Karl Bringmann and
               Danny Hermelin and
               Dvir Shabtay},
  title     = {Scheduling lower bounds via {AND} subset sum},
  journal   = {J. Comput. Syst. Sci.},
  volume    = {127},
  pages     = {29--40},
  year      = {2022},
  doi       = {10.1016/j.jcss.2022.01.005},
  bibsource = {dblp computer science bibliography, https://dblp.org}
}

@Book{Bellman1957,
  author =       "Bellman, Richard",
  title =        "Dynamic Programming",
  publisher =    "Princeton University Press",
  year =         "1957",
  address =   "Princeton, NJ, USA",
}

@incollection{chaimovichsurvey,
     author = {Chaimovich, Mark},
     title = {New structural approach to integer programming: a survey},
     booktitle = {Structure theory of set addition},
     series = {Ast\'erisque},
     pages = {341--362},
     year = {1999},
     publisher = {Soci\'et\'e math\'ematique de France},
     number = {258},
     mrnumber = {1701209},
     zbl = {0987.90060},
     language = {en},
     url = {https://www.numdam.org/item/AST_1999__258__341_0/}
}

@incollection{chaimovichalgo,
     author = {Chaimovich, Mark},
     title = {New algorithm for dense subset-sum problem},
     booktitle = {Structure theory of set addition},
     series = {Ast\'erisque},
     pages = {363--373},
     year = {1999},
     publisher = {Soci\'et\'e math\'ematique de France},
     number = {258},
     mrnumber = {1701210},
     zbl = {0987.90061},
     language = {en},
     url = {https://www.numdam.org/item/AST_1999__258__363_0/}
}

@inproceedings {freimanalgo,
    AUTHOR = {Freiman, G. A.},
     TITLE = {Subset-sum problem with different summands},
 BOOKTITLE = {Proceedings of the {T}wentieth {S}outheastern {C}onference on
              {C}ombinatorics, {G}raph {T}heory, and {C}omputing ({B}oca
              {R}aton, {FL}, 1989)},
   JOURNAL = {Congr. Numer.},
 FJOURNAL = {Congressus Numerantium. A Conference Journal on Numerical
              Themes},
    SERIES = {Congressus Numerantium},
    VOLUME = {70},
      YEAR = {1990},
     PAGES = {207--215},
      ISSN = {0384-9864},
   MRCLASS = {11D04 (11B83 11Y50 90C09)},
  MRNUMBER = {1041604},
MRREVIEWER = {Marlies Borchardt},
}

@article {freimanparti,
    AUTHOR = {Freiman, Gregory A.},
     TITLE = {On extremal additive problems of {P}aul {E}rd{\H{o}}s},
   JOURNAL = {Ars Combin.},
  FJOURNAL = {Ars Combinatoria},
    VOLUME = {26},
      YEAR = {1988},
    NUMBER = {B},
     PAGES = {93--114},
      ISSN = {0381-7032},
   MRCLASS = {11P65 (90C10)},
  MRNUMBER = {990542},
MRREVIEWER = {B\'{e}la Vizv\'{a}ri},
}

@article {cfgalgo,
    AUTHOR = {Chaimovich, Mark and Freiman, Gregory and Galil, Zvi},
     TITLE = {Solving dense subset-sum problems by using analytical number
              theory},
   JOURNAL = {J. Complexity},
  FJOURNAL = {Journal of Complexity},
    VOLUME = {5},
      YEAR = {1989},
    NUMBER = {3},
     PAGES = {271--282},
      ISSN = {0885-064X},
   MRCLASS = {11Y35 (11B13 68Q20)},
  MRNUMBER = {1018019},
       DOI = {10.1016/0885-064X(89)90025-3},
       URL = {https://doi.org/10.1016/0885-064X(89)90025-3},
}

@inproceedings{BringmannN21icalp,
  author    = {Karl Bringmann and
               Vasileios Nakos},
  title     = {Fast {$n$}-Fold {Boolean} Convolution via Additive Combinatorics},
  booktitle = {48th International Colloquium on Automata, Languages, and Programming, {ICALP} 2021},
  series    = {LIPIcs},
  volume    = {198},
  pages     = {41:1--41:17},
  publisher = {Schloss Dagstuhl - Leibniz-Zentrum f{\"{u}}r Informatik},
  year      = {2021},
  doi       = {10.4230/LIPIcs.ICALP.2021.41},
  bibsource = {dblp computer science bibliography, https://dblp.org}
}

@article{horowitz1974computing,
  title={Computing Partitions with Applications to the Knapsack Problem},
  author={Horowitz, Ellis and Sahni, Sartaj},
  journal={Journal of the ACM},
  volume={21},
  number={2},
  pages={277--292},
  year={1974},
  doi       = {10.1145/321812.321823},
}

@article{AbboudBHS22,
  author    = {Amir Abboud and
               Karl Bringmann and
               Danny Hermelin and
               Dvir Shabtay},
  title     = {{SETH}-based Lower Bounds for Subset Sum and Bicriteria Path},
  journal   = {{ACM} Trans. Algorithms},
  volume    = {18},
  number    = {1},
  pages     = {6:1--6:22},
  year      = {2022},
  doi       = {10.1145/3450524},
  bibsource = {dblp computer science bibliography, https://dblp.org}
}

@inproceedings{ChenLMZsoda24,
  author       = {Lin Chen and
                  Jiayi Lian and
                  Yuchen Mao and
                  Guochuan Zhang},
  title        = {Faster Algorithms for Bounded Knapsack and Bounded Subset Sum Via
                  Fine-Grained Proximity Results},
  booktitle    = {Proceedings of the 2024 {ACM-SIAM} Symposium on Discrete Algorithms,
                  {SODA} 2024, Alexandria, VA, USA, January 7-10, 2024},
  pages        = {4828--4848},
  publisher    = {{SIAM}},
  year         = {2024},
  doi          = {10.1137/1.9781611977912.171},
  bibsource    = {dblp computer science bibliography, https://dblp.org}
}

@inproceedings{Bringmann24,
  author       = {Karl Bringmann},
  title        = {Knapsack with Small Items in Near-Quadratic Time},
  booktitle    = {Proceedings of the 56th Annual {ACM} Symposium on Theory of Computing,
                  {STOC} 2024, Vancouver, BC, Canada, June 24-28, 2024},
  pages        = {259--270},
  publisher    = {{ACM}},
  year         = {2024},
  doi          = {10.1145/3618260.3649719},
  bibsource    = {dblp computer science bibliography, https://dblp.org}
}

@inproceedings{xiaomaofptas,
  author       = {Xiao Mao},
  title        = {{$(1-\epsilon)$}-Approximation of Knapsack in Nearly Quadratic
                  Time},
  booktitle    = {Proceedings of the 56th Annual {ACM} Symposium on Theory of Computing,
                  {STOC} 2024, Vancouver, BC, Canada, June 24-28, 2024},
  pages        = {295--306},
  publisher    = {{ACM}},
  year         = {2024},
  doi          = {10.1145/3618260.3649677},
  bibsource    = {dblp computer science bibliography, https://dblp.org}
}

@article {szemeredivu-apsumsets,
    AUTHOR = {Szemer\'{e}di, E. and Vu, V.},
     TITLE = {Long arithmetic progressions in sumsets: thresholds and
              bounds},
   JOURNAL = {J. Amer. Math. Soc.},
  FJOURNAL = {Journal of the American Mathematical Society},
    VOLUME = {19},
      YEAR = {2006},
    NUMBER = {1},
     PAGES = {119--169},
      ISSN = {0894-0347},
   MRCLASS = {11B25 (11B75 11P70)},
  MRNUMBER = {2169044},
MRREVIEWER = {Ben Joseph Green},
       DOI = {10.1090/S0894-0347-05-00502-3},
       URL = {https://doi.org/10.1090/S0894-0347-05-00502-3},
}

@inproceedings{clmzfocs24,
  author       = {Lin Chen and
                  Jiayi Lian and
                  Yuchen Mao and
                  Guochuan Zhang},
  title        = {An Improved Pseudopolynomial Time Algorithm for Subset Sum},
  booktitle    = {65th {IEEE} Annual Symposium on Foundations of Computer Science, {FOCS}
                  2024, Chicago, IL, USA, October 27-30, 2024},
  pages        = {2202--2216},
  publisher    = {{IEEE}},
  year         = {2024},
  doi          = {10.1109/FOCS61266.2024.00129},
  bibsource    = {dblp computer science bibliography, https://dblp.org}
}

@article{chen2026journal,
  title={An Improved Pseudopolynomial Time Algorithm for Subset Sum},
  author={Chen, Lin and Lian, Jiayi and Mao, Yuchen and Zhang, Guochuan},
  journal={Mathematical Programming},
  year={2026},
  doi={10.1007/s10107-026-02348-y},
  publisher={Springer}
}

@inproceedings{KoiliarisX17proc,
  author       = {Konstantinos Koiliaris and
                  Chao Xu},
  title        = {A Faster Pseudopolynomial Time Algorithm for Subset Sum},
  booktitle    = {Proceedings of the Twenty-Eighth Annual {ACM-SIAM} Symposium on Discrete
                  Algorithms, {SODA} 2017, Barcelona, Spain, Hotel Porta Fira, January
                  16-19},
  pages        = {1062--1072},
  publisher    = {{SIAM}},
  year         = {2017},
  doi          = {10.1137/1.9781611974782.68},
  bibsource    = {dblp computer science bibliography, https://dblp.org}
}

@article{conlon2023homogeneous,
  title={Homogeneous structures in subset sums and non-averaging sets},
  author={Conlon, David and Fox, Jacob and Pham, Huy Tuan},
  journal={arXiv preprint},
  year={2023},
  eprinttype    = {arXiv},
  eprint       = {2311.01416v1}
}

@book{schneider2014convex,
  address={Cambridge}, edition={2nd expanded}, series={Encyclopedia of Mathematics and its Applications}, volume={151}, title={Convex Bodies: The {Brunn--Minkowski} Theory}, publisher={Cambridge University Press}, author={Schneider, Rolf}, year={2014}, collection={Encyclopedia of Mathematics and its Applications},
   doi = {10.1017/CBO9781139003858}
}

@inproceedings{rw24,
  author       = {Tim Randolph and
                  Karol W\k{e}grzycki},
  title        = {Parameterized Algorithms on Integer Sets with Small Doubling: {Integer
                  Programming}, {Subset Sum} and {k-SUM}},
  booktitle    = {32nd Annual European Symposium on Algorithms, {ESA} 2024, September
                  2-4, 2024, Royal Holloway, London, United Kingdom},
  series       = {LIPIcs},
  volume       = {308},
  pages        = {96:1--96:19},
  publisher    = {Schloss Dagstuhl - Leibniz-Zentrum f{\"{u}}r Informatik},
  year         = {2024},
  note         = {Full version: \href{https://arxiv.org/abs/2407.18228}{\path{arXiv:2407.18228}}},
  doi          = {10.4230/LIPIcs.ESA.2024.96},
  bibsource    = {dblp computer science bibliography, https://dblp.org}
}

@incollection{bilu,
     author = {Bilu, Yuri},
     title = {Structure of sets with small sumset},
     booktitle = {Structure theory of set addition},
     series = {Ast\'erisque},
     pages = {77--108},
     year = {1999},
     publisher = {Soci\'et\'e math\'ematique de France},
     number = {258},
     mrnumber = {1701189},
     zbl = {0946.11004},
     language = {en},
     url = {https://www.numdam.org/item/AST_1999__258__77_0/}
}

@book{taovu06, place={Cambridge}, series={Cambridge Studies in Advanced Mathematics}, title={Additive Combinatorics}, publisher={Cambridge University Press}, author={Tao, Terence and Vu, Van H.}, year={2006}, collection={Cambridge Studies in Advanced Mathematics},
  volume={105},
  doi={10.1017/CBO9780511755149}
}

@inproceedings{BringmannFN22,
  author       = {Karl Bringmann and
                  Nick Fischer and
                  Vasileios Nakos},
  title        = {Deterministic and {Las Vegas} Algorithms for Sparse Nonnegative Convolution},
  booktitle    = {Proc. 2022 {ACM-SIAM} Symposium on Discrete Algorithms (SODA)},
  pages        = {3069--3090},
  year         = {2022},
  doi          = {10.1137/1.9781611977073.119},
  bibsource    = {dblp computer science bibliography, https://dblp.org}
}

@article {LLL,
    AUTHOR = {Lenstra, A. K. and Lenstra, Jr., H. W. and Lov\'{a}sz, L.},
     TITLE = {Factoring polynomials with rational coefficients},
   JOURNAL = {Math. Ann.},
  FJOURNAL = {Mathematische Annalen},
    VOLUME = {261},
      YEAR = {1982},
    NUMBER = {4},
     PAGES = {515--534},
      ISSN = {0025-5831},
   MRCLASS = {12-04 (12A20 68C20 68C25)},
  MRNUMBER = {682664},
MRREVIEWER = {Daniel Lazard},
       DOI = {10.1007/BF01457454},
       URL = {https://doi.org/10.1007/BF01457454},
}

@article {lev,
    AUTHOR = {Lev, V. F.},
     TITLE = {Consecutive integers in high-multiplicity sumsets},
   JOURNAL = {Acta Math. Hungar.},
  FJOURNAL = {Acta Mathematica Hungarica},
    VOLUME = {129},
      YEAR = {2010},
    NUMBER = {3},
     PAGES = {245--253},
      ISSN = {0236-5294,1588-2632},
   MRCLASS = {11B13 (11B25 11B75 11P70)},
  MRNUMBER = {2737725},
MRREVIEWER = {Radhakrishnan\ Nair},
       DOI = {10.1007/s10474-010-0026-6},
       URL = {https://doi.org/10.1007/s10474-010-0026-6},
}

@article{FischerW24,
  author       = {Nick Fischer and
                  Leo Wennmann},
  title        = {Minimizing Tardy Processing Time on a Single Machine in Near-Linear Time},
  journal      = {TheoretiCS},
  volume       = {4},
  pages        = {Article 14},
  year         = {2025},
  doi          = {10.46298/THEORETICS.25.14},
  bibsource    = {dblp computer science bibliography, https://dblp.org}
}

@article{pham2025sharp,
  title={Sharp Bound for the {Erd{\H{o}}s}--{Straus} Non-averaging Set Problem},
  author={Pham, Huy Tuan and Zakharov, Dmitrii},
  journal={Geometric and Functional Analysis},
  volume={35},
  number={6},
  pages={1712--1738},
  year={2025},
  doi={10.1007/s00039-025-00728-8},
  publisher={Springer}
}

@inproceedings{Jin24,
  author       = {Ce Jin},
  title        = {0-1 Knapsack in Nearly Quadratic Time},
  booktitle    = {Proceedings of the 56th Annual {ACM} Symposium on Theory of Computing,
                  {STOC} 2024, Vancouver, BC, Canada, June 24-28, 2024},
  pages        = {271--282},
  publisher    = {{ACM}},
  year         = {2024},
  doi          = {10.1145/3618260.3649618},
  bibsource    = {dblp computer science bibliography, https://dblp.org}
}

@article{clmznotes,
  author       = {Lin Chen and
                  Jiayi Lian and
                  Yuchen Mao and
                  Guochuan Zhang},
  title        = {A Note on Deterministic {FPTAS} for Partition},
  journal      = {CoRR},
  volume       = {abs/2501.12848},
  year         = {2025},
  eprinttype    = {arXiv},
  eprint       = {2501.12848},
  bibsource    = {dblp computer science bibliography, https://dblp.org}
}

@inproceedings{ChenLMZstoc24partition,
  author       = {Lin Chen and
                  Jiayi Lian and
                  Yuchen Mao and
                  Guochuan Zhang},
  title        = {Approximating Partition in Near-Linear Time},
  booktitle    = {Proceedings of the 56th Annual {ACM} Symposium on Theory of Computing,
                  {STOC} 2024, Vancouver, BC, Canada, June 24-28, 2024},
  pages        = {307--318},
  publisher    = {{ACM}},
  year         = {2024},
  doi          = {10.1145/3618260.3649727},
  bibsource    = {dblp computer science bibliography, https://dblp.org}
}

@phdthesis{Cassis_2024,
  title  = {Algorithms for Knapsacks, Paths and Strings},
  author = {Cassis, Alejandro},
  school = {Universit{\"a}t des Saarlandes},
  year   = {2024},
  doi    = {10.22028/D291-43984}
}

@inproceedings{chan26,
author = {Timothy M. Chan},
title = {Derandomizing Pseudopolynomial Algorithms for {Subset Sum}},
booktitle = {Proceedings of the 2026 Annual ACM-SIAM Symposium on Discrete Algorithms (SODA)},
year = {2026},
pages = {3600--3610},
doi = {10.1137/1.9781611978971.131},
}

@article{jinunpublishedknapsack,
  author       = {Ce Jin},
  title        = {Solving Knapsack with Small Items via {$\ell_0$}-Proximity},
  journal      = {CoRR},
  volume       = {abs/2307.09454},
  year         = {2023},
  eprinttype    = {arXiv},
  eprint       = {2307.09454},
  bibsource    = {dblp computer science bibliography, https://dblp.org}
}

@inproceedings{AxiotisBBJNTW21,
  author       = {Kyriakos Axiotis and
                  Arturs Backurs and
                  Karl Bringmann and
                  Ce Jin and
                  Vasileios Nakos and
                  Christos Tzamos and
                  Hongxun Wu},
  title        = {Fast and Simple Modular Subset Sum},
  booktitle    = {4th Symposium on Simplicity in Algorithms, {SOSA} 2021, Virtual Conference,
                  January 11-12, 2021},
  pages        = {57--67},
  publisher    = {{SIAM}},
  year         = {2021},
  doi          = {10.1137/1.9781611976496.6},
  bibsource    = {dblp computer science bibliography, https://dblp.org}
}

@inproceedings{CardinalI21,
  author       = {Jean Cardinal and
                  John Iacono},
  title        = {Modular Subset Sum, Dynamic Strings, and Zero-Sum Sets},
  booktitle    = {4th Symposium on Simplicity in Algorithms, {SOSA} 2021, Virtual Conference,
                  January 11-12, 2021},
  pages        = {45--56},
  publisher    = {{SIAM}},
  year         = {2021},
  doi          = {10.1137/1.9781611976496.5},
  bibsource    = {dblp computer science bibliography, https://dblp.org}
}

@inproceedings{Potepa21,
  author       = {Krzysztof Pot\k{e}pa},
  title        = {Faster Deterministic Modular Subset Sum},
  booktitle    = {29th Annual European Symposium on Algorithms, {ESA} 2021, Lisbon,
                  Portugal (Virtual Conference), September 6-8, 2021},
  series       = {LIPIcs},
  volume       = {204},
  pages        = {76:1--76:16},
  publisher    = {Schloss Dagstuhl - Leibniz-Zentrum f{\"{u}}r Informatik},
  year         = {2021},
  doi          = {10.4230/LIPIcs.ESA.2021.76},
  bibsource    = {dblp computer science bibliography, https://dblp.org}
}

@inproceedings{linchensoda26,
  author       = {Lin Chen and
                  Yuchen Mao and
                  Guochuan Zhang},
  title        = {Long Arithmetic Progressions in Sparse Subset Sums: {A} Computational
                  Perspective},
  booktitle    = {Proceedings of the 2026 Annual {ACM-SIAM} Symposium on Discrete Algorithms,
                  {SODA} 2026, Vancouver, BC, Canada, January 11-14, 2026},
  pages        = {3611--3638},
  publisher    = {{SIAM}},
  year         = {2026},
  doi          = {10.1137/1.9781611978971.132},
  bibsource    = {dblp computer science bibliography, https://dblp.org}
}

@inproceedings{linchenstoc25,
  author       = {Lin Chen and
                  Yuchen Mao and
                  Guochuan Zhang},
  title        = {Long Arithmetic Progressions in Sumsets and Subset Sums: Constructive
                  Proofs and Efficient Witnesses},
  booktitle    = {Proceedings of the 57th Annual {ACM} Symposium on Theory of Computing,
                  {STOC} 2025, Prague, Czechia, June 23-27, 2025},
  pages        = {2086--2097},
  publisher    = {{ACM}},
  year         = {2025},
  doi          = {10.1145/3717823.3718281},
  bibsource    = {dblp computer science bibliography, https://dblp.org}
}

@inproceedings{clmzknapsackstoc24,
  author       = {Lin Chen and
                  Jiayi Lian and
                  Yuchen Mao and
                  Guochuan Zhang},
  title        = {A Nearly Quadratic-Time {FPTAS} for Knapsack},
  booktitle    = {Proceedings of the 56th Annual {ACM} Symposium on Theory of Computing,
                  {STOC} 2024, Vancouver, BC, Canada, June 24-28, 2024},
  pages        = {283--294},
  publisher    = {{ACM}},
  year         = {2024},
  doi          = {10.1145/3618260.3649730},
  bibsource    = {dblp computer science bibliography, https://dblp.org}
}

@article{nederlofsurvey,
  title = {An invitation to ``{Fine-grained complexity of {NP}-complete problems}''},
journal = {Computer Science Review},
volume = {61},
pages = {100919},
year = {2026},
issn = {1574-0137},
doi = {10.1016/j.cosrev.2026.100919},
author = {Jesper Nederlof},
}

@inproceedings{BringmannC23,
  author       = {Karl Bringmann and
                  Alejandro Cassis},
  title        = {Faster 0-1-Knapsack via Near-Convex Min-Plus-Convolution},
  booktitle    = {31st Annual European Symposium on Algorithms, {ESA} 2023, Amsterdam,
                  The Netherlands, September 4-6, 2023},
  series       = {LIPIcs},
  volume       = {274},
  pages        = {24:1--24:16},
  publisher    = {Schloss Dagstuhl - Leibniz-Zentrum f{\"{u}}r Informatik},
  year         = {2023},
  doi          = {10.4230/LIPIcs.ESA.2023.24},
  bibsource    = {dblp computer science bibliography, https://dblp.org}
}

@article{LagariasO85,
  author       = {J. C. Lagarias and
                  Andrew M. Odlyzko},
  title        = {Solving Low-Density Subset Sum Problems},
  journal      = {J. {ACM}},
  volume       = {32},
  number       = {1},
  pages        = {229--246},
  year         = {1985},
  doi          = {10.1145/2455.2461},
  bibsource    = {dblp computer science bibliography, https://dblp.org}
}

@inproceedings{JouxW26,
  author       = {Antoine Joux and
                  Karol W\k{e}grzycki},
  title        = {Improving {Lagarias-Odlyzko} Algorithm for Average-Case Subset Sum:
                  Modular Arithmetic Approach},
  booktitle    = {43rd International Symposium on Theoretical Aspects of Computer Science,
                  {STACS} 2026, Grenoble, France, March 9-13, 2026},
  series       = {LIPIcs},
  volume       = {364},
  pages        = {57:1--57:19},
  publisher    = {Schloss Dagstuhl - Leibniz-Zentrum f{\"{u}}r Informatik},
  year         = {2026},
  doi          = {10.4230/LIPIcs.STACS.2026.57},
  bibsource    = {dblp computer science bibliography, https://dblp.org}
}

@article{linchensurvey,
author = {Lin Chen},
title = {Application of additive combinatorics in several classic combinatorial optimization problems},
publisher = {Operations Research Transactions},
year = {2025},
journal = {Operations Research Transactions},
volume = {29},
number = {3},
pages = {202--222},
note = {in Chinese},
url = {https://www.ort.shu.edu.cn/EN/abstract/abstract21520.shtml},
doi = {10.15960/j.cnki.issn.1007-6093.2025.03.010}
}

@article{alonfreiman,
    AUTHOR = {Alon, N. and Freiman, G.},
     TITLE = {On sums of subsets of a set of integers},
   JOURNAL = {Combinatorica},
  FJOURNAL = {Combinatorica. An International Journal of the J\'{a}nos
              Bolyai Mathematical Society},
    VOLUME = {8},
      YEAR = {1988},
    NUMBER = {4},
     PAGES = {297--306},
      ISSN = {0209-9683},
   MRCLASS = {11B13},
  MRNUMBER = {981887},
MRREVIEWER = {A.\ P.\ Street},
       DOI = {10.1007/BF02189086},
       URL = {https://doi.org/10.1007/BF02189086},
}

@article{Lenstra83,
  author       = {Lenstra, Jr., Hendrik W.},
  title        = {Integer Programming with a Fixed Number of Variables},
  journal      = {Mathematics of Operations Research},
  volume       = {8},
  number       = {4},
  pages        = {538--548},
  year         = {1983},
  doi          = {10.1287/moor.8.4.538},
  url          = {https://doi.org/10.1287/moor.8.4.538}
}

@incollection{Lenstra08Lattices,
  author       = {Lenstra, Jr., Hendrik W.},
  title        = {Lattices},
  booktitle    = {Algorithmic Number Theory: Lattices, Number Fields, Curves and Cryptography},
  editor       = {J. P. Buhler and P. Stevenhagen},
  series       = {Mathematical Sciences Research Institute Publications},
  volume       = {44},
  pages        = {127--182},
  publisher    = {Cambridge University Press},
  year         = {2008},
  doi          = {10.1017/9781139049801.007},
  url          = {https://library.slmath.org/books/Book44/files/06hwl.pdf}
}

@book{GLS88,
  author       = {Martin Gr{\"o}tschel and L{\'a}szl{\'o} Lov{\'a}sz and Alexander Schrijver},
  title        = {Geometric Algorithms and Combinatorial Optimization},
  series       = {Algorithms and Combinatorics},
  volume       = {2},
  publisher    = {Springer},
  year         = {1988},
  doi          = {10.1007/978-3-642-97881-4},
  url          = {https://doi.org/10.1007/978-3-642-97881-4}
}

@inproceedings{StorjohannLabahn96,
  author       = {Arne Storjohann and George Labahn},
  title        = {Asymptotically Fast Computation of {Hermite} Normal Forms of Integer Matrices},
  booktitle    = {Proceedings of the 1996 International Symposium on Symbolic and Algebraic Computation},
  pages        = {259--266},
  year         = {1996},
  publisher    = {ACM},
  doi          = {10.1145/236869.237083},
  url          = {https://cs.uwaterloo.ca/~astorjoh/p259-storjohann.pdf}
}

@article{macbeath,
  author  = {Macbeath, A. M.},
  title   = {A Theorem on Non-Homogeneous Lattices},
  journal = {Annals of Mathematics},
  series  = {Second Series},
  volume  = {56},
  number  = {2},
  pages   = {269--293},
  year    = {1952},
  doi     = {10.2307/1969800},
  url     = {https://www.jstor.org/stable/1969800}
}

@article{ewaldlarmanrogers1970,
  author  = {Ewald, G. and Larman, D. G. and Rogers, C. A.},
  title   = {The Directions of the Line Segments and of the
             {$r$}-Dimensional Balls on the Boundary of a Convex Body in
             {Euclidean} Space},
  journal = {Mathematika},
  volume  = {17},
  number  = {1},
  pages   = {1--20},
  year    = {1970},
  doi     = {10.1112/S0025579300002655}
}

@article{baranylarman1988,
  author  = {B{\'a}r{\'a}ny, Imre and Larman, D. G.},
  title   = {Convex Bodies, Economic Cap Coverings, Random Polytopes},
  journal = {Mathematika},
  volume  = {35},
  number  = {2},
  pages   = {274--291},
  year    = {1988},
  doi     = {10.1112/S0025579300015266}
}

@article{barany2000mregions,
  author  = {B{\'a}r{\'a}ny, Imre},
  title   = {The technique of {M}-regions and cap-coverings: a survey},
  journal = {Rend. Circ. Mat. Palermo (2), Suppl.},
  volume  = {65},
  pages   = {21--38},
  year    = {2000}
}

@article{aryafonsecamount2017,
  author  = {Arya, Sunil and da Fonseca, Guilherme D. and Mount, David M.},
  title   = {On the Combinatorial Complexity of Approximating Polytopes},
  journal = {Discrete \& Computational Geometry},
  volume  = {58},
  number  = {4},
  pages   = {849--870},
  year    = {2017},
  doi     = {10.1007/s00454-016-9856-5}
}

@article{AryaM25,
  author  = {Sunil Arya and David M. Mount},
  title   = {Optimal Volume-Sensitive Bounds for Polytope Approximation},
  journal = {Discrete \& Computational Geometry},
  volume  = {74},
  number  = {4},
  pages   = {839--871},
  year    = {2025},
  doi     = {10.1007/s00454-025-00780-z},
  url     = {https://doi.org/10.1007/s00454-025-00780-z}
}

@inproceedings{AbdelkaderM18,
  author    = {Ahmed Abdelkader and David M. Mount},
  title     = {Economical {Delone} Sets for Approximating Convex Bodies},
  booktitle = {16th Scandinavian Symposium and Workshops on Algorithm Theory (SWAT 2018)},
  series    = {Leibniz International Proceedings in Informatics},
  volume    = {101},
  pages     = {4:1--4:12},
  year      = {2018},
  publisher = {Schloss Dagstuhl -- Leibniz-Zentrum f{\"u}r Informatik},
  doi       = {10.4230/LIPIcs.SWAT.2018.4},
  url       = {https://doi.org/10.4230/LIPIcs.SWAT.2018.4}
}

@article{ChenHMZ26,
  author       = {Lin Chen and Tingwei Hu and Yuchen Mao and Guochuan Zhang},
  title        = {Dense Subset Sum in Multi-Dimension},
  journal      = {CoRR},
  volume       = {abs/2607.10343},
  year         = {2026},
  eprinttype   = {arXiv},
  eprint       = {2607.10343},
  note         = {To appear in {FOCS} 2026}
}

@article{submult,
  author       = {Katalin Gyarmati and M{\'a}t{\'e} Matolcsi and Imre Z. Ruzsa},
  title        = {A superadditivity and submultiplicativity property for cardinalities of sumsets},
  journal      = {Combinatorica},
  volume       = {30},
  number       = {2},
  pages        = {163--174},
  year         = {2010},
  doi          = {10.1007/s00493-010-2413-6},
  url          = {https://doi.org/10.1007/s00493-010-2413-6}
}

@misc{AlgUWOpenProblems25,
  key          = {AlgUW},
  title        = {Problem 2: Bounded Support Subset Sum},
  year         = {2025},
  author = {W{\k{e}}grzycki, Karol},
  howpublished = {{Open problem list for the AlgUW workshop in B\k{e}dlewo  28th Sept - 3rd Oct}},
  url          = {https://docs.google.com/document/d/1OjsF4TE0Tn0Xam8vFulpzSfNv_kcorxel1oQYHTFv3o},
  note         = {Version of September 30, 2025}
}

@article{AwerbuchKleinberg08,
  author       = {Baruch Awerbuch and Robert Kleinberg},
  title        = {Online Linear Optimization and Adaptive Routing},
  journal      = {Journal of Computer and System Sciences},
  volume       = {74},
  number       = {1},
  pages        = {97--114},
  year         = {2008},
  doi          = {10.1016/j.jcss.2007.04.016},
  url          = {https://doi.org/10.1016/j.jcss.2007.04.016}
}

@inproceedings{AryaFM23,
  author       = {Sunil Arya and Guilherme D. da Fonseca and David M. Mount},
  title        = {Economical Convex Coverings and Applications},
  booktitle    = {Proceedings of the 2023 Annual {ACM-SIAM} Symposium on Discrete Algorithms ({SODA})},
  pages        = {1834--1861},
  publisher    = {{SIAM}},
  year         = {2023},
  doi          = {10.1137/1.9781611977554.ch70},
  url          = {https://doi.org/10.1137/1.9781611977554.ch70}
}

@inproceedings{BringmannFN25,
  author       = {Karl Bringmann and
                  Nick Fischer and
                  Vasileios Nakos},
  editor       = {Yossi Azar and
                  Debmalya Panigrahi},
  title        = {Beating {Bellman}'s Algorithm for Subset Sum},
  booktitle    = {Proceedings of the 2025 Annual {ACM-SIAM} Symposium on Discrete Algorithms,
                  {SODA} 2025, New Orleans, LA, USA, January 12-15, 2025},
  pages        = {4596--4612},
  publisher    = {{SIAM}},
  year         = {2025},
  url          = {https://doi.org/10.1137/1.9781611978322.157},
  doi          = {10.1137/1.9781611978322.157},
  bibsource    = {dblp computer science bibliography, https://dblp.org}
}

@book{bookkpp,
  author       = {Hans Kellerer and
                  Ulrich Pferschy and
                  David Pisinger},
  title        = {Knapsack problems},
  publisher    = {Springer},
  year         = {2004},
  url          = {https://doi.org/10.1007/978-3-540-24777-7},
  doi          = {10.1007/978-3-540-24777-7},
  isbn         = {978-3-540-40286-2},
  bibsource    = {dblp computer science bibliography, https://dblp.org}
}

@article{ImpagliazzoPZ01,
  author       = {Russell Impagliazzo and
                  Ramamohan Paturi and
                  Francis Zane},
  title        = {Which Problems Have Strongly Exponential Complexity?},
  journal      = {J. Comput. Syst. Sci.},
  volume       = {63},
  number       = {4},
  pages        = {512--530},
  year         = {2001},
  doi          = {10.1006/JCSS.2001.1774},
  bibsource    = {dblp computer science bibliography, https://dblp.org}
}

@article{egyptian,
    AUTHOR = {Conlon, David and Fox, Jacob and He, Xiaoyu and Mubayi, Dhruv and Pham, Huy Tuan and Suk, Andrew and Verstra{\"{e}}te, Jacques},
     TITLE = {A question of {E}rd{\H{o}}s and {G}raham on {E}gyptian fractions},
   JOURNAL = {Discrete Anal.},
  FJOURNAL = {Discrete Analysis},
  MONTH = {December},
      YEAR = {2025},
     NUMBER = {28},
      ISSN = {2397-3129},
   MRCLASS = {11B13 (11D68)},
  MRNUMBER = {5005333},
MRREVIEWER = {E.\ J.\ Barbeau},
URL={https://discreteanalysisjournal.com/article/154329-a-question-of-erdos-and-graham-on-egyptian-fractions}
}

@article{taovu08john,
  author = {Terence Tao and Van Vu},
  title = {John-type theorems for generalized arithmetic progressions and iterated sumsets},
  journal = {Advances in Mathematics},
  volume = {219},
  number = {2},
  pages = {428--449},
  year = {2008},
  doi = {10.1016/j.aim.2008.05.002},
  eprint = {math/0701005},
  archivePrefix = {arXiv}
}

@misc{oai1,
  author = {{OpenAI}},
  title = {{Subset Sum in Time $O(2^{0.49n})$}},
  howpublished = {OpenAI Math Release preprint
                  \href{https://github.com/openai/math/blob/main/preprints/Subset-Sum-in-Time-2-power-0-49n-October-4-2026/subset-sum.pdf}{OAI:Subset-Sum-in-Time-2-power-0-49n-October-4-2026}},
  year = {2026}
}

@misc{oai2,
  author = {{OpenAI}},
  title = {{A Low-Space Algorithm for Worst-Case Subset Sum}},
  howpublished = {OpenAI Math Release preprint
                  \href{https://github.com/openai/math/blob/main/preprints/A-Low-Space-Algorithm-for-Worst-Case-Subset-Sum-September-26-2026/paper.pdf}{OAI:A-Low-Space-Algorithm-for-Worst-Case-Subset-Sum-September-26-2026}},
  year = {2026}
}

@article{JansenR23,
  author       = {Klaus Jansen and
                  Lars Rohwedder},
  title        = {On Integer Programming, Discrepancy, and Convolution},
  journal      = {Math. Oper. Res.},
  volume       = {48},
  number       = {3},
  pages        = {1481--1495},
  year         = {2023},
  doi          = {10.1287/MOOR.2022.1308},
  bibsource    = {dblp computer science bibliography, https://dblp.org}
}

@article{le2026deterministiclineartimemodularsubset,
      title={Deterministic Linear-Time Modular Subset Sum}, 
      author={Phuoc Dinh Le and Kha Le},
      journal={CoRR},
      volume={abs/2609.30992},
      year={2026},
      eprint={2609.30992},
      eprinttype={arXiv},
      primaryClass={cs.DS},
}

@article{vaughn2026simpleralgorithmsknapsacksubset,
      title={Simpler Algorithms for Knapsack, Subset Sum, and Min-Plus Convolution}, 
      author={Trevor Vaughn},
      journal={CoRR},
      volume={abs/2609.37449},
      year={2026},
      eprint={2609.37449},
      eprinttype={arXiv},
      primaryClass={cs.DS},
}

@article{alman2026trulysubquadratic3sumtruly,
      title={Truly Subquadratic {3SUM} and Truly Subcubic {APSP} via Triangles in Sparse Lopsided Graphs}, 
      author={Josh Alman and Virginia {Vassilevska Williams}},
      journal={CoRR},
      volume={abs/2610.06783},
      year={2026},
      eprint={2610.06783},
      eprinttype={arXiv},
      primaryClass={cs.DS},
}

@article{BronnimannCP93,
  author       = {Herv{\'{e}} Br{\"{o}}nnimann and
                  Bernard Chazelle and
                  J{\'{a}}nos Pach},
  title        = {How Hard Is Half-Space Range Searching},
  journal      = {Discret. Comput. Geom.},
  volume       = {10},
  pages        = {143--155},
  year         = {1993},
  url          = {https://doi.org/10.1007/BF02573971},
  doi          = {10.1007/BF02573971},
  bibsource    = {dblp computer science bibliography, https://dblp.org}
}

@article{ye2026subsetsumpartialmatch,
      title={Subset Sum via Partial Match}, 
      author={Lixi Ye and Baitian Li},
      journal={CoRR},
      volume={abs/2610.10446},
      year={2026},
      eprint={2610.10446},
      eprinttype={arXiv},
      primaryClass={cs.DS}
}

@article{chen2026subsetsumintegerprogramming,
      title={Subset Sum and Integer Programming Feasibility Test in Near-Optimal Time},
      author={Lin Chen and Yuchen Mao and Guochuan Zhang},
      journal={CoRR},
      volume={abs/2610.12953},
      year={2026},
      eprint={2610.12953},
      eprinttype={arXiv},
      primaryClass={cs.DS},
}

@article{BorgsCP01,
  author       = {Christian Borgs and
                  Jennifer T. Chayes and
                  Boris G. Pittel},
  title        = {Phase transition and finite-size scaling for the integer partitioning
                  problem},
  journal      = {Random Struct. Algorithms},
  volume       = {19},
  number       = {3-4},
  pages        = {247--288},
  year         = {2001},
  url          = {https://doi.org/10.1002/rsa.10004},
  doi          = {10.1002/RSA.10004},
  bibsource    = {dblp computer science bibliography, https://dblp.org}
}

@inproceedings{weakapproxicalp,
  author       = {Lin Chen and
                  Jiayi Lian and
                  Yuchen Mao and
                  Guochuan Zhang},
  !editor       = {Keren Censor{-}Hillel and
                  Fabrizio Grandoni and
                  Jo{\"{e}}l Ouaknine and
                  Gabriele Puppis},
  title        = {Weakly Approximating Knapsack in Subquadratic Time},
  booktitle    = {52nd International Colloquium on Automata, Languages, and Programming,
                  {ICALP} 2025, Aarhus, Denmark, July 8-11, 2025},
  series       = {LIPIcs},
  volume       = {334},
  pages        = {51:1--51:18},
  publisher    = {Schloss Dagstuhl - Leibniz-Zentrum f{\"{u}}r Informatik},
  year         = {2025},
  url          = {https://doi.org/10.4230/LIPIcs.ICALP.2025.51},
  doi          = {10.4230/LIPICS.ICALP.2025.51},
  timestamp    = {Fri, 21 Nov 2025 23:44:11 +0100},
  biburl       = {https://dblp.org/rec/conf/icalp/0009L0Z25.bib},
  bibsource    = {dblp computer science bibliography, https://dblp.org}
}
\end{document}